%% file: ShortQEC12clean.tex
\documentclass[onecolumn,prxq,nofootinbib,superscriptaddress]{revtex4-2}

\usepackage[dvips]{graphicx} % for figures
\usepackage{amsfonts}
\usepackage{amssymb}
\usepackage{amscd}
\usepackage{amsmath}    % need for subequations
\usepackage{amsthm}
\usepackage{epsfig}
\usepackage{subfigure}
\usepackage{subfloat}
\usepackage{xcolor}
\usepackage{physics}
\usepackage[most]{tcolorbox}
\usepackage{tikz}
\usetikzlibrary{quantikz}
\usepackage{tabularx}
\usepackage{float}
\usepackage{appendix}
\usepackage{makecell}
\usepackage{algorithm}
\usepackage{algorithmicx}
\usepackage{algpseudocode}
\usepackage{dsfont}
\usepackage[colorlinks, linkcolor=blue, anchorcolor=blue, citecolor=blue]{hyperref}
\usepackage{MnSymbol}
\newtcolorbox[auto counter]{mybox}[2]{
enhanced,
breakable,
label=#1,
colback=blue!5!white,
colframe=blue!75!black,
fonttitle=\bfseries,
title=Box \thetcbcounter: #2
}

\newtheorem{theorem}{Theorem}
\newtheorem{lemma}{Lemma}
\newtheorem{corollary}{Corollary}

\newtheorem{definition}{Definition}
\newtheorem{proposition}{Proposition}

\newtheorem{result}{Result}

\definecolor{ao}{rgb}{0.0, 0.5, 0.0}

\newcommand{\id}{\mathbb{I}}

\newcommand{\weight}{\mathrm{weight}}

\newcommand{\bE}{\mathbb{E}}
\newcommand{\bI}{\mathbb{I}}
\newcommand{\cD}{\mathcal{D}}
\newcommand{\cE}{\mathcal{E}}
\newcommand{\cN}{\mathcal{N}}
\newcommand{\cU}{\mathcal{U}}
\newcommand{\poly}{\mathrm{poly}}

\begin{document}
%\preprint{APS/123-QED}
%\title{One-dimensional logarithmic-depth encoding circuits make good codes}
\title{Approximate Quantum Error Correction with 1D Log-Depth Circuits}
%\date{}
\author{Guoding Liu}\thanks{These authors contributed equally to this work.}
\affiliation{Center for Quantum Information, Institute for Interdisciplinary Information Sciences, Tsinghua University, Beijing, 100084 China}
\author{Zhenyu Du}\thanks{These authors contributed equally to this work.}
\affiliation{Center for Quantum Information, Institute for Interdisciplinary Information Sciences, Tsinghua University, Beijing, 100084 China}
\author{Zi-Wen Liu}
\email{zwliu0@tsinghua.edu.cn}
\affiliation{Yau Mathematical Sciences Center, Tsinghua University, Beijing 100084, China}
\author{Xiongfeng Ma}
\email{xma@tsinghua.edu.cn}
\affiliation{Center for Quantum Information, Institute for Interdisciplinary Information Sciences, Tsinghua University, Beijing, 100084 China}

\begin{abstract}
Efficient and high-performance quantum error correction is essential for achieving fault-tolerant quantum computing. Low-depth random circuits offer a promising approach to identifying effective and practical encoding strategies. In this work, we rigorously prove through information-theoretic analysis that one-dimensional logarithmic-depth random Clifford encoding circuits can achieve high quantum error correction performance. We demonstrate that these random codes typically exhibit good approximate quantum error correction capability by proving that their encoding rate achieves the hashing bound for Pauli noise and the channel capacity for erasure errors. We show that the error correction inaccuracy decays once a threshold of logarithmic depth is exceeded, resulting in negligible recovery errors. This threshold is shown to be lower than that of the simple separate block encoding, and the decay rate is higher. We further establish that these codes are optimal by proving that logarithmic depth is necessary to maintain a constant encoding rate and high error correction performance. To prove our results, we propose decoupling theorems tailored for one-dimensional low-depth circuits. These results also imply strong decoupling and rapid thermalization properties in low-depth random circuits and have potential applications in quantum information science and physics.
\end{abstract}

\maketitle

\section{Introduction}
Quantum error correction \cite{Egan2021corrected, Gong2021correcting, Postler2022tolerant, Leverrier2022QuantumTannerCodes, Zhao2022Correcting, Acharya2023Suppressing, Acharya2024QuantumEC} is crucial for protecting quantum computers from noise, so as to enable universal fault-tolerant quantum computing~\cite{Aharonov1997FTQC,Aharonov2008FTQC,gottesman2009introductionquantumerrorcorrection,Gottesman2014constant,Xu2024BosonicFT,Nelson2025FTlowdepth}, and ultimately realize scalable quantum speedups. Significant experimental and theoretical progress has been made to reduce resource requirements and enhance the performance of quantum error correction. In experiments, surface codes on a two-dimensional system have shown remarkable success in preserving logical information over extended timescales~\cite{Acharya2023Suppressing, Acharya2024QuantumEC}. Nonetheless, the surface code has a fundamental drawback of a low encoding rate. The overhead of physical qubits to make one logical qubit is large for near-term devices~\cite{Fowler2012Surface}.

Recent research on quantum low-density parity-check (qLDPC) codes \cite{Dennis2002Topological, Fowler2012Surface, Gottesman2014constant, Kovalev2013sublinear, Fawzi2020Constant, Panteleev2022goodQLDPC} establishes constant-rate encoding with robust error-correcting capability. Nonetheless, due to the requirement of substantial long-range connectivity, the realization of high-performance qLDPC codes needs all-to-all circuit architectures or prohibitively large circuit depth in low-dimensional systems, which makes practical implementation challenging. Finding encoding schemes with low overhead and minimal device requirements is crucial for achieving practical quantum error correction in near-term quantum devices. This leads to the central question that underlies this work: What circuit architecture and depth can achieve good quantum error correction performance? Ideally, we desire simple circuit geometry and low-depth implementation. Understanding this would be crucial for both the theory and experimental implementation of quantum error correction.

Random quantum circuits offer a promising approach to evaluating the typical performance of various encoding circuits and identifying practical encoding schemes~\cite{gottesman1997stabilizer, Brown2013short, Gullans2021LowDepth,Kong2022CQEC,Darmawan2024Lowdepth,Nelson2025FTlowdepth}.
Here, we conduct the analysis using the approximate quantum error correction (AQEC)~\cite{Leung1997AQEC, Claude2005AQEC, Klesse2007AQEC,Beny2010AQEC} framework, which enables rigorous quantitative analysis for the average behaviors of random circuit ensembles and facilitates analyzing the logical error and channel capacity of generic noise models~\cite{Klesse2007AQEC,wilde2013quantum}. Under suitable conditions, AQEC ensures negligible recovery errors between the initial logical state and the state after decoding, which guarantees feasibility for practical applications and allows more flexible design of codes with various advantages~\cite{Leung1997AQEC,crepeau2005approximatequantumerrorcorrectingcodes}. Moreover, it provides insights into the role of symmetry in quantum information~\cite{Kong2022CQEC} and deepens our understanding of quantum complexity and order in many-body quantum systems~\cite{Yi2024order}.

To address the motivating question, we investigate the AQEC behaviors of random circuit encodings with locality constrained by an one-dimensional (1D) geometry. Previous studies numerically found the 1D $O(\log n)$-depth local random circuits with a brickwork structure can approximately correct a linear number of random erasure errors while preserving a constant encoding rate~\cite{Gullans2021LowDepth}. Further numerical results show that such circuits can approach the hashing bound for local Pauli noise~\cite{Darmawan2024Lowdepth}, matching the performance of random stabilizer codes~\cite{wilde2013quantum}. A simple block-encoding argument provides intuition for the robustness against local noise: the $n$-qubit system is divided into $O(n/\log n)$ blocks, each encoding $\Omega(\log n)$ qubits as a random stabilizer code. This structure corrects noise within individual blocks, making it effective against local noise.

Despite certain numerical evidence and intuition, a rigorous understanding of the encoding properties under general noises has been elusive. Intuitively, 1D random circuits with interconnections between blocks can potentially further enhance the error correction capability due to a stronger scrambling power. A comprehensive analysis would be important, but it poses significant mathematical challenges.

In this work, we present the construction of random codes implemented in 1D and logarithmic-depth circuits, along with information-theoretically rigorous proofs of their powerful AQEC performance.
We specifically analyze two types of 1D random circuit architectures: the standard 1D local random circuit with a brickwork structure~\cite{Brandao2016ApproximateDesign,Dalzell2022Anticoncentrate} and the double-layer blocked circuit composed of two layers of blocked gates~\cite{schuster2024randomunitariesextremelylow,laracuente2024approximateunitarykdesignsshallow}. We will refer to these as the 1D brickwork circuit and the double-layer blocked circuit, respectively. Compared to the separate block model discussed before, there exist connections between blocks in these two models.
There have been extensive prior studies on the approximate unitary design generation properties of such circuits~\cite{Brandao2016ApproximateDesign,Dalzell2022Anticoncentrate,schuster2024randomunitariesextremelylow,laracuente2024approximateunitarykdesignsshallow}. Here we aim to examine their application as encoding schemes, which is morally relevant to design generation but essentially a different problem (as will be further explained later). For example, the 1D brickwork circuit is known to form an approximate 2-design at $O(n)$ depth, and the double-layer blocked circuit achieves this at $O(\log n)$ depth. However, this finding alone does not guarantee a good encoding scheme at the same depth. To protect logical information with a constant encoding rate, the size of the logical space is exponential in $n$, which requires an exponentially small error in the approximate design. Since the approximate 2-design error at $O(\log n)$ depth is only polynomially small~\cite{schuster2024randomunitariesextremelylow,laracuente2024approximateunitarykdesignsshallow}, our proofs for the performance of these encoding circuits require new technical work. By demonstrating their effectiveness in decoupling and error correction tasks, our research complements the common study from a design perspective~\cite{schuster2024randomunitariesextremelylow,laracuente2024approximateunitarykdesignsshallow}.

We show that at a logarithmic depth, two encoding schemes can tolerate constant-rate stochastic noises, including local Pauli, erasure, amplitude-damping noises, and weak-correlated noises, and maintain a constant encoding rate. For the double-layer blocked circuit, we prove that they can tolerate a higher encoding rate for local noises at $\omega(\log n)$ depth. For instance, they tolerate strength-$\Vec{p}$ local Pauli noise with an encoding rate $k/n$, provided that $1-k/n-h(\Vec{p})> 0$. This achieves the well-known hashing bound of the channel capacity for local Pauli noise \cite{DiVincenzo1998channelcapacity, wilde2013quantum, Ataides2021XZZX}. For independently and identically distributed (i.i.d.)~erasure errors with error rate $p$, the encoding rate can achieve the quantum channel capacity $1-2p$.

Beyond the achievable encoding rate, we show that the recovery error of the two encoding schemes decays exponentially with the depth of encoding circuits.
The proof is based on transforming the recovery error after noise and decoding into a decoupling problem, and we establish decoupling theorems for two types of 1D low-depth circuits. Particularly, we show that the error polynomially decays for $O(\log n)$-depth circuits at a constant encoding rate and vanishes for $\omega(\log n)$-depth circuits at a higher encoding rate. Compared to a previous result that $O(\log^3 n)$-depth all-to-all Clifford circuits establish decoupling against a linear-size adversarial error~\cite{brown2013scramblingspeedrandomquantum,Brown2013short}, our work focuses on decoupling against constant-rate stochastic noise. This relaxation is more practically relevant, and we show that 1D $O(\log n)$-depth circuits generally suffice for decoupling in this case. Moreover, we compare the performance between the double-layer blocked scheme and the intuitive separate block approach. We find that for certain noise and encoding rate regimes, the double-layer blocked scheme requires a lower depth to make the recovery error decay. Also, the decay rate associated with the double-layer blocked scheme is better than the separate block approach. This advantage also holds for 1D brickwork circuits.

Furthermore, we establish lower bounds on the circuit depth required to achieve high AQEC performance. Specifically, for a $D$-dimensional circuit, we prove that the recovery error decays polynomially or that a constant encoding rate is achieved, the circuit depth must be at least $\Omega\left( (\log n)^{1/D} \right)$. Therefore, our encoding schemes are optimal since they attain the 1D lower bound $\Omega(\log n)$. This lower bound also generalizes existing results for the circuit depth required to achieve a high code distance or for other quantum error correction and fault-tolerant quantum computing metrics~\cite{baspin2023lowerboundoverheadquantum,Yi2024order}.

Beyond the crucial role in quantum computing and quantum error correction, our results on low-depth random circuits are expected to find applications in quantum many-body physics. The decoupling property and the robustness of random circuits against errors relate to their capacity to scramble information and generate entanglement~\cite{Brown2013short,brown2013scramblingspeedrandomquantum,Choi2020QECMIPT}. Haar-random unitary and random Clifford operations exhibit strong scrambling power and naturally have good quantum error correction performance. For instance, random stabilizer codes \cite{gottesman1997stabilizer} achieve the well-known Singleton bound, $n-k-2(d-1) \geq 0$~\cite{Cerf1997singleton}, with a constant encoding rate $k/n$ and a linear code distance $d = \Theta(n)$. However, Haar-random operations require circuit depth exponential in the number of qubits for implementation. Emulating the utility of Haar random unitaries using low-dimensional, low-depth random circuits is essential for many practical applications like random circuit sampling~\cite{arute2019supremacy,Bouland2019sampling,YulinWu2021Superconducting,Movassagh2023hardness,Morvan2024phase,Gao2025Advantage}, quantum device benchmarking~\cite{Elben2023toolbox,huang2020shadow,chen2021robust,liu2024auxiliaryfreereplicashadowestimation,Knill2008RB,Emerson2011prlRB,Liu2024multiqubit}, and quantum gravity~\cite{Patrick2007Blackhole,Almheiri2015QEC,Pastawski2015Holographic}. Our results also imply that subsystems of 1D low-depth random circuit states rapidly thermalize, providing a powerful tool for investigating quantum many-body dynamics~\cite{Nahum2017Random,Nahum2018randomness,Zhou2019random}.

The paper is organized as follows: In Section \ref{sec:summary}, we introduce the low-depth random circuit models and summarize our main results. We also outline the theoretical approach to analyze the AQEC performance. In Section~\ref{sc:pre}, we introduce the preliminaries for this work. In Sections~\ref{sc:upper} and~\ref{sec:lower_bound}, we elaborate on the main results about the AQEC performance of random codes, as well as the lower bounds for the circuit depth required for AQEC. Finally, in Section~\ref{sc:discuss}, we conclude with a discussion of our results and implications for future research.

\section{Circuit models and summary of results}\label{sec:summary}
In this work, we consider randomized code constructions, where the encoding operation of a code is randomly drawn from a circuit ensemble. As mentioned earlier, we mainly analyze two types of 1D random circuit ensembles and prove their decoupling properties and AQEC performance at a logarithmic depth.

The 1D brickwork circuit represents a relatively standard circuit architecture.  We consider the case with the periodic boundary condition, and let $q$ be the local dimension. The open boundary condition is similar. Without loss of generality, here we consider the total number of qudits $n$ to be even. For an $n$-qudit random unitary operation $U$ composed of $s$ two-qudit gates, $\{U^{(1)}, U^{(2)}, \cdots, U^{(s)}\}$, with $U^{(i)}$ randomly sampled from the unitary group of dimension $q^2$ with support $A^{(i)}\subset [n] = \{1, 2, \cdots, n\}$, $U$ has an expression,
\begin{equation}\label{eq:randomtwounitary}
U = U^{(s)}_{A^{(s)}}U^{(s-1)}_{A^{(s-1)}}\cdots U^{(1)}_{A^{(1)}}.
\end{equation}
For the task of decoupling, each gate $U^{(i)}$ can be sampled from a unitary 2-design group, like the Clifford group, instead of the whole unitary group.

The 1D depth-$d$ brickwork circuit is composed of alternating layers of two-qudit gates in a staggered arrangement, with the formal definition shown below.
\begin{equation}\label{eq:1DLQC}
U = \prod_{l=1}^{d} U^{[l]},
\end{equation}
where
\begin{equation}
U^{[l]} =
\begin{cases}
\begin{aligned}
\bigotimes_{i=1}^{n/2} U^{((l-1)n/2+i)}_{\{2i-1,2i\}},
\end{aligned} \quad &\text{if } \mod(l,2)=1, \\
\begin{aligned}
\bigotimes_{i=1}^{n/2} U^{((l-1)n/2+i)}_{\{2i,2i+1\}},
\end{aligned} \quad &\text{if } \mod(l,2)=0.
\end{cases}
\end{equation}
Here, $U^{[l]}$ is the unitary gate on layer $l$. The qudit index $n+1$ is identified with $1$ to match the periodic boundary condition. In our work, we set depth $d=\log\frac{n}{\varepsilon}+\frac{\log n}{\log \frac{q^2+1}{2q}} + \frac{\log(e-1)}{\log \frac{q^2+1}{2q}}+1$. The local dimension $q$ is a prefixed constant, and the depth is specified by the parameter $\varepsilon$. We denote this 1D brickwork circuit ensemble by $\mathfrak{B}_n^{\varepsilon}$. When the number of qubits and circuit depth are not of primary concern, we simply refer to this architecture as $\mathfrak{B}$. A diagram of a 1D brickwork circuit is shown in Fig.~\ref{fig:1Dlocalrandomcircuit}.
This brickwork structure is particularly well-suited for experimental implementation, as it can be realized by applying gates between nearest-neighbor qudits without complex gate compilation. Furthermore, this circuit is an excellent candidate for generating randomness and entanglement. After a depth of $d$, the information initially contained in any single qudit is scrambled across $2d$ qudits.

In this work, we will use circuit $U$ as the encoding map which encodes $2^k$-dimension logical space into $2^n$ physical space. To use this circuit for encoding, we set the local dimension of each qudit to $q = 2^{\frac{n}{k}}$, where a single qubit of information is encoded per qudit. By selecting $n/k$ as a constant, the local dimension $q$ also remains constant. This configuration yields a constant encoding rate of $k/n$. The total number of qudits is correspondingly changed to $k$, which is equivalent to $n$ qubits.

\begin{figure}[htbp!]
\raggedright
% \raggedleft
\begin{minipage}[b]{0.98\linewidth}
\subfigure[]{
\label{fig:1Dlocalrandomcircuit}
\includegraphics[width=5cm]{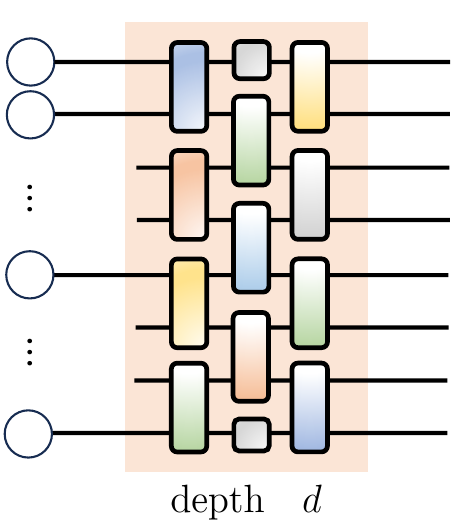}}
\subfigure[]{\label{fig:1Dlowdepthcircuit}
\includegraphics[width=10cm]{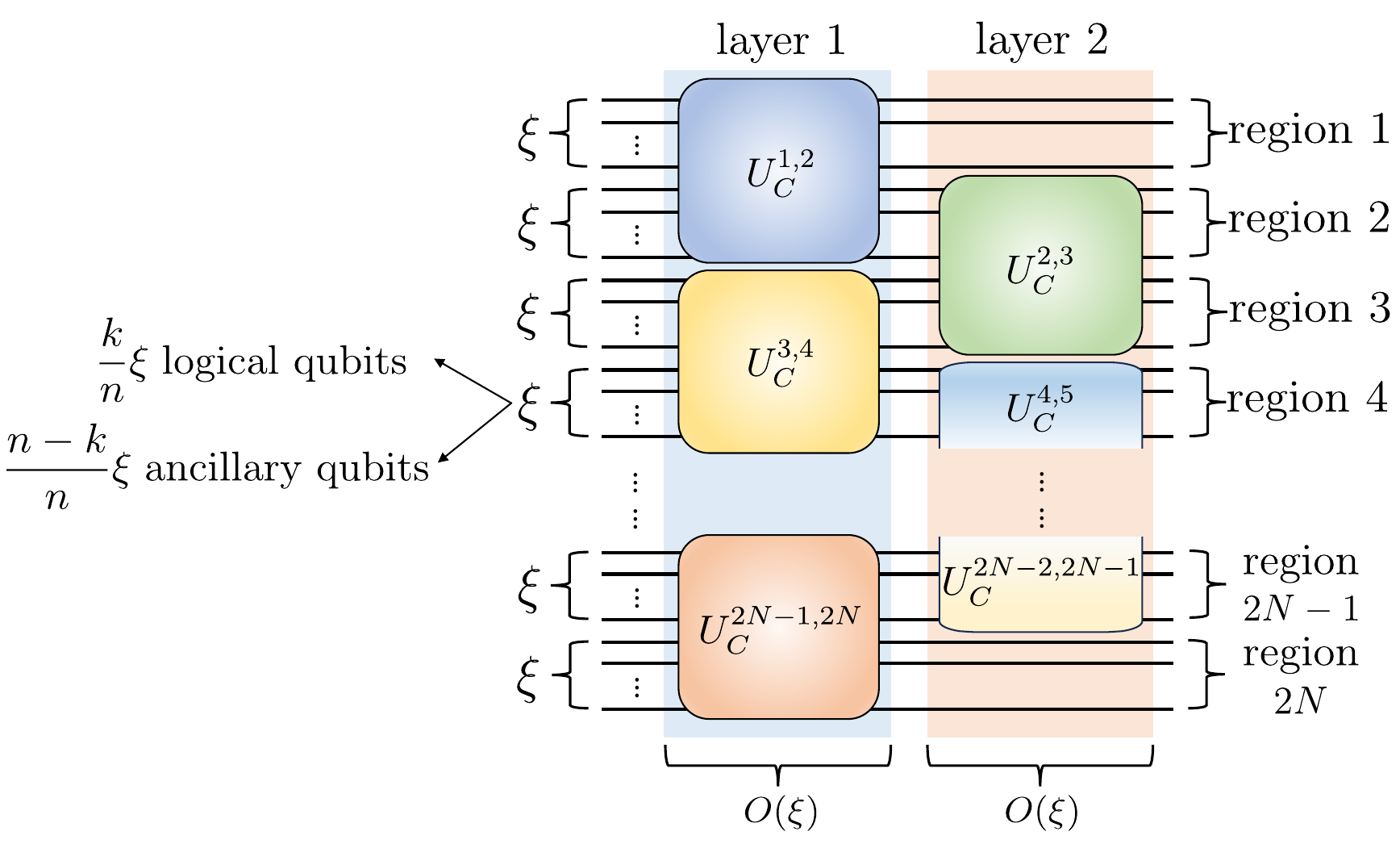}
}
\end{minipage}
\caption{(a) Diagram of the 1D brickwork circuit. The circuit is constructed by brickwise two-qudit gates. The total qudit number is $k$, and the circuit depth is $d$. When using this circuit for encoding, we set the local dimension as $q = 2^{\frac{n}{k}}$ while only one qubit of information is encoded in one qudit. In our work, $\frac{n}{k}$ is chosen as a constant, so $q$ is a constant. The encoding rate maintains $\frac{k}{n}$ in this scheme, and the total physical qubit number is $n$. Note that our results apply to arbitrary local dimensions, including 1D brickwork circuits composed of two-qubit gates acting on $n$ qubits, where each qubit encodes logical information with probability $\frac{k}{n}$. The Choi error still decays polynomially, though with a different scaling. Further discussions are provided in Appendix~\ref{appendssc:1DLRCQEC}. (b) Diagram of the 1D double-layer blocked circuit. Two layers of blocked unitary gates construct the circuit. The term $U_C^{i,i+1}$ is a random gate from a unitary 2-design ensemble on regions $i$ and $i+1$. The circuit is divided into $2N$ regions, with each region having $\xi = n/2N$ qubits. When using this circuit for encoding, the logical qubits and ancillary qubits are evenly distributed across each region. Each region contains $\frac{k}{n}\xi$ logical qubits and $\frac{n-k}{n}\xi$ ancillary qubits.}
\end{figure}

The diagram of the double-layer blocked circuit is shown in Fig.~\ref{fig:1Dlowdepthcircuit}. This circuit has a two-layer structure, with each layer consisting of blocks of random gates sampled from a unitary 2-design group, where each block has size $\xi$ and depth $d = O(\xi)$~\cite{schuster2024randomunitariesextremelylow}. In the first layer, qubits are partitioned into blocks, and random gates are applied within each block to generate entanglement. The second layer then applies random gates between these blocks, thereby creating inter-block connections. Similar to the 1D brickwork circuit, this double-layer blocked circuit scrambles the information of a single qubit across a number of qubits that is proportional to the circuit's depth. Theoretically, the double-layer blocked structure is simpler than the 1D brickwork circuit, which facilitates further theoretical analysis.

Specifically, in the double-layer blocked circuit, the $n$ qubits are divided into $2N$ regions, with each region containing $\xi = n / 2N$ qubits. In the first layer, unitaries $U_C^{2i-1, 2i}$ are applied in parallel on regions $2i-1$ and $2i$ for $1 \leq i \leq N$. In the second layer, unitaries $U_C^{2i, 2i+1}$ are applied in parallel on regions $2i$ and $2i+1$ for $1 \leq i < N$. The entire encoding unitary circuit can be written as
\begin{equation}
U = U_2\cdot U_1,
\end{equation}
where
\begin{equation}
U_1 = \bigotimes_{i=1}^N U_{C}^{2i-1, 2i}, U_2 = \bigotimes_{i=1}^{N-1} U_{C}^{2i, 2i+1}.
\end{equation}
The unitaries $U_C^{i,i+1}$ are sampled from a unitary 2-design ensemble. In this work, we take it to be the Clifford group. Since a $\xi$-qubit Clifford gate can be realized in depth $O(\xi)$ with basic CNOT, Hadamard, and phase gates in 1D circuits~\cite{Maslov2018Stabilizer,Bravyi2021Clifford}, the depth of the random circuit $U$ is $O(\xi)$. In this work, we set $\xi = \log(n/\varepsilon)$ and denote the ensemble generated by this 1D random double-layer blocked circuit as $\mathfrak{C}_n^{\varepsilon}$, where $\varepsilon$ is a parameter. We would refer to this circuit architecture as $\mathfrak{C}$ when parameters are not essential.

We use the double-layer blocked circuit as the encoding map in the following way. The circuit $U$ encodes the $k$-qubit logical system $L$ to an $n$-qubit physical system $S$. In the encoding procedure, the logical qubits and ancillary qubits are evenly distributed across each region, such that each region contains $\frac{k}{n}\xi$ logical qubits and $\frac{n-k}{n}\xi$ ancillary qubits as shown in Fig.~\ref{fig:1Dlowdepthcircuit}.

We focus on the performance of codes when the encoding circuits are of low depth, that is, when depth $d \ll n$. Due to the light-cone constraint, each qubit can interact with only $O(d)$ other qubits, limiting the spread of information across the system. Consequently, any encoding unitary from the 1D random circuits can only achieve a code distance on the order of circuit depth~\cite{Yi2024order}. Fortunately, by considering AQEC, the number of qubits protected against local noise and limited correlated noise can be much larger than the code distance. Here, the AQEC performance is quantified by the Choi error $\epsilon_{\mathrm{Choi}}$ between the identity operation and $\mathcal{D} \circ \mathcal{N} \circ \mathcal{E}$, where $\mathcal{E}$, $\mathcal{N}$, and $\mathcal{D}$ represent the encoding, noise, and decoding channels, respectively. This error measures the average distance between the initial and recovered states over all inputs, using an optimal decoding channel $\mathcal{D}$. A detailed introduction to AQEC is provided in Section~\ref{subsec:aqec}.

We rigorously analyze a host of important noise models, including local Pauli, erasure, amplitude-damping noises, and nearest neighbor $ZZ$-coupling noise. The common message is that our two 1D circuit models exhibit good quantum error correction performance at $O(\log n)$ depth for all these practical noise types. We also show that a better performance exists for the double-layer blocked circuit at $\omega(\log n)$ depth. In the following, we mainly present the results for this type of circuit.

Remarkably, we prove that the encoding schemes from the double-layer blocked circuit at $\omega(\log n)$ depth can establish a Choi error that goes to zero with an encoding rate $k/n$ under strength-$\Vec{p}$ local Pauli noise, provided that $1-k/n-h(\Vec{p})> 0$. This error decay implies that these encoding schemes can achieve the well-known hashing bound of the channel capacity for local Pauli noise\cite{DiVincenzo1998channelcapacity, wilde2013quantum, Ataides2021XZZX}. Here, $\Vec{p}=(p_I, p_X, p_Y, p_Z)$ represents the four parameters of the local Pauli noise where $p_I+p_X+p_Y+p_Z=1$. The entropy function is defined as $h(\Vec{p}) = -p_I\log p_I-p_X\log p_X-p_Y\log p_Y-p_Z\log p_Z$. When depth is $O(\log n)$, the encoding rate can achieve a smaller one, $1-f(\Vec{p})$, against the same strength of Pauli noise where $f(\Vec{p})=2\log(\sqrt{p_I}+\sqrt{p_X}+\sqrt{p_Y}+\sqrt{p_Z})$.
For strength-$p$ local erasure error, where each qubit is i.i.d.~erased with probability $p$, the encoding rate can also be constant and achieve the quantum channel capacity $1-2p$ \cite{Bennett1997Erasure}. For strength-$p$ local amplitude damping noise, the encoding rate can reach $h(\frac{1-p}{2})-h(\frac{p}{2})$. For nearest neighbor $ZZ$-coupling noise, we show that the encoding rate can achieve a constant $1-2\log(\sqrt{1-p}+\sqrt{p})$. The results for 1D brickwork circuits are the same when circuit depth is $O(\log n)$. The results are summarized in Result~\ref{result:EncodingRateBound} and depicted in Fig.~\ref{fig:EncodingRateBound}. We present the formal results in Corollaries~\ref{coro:pauli_nonsmoothing},~\ref{coro:pauli},~\ref{coro:iiderasure_nonsmoothing},~\ref{coro:iiderasure},~\ref{coro:amp_nonsmoothing},~\ref{coro:amp},~\ref{coro:zzcouple}, and Appendix~\ref{appendssc:1DLRCQEC}.

\begin{result}[Encoding rate of the double-layer blocked circuit $\mathfrak{C}$ and 1D brickwork circuit $\mathfrak{B}$, informal]\label{result:EncodingRateBound}
\begin{enumerate}
\item In the large $n$ limit, the expected Choi error for the random codes from $\mathfrak{C}$ and $\mathfrak{B}$ vanishes with encoding rate $k/n < 1-f(\vec{p})$ for strength-$\vec{p}$ local Pauli noise, $k/n < 1-\log(1+3p)$ for strength-$p$ local erasure errors, $k/n < -\log(\frac{1}{2-p}+\sqrt{\frac{p}{2-p}})$ for strength-$p$ local amplitude damping noise, and $k/n < 1-2c\log(\sqrt{1-p}+\sqrt{p})$ for strength-$p$ correlated noise at circuit depth $O(\log n)$. Here, $c=1$ for $\mathfrak{C}$ and $c = (1+\frac{k}{n})$ for $\mathfrak{B}$.
\item In the large $n$ limit, the expected Choi error for the random codes from $\mathfrak{C}$ vanishes with encoding rate $k/n < 1-h(\vec{p})$ for strength-$\vec{p}$ local Pauli noise, $k/n < 1-2p$ for strength-$p$ local erasure errors, and $k/n < h(\frac{1-p}{2})-h(\frac{p}{2})$ for strength-$p$ local amplitude damping noise at circuit depth $\omega(\log n)$.
\end{enumerate}

\end{result}

\begin{figure}[htbp!]
\centering
\includegraphics[width=.9\textwidth]{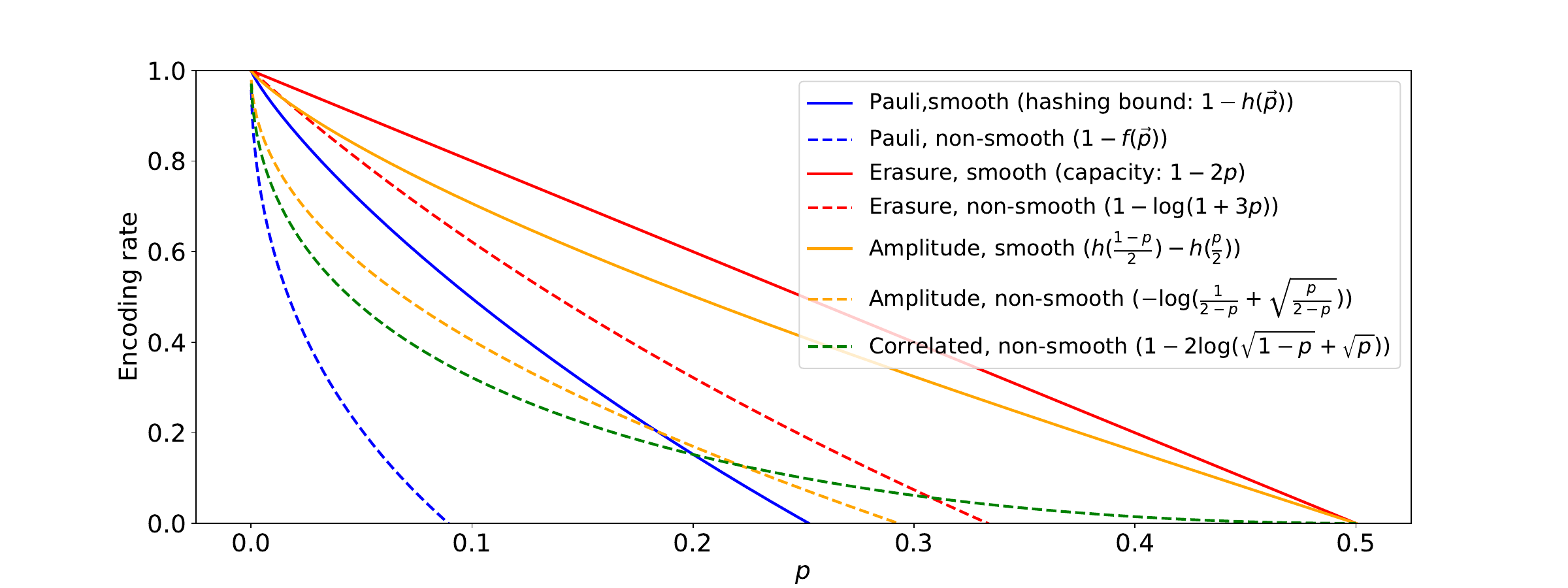}
\caption{The encoding rates of different bounds. For Pauli noise, the encoding rates given by the smooth and non-smooth decoupling theorems are bounded by $1-h(\Vec{p})$ and $1-f(\Vec{p})$, respectively. They are depicted in blue with solid and dashed lines, respectively. For the strength-$p$ depolarizing noise, the four parameters of the Pauli noise are $\vec{p} = (1-3p/4,p/4,p/4,p/4)$. The hashing bound $1-h(\Vec{p})$ is higher than the bound $1-f(\Vec{p})$. For erasure errors, the encoding rates given by the smooth and non-smooth decoupling theorems are bounded by $1-2p$ and $1-\log(1+3p)$ and depicted in red with solid and dashed lines, respectively. The bound $1-2p$ is the quantum channel capacity of the erasure error, and it is higher than the bound $1-\log(1+3p)$. For amplitude damping noise, the encoding rates given by the smooth and non-smooth decoupling theorems are bounded by $h(\frac{1-p}{2})-h(\frac{p}{2})$ and $-\log(\frac{1}{2-p}+\sqrt{\frac{p}{2-p}})$ and depicted in orange with solid and dashed lines, respectively. For strength-$p$ nearest-neighbor correlated noise, only the non-smooth decoupling theorem is applicable. The achievable encoding rate is $1-2\log(\sqrt{1-p}+\sqrt{p})$, depicted in green.}
\label{fig:EncodingRateBound}
\end{figure}

We further strengthen the result by showing that these codes exhibit a Choi error that decays exponentially with circuit depth.
\begin{result}[Choi error of encoding schemes from the double-layer blocked circuit $\mathfrak{C}$ and 1D brickwork circuit $\mathfrak{B}$, informal]\label{result:Choi_error}
The expected Choi error of the random codes formed by $\mathfrak{C}_n^{\varepsilon}$ or $\mathfrak{B}_n^{\varepsilon}$ with depth $O(\log n/\varepsilon)$ satisfies
\begin{equation}
\begin{split}
\mathbb{E}_{U\sim \mathfrak{C}_n^{\varepsilon}}\epsilon_{\mathrm{Choi}} &= O\left((\varepsilon^{1-\frac{k}{n}}n^{\frac{k}{n}})^{\frac{1}{4}}\right);\\
\mathbb{E}_{U\sim \mathfrak{B}_n^{\varepsilon}}\epsilon_{\mathrm{Choi}} &=O((\frac{\varepsilon}{k})^{\frac{n}{4k}-\frac{1}{2}}),
\end{split}
\end{equation}
with an encoding rate $k/n < 1 - f(\vec{p})$ for strength-$\vec{p}$ local Pauli noise, $k/n < 1-\log(1+3p)$ for strength-$p$ local erasure errors, $k/n < -\log(\frac{1}{2-p}+\sqrt{\frac{p}{2-p}})$ for strength-$p$ local amplitude damping noise, and $k/n < 1-2c\log(\sqrt{1-p}+\sqrt{p})$ for strength-$p$ correlated noise.
\end{result}

\noindent The formal version of this result is given in Corollaries~\ref{coro:pauli_nonsmoothing},~\ref{coro:iiderasure_nonsmoothing},~\ref{coro:amp_nonsmoothing},~\ref{coro:zzcouple}, and Appendix~\ref{appendssc:1DLRCQEC}. A direct consequence of Result \ref{result:Choi_error} is that 1D encoding schemes with depth $O(\log n)$ establish a Choi error $\epsilon_{\mathrm{Choi}} = 1/\poly(n)$. This small Choi error is particularly valuable for quantum error correction over extended periods. When the encoding rate is above the value mentioned in Result~\ref{result:Choi_error}, the scaling of the Choi error will be $O((\frac{\log n}{\log \frac{n}{\varepsilon}})^{\frac{1}{4}})$, as discussed in Appendix~\ref{app:hashingbound}.

Furthermore, we establish the lower bound on the circuit depth required to achieve high AQEC performance under the local depolarizing channel, a type of local Pauli noise. This result provides a lower bound for encoding circuit depth against local Pauli noise and can be easily extended to other local errors.

\begin{result}[Lower bound for circuit depth, informal]
For a circuit $U$ that encodes $k$ logical qubits with depth $d$ and has a Choi error $\epsilon$ against a constant-strength local depolarizing channel, then:
\begin{enumerate}
\item {If $U$ is a D-dimensional circuit, then $d = \Omega\left((\log (\frac{1}{\epsilon}) + \log k)^{1/D}\right)$.}
\item {If $U$ is an all-to-all circuit, then $d = \Omega\left(\log [\log (\frac{1}{\epsilon}) + \log k]\right)$.}
\end{enumerate}
\end{result}
\noindent This result matches the upper bound in Results \ref{result:EncodingRateBound} and \ref{result:Choi_error}, confirming the optimality of our encoding scheme. We present the formal version of this result in Theorem \ref{thm:depolarizing_lower_bound}.
% This lower bound generalizes existing results for the circuit depth required to achieve a high code distance or for other quantum error correction and fault-tolerant quantum computing metrics~\cite{baspin2023lowerboundoverheadquantum,Yi2024order}. We present the formal version of this result in Theorem \ref{thm:depolarizing_lower_bound} and prove it in Appendix \ref{app:lower_bound}.

To prove the AQEC performance established in Results \ref{result:EncodingRateBound} and \ref{result:Choi_error}, we employ the complementary channel approach~\cite{Beny2010AQEC}, which evaluates the Choi error by characterizing how much information leaks from the system to the environment after encoding and noise. The issue becomes how decoupled the environment and the other copy of the system are after inputting an entangled state across two copies of systems, with one undergoing encoding and the complementary channel from system to environment. This decoupling task is illustrated in Fig.~\ref{fig:decoupling}. To address this issue, we prove decoupling theorems tailored for $\mathfrak{C}$ and $\mathfrak{B}$. Compared to previous decoupling results~\cite{Dupuis2014Decoupling,Szehr2013Decoupling,brown2013scramblingspeedrandomquantum}, the new decoupling theorem considers initial states that have a tensor-product structure. This relaxation enables the establishment of the decoupling for 1D low-depth circuits instead of high-dimensional or long-depth circuits. Meanwhile, this relaxation is sufficient for evaluating the Choi error. We present the detailed results in Theorems~\ref{thm:nonsmoothdecoupling},~\ref{thm:1Dlocalrandomcircuit}, and~\ref{thm:smoothdecoupling}.

\begin{figure}[htbp!]
\centering
\includegraphics[width=.48\textwidth]{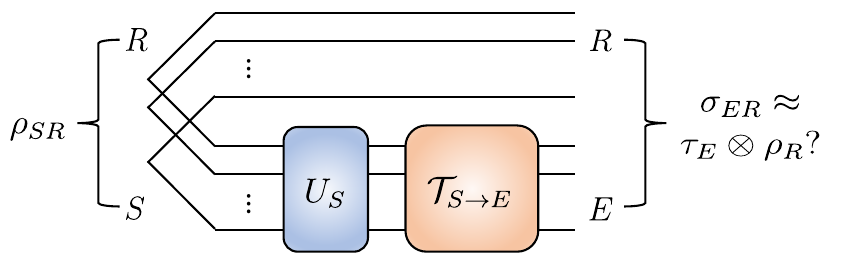}
\caption{The diagram of the decoupling task. The process begins with an entangled state, $\rho_{SR}$, where $R$ is a reference system identical to $S$. The system $S$ undergoes an encoding unitary operation, $U_S$, followed by a channel $\mathcal{T}_{S\rightarrow E}$ that maps $S$ to the environment $E$. The resulting output state is $\sigma_{ER}$. The decoupling task examines whether $\sigma_{ER}$ is sufficiently close to a tensor-product state $\tau_E \otimes \rho_R$ for certain choices of $U_S$ and $\mathcal{T}_{S\rightarrow E}$. The role of $U_S$ is to spread information within $S$, making it difficult for the environment $E$ to extract information about $S$ and $R$ through $\mathcal{T}_{S\rightarrow E}$. When $E$ has no information about $SR$, it becomes decoupled from $R$.
}
\label{fig:decoupling}
\end{figure}

\begin{result}[Decoupling theorems for the double-layer blocked circuit $\mathfrak{C}$ and 1D brickwork circuit $\mathfrak{B}$, informal]\label{thm:informaldecoupling}
Given three quantum systems $S$, $R$, and $E$,  denote the reduced density matrix of $\rho_{SR}$ on system $R$ as $\rho_R$ and the Choi-Jamio{\l}kowski representation of $\mathcal{T}_{S\rightarrow E}$ as $\tau_{SE}$ with reduced density matrix on system $E$ as $\tau_E$. Suppose $U_S$ is drawn from $\mathfrak{C}_n^{\varepsilon}$ or $\mathfrak{B}_n^{\varepsilon}$, and $\rho_{SR}$ has a tensor-product structure over different regions. Then in the large $n$ limit,
\begin{align}
\mathbb{E}_{U_S\sim \mathfrak{C}_n^{\varepsilon}}\Vert \mathcal{T}_{S\rightarrow E}(U_S\rho_{SR}U_{S}^{\dagger}) - \tau_E\otimes \rho_R \Vert_1 &= O(\varepsilon^{c_1}n^{c_2});\\
\mathbb{E}_{U_S\sim \mathfrak{B}_n^{\varepsilon}}\Vert \mathcal{T}_{S\rightarrow E}(U_S\rho_{SR}U_{S}^{\dagger}) - \tau_E\otimes \rho_R \Vert_1 &= O(\varepsilon^{c_3}n^{c_4}),
\end{align}
if $H_2(S|E)_{\tau_{SE}}+H_2(S|R)_{\rho_{SR}} > 0$, and $\varepsilon = O(1/\poly(n))$. Meanwhile,
\begin{equation}
\mathbb{E}_{U_S\sim \mathfrak{C}_n^{\varepsilon}}\Vert \mathcal{T}_{S\rightarrow E}(U_S\rho_{SR}U_{S}^{\dagger}) - \tau_E\otimes \rho_R \Vert_1 = O(\delta+\varepsilon^{c'_1}n^{c'_2})
\end{equation}
if $H_2^{\delta}(S|E)_{\tau_{SE}}+H_2(S|R)_{\rho_{SR}} > 0$, and $\varepsilon = o(1/\poly(n))$. The coefficients $c_1$, $c_2$, $c_3$, $c_4$, $c'_1$, and $c'_2$ are constant. Here, $H_2$ and $H_2^{\delta}$ are the conditional collision entropy and the smoothed version, respectively, which we will introduce in Section~\ref{ssc:notation}.
\end{result}

In the above theorem, the term $H_2(S|R)_{\rho_{SR}}$ represents the initial correlation between system $S$ and reference $R$ and equals $-k$ for $k$ EPR pairs. The term $H_2(S|E)_{\tau_{SE}}$ evaluates how the quantum channel from $S$ to $E$ preserves the correlation or the amount of information leakage from $S$ to $E$. If there is no information leakage, then environment $E$ knows nothing about $S$ and $R$, $H_2(S|E)_{\tau_{SE}}$ achieves the maximum, and $E$ and $R$ will be decoupled. Normally, a higher noise error rate implies a higher information leakage. Thus, $H_2(S|E)_{\tau_{SE}}$ decays when noise strength increases. It gives the lower encoding rate bounds like $n(1-f(\vec{p}))$ for Pauli noise. The smoothed entropy $H_2^{\delta}(S|E)_{\tau_{SE}}$ gives the higher encoding rate bounds.
The non-smooth version also implies the exponential decay of Choi error with circuit depth, as stated in Result~\ref{result:Choi_error}. Thus, $O(\log n)$-depth circuits exhibit negligible small Choi error. On the other hand, the smooth version shows the Choi error approaches 0 for an $\omega(\log n)$-depth circuit with a larger encoding rate.

The new decoupling theorems are expected to have important implications in quantum many-body physics and quantum information tasks, as outlined in the discussion section. We also generalize these results to bound the error for an arbitrary input state where $\rho_{SR}$ does not have a tensor-product structure (Proposition \ref{prop:decoupling}). Additionally, we propose a randomized encoding scheme that exhibits good recovery fidelity for arbitrary input states (Proposition \ref{prop:pauli_twirling}).

\section{Preliminaries}\label{sc:pre}
In this section, we introduce the preliminaries of our work. First, we present the basic notations and quantities used throughout this paper. Then, we introduce fundamental concepts and results in approximate quantum error correction. Next, we discuss the related concept of random unitaries and quantum design.
% An additional technical review about quantum entropy and decoupling theorem is shown in Appendix~\ref{app:entropy}.

\subsection{Basic notations and quantities}\label{ssc:notation}
Let us start with the basic notations used in this work. We denote the Hilbert space as $\mathcal{H}$, the set of the states on $\mathcal{H}$ as $\mathcal{D}(\mathcal{H})$, the set of the general matrices on $\mathcal{H}$ as $\mathcal{M}(\mathcal{H})$. For a single qudit system with local dimension $q$, we denote the Hilbert space as $\mathcal{H}_q$ where $\dim \mathcal{H}_q = q$. The basis of $\mathcal{H}_q$ is $\{\ket{1}, \ket{2}, \cdots, \ket{q}\}$. The Hilbert space of the $n$-qudit system is $\mathcal{H}_q^{\otimes n}$. The Hilbert space $\mathcal{H}_{AB}$ of the combined system $AB$ of subsystems $A$ and $B$ is $\mathcal{H}_{AB} = \mathcal{H}_A\otimes \mathcal{H}_B$. The dimension of $\mathcal{H}_A$ is denoted as $\abs{A}$. For a quantum state $\rho_{AB}\in \mathcal{D}(\mathcal{H}_{AB})$, its reduced density matrix on system $B$ is denoted as $\rho_B = \tr_A(\rho_{AB})$, where $\tr_A$ is the partial trace over subsystem $A$. The unitary operation $U$ on system $A$ is denoted as $U_A$ with the subscript to represent the support of $U$. A completely positive and trace-preserving (CPTP) map $\mathcal{T}$ that maps a state from $\mathcal{H}_A$ to $\mathcal{H}_B$ is denoted as  $\mathcal{T}_{A \rightarrow B}$.
The Choi-Jamio{\l}kowski representation of a CPTP map $\mathcal{T}_{A\rightarrow B}$ is a quantum state $\tau_{AB}$, which is defined as
\begin{equation}
\tau_{AB} = \mathcal{T}_{A'\rightarrow B}(\ketbra{\hat{\phi}}_{AA'}),
\end{equation}
where $A'$ is a copy of system $A$, and $\ketbra{\hat{\phi}}_{AA'}$ is the maximally entangled state on system $AA'$ with
\begin{equation}
\ket{\hat{\phi}}_{AA'} = \frac{1}{\sqrt{\abs{A}}}\sum_{i=0}^{\abs{A}-1}\ket{i}_A\ket{i}_{A'}.
\end{equation}

We use several measures to quantify the distance between two states. The fidelity is defined as
\begin{equation}
F(\rho, \sigma) = \Vert \sqrt{\rho}\sqrt{\sigma}\Vert_1 = \tr\sqrt{\sqrt{\rho}\sigma\sqrt{\rho}}.
\end{equation}
For pure state $\ket{\psi}$ and mixed state $\sigma$, the fidelity is given by
\begin{equation}\label{eq:fidelity_pure_mixed}
F(\ketbra{\psi}, \sigma) = \sqrt{\bra{\psi} \sigma \ket{\psi}}.
\end{equation}
From the fidelity, we can define the purified distance:
\begin{equation}
P(\rho,\sigma) = \sqrt{1-F(\rho, \sigma)^2},
\end{equation}
which satisfies~\cite{nielsen2010quantum}
\begin{equation}\label{eq:tracedistanceineq}
\frac{1}{2}\Vert \rho - \sigma \Vert_1 \leq P(\rho, \sigma) \leq \sqrt{\Vert \rho - \sigma\Vert_1}.
\end{equation}

The trace distance between $\rho$ and $\sigma$ is $\frac{1}{2}\Vert \rho - \sigma \Vert_1$ . The definitions of fidelity $F$ and trace distance can also be generalized to arbitrary positive matrices beyond states by substituting $\rho$ and $\sigma$ with general matrices. For two quantum channels $\mathcal{T}_{A\rightarrow B}$ and $\mathcal{T}'_{A\rightarrow B}$, we normally use diamond distance to characterize their distance, defined as
\begin{equation}
\Vert \mathcal{T}_{A\rightarrow B} - \mathcal{T}'_{A\rightarrow B} \Vert_{\diamond} = \max_{\rho_{AR}} \Vert \mathcal{T}_{A\rightarrow B}(\rho_{AR}) - \mathcal{T}'_{A\rightarrow B}(\rho_{AR})  \Vert_1,
\end{equation}
where $R$ is a reference system without dimension restriction. Nonetheless, it can be shown that considering a reference system $R$ with the same dimension as $A$ is sufficient for maximization. The definition of the diamond distance can also be generalized for two completely positive maps.

In this work, we will massively use the conditional entropy. The conditional collision entropy of $A$ given $B$ is defined as
\begin{equation}
H_2(A|B)_{\rho} = -\log \tr[\left((\id_A\otimes \rho_B)^{-1/4}\rho_{AB}(\id_A\otimes \rho_B)^{-1/4}\right)^2].
\end{equation}
Its smoothed version is
\begin{equation}
H^{\delta}_{2}(A|B)_{\rho} = \sup_{\hat{\rho}_{AB}\in \mathcal{B}^{\delta}(\rho_{AB})}H_{2}(A|B)_{\hat{\rho}},
\end{equation}
where $\mathcal{B}^{\delta}(\rho) = \{\rho', P(\rho', \rho)\leq \delta\}$ is the set of all $\delta$-close states of $\rho$. For more details about quantum entropies refer to Appendix~\ref{app:entropy}.

Below, we introduce a mathematical lemma used in this work, which turns the evaluation of the 1-norm of a matrix into a second-moment quantity.

\begin{lemma}[Lemma 3.7 in Ref.~\cite{Dupuis2014Decoupling}]\label{lemma:1normbound}
Given a Hilbert space $\mathcal{H}$, let $M\in \mathcal{M}(\mathcal{H})$ and $\sigma\in \mathcal{D}(\mathcal{H})$. We have that
\begin{equation}
\Vert M \Vert_1\leq \sqrt{\tr(\sigma^{-1/4}M\sigma^{-1/2}M^{\dagger}\sigma^{-1/4})}.
\end{equation}
If $M$ is Hermitian, this simplifies to
\begin{equation}
\Vert M \Vert_1\leq \sqrt{\tr((\sigma^{-1/4}M\sigma^{-1/4})^2)}.
\end{equation}
\end{lemma}

\subsection{Approximate quantum error correction}\label{subsec:aqec}
In this subsection, we introduce basic concepts in approximate quantum error correction. To protect logical information, the logical system $L$, consisting of $k$ qubits, is encoded into the $n$-qubit physical system $S$ via an encoding channel $\mathcal{E}_{L \rightarrow S}$. The logical system may also be entangled with a $k$-qubit reference system $R$. The code rate is given by $k/n$. After encoding, the physical system $S$ undergoes a noise channel $\mathcal{N}$. Finally, a decoding operation $\mathcal{D}_{S \rightarrow L}$ is applied to recover the information stored in the logical system. For an initial state $\rho_{LR}$, the recovered state is given by
\begin{equation}
[(\mathcal{D}\circ\mathcal{N}\circ\mathcal{E})_L\otimes I_R](\rho_{LR}).
\end{equation}
For a unitary encoding, $\mathcal{E}_{L\rightarrow S}(\rho) = U_S\rho_L\otimes \ketbra{\mathbf{0}}_{S\backslash L}U_S^{\dagger}$, as demonstrated by Fig.~\ref{fig:qecdiagram}.

\begin{figure}[htbp!]
\centering
\includegraphics[width=.3\textwidth]{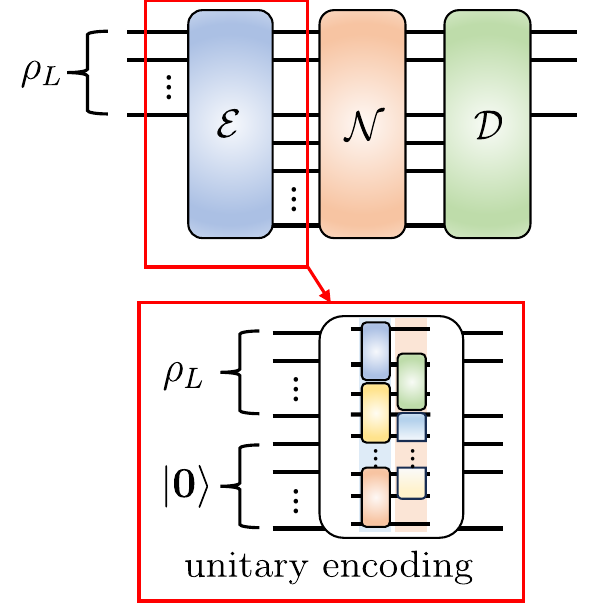}
\caption{Diagram of the procedure of quantum error correction. The initial logical information is stored in a quantum state $\rho_L$. Before undergoing noise $\mathcal{N}$, $\rho_L$ is first encoded by $\mathcal{E}$ with $\mathcal{E}(\rho) = U\rho\otimes \ketbra{\mathbf{0}}U^{\dagger}$ for unitary encoding operation $U$. Finally, one applies decoding operation $\mathcal{D}$ to retrieve logical information. More generally, $\rho_L$ may be entangled with a reference system $R$.}
\label{fig:qecdiagram}
\end{figure}

One can quantify the performance of quantum error correction by the fidelity between the initial and final states. In this work, we focus on the entanglement fidelity, or Choi fidelity, which is defined as
\begin{equation}
F_{\mathrm{Choi}} = \max_{\mathcal{D}} F\left(\hat{\phi}_{LR}, [(\mathcal{D}\circ\mathcal{N}\circ\mathcal{E})_L\otimes I_R](\hat{\phi}_{LR})\right),
\end{equation}
where $\hat{\phi}_{LR}$ is the maximally entangled state between the logical system and the reference system. The Choi error is defined as:
\begin{equation}\label{eq:choierrordef}
\begin{split}
\epsilon_{\mathrm{Choi}} &= \sqrt{1 - F^2_{\mathrm{Choi}}}\\
&= \min_{\mathcal{D}} P\left(\hat{\phi}_{LR}, [(\mathcal{D}\circ\mathcal{N}\circ\mathcal{E})_L\otimes I_R](\hat{\phi}_{LR})\right).
\end{split}
\end{equation}

The Choi fidelity quantifies how well the code can preserve entanglement between the logical and reference systems. It is also closely related to the max-average state fidelity $F_{\mathrm{ave}}$, which reflects the average performance of the AQEC codes over all input states:
\begin{equation}
F_{\mathrm{ave}}^2 = \max_{\mathcal{D}} \int d\psi \, F^2\left(\psi, (\mathcal{D}\circ\mathcal{N}\circ\mathcal{E})_L(\psi)\right).
\end{equation}
Here, $d\psi$ denotes the Haar random distribution of pure states over $\mathcal{H}_L$. The average fidelity $F_{\mathrm{ave}}$ and Choi fidelity satisfy the relation \cite{horodecki1999generalteleportationchannelsinglet, Gilchrist2005DistanceMeasures}:
\begin{equation}
F_{\mathrm{ave}}^2 = \frac{dF_{\mathrm{Choi}}^2 + 1}{d + 1},
\end{equation}
where $d = |L| = 2^k$ is the dimension of the logical system.

Directly evaluating Eq.~\eqref{eq:choierrordef} is challenging. The Choi error is often analyzed using the complementary channel formalism \cite{Beny2010AQEC,Kong2022CQEC}. In essence, the complementary channel formalism quantifies recovery performance by examining the information leakage contained in the environment. Specifically, according to Stinespring dilation, the noise channel can be modeled by interacting the physical system $S$ with the environment $E$ through a unitary $U$, followed by tracing out the environment:
\begin{equation}
\mathcal{N}_S(\rho_S) = \tr_E(U_{SE}\rho_S\otimes \ketbra{0}_EU_{SE}^{\dagger}).
\end{equation}
If we trace out the system instead of the environment, we obtain the complementary channel of $\mathcal{N}_S$:
\begin{equation}
\hat{\mathcal{N}}_{S\rightarrow E} = \tr_S(U_{SE}\rho_S\otimes \ketbra{0}_EU_{SE}^{\dagger}).
\end{equation}
For a noise channel with Kraus operators $\{K_i\}$,
\begin{equation}
\mathcal{N}_S(\rho) = \sum_i K_i \rho K_i^{\dagger},
\end{equation}
the complementary channel $\hat{\mathcal{N}}_{S\rightarrow E}$ can be written as~\cite{Devetak2005Capacity,Beny2010AQEC}:
\begin{equation}
\hat{\mathcal{N}}_{S\rightarrow E} = \sum_{ij}\tr(K_i\rho K_j^{\dagger})\ketbra{i}{j}.
\end{equation}

Note that the complementary channel of $\mathcal{N}\circ \mathcal{E}$ can be obtained by
\begin{equation}
\widehat{\mathcal{N}\circ \mathcal{E}}_{L\rightarrow E} = \hat{\mathcal{N}}_{S\rightarrow E}\circ\mathcal{E}_{L\rightarrow S}.
\end{equation}
By applying the complementary channel $\hat{\mathcal{N}}_{S \rightarrow E}$, we obtain the state in the environment, which captures the information leaked from the system to the environment. The Choi error can then be expressed as \cite{Beny2010AQEC, Kong2022CQEC}:
\begin{equation}\label{eq:choicomplementary}
\epsilon_{\mathrm{Choi}} = \min_{\zeta} P\left( (\widehat{\mathcal{N}\circ \mathcal{E}}_{L\rightarrow E}\otimes I_R)(\ketbra{\hat{\phi}}_{LR}), (\mathcal{T}^{\zeta}_{L\rightarrow E}\otimes I_R)(\ketbra{\hat{\phi}}_{LR}) \right),
\end{equation}
where $\mathcal{T}^{\zeta}$ is a quantum channel that maps any input state into a fixed state $\zeta$, satisfying:
\begin{equation}\label{eq:zeta_complementary}
(\mathcal{T}^{\zeta}_{L\rightarrow E}\otimes I_R)(\rho_{LR}) = \zeta_E\otimes \tr_L(\rho_{LR}).
\end{equation}
Eq.~\eqref{eq:choicomplementary} indicates that if the state in the environment is close to a fixed state in the tensor product with the reference system, the Choi error will be small, implying that entanglement is well preserved.

Based on the Choi error, one can define the achievable rate and the channel capacity of local noise channels~\cite{Klesse2007AQEC}. Note that the Choi error is the function of the logical qubit number $k$, the encoding map $\mathcal{E}$, and the noise channel $\mathcal{N}$. We can denote it as $\epsilon_{\mathrm{Choi}} = \epsilon_{\mathrm{Choi}}(k, \mathcal{E}, \mathcal{N})$. For a single-qubit noise channel $\mathcal{N}_1$, if there exists a sequence of logical qubit numbers $\{k_n\}$ and encoding maps $\{\mathcal{E}_n\}$ such that
\begin{equation}
\lim_{n\rightarrow \infty} \epsilon_{\mathrm{Choi}}(k_n, \mathcal{E}_n, \mathcal{N}_1^{\otimes n}) = 0,
\end{equation}
then we call
\begin{equation}
R = \limsup_{n\rightarrow \infty} \frac{k_n}{n}
\end{equation}
as an achievable rate of $\mathcal{N}_1$. The quantum channel capacity of $\mathcal{N}_1$, denoted as $Q(\mathcal{N}_1)$, is the supremum over all achievable rates.

\subsection{Noise model}
In this work, we examine a wide variety of important noise models: Pauli noises~\cite{Bennett1996QEC}, erasure errors~\cite{Grassl1997erasure}, amplitude damping noises~\cite{Giovannetti2005amplitude}, and correlated noises. Firstly, we introduce Pauli noises. For a single qubit, the Pauli noise is defined as
\begin{equation}\label{eq:pauli_noise}
\mathcal{N}_p(\rho) = p_I\rho+p_XX\rho X+p_YY\rho Y+p_ZZ\rho Z.
\end{equation}
We denote $\Vec{p} = (p_I, p_X, p_Y, p_Z)$ and refer to this noise as the strength-$\Vec{p}$ Pauli noise. When the probabilities satisfy $p_X=p_Y=p_Z=p/3$, the Pauli noise simplifies to the depolarizing noise:
\begin{equation}\label{eq:depolarizing_channel}
\mathcal{N}_d(\rho) = (1-p) \rho + p \frac{\bI}{2}.
\end{equation}
which we refer to as the strength-$p$ depolarizing noise. For an $n$-qubit state, we consider the i.i.d.~local Pauli noise on each qubit. The overall noise channel is then the tensor product of $n$ independent Pauli noises, expressed as $\mathcal{N} = \mathcal{N}_p^{\otimes n}$.

Another noise model of interest is the erasure error, which completely traces out a subsystem $T$ while simultaneously recording the location of the erased subsystem in a separate register $C$. The erasure error model is important due to its strong connection with code distance \cite{nielsen2010quantum,Grassl1997erasure}. For an erasure noise act on a fixed subsystem $T$, the register $C$ can be ignored and the noise is expressed as:
\begin{equation}
\mathcal{N}(\rho) = \tr_{T}(\rho).
\end{equation}
The complementary channel of the erasure error is
\begin{equation}
\hat{\mathcal{N}}(\rho) = \tr_{\Bar{T}}(\rho),
\end{equation}
where $\bar{T} = S \backslash T$ denotes the complement of subsystem $T$. This implies that the environment obtains the state in the traced-out subsystem.
For the random $t$-erasure error, the locations of the erasured subsystem are randomly selected among all possible $t$-qubit subsystems, and the noise is expressed as
\begin{equation}\label{eq:random_t_erasure}
\mathcal{N}(\rho) = \frac{1}{\binom{n}{t}}\sum_{\abs{T}=t} \ketbra{T}_C \otimes \tr_{T}(\rho).
\end{equation}
Here, each $T$ represents a distinct $t$-qubit subsystem being erased, and $\ket{T}_C$ records the specific location of the erasure in the register $C$. The complementary channel of this random $t$-erasure noise is~\cite{Faist2020AQEC}
\begin{equation}
\hat{\mathcal{N}}(\rho) = \frac{1}{\binom{n}{t}}\sum_{\abs{T}=t}\ketbra{T}_C\otimes \tr_{\bar{T}}(\rho),
\end{equation}
where $\ket{T}_C$ continues to label the location of the erasure errors.

Instead of considering erasure errors with a fixed number of erasures, one can also study i.i.d.~erasure errors. Particularly, we consider that each qubit experiences an erasure error i.i.d.~with probability $p$. For a single-qubit state $\rho$, the noise is described by
\begin{equation}\label{eq:iiderasure}
\mathcal{N}_e(\rho) = (1-p)\rho + p \ketbra{2},
\end{equation}
where state $\ket{2}$ is orthogonal to the space supporting $\rho$ and represents the erasure. For an $n$-qubit system, the total erasure noise is modeled as $\mathcal{N}_e^{\otimes n}$.

The local amplitude damping noise is also an i.i.d.\ noise model. We first define the strength-$p$ amplitude damping noise on a single qubit:
\begin{equation}
\mathcal{N}_a(\rho) = K_0\rho K_0^{\dagger} + K_1\rho K_1^{\dagger},
\end{equation}
with
\begin{equation}
K_0 = \begin{pmatrix}
1 & 0\\
0 & \sqrt{1-p}
\end{pmatrix},
K_1 = \begin{pmatrix}
0 & \sqrt{p}\\
0 & 0
\end{pmatrix}.
\end{equation}
For an $n$-qubit system, the total amplitude damping noise is $\mathcal{N}_a^{\otimes n}$.

The quantum channel capacity of these noise models is of fundamental importance in quantum information theory. For Pauli noise, a well-known lower bound on the quantum channel capacity is the hashing bound:
\begin{equation}\label{eq:hashingbound}
Q(\mathcal{N}_p)\geq \max\{ 0, 1 - h(\Vec{p}) \},
\end{equation}
where
\begin{equation}
h(\Vec{p}) = -p_I\log p_I - p_X\log p_X - p_Y\log p_Y - p_Z\log p_Z,
\end{equation}
is the Shannon entropy of the probability vector $\Vec{p}$. Moreover, it has been shown that the inequality in Eq.~\eqref{eq:hashingbound} is strict in some regimes \cite{DiVincenzo1998channelcapacity, wilde2013quantum, Ataides2021XZZX}. Existing upper and lower bounds for the quantum capacity of Pauli noise do not match, and the exact value of the quantum channel capacity for Pauli noise remains unknown, even for the depolarizing noise. The hashing bound is the most famous lower bound for a generic Pauli noise, which is achievable for random stabilizer codes~\cite{wilde2013quantum}.
% \lgd{This bound is also the highest achievable rate for Quantum CSS codes~\cite{Bennett1996QEC}.}
Finding encoding schemes simpler than the random stabilizer construction to achieve the hashing bound deserves further exploration.

For erasure errors, the quantum channel capacity is much simpler and can be computed exactly~\cite{Bennett1997Erasure}:
\begin{equation}
Q(\mathcal{N}_e) = \max\{ 0, 1 - 2p\}.
\end{equation}
This result matches the Singleton bound $n-k-2t\geq 0$, which yields $k/n\leq 1-\frac{2t}{n}$. The erasure probability $p$ in the i.i.d.~case corresponds to $\frac{t}{n}$.

For amplitude damping noise, the quantum channel capacity can also be computed exactly~\cite{Giovannetti2005amplitude}. For $p > \frac{1}{2}$, $Q(\mathcal{N}_a) = 0$. For $p\leq \frac{1}{2}$,
\begin{equation}\label{eq:ampcapacity}
Q(\mathcal{N}_a) = \max_{0\leq x\leq 1}\{ h((1-p)x)-h(px)\},
\end{equation}
where $h$ is the binary entropy function.

The above noise models are all i.i.d.\ single-qubit noises, which enjoy a tensor-product structure. In this work, we also study correlated noise models, which are also practically important. Particularly, we introduce strength-$p$ 1D nearest neighbor $ZZ$-coupling noise on $n$ qubits:
\begin{equation}
\mathcal{N}_{zz}(\rho) = \circ_{i=1}^{n-1}((1-p)\mathcal{I}+p\mathcal{Z}_i\mathcal{Z}_{i+1})(\rho),
\end{equation}
where $\mathcal{I}$ is the identity channel, and $\mathcal{Z}_i$ is the Pauli $Z$ channel on qubit $i$. Note that for correlated noise, currently, there is no corresponding quantum channel capacity. But there can exist a threshold for $p$ such that the logical error rate goes to $0$ for this noise model.

\subsection{Twirling operations and approximate design}
In this part, we introduce the concept of twirling and design associated with random unitary operations. The twirling operation is defined over an ensemble of quantum unitary operations, $\mathfrak{S}$. We consider the following $t$-th order twirling operation over $\mathfrak{S}$:
\begin{equation}
\Phi^t_{\mathfrak{S}}(O) = \mathbb{E}_{U\sim\mathfrak{S}} U^{\otimes t}O U^{\dagger\otimes t}.
\end{equation}
where $U$ is drawn from the ensemble $\mathfrak{S}$; $O$ is the matrix being twirled, and its support may overpass $U^{\otimes t}$. The expectation is taken over the random unitary operations from $\mathfrak{S}$. When the ensemble $\mathfrak{S}$ corresponds to the Haar random unitaries $\mathfrak{U}_n$ on $n$ qubits, the first-order twirling operation over $\mathfrak{U}_n$ can be obtained from the Schur-Weyl duality \cite{fulton2004Representation}:
\begin{equation}\label{eq:firsttwirling}
\mathbb{E}_{U\sim \mathfrak{U}_n} U O U^{\dagger} = \tr_U(O)\frac{\id_{q}}{q},
\end{equation}
where $\tr_U$ denotes the trace over the subsystem being twirled, $q = \dim U$ is the dimension of Hilbert space that the unitary acts on, and $\id_{q}$ is the identity operator with dimension $q$. The second-order twirling operation over $\mathfrak{U}_n$ is
\begin{equation}\label{eq:secondtwirling}
\begin{split}
\mathbb{E}_{U\sim \mathfrak{U}_n} U^{\otimes 2} O U^{\dagger \otimes 2} &= B_1\tr_{U^{\otimes 2}}(OB_1)+ B_2\tr_{U^{\otimes 2}}(OB_2)\\
&= \tr_{U^{\otimes 2}}(O \frac{\id_{q^2}-q^{-1} F}{q^2-1})\id_{q^2}+\tr_{U^{\otimes 2}}(O \frac{F-q^{-1} \id_{q^2}}{q^2-1})F,
\end{split}
\end{equation}
where
\begin{equation}
B_1 = \frac{\id_{q^2}}{q}, B_2 = \frac{F-B_1}{\sqrt{q^2-1}},
\end{equation}
with $\tr_{U^{\otimes 2}}$ denoting the trace over the subsystem of being twirled, $\id_{q^2}$ is the identity operator on the $q^2$-dimension subsystem, and $F$ denotes the SWAP operator on the $q^2$-dimension subsystem.

Though the Haar random unitary ensemble possesses good twirling properties, it requires exponential resources to implement in practice. An approximate version of the Haar random unitary ensemble is usually used for practical applications. An ensemble $\mathfrak{S}$ is called a unitary $t$-design if $\Phi^t_{\mathfrak{S}} = \Phi^t_{\mathfrak{U}_n}$ and an $\varepsilon$-approximate unitary $t$-design~\cite{dankert2005efficientsimulationrandomquantum,Gross2007deisgn,Dankert2009design,Brandao2016ApproximateDesign} if
\begin{equation}
(1-\varepsilon)\Phi^t_{\mathfrak{U}_n}\leq \Phi^t_{\mathfrak{S}} \leq (1+\varepsilon)\Phi^t_{\mathfrak{U}_n},
\end{equation}
where $\Phi_1 \leq \Phi_2$ means that $\Phi_2-\Phi_1$ is a completely positive map~\cite{Brandao2016ApproximateDesign}. The term $\varepsilon$ is referred to as the relative error. An approximate design can also be defined in an additive error version:
\begin{equation}
\Vert \Phi^t_{\mathfrak{S}} - \Phi^t_{\mathfrak{U}_n} \Vert_{\diamond}\leq \varepsilon.
\end{equation}
It has been shown that an approximate design with relative error $\varepsilon$ is an approximate design with additive error $2\varepsilon$~\cite{schuster2024randomunitariesextremelylow}. Hereafter, unless stated otherwise, the error of an approximate design refers to relative error. An exact unitary $t$-design is an approximate unitary $t$-design with $\varepsilon = 0$. It was shown that the $n$-qubit Clifford group is an exact unitary 3-design~\cite{Webb2016Clifford3design,Zhu2017MultiqubitClifford}, and when the local dimension is odd prime instead of $2$, the Clifford group is an exact unitary 2-design~\cite{DiVincenzo2002hiding,Dur2005depolarization,dankert2005efficientsimulationrandomquantum,Gross2007deisgn,Dankert2009design}. As a remark, $\mathfrak{C}$ and $\mathfrak{B}$ are exact unitary 1-designs since they contain the local Pauli group, which is already a unitary 1-design.

\section{Log-depth circuit to achieve decoupling and high AQEC performance}\label{sc:upper}
In this section, we investigate the AQEC performance of the unitary operations generated by 1D log-depth circuits using the Choi error metric.

Before elaborating on the detailed results, we note that a recent surprising result shows that the 1D double-layer blocked circuit $\mathfrak{C}_n^{\varepsilon}$ forms a $9\varepsilon$-approximate unitary 3-design and $4\varepsilon$-approximate unitary 2-design~\cite{schuster2024randomunitariesextremelylow}. However, the property of approximate design, in general, does not guarantee good AQEC performance unless $\varepsilon$ is smaller than an exponential decay function. The decoupling theorem for approximate unitary designs~\cite{Szehr2013Decoupling} includes an exponential factor, as stated in Lemma \ref{lemma:appro2designdecoup} in Appendix \ref{app:entropy}. The exponential small $\varepsilon$ implies a circuit depth linear in the qubit number, which is not desired. To address this issue, we develop new decoupling theorems. We use these results to show that $\mathfrak{C}_n^{\varepsilon}$ exhibits a logical qubit count that scales linearly with $n$ while maintaining a small Choi error. Particularly, $\mathfrak{C}_n^{\varepsilon}$ can achieve the hashing bound for Pauli noises and the quantum capacity for erasure errors even if $\varepsilon$ is not exponentially small.

We begin by reformulating the analysis of the Choi error as a decoupling problem. Next, we present the decoupling theorems for 1D double-layer blocked circuit $\mathfrak{C}$ and 1D brickwork circuit $\mathfrak{B}$ as a key technical contribution. By applying the decoupling theorem to Pauli noise, erasure errors, amplitude damping noise, and correlated noise, we establish the favorable AQEC performance. Furthermore, we prove a better AQEC performance of $\mathfrak{C}$ and $\mathfrak{B}$ than the block-encoding scheme. We also extend these results to more general error metrics beyond the Choi error. In the next section, we derive the lower bound of the circuit depth required to achieve a small Choi error in AQEC, which matches the circuit depth needed for $\mathfrak{C}$ and $\mathfrak{B}$.

\subsection{Approximate quantum error correction and decoupling theorem}
We start by analyzing the expected Choi error. With the complementary channel approach, we have that
\begin{equation}
\begin{split}
\epsilon_{\mathrm{Choi}} =& \min_{\zeta} P\left( \hat{\mathcal{N}}_{S\rightarrow E}(U\ketbra{\hat{\phi}}_{LR}\otimes \ketbra{0^{n-k}} U^{\dagger}), \zeta_E\otimes \frac{\id_R}{2^k} \right)\\
\leq& P\left( \hat{\mathcal{N}}_{S\rightarrow E}(U\ketbra{\hat{\phi}}_{LR}\otimes \ketbra{0^{n-k}} U^{\dagger}), \mathbb{E}_{U\sim \mathfrak{U}_n}\hat{\mathcal{N}}_{S\rightarrow E}(U\ketbra{\hat{\phi}}_{LR}\otimes \ketbra{0^{n-k}} U^{\dagger}) \right)\\
&+ \min_{\zeta} P\left(\mathbb{E}_{U\sim \mathfrak{U}_n} \hat{\mathcal{N}}_{S\rightarrow E}(U\ketbra{\hat{\phi}}_{LR}\otimes \ketbra{0^{n-k}} U^{\dagger}), \zeta_E\otimes \frac{\id_R}{2^k} \right)\\
=& P\left( \hat{\mathcal{N}}_{S\rightarrow E}(U\ketbra{\hat{\phi}}_{LR}\otimes \ketbra{0^{n-k}} U^{\dagger}), \hat{\mathcal{N}}_{S\rightarrow E}(\frac{\id_S}{2^n})\otimes \frac{\id_R}{2^k} \right) + \min_{\zeta} P\left(\hat{\mathcal{N}}_{S\rightarrow E}(\frac{\id_S}{2^n})\otimes \frac{\id_R}{2^k}, \zeta_E\otimes \frac{\id_R}{2^k} \right)\\
=& P\left( \hat{\mathcal{N}}_{S\rightarrow E}(U\ketbra{\hat{\phi}}_{LR}\otimes \ketbra{0^{n-k}} U^{\dagger}), \hat{\mathcal{N}}_{S\rightarrow E}(\frac{\id_S}{2^n})\otimes \frac{\id_R}{2^k} \right).
\end{split}
\end{equation}
The first inequality comes from the triangle inequality of the purified distance. The second equality uses the property of the twirling over Haar random unitary gates. The third equality holds by taking the state $\zeta$ as $\hat{\mathcal{N}}_{S\rightarrow E}(\frac{\id_S}{2^n})$.

Our target is to bound $\mathbb{E}_{U}\epsilon_{\mathrm{Choi}}$. We utilize the relation between purified distance and trace distance in Eq.~\eqref{eq:tracedistanceineq} to get
\begin{equation}\label{eq:choiexpectation}
\begin{split}
\mathbb{E}_{U} \epsilon_{\mathrm{Choi}} \leq& \mathbb{E}_{U} \sqrt{\Vert \hat{\mathcal{N}}_{S\rightarrow E}(U\ketbra{\hat{\phi}}_{LR}\otimes \ketbra{0^{n-k}} U^{\dagger}) - \hat{\mathcal{N}}_{S\rightarrow E}(\frac{\id_S}{2^n})\otimes \frac{\id_R}{2^k} \Vert_1}\\
\leq& \sqrt{\mathbb{E}_{U} \Vert \hat{\mathcal{N}}_{S\rightarrow E}(U\ketbra{\hat{\phi}}_{LR}\otimes \ketbra{0^{n-k}} U^{\dagger}) - \hat{\mathcal{N}}_{S\rightarrow E}(\frac{\id_S}{2^n})\otimes \frac{\id_R}{2^k} \Vert_1}.
\end{split}
\end{equation}
The remaining part is evaluating the term in the square root, which corresponds to the decoupling issue: how the system $E$ and $R$ decouple after being acted on by a unitary and channel $\hat{\mathcal{N}}_{S\rightarrow E}$.

We now present our decoupling theorems. For the double-layer blocked ensemble $\mathfrak{C}$, we establish both non-smooth and smooth versions, and for the 1D brickwork circuit ensemble $\mathfrak{B}$, we prove the non-smooth version. We demonstrate how the environmental and reference systems are decoupled after the physical system, encoded by unitaries from $\mathfrak{C}$ or $\mathfrak{B}$, interacts with a noise channel. This step lays the foundation for proving that the encoding rate can be maintained at a constant encoding rate with a small Choi error. Subsequently, we introduce the smooth decoupling theorem for $\mathfrak{C}$. This refinement allows us to demonstrate that the encoding rate for the local Pauli noise channel can reach the hashing bound, and the encoding rate for the local erasure error can reach the channel capacity.

We first present the non-smooth decoupling theorem for the double-layer blocked circuit $\mathfrak{C}$ below.
\begin{theorem}[Non-smooth decoupling theorem for 1D double-layer blocked circuits]\label{thm:nonsmoothdecoupling}
Given three quantum systems $S = S_1S_2\cdots S_{2N}$, $R = R_1R_2\cdots R_{2N}$, and $E = E_1E_2\cdots E_{2N}$, we denote $S_i$ as the subsystem of region $i$ inside $S$, and the same for $R_i$ and $E_i$. Without loss of generality, we assume that $S$ is composed of $n$ qubits, and each $S_i$ has an equal size $\xi = \frac{n}{2N}$ qubits. Let quantum state $\rho_{SR}$ be
$\rho_{SR} = \bigotimes_{i=1}^{2N}\rho_{S_iR_i}$ with a tensor-product structure and $\mathcal{T}_{S\rightarrow E}$ be a completely positive map. Denote the reduced density matrix of $\rho_{SR}$ on system $R$ as $\rho_R$ and the Choi-Jamio{\l}kowski representation of $\mathcal{T}_{S\rightarrow E}$ as $\tau_{SE}$ with reduced density matrix on system $E$ as $\tau_E$. Suppose $U_S$ is drawn from $\mathfrak{C}_n^{\varepsilon}$ as shown in Fig.~\ref{fig:1Dlowdepthcircuit}. We obtain that
\begin{equation}
\mathbb{E}_{U_S\sim \mathfrak{C}_n^{\varepsilon}}\Vert \mathcal{T}_{S\rightarrow E}(U_S\rho_{SR}U_{S}^{\dagger}) - \tau_E\otimes \rho_R \Vert_1\leq \sqrt{2^{-H_2(S|E)_{\tau_{SE}}-H_2(S|R)_{\rho_{SR}}} + c \max_{A\subseteq [2N]} 2^{-H_2(A|R)_{\rho_{AR}}-H_2(A|E)_{\tau_{AE}}}},
\end{equation}
where
\begin{equation}
c = 2((1+\frac{2^{\xi+1}}{2^{2\xi}+1}\max_i(\max(2^{-H_2(S_i|R_i)_{\rho_{S_iR_i}}}, 2^{H_2(S_i|R_i)_{\rho_{S_iR_i}}})))^{N-1}-1).
\end{equation}
Furthermore, if $\mathcal{T}_{S\rightarrow E}$ has a tensor-product structure, $\mathcal{T}_{S\rightarrow E} = \bigotimes_{i=1}^{2N}\mathcal{T}_{S_i\rightarrow E_i}$, we obtain that
\begin{equation}
\begin{split}
&\mathbb{E}_{U_S\sim \mathfrak{C}_n^{\varepsilon}}\Vert \mathcal{T}_{S\rightarrow E}(U_S\rho_{SR}U_{S}^{\dagger}) - \tau_E\otimes \rho_R \Vert_1\\
\leq& \sqrt{2^{-\sum_{i=1}^{2N}H_2(S_i|E_i)_{\tau_{S_iE_i}}-\sum_{i=1}^{2N}H_2(S_i|R_i)_{\rho_{S_iR_i}}}+ c\prod_{i=1}^{2N} \max(1, 2^{-H_2(S_i|E_i)_{\tau_{S_iE_i}}-H_2(S_i|R_i)_{\rho_{S_iR_i}}})}.
\end{split}
\end{equation}
\end{theorem}

\begin{proof}[Proof sketch]
The proof of Theorem~\ref{thm:nonsmoothdecoupling} begins with Lemma~\ref{lemma:1normbound}, where the inequality $\Vert M \Vert_1 \leq \sqrt{\tr\left( (\sigma^{-1/4} M \sigma^{-1/4})^2 \right)}$ is used to convert the 1-norm distance into the square of a matrix. By expanding this square, we apply both the first- and second-order twirling operations conducted by $\mathfrak{C}$. The analysis of the first-order twirling is straightforward, as $\mathfrak{C}$ is a unitary 1-design. However, the second-order twirling requires evaluating $\mathbb{E}_U \tr\left[ \left( \widetilde{\mathcal{T}}_{S \to E}\left(U_S \widetilde{\rho}_{SR} U_S^{\dagger}\right) \right)^2 \right]$ and bounding the difference between the expectation over $\mathfrak{C}$ and the Haar random unitary group $\mathfrak{U}_n$. The property of $\mathfrak{C}$ being an approximate 2-design does not guarantee that this term is small unless $\varepsilon$ is exponentially small. To address this, we refine the evaluation by leveraging the double-layer blocked structure of $\mathfrak{C}$. The additional condition that the state $\rho_{SR}$ exhibits a tensor-product structure simplifies the calculation, and this condition holds for EPR input states. The full proof can be found in Appendix~\ref{app:proof_nonsmooth}.
\end{proof}

The non-smooth decoupling theorem for the 1D brickwork circuits $\mathfrak{B}_n^{\varepsilon}$ is similar to that for $\mathfrak{C}_n^{\varepsilon}$, but with a different decoupling error. The results are shown below.
\begin{theorem}[Decoupling with 1D brickwork circuits]\label{thm:1Dlocalrandomcircuit}
Given three quantum systems, $S = S_1S_2\cdots S_{n}$, $R = R_1R_2\cdots R_{n}$, and $E = E_1E_2\cdots E_{n}$, we denote $S_i$ as the subsystem of region $i$ inside $S$, and the same for $R_i$ and $E_i$. Here, $S_i$ is a qudit with local dimension $q$. Let quantum state $\rho_{SR}$ be
$\rho_{SR} = \bigotimes_{i=1}^{n}\rho_{S_iR_i}$ with a tensor-product structure and $\mathcal{T}_{S\rightarrow E}$ be a completely positive map. Denote the reduced density matrix of $\rho_{SR}$ on system $R$ as $\rho_R$ and the Choi-Jamio{\l}kowski representation of $\mathcal{T}_{S\rightarrow E}$ as $\tau_{SE}$ with reduced density matrix on system $E$ as $\tau_E$. Set $\eta = \frac{q}{q^2+1}$. Suppose $U_S$ is drawn from $\mathfrak{B}_n^{\varepsilon}$ with depth $d=\log\frac{n}{\varepsilon}+\frac{\log n}{\log \frac{1}{2\eta}} + \frac{\log(e-1)}{\log \frac{1}{2\eta}}+1$. We obtain that
\begin{equation}
\mathbb{E}_{U_S\sim \mathfrak{B}_n^{\varepsilon}}\Vert \mathcal{T}_{S\rightarrow E}(U_S\rho_{SR}U_{S}^{\dagger}) - \tau_E\otimes \rho_R \Vert_1
\leq \sqrt{2^{-H_2(S|E)_{\tau_{SE}}-H_2(S|R)_{\rho_{SR}}} + (\frac{\varepsilon}{n})^{\log\frac{1}{2\eta\rho_m}} \max_{A\subseteq [n]} 2^{-H_2(A|R)_{\rho_{AR}}-H_2(A|E)_{\tau_{AE}}}},
\end{equation}
where
\begin{equation}
\rho_m = \max_{i\in [n]}(2^{H_2(S_i|R_i)_{\rho_{S_iR_i}}}, 2^{-H_2(S_i|R_i)_{\rho_{S_iR_i}}}).
\end{equation}
\end{theorem}

\begin{proof}[Proof sketch]
The proof of Theorem~\ref{thm:1Dlocalrandomcircuit} is similar to that of Theorem~\ref{thm:nonsmoothdecoupling}. The key proof is to evaluate the second-order twirling $\mathbb{E}_U \tr\left[ \left( \widetilde{\mathcal{T}}_{S \to E}\left(U_S \widetilde{\rho}_{SR} U_S^{\dagger}\right) \right)^2 \right]$. To address this for $\mathfrak{B}$, we utilize the domain wall technique~\cite{Dalzell2022Anticoncentrate}, which enables us to evaluate the second moment twirling induced by 1D brickwork circuits. The full proof can be found in Appendix~\ref{appendssc:1DLRC}.
\end{proof}

% We remark that the property of $\mathfrak{C}_n^{\varepsilon}$ being an $\varepsilon$-approximate 2-design does not guarantee the decoupling of the two systems, due to the exponentially large factor in Lemma~\ref{lemma:appro2designdecoup}. Our refined analysis leverages the inner structure of $\mathfrak{C}_n^{\varepsilon}$ to eliminate this exponential factor.

With the non-smooth decoupling theorem, one can show that the Choi error can be polynomially small when $\varepsilon$ is polynomially small for Pauli noise, erasure errors, amplitude damping noise, and correlated noise, as demonstrated later. In this case, the depths of the double-layer blocked circuit $\mathfrak{C}_n^{\varepsilon}$ and 1D brickwork circuit $\mathfrak{B}_n^{\varepsilon}$ are still $O(\log n)$.

For the double-layer blocked ensemble $\mathfrak{C}$, we can further introduce a smooth decoupling theorem. Particularly, the smooth decoupling theorem replaces the conditional collision entropy $H_2$ with the smoothed conditional collision entropy $H_2^{\delta}$. The proof follows a standard smoothing technique, with a full derivation provided in Appendix~\ref{app:smoothing}. This smooth version allows us to demonstrate a higher achievable encoding rate for $\mathfrak{C}$ when the noise channel is local. The collision entropy $H_2^{\delta}({\rho^{\otimes \xi}})$ can approximate the tighter von Neumann entropy $H(\rho^{\otimes \xi})$ in the large $n$ limit. In contrast, this approach cannot increase the achievable encoding rate for 1D brickwork circuit $\mathfrak{B}$, since the upper bound contains the term $2^{H_2(S_iR_i)_{\rho_{S_iR_i}}}$ evaluated on a single qudit. In the following, we mainly present the results of $\mathfrak{C}$.

\begin{theorem}[Smooth decoupling theorem for 1D double-layer blocked circuits]\label{thm:smoothdecoupling}
With the same setting in Theorem~\ref{thm:nonsmoothdecoupling} and the tensor-product structure of $\mathcal{T}_{S\rightarrow E}$, for any $0 < \delta < 1$, we have that
\begin{equation}
\begin{split}
&\mathbb{E}_{U\sim \mathfrak{C}_n^{\varepsilon}}\Vert \mathcal{T}_{S\rightarrow E}(U_S\rho_{SR}U_{S}^{\dagger}) - \tau_E\otimes \rho_R \Vert_1\leq 16N\delta+\\
&\sqrt{2^{-\sum_{i=1}^{2N}H^{\delta}_{2}(S_i|E_i)_{\tau_{S_iE_i}}-\sum_{i=1}^{2N}H^{\delta}_{2}(S_i|R_i)_{\rho_{S_iR_i}}} + c'\prod_{i=1}^{2N} \max(1, 2^{-H^{\delta}_{2}(S_i|E_i)_{\tau_{S_iE_i}}-H^{\delta}_{2}(S_i|R_i)_{\rho_{S_iR_i}}})},
\end{split}
\end{equation}
where
\begin{equation}
c' = 2((1+\frac{2^{\xi+1}}{2^{2\xi}+1}\max_i(\max(2^{-H^{\delta}_{2}(S_i|R_i)_{\rho_{S_iR_i}}}, 2^{H^{\delta}_2(S_i|R_i)_{\rho_{S_iR_i}}})))^{N-1}-1).
\end{equation}
If we only smooth $H_2(S_i|E_i)$, we will obtain
\begin{equation}\label{eq:H2smoothSE}
\begin{split}
&\mathbb{E}_{U\sim \mathfrak{C}_n^{\varepsilon}}\Vert \mathcal{T}_{S\rightarrow E}(U_S\rho_{SR}U_{S}^{\dagger}) - \tau_E\otimes \rho_R \Vert_1\leq 8N\delta+\\
&\sqrt{2^{-\sum_{i=1}^{2N}H^{\delta}_{2}(S_i|E_i)_{\tau_{S_iE_i}}-\sum_{i=1}^{2N}H_{2}(S_i|R_i)_{\rho_{S_iR_i}}} + c\prod_{i=1}^{2N} \max(1, 2^{-H^{\delta}_{2}(S_i|E_i)_{\tau_{S_iE_i}}-H_{2}(S_i|R_i)_{\rho_{S_iR_i}}})},
\end{split}
\end{equation}
where
\begin{equation}
c = 2((1+\frac{2^{\xi+1}}{2^{2\xi}+1}\max_i(\max(2^{-H_2(S_i|R_i)_{\rho_{S_iR_i}}}, 2^{H_2(S_i|R_i)_{\rho_{S_iR_i}}})))^{N-1}-1).
\end{equation}
\end{theorem}

In the following, we introduce the results of applying the decoupling theorems to prove the AQEC performance of $\mathfrak{C}$ and $\mathfrak{B}$. We mainly present the results associated with $\mathfrak{C}$ since its proven achievable encoding rate is better. We also show the advantage of $\mathfrak{C}$ and $\mathfrak{B}$ compared to the block-encoding method in Section~\ref{ssc:block}.

\subsection{Pauli error}
We now analyze the performance of random codes under local Pauli noise. This type of noise is particularly important as it is frequently used to model noise encountered in quantum computing~\cite{Terhal2015QEC,Flammia2020Estimation}. Notably, any noise can be turned into Pauli noise by randomized compiling~\cite{Wallman2016compiling}. Specifically, we analyze the Pauli noise model with parameter $\Vec{p} = (p_I, p_X, p_Y, p_Z)$, as defined in Eq.~\eqref{eq:pauli_noise}. By applying the non-smooth decoupling theorem, we prove the following result.

\begin{corollary}[AQEC performance of $\mathfrak{C}$ for Pauli noise]\label{coro:pauli_nonsmoothing}
In the large $n$ limit, suppose $k$ and $\Vec{p}$ satisfy $k/n\leq 1-f(\Vec{p})$ and $(\frac{\varepsilon}{n})^{1-\frac{k}{n}} = o(\frac{1}{n})$, the expected Choi error of the random codes formed by $\mathfrak{C}_n^{\varepsilon}$ against the strength-$\Vec{p}$ local Pauli noise channel satisfies
\begin{equation}
\mathbb{E}_{U\sim \mathfrak{C}_n^{\varepsilon}}\epsilon_{\mathrm{Choi}} \leq \left(2^{-n(1-f(\Vec{p})-\frac{k}{n})} + \frac{4\varepsilon^{1-\frac{k}{n}}n^{\frac{k}{n}}}{\log (n/\varepsilon)}\right)^{\frac{1}{4}},
\end{equation}
where $f(\Vec{p})= 2\log (\sqrt{p_I} + \sqrt{p_X} + \sqrt{p_Y} + \sqrt{p_Z})$. That is, the $1D$ $O(\log n)$-depth double-layer blocked circuits can achieve an encoding rate $1 - f(\Vec{p})$.
\end{corollary}

\begin{proof}[Proof sketch]
To prove this corollary, we straightforwardly apply the non-smooth decoupling theorem or Theorem~\ref{thm:nonsmoothdecoupling}. We substitute the state $\tau_{S_iE_i}$ and $\rho_{S_iR_i}$ in Theorem~\ref{thm:nonsmoothdecoupling} with the Choi state of local Pauli noise and the maximally entangled state, respectively. For local Pauli noise, the conditional collision entropy $H_2(S|E)_{\tau_{SE}}$ is given by $n(1-f(\Vec{p}))$. Meanwhile, the entropy $H_2(S|R)_{\rho_{SR}}$ captures the entanglement of the initial state and is equal to $-k$, where $k$ denotes the number of Bell states in $\ketbra{\hat{\phi}}$. Then, we get the first term $2^{-n(1-f(\Vec{p})-\frac{k}{n})}$ in the square root. The second term comes from the bound of coefficient $c$ with standard inequalities. See Appendix~\ref{app:pauli_nonsmoothing} for the detailed proof.
\end{proof}

Corollary~\ref{coro:pauli_nonsmoothing} shows that the Choi error can be made polynomially small by choosing $\varepsilon$ to be polynomially small. The exponential decay term $2^{-n(1-f(\Vec{p})-\frac{k}{n})}$ corresponds to the result associated with random stabilizer codes, as discussed in Appendix~\ref{app:Clif}. The polynomial decay term arises from the approximation error of $\mathfrak{C}$ compared to the Haar-random ensemble. As a consequence, we show that 1D double-layer blocked $O(\log n)$-depth random circuits maintain a constant encoding and error-correcting rate under local Pauli noise. The encoding rate reaches $1 - 2 \log (\sqrt{p_I} + \sqrt{p_X} + \sqrt{p_Y} + \sqrt{p_Z})$ for local Pauli noise with a polynomially small Choi error. The polynomial decay of the Choi error is particularly beneficial when the logical state needs to be protected over extended periods.

Similar to Corollary~\ref{coro:pauli_nonsmoothing}, for the 1D brickwork circuits $\mathfrak{B}_n^{\varepsilon}$, we can also show that the encoding scheme can reach the encoding rate $1 - 2 \log (\sqrt{p_I} + \sqrt{p_X} + \sqrt{p_Y} + \sqrt{p_Z})$ for local Pauli noise. The difference is that the polynoimal decay term within the bracket changes to $(\frac{\varepsilon}{k})^{\frac{n}{k}-2}$. The full analysis is shown in Appendix~\ref{appendssc:1DLRCQEC}.

When we apply the smooth decoupling theorem, or Theorem~\ref{thm:smoothdecoupling}, we obtain a higher achievable encoding rate, known as the hashing bound. Specifically, we use Eq.~\eqref{eq:H2smoothSE} by smoothing only $H_2(S_i|E_i)$.

\begin{corollary}[Random codes from $\mathfrak{C}$ achieve the hashing bound for Pauli noise]\label{coro:pauli}
For any $\delta > 0$, in the large $n$ limit, suppose $k$ and $\Vec{p}$ satisfy $k/n\leq 1-h(\Vec{p})-\delta$ and $\lim_{n\rightarrow \infty} \frac{\log n}{\log \frac{n}{\varepsilon}} = 0$, the expected Choi error of the random codes formed by $\mathfrak{C}_n^{\varepsilon}$ against the strength-$\Vec{p}$ local Pauli noise channel satisfies
\begin{equation}
\mathbb{E}_{U\sim \mathfrak{C}_n^{\varepsilon}}\epsilon_{\mathrm{Choi}} \leq \sqrt{8\delta+\sqrt{2^{-n(1-h(\Vec{p})-\frac{k}{n}-\delta)} + \frac{4\varepsilon^{1-\frac{k}{n}}n^{\frac{k}{n}}}{\log (n/\varepsilon)}}}.
\end{equation}
Here, $h(\Vec{p}) = -p_I\log p_I - p_X\log p_X - p_Y\log p_Y - p_Z\log p_Z$. That is, the $1D$ $\omega(\log n)$-depth double-layer blocked circuits can achieve the hashing bound $1 - h(\Vec{p})$.
\end{corollary}

\begin{proof}[Proof sketch]
The proof of Corollary \ref{coro:pauli} closely follows the proof of Corollary \ref{coro:pauli_nonsmoothing}, with the key distinction being the evaluation of the smoothed entropy. By the asymptotic equipartition property (AEP) of conditional entropy, the smoothed collision entropy converges to the von Neumann entropy. For local Pauli noise, the von Neumann entropy $H(S|E)_{\tau_{SE}}$ is given by $n(1-h(\Vec{p}))$. The entropy $H_2(S|R)_{\rho_{SR}}$ captures the initial entanglement and is equal to $-k$. Combining these results, we obtain that the encoding rate achievable under local Pauli noise approaches the hashing bound. The detailed analysis is provided in Appendix~\ref{app:hashingbound}.
\end{proof}

Note that, compared to Corollary~\ref{coro:pauli_nonsmoothing}, Corollary~\ref{coro:pauli} introduces an additional requirement: $\lim_{n \to \infty} \frac{\log n}{\log \frac{n}{\varepsilon}} = 0$. This condition arises from the AEP when the smooth conditional collision entropy approximates the conditional Shannon entropy. The condition ensures that the approximation error in AEP diminishes as $n$ increases. It can be satisfied by choosing $\log (\frac{n}{\varepsilon})$ to be any function that grows slightly faster than $\log n$, such as $(\log n) \log \log \log n$.

While smoothing the conditional collision entropy introduces an additional error term of $8\delta$, the results remain significant. Since $\delta$ can be chosen arbitrarily small, we can take the limit $\delta \to 0$ and $n \to \infty$. In this limit, the Choi error approaches 0. Therefore, for Pauli noise parameterized by $\Vec{p} = (p_I, p_X, p_Y, p_Z)$, the Choi error diminishes to zero if the encoding rate $k/n$ is less than $1 - h(\Vec{p})$. Hence, we conclude that the $\omega(\log n)$-depth encoding circuit from $\mathfrak{C}$ can achieve the hashing bound. The hashing bound $1-h(\Vec{p})$ can be much higher than the previous bound $1-f(\Vec{p})$ from the non-smooth decoupling theorem.
% We provide a numerical figure to show the two bounds under strength-$p$ depolarizing noise in Fig.~\ref{fig:EncodingRateBound}. For a fixed error parameter, the encoding rate can differ over 0.5.

% \begin{figure}[htbp!]
% \centering
% \includegraphics[width=.5\textwidth]{./figure/EncodingRatevsErrorRate.pdf}
% \caption{The encoding rates of different bounds. For Pauli noise, the encoding rates given by the smooth and non-smooth decoupling theorems are bounded by $1-h(\Vec{p})$ and $1-f(\Vec{p})$, respectively. There are depicted in blue with solid and dashed lines, respectively. For the strength-$p$ depolarizing noise, the four parameters of the Pauli noise are $\vec{p} = (1-3p/4,p/4,p/4,p/4)$. The hashing bound $1-h(\Vec{p})$ is higher than the bound $1-f(\Vec{p})$. For erasure errors, the encoding rates given by the smooth and non-smooth decoupling theorems are bounded by $1-2p$ and $1-\log(1+3p)$ and depicted in red with dashed and solid lines, respectively. The bound $1-2p$ is the quantum channel capacity of the erasure error, and it is higher than the bound $1-\log(1+3p)$.}
% \label{fig:EncodingRateBound}
% \end{figure}

\subsection{Erasure error}
We now analyze the performance of random codes under erasure errors, an important error model extensively studied in prior work \cite{Gullans2021LowDepth, Kong2022CQEC} and usually used to assess the performance of quantum codes. We first consider strength-$p$ i.i.d.~erasure errors defined in Eq.~\eqref{eq:iiderasure}. The derivation is similar to that for Pauli noise, with the differences arising from the noise channel and the evaluation of $H_2(S|E)_{\tau_{SE}}$ and $H(S|E)_{\tau_{SE}}$. We provide the results below.

\begin{corollary}[AQEC performance of $\mathfrak{C}$ for i.i.d.~erasure error]\label{coro:iiderasure_nonsmoothing}
In the large $n$ limit, suppose $k$ and $p$ satisfy $k/n\leq 1-\log(1+3p)$ and $(\frac{\varepsilon}{n})^{1-\frac{k}{n}} = o(\frac{1}{n})$, the expected Choi error of the random codes formed by $\mathfrak{C}_n^{\varepsilon}$ against the strength-$p$ i.i.d.~erasure error satisfies
\begin{equation}
\mathbb{E}_{U\sim \mathfrak{C}_n^{\varepsilon}}\epsilon_{\mathrm{Choi}} \leq \left(2^{-n(1-\log(1+3p)-\frac{k}{n})} + \frac{4\varepsilon^{1-\frac{k}{n}}n^{\frac{k}{n}}}{\log (n/\varepsilon)}\right)^{\frac{1}{4}},
\end{equation}
That is, the $1D$ $O(\log n)$-depth double-layer blocked circuits can achieve an encoding rate $1 - \log(1+3p)$ for strength-$p$ erasure error.
\end{corollary}

The detailed proof is available in Appendix~\ref{app:iiderasure_nondecouple}. Corollary~\ref{coro:iiderasure_nonsmoothing} demonstrates that the Choi error can be made polynomially small when $\varepsilon$ is polynomially small. The exponential decay term $2^{-n(1-\log(1+3p)-\frac{k}{n})}$ arises from the result associated with random stabilizer codes, while the polynomial decay term comes from the approximation error of $\mathfrak{C}_n^{\varepsilon}$. Consequently, we show that 1D double-layer blocked $O(\log n)$-depth random circuits achieve an encoding rate of $1 - \log(1 + 3p)$ for i.i.d.~erasure error with a vanishing Choi error.

Similar to Corollary~\ref{coro:iiderasure_nonsmoothing}, for the 1D brickwork circuits $\mathfrak{B}_n^{\varepsilon}$, the encoding rate $1 - \log(1 + 3p)$ is also achievable for i.i.d.~erasure error with a small Choi error $O((\frac{\varepsilon}{k})^{\frac{n}{4k}-\frac{1}{2}})$. The full analysis is shown in Appendix~\ref{appendssc:1DLRCQEC}.

If we apply the smooth decoupling theorem, or Theorem~\ref{thm:smoothdecoupling}, we find that $\mathfrak{C}$ can further achieve the quantum channel capacity for the erasure error, which is $1 - 2p$. This result is shown below. Specifically, we use Eq.~\eqref{eq:H2smoothSE} by smoothing only $H_2(S_i|E_i)$.

\begin{corollary}[Random codes from $\mathfrak{C}$ achieve the quantum channel capacity for i.i.d.~erasure error]\label{coro:iiderasure}
Given $\delta > 0$, in the large $n$ limit, suppose $k$ and $p$ satisfy $k/n\leq 1-2p-\delta$ and $\lim_{n\rightarrow \infty} \frac{\log n}{\log \frac{n}{\varepsilon}} = 0$, the expected Choi error of the random codes formed by $\mathfrak{C}_n^{\varepsilon}$ against the strength-$p$ i.i.d.~erasure error satisfies
\begin{equation}
\mathbb{E}_{U\sim \mathfrak{C}_n^{\varepsilon}}\epsilon_{\mathrm{Choi}} \leq \sqrt{8\delta+\sqrt{2^{-n(1-2p-\frac{k}{n}-\delta)} + \frac{4\varepsilon^{1-\frac{k}{n}}n^{\frac{k}{n}}}{\log (n/\varepsilon)}}}.
\end{equation}
That is, the $1D$ $\omega(\log n)$-depth double-layer blocked circuits can achieve the erasure channel capacity $1 - 2p$.
\end{corollary}

The detailed proof is provided in Appendix~\ref{app:iiderasure}. The additional requirement $\lim_{n \to \infty} \frac{\log n}{\log \frac{n}{\varepsilon}} = 0$ arises from AEP, just as it does for Pauli noise. By taking the limits $\delta \to 0$ and $n \to \infty$, the Choi error diminishes to zero if the encoding rate $k/n$ is less than $1 - 2p$. Hence, we conclude that the $\omega(\log n)$-depth encoding circuit from $\mathfrak{C}$ can achieve the quantum channel capacity $1 - 2p$ for i.i.d.~erasure error.
% We also compare the two bounds $1-2p$ and $1-\log(1+3p)$ numerically in Fig.~\ref{fig:EncodingRateBound}.

In addition to i.i.d.~erasure, we also consider the fixed-number erasure error, which was studied in~\cite{Kong2022CQEC}. This noise model is defined in Eq.~\eqref{eq:random_t_erasure}, where the number of erasures is $t$. In this error model, a subsystem of size $t$ is randomly selected and traced out. Similar to Corollary~\ref{coro:iiderasure_nonsmoothing}, we show that these codes maintain a constant rate of information storage and error-correction capabilities. Specifically, given the logical qubit count $k$ and the number of erasure errors $t$, we have that for sufficiently large $n$, as long as $1 - \frac{k}{n} - \log(1+\frac{3t}{n}) \geq 0$, and $(\frac{\varepsilon}{n})^{1-\frac{k}{n}} = o(\frac{1}{n})$, the expected Choi error over $\mathfrak{C}_n^{\varepsilon}$ against the random $t$ erasure errors satisfies
\begin{equation}\label{eq:randomterasure}
\mathbb{E}\epsilon_{\mathrm{Choi}} \leq \left(2^{-(n-2t-k)}+\frac{4\varepsilon^{1-\frac{k}{n}}n^{\frac{k}{n}}}{\log (n/\varepsilon)}\right)^{\frac{1}{4}}.
\end{equation}

Note that $\varepsilon$ is only required to be polynomially small to satisfy $(\frac{\varepsilon}{n})^{1-\frac{k}{n}} = o(\frac{1}{n})$. Consequently, the circuit depth is $O(\log(\frac{n}{\varepsilon})) = O(\log n)$, and the Choi error decays polynomially as the number of qubits $n$ increases. The maximal number of erasure errors can be $\frac{n}{3}$. The full detail is available in Appendix~\ref{app:proof_erasure}.

For random stabilizer code, the Choi error can be bounded by the first term $2^{-\frac{n-2t-k}{4}}$. Thus, random stabilizer codes perform well when $n-2t-k > 0$, which saturates the quantum Singleton bound, $n-k-2(d-1)\geq 0$, for an $[n, k, d]$ code, as distance-$d$ code can guarantee to correct $d-1$ erasure errors~\cite{Cerf1997singleton}. Here, the codes from $\mathfrak{C}$ require the condition $n-k-n\log(1+\frac{3t}{n}) \geq 0$, which is close to the Singleton bound $n-k-2t\geq 0$ up to a constant factor. This raises a question regarding whether the low-depth codes can achieve the ``Singleton bound" for AQEC under random $t$-erasure errors. Meanwhile, notice that the condition $n-k-n\log(1+\frac{3t}{n}) \geq 0$ is in accord with the requirement when applying the non-smooth decoupling theorem to i.i.d.~erasure error. Also, applying the smooth decoupling theorem to i.i.d.~erasure error gives an encoding rate $1-2p$, which is exactly the Singleton bound. It is possible to enhance Eq.~\eqref{eq:randomterasure} by ``smoothing" the result, which we leave for future work.

\subsection{Amplitude damping noise}
We now analyze the performance of random codes under amplitude damping noise~\cite{Giovannetti2005amplitude}. The derivation is similar to that for Pauli noise and erasure errors, with differences in the noise channel and the evaluation of $H_2(S|E)_{\tau_{SE}}$ and $H(S|E)_{\tau_{SE}}$. We provide the results below.

\begin{corollary}[AQEC performance of $\mathfrak{C}$ for amplitude damping noise]\label{coro:amp_nonsmoothing}
In the large $n$ limit, suppose $k$ and $p$ satisfy $k/n\leq -\log( \frac{1}{2-p} + \sqrt{\frac{p}{2-p}} )$ and $(\frac{\varepsilon}{n})^{1-\frac{k}{n}} = o(\frac{1}{n})$, the expected Choi error of the random codes formed by $\mathfrak{C}_n^{\varepsilon}$ against the strength-$p$ local amplitude damping noise satisfies
\begin{equation}
\mathbb{E}_{U\sim \mathfrak{C}_n^{\varepsilon}}\epsilon_{\mathrm{Choi}} \leq \left(2^{-n(-\log( \frac{1}{2-p} + \sqrt{\frac{p}{2-p}} )-\frac{k}{n})} + \frac{4\varepsilon^{1-\frac{k}{n}}n^{\frac{k}{n}}}{\log (n/\varepsilon)}\right)^{\frac{1}{4}},
\end{equation}
That is, the $1D$ $O(\log n)$-depth double-layer blocked circuits can achieve an encoding rate $-\log( \frac{1}{2-p} + \sqrt{\frac{p}{2-p}} )$ for strength-$p$ local amplitude damping noise.
\end{corollary}

The detailed proof is available in Appendix~\ref{app:amp_nondecouple}. Corollary~\ref{coro:amp_nonsmoothing} demonstrates that the Choi error can be made polynomially small when $\varepsilon$ is polynomially small.

Similar to Corollary~\ref{coro:amp_nonsmoothing}, for the 1D brickwork circuits $\mathfrak{B}_n^{\varepsilon}$, the encoding rate $-\log( \frac{1}{2-p} + \sqrt{\frac{p}{2-p}} )$ is also achievable for local amplitude damping noise with a small Choi error $O((\frac{\varepsilon}{k})^{\frac{n}{4k}-\frac{1}{2}})$. The full analysis is shown in Appendix~\ref{appendssc:1DLRCQEC}.

If we apply the smooth decoupling theorem, we find that $\mathfrak{C}$ can achieve a better encoding rate $h(\frac{1-p}{2})-h(\frac{p}{2})$. This result is shown below.

\begin{corollary}[Better AQEC performance of $\mathfrak{C}$ for amplitude damping noise]\label{coro:amp}
Given $\delta > 0$, in the large $n$ limit, suppose $k$ and $p$ satisfy $k/n\leq h(\frac{1-p}{2})-h(\frac{p}{2})-\delta$ and $\lim_{n\rightarrow \infty} \frac{\log n}{\log \frac{n}{\varepsilon}} = 0$, the expected Choi error of the random codes formed by $\mathfrak{C}_n^{\varepsilon}$ against the strength-$p$ local amplitude damping noise satisfies
\begin{equation}
\mathbb{E}_{U\sim \mathfrak{C}_n^{\varepsilon}}\epsilon_{\mathrm{Choi}} \leq \sqrt{8\delta+\sqrt{2^{-n(h(\frac{1-p}{2})-h(\frac{p}{2})-\frac{k}{n}-\delta)} + \frac{4\varepsilon^{1-\frac{k}{n}}n^{\frac{k}{n}}}{\log (n/\varepsilon)}}}.
\end{equation}
That is, the $1D$ $\omega(\log n)$-depth double-layer blocked circuits can achieve the encoding rate $h(\frac{1-p}{2})-h(\frac{p}{2})$.
\end{corollary}

The detailed proof is provided in Appendix~\ref{app:ampsmooth}. The bound of the encoding rate is $h(\frac{1-p}{2})-h(\frac{p}{2})$, which is the same as the bound for random stabilizer codes. Nonetheless, this rate may be lower than the quantum capacity of local amplitude damping noise given in Eq.~\eqref{eq:ampcapacity}. The discrepancy arises because both $\mathfrak{C}$ and random stabilizer codes are symmetric when interchanging the computational basis, whereas amplitude damping noise is inherently asymmetric.

\subsection{Correlated noise}
In the above, we show the AQEC performance of 1D log-depth circuits for local noises. Below, we demonstrate that our analysis extends beyond this type. By applying Theorem~\ref{thm:nonsmoothdecoupling}, we show that codes from $\mathfrak{C}$ can achieve a constant rate of encoding against correlated noise. The derivation is similar to that for local noise, with the differences arising from the noise channel and the evaluation of $H_2(S|E)_{\tau_{SE}}$. The result is summarized below.

\begin{corollary}[Random codes from $\mathfrak{C}$ possess a threshold for 1D nearest neighbor $ZZ$-coupling noise]\label{coro:zzcouple}
In the large $n$ limit, suppose $k$ and $p$ satisfy $k/n < 1-2\log(\sqrt{1-p}+\sqrt{p})$ and $(\frac{\varepsilon}{n})^{1-\frac{k}{n}} = o(\frac{1}{n})$, the expected Choi error of the random codes formed by $\mathfrak{C}_n^{\varepsilon}$ against the strength-$p$ nearest neighbor $ZZ$-coupling noise satisfies
\begin{equation}
\mathbb{E}_{U\sim \mathfrak{C}_n^{\varepsilon}}\epsilon_{\mathrm{Choi}} \leq \left(2^{-n(1-2\log(\sqrt{1-p}+\sqrt{p})-\frac{k}{n})} + \frac{4\varepsilon^{1-\frac{k}{n}}n^{\frac{k}{n}}}{\log (n/\varepsilon)}\right)^{\frac{1}{4}}.
\end{equation}
That is, the $1D$ $O(\log n)$-depth double-layer blocked circuits can achieve an encoding rate $1 - 2\log(\sqrt{1-p}+\sqrt{p})$ for strength-$p$ 1D nearest neighbor $ZZ$-coupling noise.
\end{corollary}

The detailed proof is available in Appendix~\ref{app:proof_corr}. Corollary~\ref{coro:zzcouple} demonstrates that the Choi error is polynomially small when $\varepsilon$ is polynomially small. The exponential decay term $2^{-n(1-2\log(\sqrt{1-p}+\sqrt{p})-\frac{k}{n})}$ is associated with random stabilizer codes. Consequently, we show that 1D double-layer blocked $O(\log n)$-depth random circuits achieve an encoding rate of $1-2\log(\sqrt{1-p}+\sqrt{p})$ for nearest-neighbor $ZZ$-coupling noise with a polynomially small Choi error. For 1D brickwork circuit $\mathfrak{B}$, we prove a similar result with a modification of the bound to $1-2(1+\frac{k}{n})\log(\sqrt{1-p}+\sqrt{p})$. The modification originates from a finite local dimension for $\mathfrak{B}$. The full analysis is shown in Appendix~\ref{appendssc:1DLRCQEC}.

\subsection{Advantages of the 1D double-layer blocked encoding scheme over the block-encoding scheme}\label{ssc:block}
In addition to proving that our 1D circuit encoding scheme can achieve a high AQEC performance against local noises and correlated noise, we now specifically discuss its advantage over the block-encoding scheme.

Previous works show that the one-dimensional block-encoding scheme can have a threshold against random erasure error~\cite{Gullans2021LowDepth} and can achieve the hashing bound against Pauli noise~\cite{Darmawan2024Lowdepth} at a logarithmic depth. Here, we show that although the achievable encoding rate for the block-encoding method is the same as that for the double-layer blocked circuit $\mathfrak{C}$, $\mathfrak{C}$ can have advantages in the required circuit depth and logical error rate.

Particularly, we will show that to correct a given noise with a fixed encoding rate, the required circuit depth of $\mathfrak{C}$ is lower, and the Choi error is lower for $\mathfrak{C}$. In other words, given certain noise channel the circuit depth the double-layer blocked encoding scheme can accomplish error correction but the block-encoding scheme cannot.

Using our notation, the encoding unitary operation of the block-encoding method can be written as
\begin{equation}
U = \bigotimes_{i=1}^{\frac{N}{2}}U_C^{4i-3,4i-2,4i-1,4i},
\end{equation}
where $U_C^{4i-3,4i-2,4i-1,4i}$ is a random Clifford gate on regions $4i-3,4i-2,4i-1,4i$. The location of regions is shown in Fig.~\ref{fig:1Dlowdepthcircuit}. The encoding unitary $U$ is a tensor product of parallel Clifford gates. Note that the depth to realize an $n$-qubit Clifford gate depends linearly on $n$. To make a fair comparison, we let the block size be $4\xi$ so that the depths to realize the block-encoding and the double-layer blocked encoding schemes are the same. We denote the set of the random block-encoding unitary operations as $\mathfrak{S}_n^B$.

In Appendix~\ref{app:block}, we compute the Choi error of the block-encoding scheme and compare it to the double-layer blocked encoding scheme. For local erasure error, we rigorously establish an advantage of the double-layer blocked encoding scheme over the block-encoding scheme. We evaluate the lower bound of the Choi error for the block-encoding method as follows.

\begin{equation}\label{eq:blocklowerbound}
\begin{split}
\mathbb{E}_{U\sim \mathfrak{S}_n^B}\epsilon_{\mathrm{Choi}}
\geq \min \left\{ n \left(\frac{\varepsilon}{n}\right)^{8D( \frac{1-\frac{k}{n}}{2} \| p )+o(1)}, \Omega(1) \right\},
\end{split}
\end{equation}
where $D$ is the binary relative entropy. Particularly, the $\Omega(1)$ term is a positive constant. When $n \left(\frac{\varepsilon}{n}\right)^{8D( \frac{1-\frac{k}{n}}{2} \| p )+o(1)}$ vanishes, the Choi error also vanishes. If $n \left(\frac{\varepsilon}{n}\right)^{8D( \frac{1-\frac{k}{n}}{2} \| p )+o(1)}$ does not converge, the Choi error will have a non-vanishing value, and the error correction fails.

Recalling Corollary~\ref{coro:iiderasure_nonsmoothing}, the Choi error for the double-layer blocked encoding scheme is:
\begin{equation}\label{eq:twolayerupperbound}
\mathbb{E}_{U\sim \mathfrak{C}_n^{\varepsilon}}\epsilon_{\mathrm{Choi}} \leq \left(2^{-n(1-\log(1+3p)-\frac{k}{n})} + \frac{4\varepsilon^{1-\frac{k}{n}}n^{\frac{k}{n}}}{\log (n/\varepsilon)}\right)^{\frac{1}{4}}.
\end{equation}

Based on Eqs.~\eqref{eq:blocklowerbound} and~\eqref{eq:twolayerupperbound}, there exists regimes of $p$, encoding rate $\frac{k}{n}$, and $\varepsilon$ such that
\begin{align}
\label{eq:blockinfty}\lim_{n\rightarrow \infty}\mathbb{E}_{U\sim \mathfrak{S}_n^B}\epsilon_{\mathrm{Choi}}
&\geq\lim_{n\rightarrow \infty} \min\left\{n \left(\frac{\varepsilon}{n}\right)^{8D( \frac{1-\frac{k}{n}}{2} \| p )+o(1)}, \Omega(1) \right\} = \Omega(1)\\
\label{eq:twolayervanish}\lim_{n\rightarrow \infty}\mathbb{E}_{U\sim \mathfrak{C}_n^{\varepsilon}}\epsilon_{\mathrm{Choi}}&\leq  \lim_{n\rightarrow \infty}\left(2^{-n(1-\log(1+3p)-\frac{k}{n})} + \frac{4\varepsilon^{1-\frac{k}{n}}n^{\frac{k}{n}}}{\log (n/\varepsilon)}\right)^{\frac{1}{4}}= 0.
\end{align}
A explicit example is when $p = 0.25$, $\frac{k}{n} = 0.2$, $\varepsilon = n^{-0.375}$. In this case,
\begin{align}
\mathbb{E}_{U\sim \mathfrak{S}_n^B}\epsilon_{\mathrm{Choi}}
&\geq \min\{n^{0.14}, \Omega(1) \} = \Omega(1);\\
\mathbb{E}_{U\sim \mathfrak{C}_n^{\varepsilon}}\epsilon_{\mathrm{Choi}}&\leq n^{-0.025}.
\end{align}
Thus, given $p = 0.25$ and $\frac{k}{n} = 0.2$, the block-encoding method requires a smaller $\varepsilon$ to make the Choi error vanish, inducing a deeper circuit depth. Hence, we demonstrate that the required circuit depth of $\mathfrak{C}_n^{\varepsilon}$ is lower, and the Choi error is also lower.

We show the full analysis of the regime where Eqs.~\eqref{eq:blockinfty} and~\eqref{eq:twolayervanish} hold in the Appendix. One can also find the regime where $\lim_{n\rightarrow \infty}(\mathbb{E}_{U\sim \mathfrak{C}_n^{\varepsilon}}\epsilon_{\mathrm{Choi}}/\mathbb{E}_{U\sim \mathfrak{S}_n^B}\epsilon_{\mathrm{Choi}}) = 0$. In this regime, the double-layer blocked scheme outperforms the separate block approach. Thus, the double-layer blocked encoding scheme is proven to have an advantage in certain error regimes against erasure errors. We expect this result to hold for more generic noise models, which is numerically verified for local Pauli noise~\cite{Darmawan2024Lowdepth}. We provide more details and discussions in Appendix~\ref{app:block}.

The above comparison can be generalized for 1D brickwork circuits $\mathfrak{B}$ and block-encoding ensemble $\mathfrak{S}_n^B$, where the former is upper bounded by $O((\frac{\varepsilon}{k})^{\frac{n}{4k}-\frac{1}{2}})$. Evaluating the regime where this upper bound is smaller than $n \left(\frac{\varepsilon}{n}\right)^{8D( \frac{1-\frac{k}{n}}{2} \| p )+o(1)}$ suffices for showing the advantage for $\mathfrak{B}$. Clearly, when the encoding rate $\frac{k}{n}$ is a small constant, $\mathfrak{B}$ will outperform the block approach.

One can observe that the advantages of $\mathfrak{C}$ and $\mathfrak{B}$ are both manifested in the low encoding rate regime. This phenomenon is expected to be related to the stronger scrambling power of $\mathfrak{C}$ and $\mathfrak{B}$ compared to $\mathfrak{S}_n^B$.
If the encoding rate is large, like 0.5, a one-layer gate suffices to scramble the logical information to the whole system to protect it. A stronger scrambling power will be useless. In contrast, when the logical qubit number is small, $\mathfrak{C}$ and $\mathfrak{B}$ can scramble their information to the whole system faster and then manifest an advantage.

\subsection{Approximate quantum error correction and decoupling for arbitrary input states}
We have previously established decoupling theorems for codes from $\mathfrak{C}$ and $\mathfrak{B}$ with input states that exhibit a product structure across different regions. Using these theorems, we demonstrated the AQEC performance of random codes with the maximally entangled state, corresponding to the case of random inputs. In this part, we extend these results to show that similar decoupling theorems and error-correcting properties hold for arbitrary input states.

We begin by noting that if a unitary encoding $U$ achieves a small Choi error, it also produces a comparably small recovery error for any input state when the encoding operation is combined with a local random unitary operation. Specifically, we select the local random unitary to be a random Pauli matrix and define the encoding and decoding operations as
\begin{align}
\mathcal{E}_P(\psi) &= U[(P\psi P^{\dagger})_L \otimes \ketbra{0}_{S\backslash L}] U^{\dagger},\\
\mathcal{D}_P(\psi) &= P\mathcal{D}(\psi)P^{\dagger}.
\end{align}
Such a random Pauli encoding twirls the original channel $\mathcal{D}\circ\mathcal{N}\circ\mathcal{E}$ into a Pauli channel, ensuring that the Choi error provides a lower bound on the recovery error for arbitrary input states. We establish this result using entanglement fidelity.

\begin{proposition}[Recovery fidelity via Pauli twirling] \label{prop:pauli_twirling}
Consider an encoding and decoding procedure $\mathcal{D}\circ\mathcal{N}\circ\mathcal{E}$ with entanglement fidelity $F_{\mathrm{Choi}}$. For any input state $\psi_{LR}$, we have
\begin{equation}
\mathbb{E}_P F(\psi_{LR}, [(\mathcal{D}_P\circ\mathcal{N}\circ\mathcal{E}_P)_L \otimes I_R](\psi_{LR})) \ge F_{\mathrm{Choi}},
\end{equation}
where the expectation $\mathbb{E}_P$ is taken over all Pauli matrices $P$ uniformly at random.
\end{proposition}

\begin{proof}[Proof sketch]
Applying Uhlmann’s theorem \cite{Uhlmann1976transition, Beny2010AQEC}, we express the fidelity of the two states in terms of the fidelity between corresponding purifications. By performing a Pauli twirl on these purifications using random Pauli matrices, we reduce the expression to the entanglement fidelity, thereby establishing the lower bound.  See Appendix \ref{app:arbitrary_input} for the detailed derivation.
\end{proof}

The encoding and decoding procedures using random Pauli matrices are independent of the input state, thus, they consistently provide high recovery fidelity for arbitrary input states. This quality offers additional insight into the AQEC performance of random codes. We define a unitary ensemble to be right-invariant under Pauli matrices if, for any Pauli gate $P$ and unitary $U \in \mathfrak{U}$, the probability $p_{\mathfrak{U}}(U)$ of drawing $U$ from $\mathfrak{U}$ is the same as the probability $p_{\mathfrak{U}}(UP)$ of drawing $UP$ from $\mathfrak{U}$:
\begin{equation}
p_{\mathfrak{U}}(U) = p_{\mathfrak{U}}(UP).
\end{equation}
The Pauli twirling operation is naturally incorporated in a right-invariant ensemble without incurring additional overhead. Consequently, we expect such an ensemble to exhibit robust AQEC performance for any input state. Consider the recovery error $\epsilon_{\psi}$ for an arbitrary input state $\psi$:
\begin{equation}
\epsilon_{\psi} = \min_{\cD} P(\psi_{LR}, [(\cD \circ \cN \circ \cE)_L \otimes I_R] (\psi_{LR}))
\end{equation}
A direct consequence of Proposition \ref{prop:pauli_twirling} is that the expected recovery error for any arbitrary input state is bounded above by the expected Choi error for any right-invariant unitary ensemble $\mathfrak{U}$.

\begin{corollary} \label{cor:recover_error_arbitrary}
If the unitary ensemble $\mathfrak{U}$ is right-invariant under Pauli rotations, then for any state $\psi_{LR}$, the average error $\bE_{\mathfrak{U}} \epsilon_{\psi}$ does not exceed the average Choi error:
\begin{equation}\label{eq:epsilon_psi_le_Choi}
\bE_{\mathfrak{U}} \epsilon_{\psi} \le \bE_{\mathfrak{U}} \epsilon_{\mathrm{Choi}}.
\end{equation}
\end{corollary}

\begin{proof} This comes from direct calculation:
\begin{equation}
\begin{split}
\bE_{\mathfrak{U}} \epsilon_{\psi} &= \bE_{\mathfrak{U}}  \min_{\cD} P(\psi_{LR}, [(\cD \circ \cN \circ \cE)_L \otimes I_R] (\psi_{LR})) \\
&=  \bE_{\mathfrak{U}}\bE_{P}   \min_{\cD} P(\psi_{LR}, [(\cD \circ \cN \circ \cE)_L \otimes I_R] (P\psi_{LR}P^{\dagger})) \\
&\le \bE_{\mathfrak{U}}\bE_{P}   P(\psi_{LR}, [(\cD_P^{(0)} \circ \cN \circ \cE)_L \otimes I_R] (P\psi_{LR}P^{\dagger}))\\
&= \bE_{\mathfrak{U}}\bE_{P}   P(\psi_{LR}, [(\cD_P^{(0)} \circ \cN \circ \cE_P)_L \otimes I_R] (\psi_{LR}))\\
&\le \bE_{\mathfrak{U}} \epsilon_{\mathrm{Choi}}.
\end{split}
\end{equation}
The second line follows from the fact that $\mathfrak{U}$ is right-invariant under Pauli matrices. In the third line, we choose $\cD^{(0)}_{P}(\psi) = P \mathcal{D}^{(0)}(\psi)P^{\dagger}$ as the recovery channel, where $\mathcal{D}^{(0)}$ is the optimal recovery channel that maximizes the entanglement fidelity of  $\cD \circ \mathcal{N} \circ \mathcal{E}$. This inequality holds because the specific recovery channel $\cD^{(0)}_{P}$ can not have a smaller recovery error than the optimal recovery channel.  The fifth line comes from Proposition \ref{prop:pauli_twirling}.
\end{proof}

The recovery error for arbitrary input state established in Corollary \ref{cor:recover_error_arbitrary} is closely related to the decoupling between the reference system and environment, as described by the complementary channel formalism (see Eqs.~\eqref{eq:choicomplementary},~\eqref{eq:zeta_complementary}). The recovery error can be expressed as:
\begin{equation}\label{eq:decoupling_arbitrary_state}
\epsilon_{\psi} = \min_{\zeta} P\left( (\widehat{\mathcal{N}\circ \mathcal{E}}_{L\rightarrow E}\otimes I_R)(\psi_{LR}), \zeta_E\otimes \tr_L(\psi_{LR}) \right),
\end{equation}
where a low value of $\epsilon_{\psi}$ indicates that $(\widehat{\mathcal{N} \circ \mathcal{E}}_{L \rightarrow E} \otimes I_R)(\psi_{LR})$ is close to a tensor product state $\zeta_E \otimes \tr_L(\psi_{LR})$. By Combining Eq.~\eqref{eq:decoupling_arbitrary_state}, Eq.~\eqref{eq:tracedistanceineq} with Corollary~\ref{cor:recover_error_arbitrary}, we conclude that the environment is approximately decoupled from the reference system, with the closeness bounded by the Choi error.

\begin{proposition}[Decoupling for arbitrary input states]\label{prop:decoupling}
If the unitary ensemble $\mathfrak{U}$ is right-invariant under Pauli rotations, then for any state $\psi_{LR}$, the expected closeness of decoupling is bounded above by the expected Choi error:
\begin{equation}
\begin{split}
&\mathbb{E}_{U_S \sim \mathfrak{U}} \min_{\tau_E}\Vert \hat{\mathcal{N}}_{S\rightarrow E}[U_S(\psi_{LR}\otimes \ketbra{0}_{S\backslash L})U_{S}^{\dagger}] - \tau_E\otimes \rho_R \Vert_1 \le 2 \mathbb{E}_{U_S \sim \mathfrak{U}} \epsilon_{\mathrm{Choi}}.
\end{split}
\end{equation}
Here, $\epsilon_\mathrm{Choi}$ represents the Choi error for the unitary encoding $U_S$ under the noise channel $\mathcal{N}$. The coefficient $2$ comes from the relationship between the $1$-norm distance and the purified distance.
\end{proposition}

Note that the ensembles $\mathfrak{C}$ and $\mathfrak{B}$ are right-invariant under Pauli rotations. This decoupling result extends the established decoupling theorems, Theorem~\ref{thm:nonsmoothdecoupling}, Theorem~\ref{thm:1Dlocalrandomcircuit}, and Theorem~\ref{thm:smoothdecoupling}, by eliminating the requirement for the input state to be in a tensor-product form. Different from the previous two theorems, the environmental state $\tau_E$ does not have an explicit form in this generalization.

Note that the state in Proposition \ref{prop:pauli_twirling} is arbitrary. One can take the minimal one over all states. Then, $\min_{\psi} \mathbb{E}_P F(\psi_{LR}, [(\mathcal{D}_P\circ\mathcal{N}\circ\mathcal{E}_P)_L \otimes I_R](\psi_{LR}))$ is close to 1 when the Choi fidelity is close to 1. We remark that although the encoding and decoding procedure $\mathcal{E}_P$ and $\mathcal{D}_P$ do not depend on the input state, the recovery fidelity derived in Proposition \ref{prop:pauli_twirling} do not provide a bound for the normal worst-case fidelity or worst-case error.
The definition of the worst-case error requires the encoding and decoding operations to be fixed, without any probability mixing:
\begin{equation}
\begin{split}
F_{\mathrm{worst}} &= \max_{\mathcal{D}}\min_{R,\rho_{LR}}F\left(\rho_{LR}, [(\mathcal{D}\circ\mathcal{N}\circ\mathcal{E})_L\otimes I_R](\rho_{LR})\right). \\
\epsilon_{\mathrm{worst}} &= \sqrt{1 - F^2_{\mathrm{worst}}}\\
&= \min_{\mathcal{D}}\max_{R,\rho_{LR}} P\left(\rho_{LR}, [(\mathcal{D}\circ\mathcal{N}\circ\mathcal{E})_L\otimes I_R](\rho_{LR})\right).
\end{split}
\end{equation}
That is, $F_{\mathrm{worst}}$ is defined for a specific encoding $\mathcal{E}$. In contrast, $\min_{\psi} \mathbb{E}_P F(\psi_{LR}, [(\mathcal{D}_P\circ\mathcal{N}\circ\mathcal{E}_P)_L \otimes I_R](\psi_{LR}))$ considers an ensemble of encodings and decodings $\{\mathcal{E}_P, \mathcal{D}_P\}_P$ and shows that this ensemble achieves high expected recovery fidelity for any input state. The worst-case fidelity is not guaranteed for a particular choice of $\mathcal{E}_P$.

Nonetheless, when the Choi error is small, it can be shown that there exists a code subspace that exhibits good worst-case error for a fixed encoding procedure. Specifically, we use the code space $C = \mathrm{span}\{U_S\ket{i}_L\ket{0}_{S\backslash L}\}$ spanned by the logical state to represent a coding scheme, and the worst-case error is only dependent on $C$. It has been shown that when the Choi fidelity of the original encoding is $1 - \epsilon$, there exists a large subcode $C' \subset C$ such that:
\begin{equation}
\dim C' \ge \lfloor \frac{1}{2}\dim C \rfloor,
\end{equation}
and the worst-case fidelity for this subcode satisfies \cite{Klesse2007AQEC}:
\begin{equation}
F_{\mathrm{worst}}(C') > 1 - 2\epsilon.
\end{equation}
This implies that as long as the Choi error of the encoding is small, the worst-case error in a large subspace $C' \subset C$ is also small, thus ensuring good error correction performance over a significant portion of the code space.

\section{Lower bound of circuit depth in AQEC}\label{sec:lower_bound}

After establishing the AQEC performance of a logarithmic-depth circuit for local noise, we now derive lower bounds on the circuit depth required to achieve good AQEC performance. Here, we focus on the lower bound for local Pauli noise, specifically i.i.d.~depolarizing noise, which can be easily extended to other local noise models. We show that the scaling of circuit depth in the double-layer blocked encoding $\mathfrak{C}$ and 1D brickwork encoding $\mathfrak{B}$ is asymptotically optimal. Our analysis also provides a general result encompassing both $D$-dimensional and all-to-all circuit architectures.

\begin{theorem}[Lower bound for circuit depth]\label{thm:depolarizing_lower_bound}
For a circuit $U$ that encodes $k$ logical qubits with depth $d$, if $U$ has a Choi error $\epsilon < 0.1$ against the strength-$p$ local depolarizing channel, then:
\begin{enumerate}
\item {If $U$ is a D-dimensional circuit, then
\begin{equation}
(2d)^D \log\left(\frac{1}{p}\right) + 2D \log(2d) \geq \log\left(\frac{3}{8 \epsilon_{\mathrm{Choi}}^2}\right) + \log k.
\end{equation}
In the large $n$ and $d$ limit, this implies
\begin{equation}
d \geq \frac{1}{2} \left[\frac{\log\left(\frac{3}{8\epsilon^2}\right) + \log k}{(1+o(1))\log\left(\frac{1}{p}\right)}\right]^{1/D}.
\end{equation}}
\item {If $U$ is an all-to-all circuit, then
\begin{equation}
2^d \log\left(\frac{1}{p}\right) + 2d \geq \log\left(\frac{3}{8 \epsilon_{\mathrm{Choi}}^2}\right) + \log k.
\end{equation}
In the large $n$ and $d$ limit, this implies
\begin{equation}
d \geq \log\left[\frac{\log\left(\frac{3}{8\epsilon^2}\right) + \log k}{(1+o(1))\log\left(\frac{1}{p}\right)}\right].
\end{equation}}
\end{enumerate}
\end{theorem}

\begin{proof}[Proof Sketch]
We outline the main idea of the proof and refer to Appendix \ref{app:lower_bound} for detailed arguments. The proof leverages the analysis of the light cones of each qubit within the circuit. For a circuit of depth \( d \), the light cone of each qubit is bounded by a size \( M \) that depends on the circuit depth $d$ and circuit architecture. For instance, in a 1D circuit, \( M = 2d \).

We first prove that there exist at least \( \frac{k}{M^2} \) logical qubits with disjoint light cones, and we denote this subset as \( J \). That is, the light cones of any two qubits in $J$ do not intersect with each other. If any of these light cones are totally affected by noise, then any recovery channel can not have a small recovery error. Since the light cones in \( J \) are disjoint, the events of them being totally affected are independent. The probability that none of these light cones are completely traced out and replaced by the maximally mixed state is upper bounded by
\begin{equation}
p_{\mathrm{None}} \le (1 - p^M)^{|J|}.
\end{equation}

We also prove that the Choi error is lower bounded by the probability that at least one of these light cones is entirely traced out:
\begin{equation}
\epsilon^2 \geq \frac{3}{4} \left(1 - p_{\mathrm{None}}\right).
\end{equation}

Using standard inequalities and substituting \( M \) with a function of the circuit depth \( d \) that depends on the circuit's structure, we derive the desired lower bounds on the circuit depth.
\end{proof}

The light cone argument has been utilized in previous works to establish lower bounds for quantum error-correcting codes, such as determining the number of physical qubits required for fault-tolerant quantum computation \cite{baspin2023lowerboundoverheadquantum} and the circuit depth necessary to achieve specific subsystem variance properties of codes \cite{Yi2024order}. In our work, we extend these lower bounds by considering the AQEC performance under a practical noise model, specifically establishing circuit depth lower bounds under local noise. While our results are derived for local depolarizing noise, they can be easily extended to other error models. Furthermore, our analysis applies to arbitrary circuit architectures characterized by an adjacency graph \(\mathsf{G}\). In such cases, the upper bound \(M\) on the size of the light cones depends on both the circuit depth and the adjacency graph, \(M = M(\mathsf{G}, d)\), while the remaining parts of the proof remain unchanged.

Using Theorem \ref{thm:depolarizing_lower_bound}, we now show that the circuit depth scaling of $\mathfrak{C}$ and $\mathfrak{B}$ is optimal up to constant factors. Theorem~\ref{thm:depolarizing_lower_bound} shows that the depth of a 1D circuit should be $\Omega(\log n)$ to achieve a polynomially small Choi error or a constant encoding rate. Meanwhile, 1D circuits  $\mathfrak{C}$ and $\mathfrak{B}$ can achieve such error scaling and encode \(k = \Theta(n)\) logical qubits with circuit depth $O(\log n)$. The depth is optimal up to a constant factor.

Moreover, 1D double-layer blocked circuits $\mathfrak{C}$ are nearly optimal for achieving the hashing bound, as the circuit depth $\omega(\log n)$ given in Corollary \ref{coro:pauli} is close to the lower bound $\Omega(\log n)$ given by Theorem \ref{thm:depolarizing_lower_bound}. The scaling $\omega(\log n)$ arises from suppressing errors in AEP concentration, suggesting that improving the circuit depth scaling to $O(\log n)$ may be inherently challenging. Exploring potential methods to enhance this bound is an interesting question for future research.

Moreover, we can implement the codes from \(\mathfrak{C}\) using an all-to-all circuit architecture to further reduce the circuit depth. In an all-to-all architecture, any \(\xi\)-qubit Clifford unitary can be implemented in \(O(\log \xi)\) depth with \(O(\xi^2)\) ancillary qubits \cite{Moore2001Parallel, Jiang2020OptimalSpacedepth}. Consequently, the Clifford encoding \(\mathfrak{C}\) can be realized with a depth of \(O(\log \xi) = O\left(\log \log \left(\frac{n}{\epsilon}\right)\right)\), utilizing a total of \(O(n\xi) = O\left(n \log\left(\frac{n}{\epsilon}\right)\right)\) ancillary qubits. This encoding scheme effectively encodes \(k = O(n)\) logical qubits with a circuit depth of \( O\left(\log \log \left(\frac{n}{\epsilon}\right)\right)\), which is comparable to the circuit depth lower bound \(\Omega(\log \log n)\) established in Theorem \ref{thm:depolarizing_lower_bound}. This comparison indicates that our results provide a tight lower bound applicable to a broad range of circuit architectures.
Notably, in such an all-to-all encoding strategy, the encoding rate becomes \(O\left(\frac{k}{n\xi}\right) = O\left(\log^{-1} \left(\frac{n}{\epsilon}\right)\right)\), which has a gap to a constant encoding rate. Whether all-to-all circuits with \(O(\log \log n)\) depth can achieve a constant encoding rate while maintaining good AQEC performance remains an intriguing direction for future research.

\section{Conclusion and discussion}\label{sc:discuss}
In this work, we establish the AQEC performance of 1D logarithmic-depth encoding circuits with brickwork and double-layer blocked architecture against a wide range of practical noise models. We show that logarithmic depth is both necessary and sufficient to achieve a constant encoding rate and error-correcting capability while keeping the AQEC error polynomially small. In particular, we prove that for the double-layer blocked circuit, $\omega(\log n)$ depth suffices to achieve the well-known hashing bound under local Pauli noise and the channel capacity under local erasure errors. %It is interesting to extend this result to $\mathfrak{B}_n^{\varepsilon}$.
Although the double-layer blocked circuit is shown to be able to form an approximate design~\cite{schuster2024randomunitariesextremelylow} with a polynomially small error, our analysis does not rely on this property. The decoupling theorem for approximate designs~\cite{Szehr2013Decoupling} generally requires an exponentially small error to achieve a negligible decoupling error. This raises an interesting question of whether more general conditions, such as forming an approximate design with polynomially small error, suffice to ensure good AQEC performance for random codes. Meanwhile, one can further unify the properties of approximate design and decoupling as the scrambling power of a random circuit ensemble and investigate more properties of $\mathfrak{C}$ like out-of-time-ordered correlation~\cite{Swingle2018OTOC,Garcia2021scrambling}.
For clarity, Table~\ref{tab:summary} summarizes the depth scaling of two classes of 1D circuits and all-to-all circuits with respect to decoupling and design properties.

\begin{table}[!ht]
	\centering
	\caption{The table summarizes the circuit depth scaling of three types of circuits with respect to decoupling against local and weakly correlated noise, as well as their approximate design properties. The second column presents our results on decoupling. The third column reports results on approximate designs, while the second row indicates areas where further improvements are possible.}
	\resizebox{.45\textwidth}{!}{
		\begin{tabular}{ccc}\hline
			& Decoupling & Approx. design  \\\hline
			1D brickwork & $\Theta(\log n)$  & $O(n)$~\cite{Brandao2016ApproximateDesign,Dalzell2022Anticoncentrate}  \\
			1D Double-layer blocked & $\Theta(\log n)$ & $\Theta(\log n)$~\cite{schuster2024randomunitariesextremelylow,laracuente2024approximateunitarykdesignsshallow} \\
			All-to-all & $\Theta(\mathrm{loglog} n)$ & $\Theta(\mathrm{loglog} n)$~\cite{schuster2024randomunitariesextremelylow,laracuente2024approximateunitarykdesignsshallow} \\
			\hline
		\end{tabular}
	}
\label{tab:summary}
\end{table}

We also introduce a random Pauli rotation before unitary encoding to show that approximate decoupling and small recovery errors are maintained for arbitrary input states. Developing new techniques to bound the worst-case recovery error for random codes without resorting to random Pauli matrices, possibly by further exploiting the internal structure of these encoding circuits, remains an interesting avenue for future work.

Moreover, exploring low-depth circuits in higher-dimensional architectures, such as 2D or all-to-all circuits, is also important. With improved connectivity in these architectures, the required circuit depth can be expected to decrease. For instance, as discussed in Section~\ref{sec:lower_bound}, all-to-all circuits can use $O(n \log n)$ physical qubits to encode $O(n)$ logical qubits with $O(\log \log n)$ circuit depth while preserving good AQEC performance. An intriguing question is whether the encoding rate can be made constant in the all-to-all circuit while still achieving good AQEC performance. Additionally, it is interesting to investigate whether the hashing bound for 1D circuits can be improved from $\omega(\log n)$ to $O(\log n)$. We also conjecture that, in a $D$-dimensional circuit architecture, circuit depths on the order of $O\bigl((\log n)^{1/D}\bigr)$ suffice to achieve a constant encoding rate and error-correcting capability in the AQEC setting.

%While this work demonstrates that one-dimensional low-depth circuits  exhibit good AQEC properties, further improvements in e.g.~encoding rate and logical operation sets are promising when focusing on AQEC rather than on code distance.
It is worth noting that working with AQEC instead of focusing on the exact code distance offers promising avenues for improvements in various aspects of code performance such as encoding power and logical operation sets.
For example, AQEC  may allow enhanced parameters for geometrically local codes beyond known fundamental limits such as the Bravyi-Poulin-Terhal bound~\cite{Bravyi2010Tradeoffs}. Moreover, certain covariant codes under AQEC possess a continuous set of logical operations, a feature that is highly desirable for fault tolerance~\cite{Kong2022CQEC}.
Beyond evaluating AQEC performance, achieving fault-tolerant implementations of codes from low-depth Clifford encoding circuits~\cite{Nelson2025FTlowdepth} is also crucial in practice, which requires further analysis of the fault-tolerant implementation of quantum gates and the development of improved decoding algorithms for such codes.

Beyond AQEC, our analysis of 1D random circuits extends the previous domain-wall technique~\cite{Dalzell2022Anticoncentrate} and can be applied to study other properties of random circuits, such as their generic second-moment behavior~\cite{heinrich2025anticoncentrationalmostneed,belkin2025apparentuniversalbehaviorsecond}. We have developed decoupling theorems for low-depth circuits, which potentially have important applications in quantum information theory \cite{Dupuis2014Decoupling, Berta2011ReverseShannon}. Decoupling theorems have applications in quantum communication tasks like quantum state merging \cite{Horodecki2005partialquantuminformation} and in quantum cryptography, such as privacy amplification \cite{dupuis2022privacyamplificationdecouplingsmoothing}. We also envision a wide range of exciting applications in many-body physics and quantum thermodynamics~\cite{Linden2009thermalequilibrium, Rio2011thermodynamic}. For instance, a linear number of random erasure errors can be corrected, implying that the subsystem of a 1D log-depth circuit state is highly likely to thermalize. Our results may also offer insights for studies of the black hole information paradox \cite{Patrick2007Blackhole} and quantum many-body entanglement \cite{Brand_o_2014arealaw}.

% While our work focuses on the AQEC properties of low-depth Clifford encoding circuits, other stabilizer codes, such as qLDPC codes, also have significant importance. Seminal results \cite{Panteleev2022goodQLDPC, Leverrier2022QuantumTannerCodes} construct good qLDPC codes with $O(n)$ logical qubits and code distance. Further investigating their performance under practical noise models \cite{Kovalev2013sublinear} and in the context of AQEC offers many opportunities for both theoretical and practical advancements in quantum error correction. Beyond evaluating AQEC performance, achieving fault-tolerant implementations of codes from low-depth Clifford encoding circuits~\cite{Nelson2025FTlowdepth} is also crucial in practice, which requires further analysis of the fault-tolerant implementation of quantum gates and the development of improved decoding algorithms for such codes.

\begin{acknowledgements}
We thank Junjie Chen, Chushi Qin, Zitai Xu, and Xiangran Zhang for their insightful discussions. G.L., Z.D., and X.M. are supported by the National Natural Science Foundation of China Grant No.~12575023, and the Innovation Program for Quantum Science and Technology Grant No.~2021ZD0300804 and No.~2021ZD0300702. Z.-W.L. is supported in part by NSFC under Grant No.~12475023, Dushi Program, and startup funding from YMSC.
\end{acknowledgements}

\newpage

\appendix
Here, we briefly outline the contents of the Appendices. Appendix~\ref{app:entropy} introduces the quantum entropies used in this work and reviews the decoupling theorems for the approximate 2-design group. Appendix~\ref{app:decoupling} provides detailed proofs of the decoupling theorem for 1D double-layer blocked circuits and related AQEC results. Appendix~\ref{appendssc:1DLRC} presents the results associated with 1D brickwork circuits. Finally, Appendix~\ref{app:lower_bound} presents the proof details of the circuit lower bound required to achieve a certain AQEC capability.

\section{Quantum entropies and decoupling theorems}\label{app:entropy}
In this section, we review the related concepts of quantum entropies and the decoupling theorem, which is an important tool for analyzing the AQEC performance of random circuits. First, we introduce several important entropies that we will use. Here, we mainly consider the sandwiched conditional R$\acute{e}$nyi entropy~\cite{rubboli2024quantumconditionalentropies}
\begin{equation}
\widetilde{H}^{\downarrow}_{\alpha}(A|B)_{\rho} = -\widetilde{D}_{\alpha}(\rho_{AB}\Vert \id_A\otimes \rho_B),
\end{equation}
where $\widetilde{D}_{\alpha}$ denotes the sandwiched R$\acute{e}$nyi entropy
\begin{equation}
\widetilde{D}_{\alpha}(\rho\Vert \sigma) = \frac{1}{\alpha-1}\log \tr[(\sigma^{\frac{1-\alpha}{2\alpha}}\rho \sigma^{\frac{1-\alpha}{2\alpha}})^{\alpha}].
\end{equation}
Here and in the whole work, the base of the logarithm is 2. The notation of $\downarrow$ in $\widetilde{H}^{\downarrow}_{\alpha}(A|B)_{\rho}$ is used to distinguish it from another definition of sandwiched conditional R$\acute{e}$nyi entropy,
\begin{equation}
\widetilde{H}^{\uparrow}_{\alpha}(A|B)_{\rho} = -\inf_{\sigma_B}\widetilde{D}_{\alpha}(\rho_{AB}\Vert \id_A\otimes \sigma_B).
\end{equation}
Clearly, $\widetilde{H}^{\uparrow}_{\alpha}(A|B)_{\rho}\geq \widetilde{H}^{\downarrow}_{\alpha}(A|B)_{\rho}$, and these two quantities are different in general.
% In this work, for simplicity, we only use $\widetilde{H}^{\downarrow}_{\alpha}(A|B)_{\rho}$ and simply denote it as $H_{\alpha}(A|B)_{\rho}$.
When $\alpha$ takes value $\infty$, $2$, and $1/2$, we get the conditional min-entropy, the conditional collision entropy, and the conditional max-entropy, respectively, as shown below.

\begin{definition}\label{def:renyientropy}
Given a quantum state $\rho_{AB}$, the conditional min-entropy of $A$ given $B$ is defined as
\begin{align}
\widetilde{H}^{\downarrow}_{\mathrm{min}}(A|B)_{\rho} &= \sup\{\lambda\in \mathbb{R}: 2^{-\lambda}\cdot \id_A\otimes \rho_B-\rho_{AB}\geq 0\};\\
\widetilde{H}^{\uparrow}_{\mathrm{min}}(A|B)_{\rho} &= \sup_{\sigma_B}\sup\{\lambda\in \mathbb{R}: 2^{-\lambda}\cdot \id_A\otimes \sigma_B-\rho_{AB}\geq 0\};
\end{align}
the conditional collision entropy of $A$ given $B$ is defined as
\begin{align}
\widetilde{H}^{\downarrow}_2(A|B)_{\rho} &= -\log \tr[\left((\id_A\otimes \rho_B)^{-1/4}\rho_{AB}(\id_A\otimes \rho_B)^{-1/4}\right)^2];\\
\widetilde{H}^{\uparrow}_2(A|B)_{\rho} &= \sup_{\sigma_B}-\log \tr[\left((\id_A\otimes \sigma_B)^{-1/4}\rho_{AB}(\id_A\otimes \sigma_B)^{-1/4}\right)^2];
\end{align}
the conditional max-entropy of $A$ given $B$ is defined as
\begin{align}
\widetilde{H}^{\downarrow}_{\mathrm{max}}(A|B)_{\rho} &= \log F(\rho_{AB}, \id_A\otimes \rho_B)^2;\\
\widetilde{H}^{\uparrow}_{\mathrm{max}}(A|B)_{\rho} &= \sup_{\sigma_B}\log F(\rho_{AB}, \id_A\otimes \sigma_B)^2.
\end{align}
\end{definition}
When $\alpha$ goes to $1$, $\widetilde{H}^{\uparrow/\downarrow}_{\alpha}(A|B)_{\rho}$ will converge to the conditional von Neumann entropy, which is defined as
\begin{equation}
H(A|B) = H(AB)-H(B),
\end{equation}
where $H$ is the von Neumann entropy with $H(\rho) = -\tr\rho\log \rho$.

We also consider the smoothing version of the conditional entropy. We call two quantum states $\rho$ and $\sigma$ $\delta$-close if their purified distance $P(\rho,\sigma)\leq \delta$, and we denote $\mathcal{B}^{\delta}(\rho) = \{\rho', P(\rho', \rho)\leq \delta\}$ as the set of all $\delta$-close states of $\rho$. Then, we have the following smooth conditional entropies.
\begin{definition}\label{def:smoothentropy}
The $\delta$-smooth sandwiched conditional R$\acute{e}$nyi entropy of $A$ given $B$ for state $\rho_{AB}$ is the maximal sandwiched conditional R$\acute{e}$nyi entropy of states in $\mathcal{B}^{\delta}(\rho_{AB})$:
\begin{align}
\widetilde{H}^{\downarrow\delta}_{\alpha}(A|B)_{\rho} &= \sup_{\hat{\rho}_{AB}\in \mathcal{B}^{\delta}(\rho_{AB})}\widetilde{H}^{\downarrow}_{\alpha}(A|B)_{\hat{\rho}};\\
\widetilde{H}^{\uparrow\delta}_{\alpha}(A|B)_{\rho} &= \sup_{\hat{\rho}_{AB}\in \mathcal{B}^{\delta}(\rho_{AB})}\widetilde{H}^{\uparrow}_{\alpha}(A|B)_{\hat{\rho}}.
\end{align}
Particularly, the $\delta$-smooth conditional min-entropy of $A$ given $B$ is defined as
\begin{align}
\widetilde{H}^{\downarrow \delta}_{\mathrm{min}}(A|B)_{\rho} &= \sup_{\hat{\rho}_{AB}\in \mathcal{B}^{\delta}(\rho_{AB})}\widetilde{H}^{\downarrow}_{\mathrm{min}}(A|B)_{\hat{\rho}};\\
\widetilde{H}^{\uparrow \delta}_{\mathrm{min}}(A|B)_{\rho} &= \sup_{\hat{\rho}_{AB}\in \mathcal{B}^{\delta}(\rho_{AB})}\widetilde{H}^{\uparrow}_{\mathrm{min}}(A|B)_{\hat{\rho}};
\end{align}
the $\delta$-smooth conditional collision entropy of $A$ given $B$ is defined as
\begin{align}
\widetilde{H}^{\downarrow \delta}_{2}(A|B)_{\rho} &= \sup_{\hat{\rho}_{AB}\in \mathcal{B}^{\delta}(\rho_{AB})}\widetilde{H}^{\downarrow}_{2}(A|B)_{\hat{\rho}};\\
\widetilde{H}^{\uparrow \delta}_{2}(A|B)_{\rho} &= \sup_{\hat{\rho}_{AB}\in \mathcal{B}^{\delta}(\rho_{AB})}\widetilde{H}^{\uparrow}_{2}(A|B)_{\hat{\rho}};
\end{align}
the $\delta$-smooth conditional max-entropy of $A$ given $B$ is defined as
\begin{align}
\widetilde{H}^{\downarrow \delta}_{\mathrm{max}}(A|B)_{\rho} &= \sup_{\hat{\rho}_{AB}\in \mathcal{B}^{\delta}(\rho_{AB})}\widetilde{H}^{\downarrow}_{\mathrm{max}}(A|B)_{\hat{\rho}};\\
\widetilde{H}^{\uparrow \delta}_{\mathrm{max}}(A|B)_{\rho} &= \sup_{\hat{\rho}_{AB}\in \mathcal{B}^{\delta}(\rho_{AB})}\widetilde{H}^{\uparrow}_{\mathrm{max}}(A|B)_{\hat{\rho}}.
\end{align}
\end{definition}

% \begin{lemma}\label{lemma:entropymonotone}
% Given quantum state $\rho_{AB}$, we have that $H_2(A|B)_{\rho}\geq H_{\mathrm{min}}(A|B)_{\rho}$.
% \end{lemma}
% \begin{proof}
% Note that $\rho_{AB}\leq 2^{-H_{\mathrm{min}}(A|B)_{\rho}}\cdot \id_A\otimes \rho_B$. Thus,
% \begin{equation}
% \begin{split}
% 2^{-H_2(A|B)_{\rho}} &= \tr[\left((\id_A\otimes \rho_B)^{-1/4}\rho_{AB}(\id_A\otimes \rho_B)^{-1/4}\right)^2]\\
% &= \tr[(\id_A\otimes \rho_B)^{-1/2}\rho_{AB}(\id_A\otimes \rho_B)^{-1/2}]\\
% &\leq 2^{-H_{\mathrm{min}}(A|B)_{\rho}}\tr(\rho_{AB})\\
% &= 2^{-H_{\mathrm{min}}(A|B)_{\rho}}.
% \end{split}
% \end{equation}
% After taking the logarithm, we finish the proof.
% \end{proof}

Below, we list several relationships among different entropies.

\begin{lemma}[Corollary 4 in Ref.~\cite{Tomamichel2014renyi}]\label{lemma:entropymonotone2min}
Given quantum state $\rho_{AB}$, we have that $\widetilde{H}^{\downarrow}_{2}(A|B)_{\rho}\geq \widetilde{H}^{\uparrow}_{\min}(A|B)_{\rho}$.
\end{lemma}

Based on Lemma~\ref{lemma:entropymonotone2min}, we have $\widetilde{H}^{\downarrow\delta}_{2}(A|B)_{\rho}\geq \widetilde{H}^{\uparrow\delta}_{\min}(A|B)_{\rho}$. Meanwhile, since $\widetilde{D}$ is monotonically increasing in $\alpha$~\cite{Lennert2013renyi,Beigi2013renyi}, we have
\begin{lemma}\label{lemma:entropymonotone2max}
Given quantum state $\rho_{AB}$, we have that $\widetilde{H}^{\downarrow}_{2}(A|B)_{\rho}\leq \widetilde{H}^{\downarrow}_{\max}(A|B)_{\rho}\leq \widetilde{H}^{\uparrow}_{\max}(A|B)_{\rho}$.
\end{lemma}
Thus, $\widetilde{H}^{\downarrow\delta}_{2}(A|B)_{\rho}\leq \widetilde{H}^{\uparrow\delta}_{\max}(A|B)_{\rho}$. Note that the conditional min-entropy and the conditional max-entropy both have an asymptotic equipartition property~\cite{Tomamichel2009QAEP,tomamichel2013frameworknonasymptoticquantuminformation}, as shown below.

\begin{lemma}[AEP~\cite{Tomamichel2009QAEP,tomamichel2013frameworknonasymptoticquantuminformation}]\label{lem:AEP} For quantum state $\rho_{AB}$, we have
\begin{align}
\lim_{\delta\rightarrow 0}\lim_{n\rightarrow \infty}\frac{1}{n}\widetilde{H}^{\uparrow \delta}_{\mathrm{min}}(A|B)_{\rho^{\otimes n}} &= H(A|B)_{\rho};\\
\lim_{\delta\rightarrow 0}\lim_{n\rightarrow \infty}\frac{1}{n}\widetilde{H}^{\uparrow \delta}_{\mathrm{max}}(A|B)_{\rho^{\otimes n}} &= H(A|B)_{\rho}.
\end{align}
\end{lemma}
Since $\widetilde{H}^{\uparrow\delta}_{\min}(A|B)_{\rho}\leq \widetilde{H}^{\downarrow\delta}_{2}(A|B)_{\rho}\leq \widetilde{H}^{\uparrow\delta}_{\max}(A|B)_{\rho}$, we have
\begin{equation}
\lim_{\delta\rightarrow 0}\lim_{n\rightarrow \infty}\frac{1}{n}\widetilde{H}^{\downarrow\delta}_{2}(A|B)_{\rho^{\otimes n}} = H(A|B)_{\rho}.
\end{equation}

More exactly, for any $0<\delta<1$, as long as $n\geq -\frac{8}{5}\log(1-\sqrt{1-\delta^2})$, we have that
\begin{align}
\frac{1}{n}\widetilde{H}^{\uparrow \delta}_{\mathrm{min}}(A|B)_{\rho^{\otimes n}} &\geq H(A|B)_{\rho} - \frac{4\log \gamma \sqrt{\log \frac{2}{\delta^2}}}{\sqrt{n}};\\
\frac{1}{n}\widetilde{H}^{\uparrow \delta}_{\mathrm{max}}(A|B)_{\rho^{\otimes n}} &\leq H(A|B)_{\rho} + \frac{4\log \gamma \sqrt{\log \frac{2}{\delta^2}}}{\sqrt{n}},
\end{align}
where
\begin{equation}
\gamma \leq \sqrt{2^{-\widetilde{H}^{\uparrow \delta}_{\mathrm{min}}(A|B)_{\rho}}}+\sqrt{2^{\widetilde{H}^{\uparrow \delta}_{\mathrm{max}}(A|B)_{\rho}}}+1.
\end{equation}
Thus,
\begin{equation}
H(A|B)_{\rho} - \frac{4\log \gamma \sqrt{\log \frac{2}{\delta^2}}}{\sqrt{n}} \leq \frac{1}{n}\widetilde{H}^{\downarrow\delta}_{2}(A|B)_{\rho^{\otimes n}} \leq H(A|B)_{\rho} + \frac{4\log \gamma \sqrt{\log \frac{2}{\delta^2}}}{\sqrt{n}}.
\end{equation}

Below, we present a lemma for calculating quantum conditional entropies with the purification of a quantum state.
\begin{lemma}[Dual relations for quantum conditional entropies, Theorem~2 in Ref.~\cite{Tomamichel2014renyi}]\label{lemma:entropydual}
Given tripartite pure state $\rho_{ABC}$ with reduced density matrices on system $AB$ and $AC$ as $\rho_{AB}$ and $\rho_{AC}$, respectively, we have that
\begin{equation}
\widetilde{H}^{\downarrow}_{\alpha}(A|B)_{\rho_{AB}} = -\frac{1}{\alpha-1} \log \tr_C((\tr_A \rho_{AC}^{\frac{1}{\alpha}})^{\alpha}).
\end{equation}
\end{lemma}

In Ref.~\cite{Dupuis2014Decoupling}, the authors prove a decoupling theorem, which describes how the two systems are decoupled after one system is acted on by a random unitary and a channel. This theorem is useful in analyzing the performance of a random quantum error correction scheme, which is shown below.
\begin{lemma}[Decoupling theorem~\cite{Dupuis2014Decoupling}]\label{lemma:haardecoupling}
Given quantum state $\rho_{SR}$ and quantum channel $\mathcal{T}_{S\rightarrow E}$, denote the reduced density matrix of $\rho_{SR}$ on system $R$ as $\rho_R$ and the Choi-Jamio{\l}kowski representation of $\mathcal{T}_{S\rightarrow E}$ as $\tau_{SE}$ with reduced density matrix on system $E$ as $\tau_E$. We obtain that,
\begin{equation}
\mathbb{E}_{U\sim \mathfrak{U}_n}\Vert \mathcal{T}_{S\rightarrow E}(U_S\rho_{SR}U_{S}^{\dagger}) - \tau_E\otimes \rho_R \Vert_1
\leq 2^{-\frac{1}{2}\widetilde{H}^{\uparrow}_{2}(S|R)_{\rho_{SR}}-\frac{1}{2}\widetilde{H}^{\uparrow}_{2}(S|E)_{\tau_{SE}}},
\end{equation}
and
\begin{equation}
\mathbb{E}_{U\sim \mathfrak{U}_n}\Vert \mathcal{T}_{S\rightarrow E}(U_S\rho_{SR}U_{S}^{\dagger}) - \tau_E\otimes \rho_R \Vert_1
\leq 2^{-\frac{1}{2}\widetilde{H}^{\uparrow\delta}_{\min}(S|R)_{\rho_{SR}}-\frac{1}{2}\widetilde{H}^{\uparrow\delta}_{\min}(S|E)_{\tau_{SE}}}+12\delta,
\end{equation}
where $\mathbb{E}_{U\sim\mathfrak{U}_n}$ denotes the expectation over the Haar random unitary operations. This result also holds for a unitary $2$-design group instead of $\mathfrak{U}_n$~\cite{Szehr2013Decoupling}.
\end{lemma}

The term $\widetilde{H}^{\uparrow}_{2}(S|R)_{\rho_{SR}}$ characterizes the initial correlation between the system and the reference part, and $\widetilde{H}^{\uparrow}_{2}(S|E)_{\tau_{SE}}$ characterizes how the channel $\mathcal{T}_{S\rightarrow E}$ preserves the correlation. When applying the decoupling theorem to AQEC, one can view the state $\rho_{SR}$ to be the initial logical state, the unitary operation $U_S$ to be the encoding map, and the channel $\mathcal{T}_{S\rightarrow E}$ to be the complementary channel from system to environment. Then the decoupling error on the right-hand side can be viewed as the Choi error.

If the system is twirled by an $\varepsilon$-approximate 2-design $\mathfrak{S}$ instead of $\mathfrak{U}_n$, we also have a result~\cite{Szehr2013Decoupling}.

\begin{lemma}[Decoupling theorem for an approximate 2-design~\cite{Szehr2013Decoupling}]\label{lemma:appro2designdecoup}
Given quantum state $\rho_{SR}$ and quantum channel $\mathcal{T}_{S\rightarrow E}$, denote the reduced density matrix of $\rho_{SR}$ on system $R$ as $\rho_R$ and the Choi-Jamio{\l}kowski representation of $\mathcal{T}_{S\rightarrow E}$ as $\tau_{SE}$ with reduced density matrix on system $E$ as $\tau_E$. We obtain that,
\begin{equation}
\mathbb{E}_{U\sim \mathfrak{S}}\Vert \mathcal{T}_{S\rightarrow E}(U_S\rho_{SR}U_{S}^{\dagger}) - \tau_E\otimes \rho_R \Vert_1
\leq \sqrt{1+4\varepsilon \abs{S}^4}2^{-\frac{1}{2}\widetilde{H}^{\uparrow}_{\min}(S|R)_{\rho_{SR}}-\frac{1}{2}\widetilde{H}^{\uparrow}_{\min}(S|E)_{\tau_{SE}}},
\end{equation}
where $\mathbb{E}_{U\sim\mathfrak{S}}$ denotes the expectation over $\mathfrak{S}$, and $\varepsilon$ is the additive error of the approximate design.
\end{lemma}

Due to the coefficient $4\varepsilon \abs{S}^4$, for a $k$-qubit system $S$, $\varepsilon$ needs to be exponentially small with $k$ to make the decoupling error vanish. Thus, an approximate 2-design ensemble with polynomially decreasing error is not guaranteed to have a good performance in quantum error correction.

In this work, for the collision entropy, we will usually consider the case of arrow down. For the min-entropy and the max-entropy, we will usually consider the case of arrow up. For simplicity, below without additional clarification, we denote
\begin{align}
H_{\mathrm{min}}(A|B)_{\rho} &= \widetilde{H}^{\uparrow}_{\mathrm{min}}(A|B)_{\rho};\\
H_2(A|B)_{\rho} &= \widetilde{H}^{\downarrow}_{2}(A|B)_{\rho};\\
H_{\mathrm{max}}(A|B)_{\rho}&=\widetilde{H}^{\uparrow}_{\mathrm{max}}(A|B)_{\rho}.
\end{align}
The smooth conditional entropies are modified correspondingly.

\section{Decoupling theorem and AQEC for 1D double-layer blocked low-depth circuits}\label{app:decoupling}
In this section, we first prove the non-smooth decoupling theorem or Theorem~\ref{thm:nonsmoothdecoupling} and then use Theorem~\ref{thm:nonsmoothdecoupling} to prove the smooth decoupling theorem or Theorem~\ref{thm:smoothdecoupling}. Then, we apply the decoupling theorem for 1D double-layer blocked low-depth circuits to local and correlated noises. After that, we prove the AQEC performance against the fixed-number random erasure errors for 1D low-depth circuits. Subsequently, we give the proof of Proposition~\ref{prop:pauli_twirling}. We also present the AQEC results with random Clifford encoding and block-encoding methods as a comparison.

\subsection{Proof of non-smooth decoupling theorem or Theorem~\ref{thm:nonsmoothdecoupling}}\label{app:proof_nonsmooth}
\begin{proof}
Based on Lemma~\ref{lemma:1normbound}, for any quantum states $\sigma_E\in \mathcal{D}(\mathcal{H}_E)$ and $\zeta_R\in \mathcal{D}(\mathcal{H}_R)$,
\begin{equation}
\Vert \mathcal{T}_{S\rightarrow E}(U_S\rho_{SR}U_{S}^{\dagger}) - \tau_E\otimes \rho_R \Vert_1 \leq \sqrt{\tr[ \left( (\sigma_E\otimes \zeta_R)^{-1/4}(\mathcal{T}_{S\rightarrow E}(U_S\rho_{SR}U_{S}^{\dagger}) - \tau_E\otimes \rho_R) (\sigma_E\otimes \zeta_R)^{-1/4} \right)^2 ]}.
\end{equation}
We define completely positive map $\widetilde{\mathcal{T}}_{S\rightarrow E} = \sigma_E^{-1/4}\mathcal{T}_{S\rightarrow E} \sigma_E^{-1/4}$ with the corresponding Choi-Jomio{\l}kowski representation $\widetilde{\tau}_{SE} = \sigma_E^{-1/4}\tau_{SE} \sigma_E^{-1/4}$ and $\widetilde{\rho}_{SR} = \zeta_R^{-1/4}\rho_{SR}\zeta_R^{-1/4}$. Recall that $S = S_1S_2\cdots S_{2N}$ and $R = R_1R_2\cdots R_{2N}$. In the following, we may use $A\subseteq [2N]$ or $S_A$ to represent the subsystem $\bigcup_{i\in A}S_i$. Also, $\rho_{SR}$ has a structure of tensor product, $\rho_{SR} = \bigotimes_{i=1}^{2N}\rho_{S_iR_i}$. Below, we let $\zeta_R$ also have a tensor product structure:
\begin{equation}
\zeta_R = \bigotimes_{i=1}^{2N} \zeta_{R_i}.
\end{equation}
Thanks to the tensor product structure of $\rho_{SR}$ and $\zeta_R$, we have that
\begin{equation}
\widetilde{\rho}_{SR} = \bigotimes_{i=1}^{2N}\widetilde{\rho}_{S_iR_i},
\end{equation}
where $\widetilde{\rho}_{S_iR_i} = \zeta_{R_i}^{-1/4}\rho_{S_iR_i}\zeta_{R_i}^{-1/4}$.

We then get
\begin{equation}
\Vert \mathcal{T}_{S\rightarrow E}(U_S\rho_{SR}U_{S}^{\dagger}) - \tau_E\otimes \rho_R \Vert_1 \leq \sqrt{\tr[ \left( \widetilde{\mathcal{T}}_{S\rightarrow E}(U_S\widetilde{\rho}_{SR}U_{S}^{\dagger}) - \widetilde{\tau}_E\otimes \widetilde{\rho}_R  \right)^2 ]}.
\end{equation}

Using Jensen's inequality, we have that
\begin{equation}
\mathbb{E}_U\Vert \mathcal{T}_{S\rightarrow E}(U_S\rho_{SR}U_{S}^{\dagger}) - \tau_E\otimes \rho_R \Vert_1 \leq \sqrt{\mathbb{E}_U\tr[ \left( \widetilde{\mathcal{T}}_{S\rightarrow E}(U_S\widetilde{\rho}_{SR}U_{S}^{\dagger}) - \widetilde{\tau}_E\otimes \widetilde{\rho}_R  \right)^2 ]}.
\end{equation}
Below, we focus on the term inside the square root.
\begin{equation}
\begin{split}
&\mathbb{E}_U\tr[ \left( \widetilde{\mathcal{T}}_{S\rightarrow E}(U_S\widetilde{\rho}_{SR}U_{S}^{\dagger}) - \widetilde{\tau}_E\otimes \widetilde{\rho}_R  \right)^2]\\
=&\mathbb{E}_U\tr[ \left( \widetilde{\mathcal{T}}_{S\rightarrow E}(U_S\widetilde{\rho}_{SR}U_{S}^{\dagger}) \right)^2] + \tr[\left( \widetilde{\tau}_E\otimes \widetilde{\rho}_R  \right)^2] - 2\mathbb{E}_U\tr[\widetilde{\mathcal{T}}_{S\rightarrow E}(U_S\widetilde{\rho}_{SR}U_{S}^{\dagger}) (\widetilde{\tau}_E\otimes \widetilde{\rho}_R)].
\end{split}
\end{equation}

Since $U$ is randomly sampled from $\mathfrak{C}_n^{\varepsilon}$, which forms a unitary 1-design, we have
\begin{equation}
\mathbb{E}_U\widetilde{\mathcal{T}}_{S\rightarrow E}(U_S\widetilde{\rho}_{SR}U_{S}^{\dagger}) = \widetilde{\tau}_E\otimes \widetilde{\rho}_R,
\end{equation}
and
\begin{equation}
\tr[\left( \widetilde{\tau}_E\otimes \widetilde{\rho}_R  \right)^2] - 2\mathbb{E}_U\tr[\widetilde{\mathcal{T}}_{S\rightarrow E}(U_S\widetilde{\rho}_{SR}U_{S}^{\dagger}) (\widetilde{\tau}_E\otimes \widetilde{\rho}_R)] = -\tr(\widetilde{\tau}^2_E)\tr(\widetilde{\rho}^2_R).
\end{equation}
For the remaining part $\mathbb{E}_U\tr[ \left( \widetilde{\mathcal{T}}_{S\rightarrow E}(U_S\widetilde{\rho}_{SR}U_{S}^{\dagger}) \right)^2]$, we use $\tr A^2 = \tr A^{\otimes 2}F$ to evaluate it with $F$ a SWAP operator and separate the two layers of the encoding unitary gate. Then,
\begin{equation}
\begin{split}
\mathbb{E}_U\tr[ \left( \widetilde{\mathcal{T}}_{S\rightarrow E}(U_S\widetilde{\rho}_{SR}U_{S}^{\dagger}) \right)^2]
&= \tr_S( \mathbb{E}_{U_2}[U_2^{\otimes 2} \mathbb{E}_{U_1}[U_1^{\otimes 2}\tr_R(\widetilde{\rho}_{SR}^{\otimes 2} F_R)U_1^{\dagger \otimes 2}] U_2^{\dagger\otimes 2}](\widetilde{\mathcal{T}}_{S\rightarrow E}^{\dagger})^{\otimes 2}(F_E))\\
&= \tr_S( M_2M_1( \tr_R(\widetilde{\rho}_{SR}^{\otimes 2} F_R) )  (\widetilde{\mathcal{T}}_{S\rightarrow E}^{\dagger})^{\otimes 2}(F_E)),
\end{split}
\end{equation}
where $U_1$ and $U_2$ are the first and second layers of the encoding unitary $U$, respectively, and $\widetilde{\mathcal{T}}_{S\rightarrow E}^{\dagger}$ is the adjoint channel of $\widetilde{\mathcal{T}}_{S\rightarrow E}$. We also define $M_i = \mathbb{E}_{U_i} U_i^{\otimes 2}\cdot U_i^{\dagger\otimes 2}$ as the second-order twirling operation of layer $i$. Thus, the key technical part is to evaluate the action of $M_2M_1$.

Note that when $\xi=1$, $U_i$ is generated by two-qubit gates. In this case, $(M_2M_1)^{l}$ simulates $2l$-depth 1D brickwork circuits~\cite{Dalzell2022Anticoncentrate}. In this part, we focus mainly on the double-layer blocked structure and $l=1$, which enables a simpler evaluation. The case of 1D brickwork circuits can be analyzed using the domain wall approach~\cite{Dalzell2022Anticoncentrate}, which we will discuss in Appendix~\ref{appendssc:1DLRC}.

In the following, for simplicity, we set $O_{SR} = \tr_R(\widetilde{\rho}_{SR}^{\otimes 2} F_R)$ and $O_{SE} = (\widetilde{\mathcal{T}}_{S\rightarrow E}^{\dagger})^{\otimes 2}(F_E)$. We also denote $O_{S_iR_i} = \tr_{R_i}(\widetilde{\rho}_{S_iR_i}^{\otimes 2} F_{R_i})$. Thus, $O_{SR} = \bigotimes_{i}O_{S_iR_i}$. Our task is to evaluate $\tr(M_2M_1(O_{SR})O_{SE})$.

Note that $S = S_1S_2\cdots S_{2N}$ and $R = R_1R_2\cdots R_{2N}$. Based on Eq.~\eqref{eq:secondtwirling} and the tensor product structures of $U_1$, we have that
\begin{equation}
M_1(O_{SR}) = \mathbb{E}_{U_1}[U_1^{\otimes 2}O_{SR}U_1^{\dagger \otimes 2}] = \frac{1}{[(2^{2\xi})^2-1]^N} \sum_{\vec{\gamma}\in \{\id\otimes\id, F\otimes F\}^{\otimes N} } \tr(O_{SR} g_{\vec{\gamma}})\vec{\gamma} = \sum_{\vec{\gamma}\in \{\id\otimes\id, F\otimes F\}^{\otimes N} }c_{\vec{\gamma}}\vec{\gamma}.
\end{equation}
Here, $\id$ is the identity operation on two copies of a subsystem $i$, and $F$ is the SWAP operation on two copies of a subsystem $i$; $\vec{\gamma}$ can be viewed as a vector with $2N$ components. The $j$-th component of $g_{\vec{\gamma}}$, denoted as $g_{\gamma_j}$ satisfies: $g_{\gamma_{2j-1}}\otimes g_{\gamma_{2j}} = \id\otimes\id-2^{-2\xi}F\otimes F$ if $\gamma_{2j-1}\otimes \gamma_{2j} = \id\otimes\id$, and $g_{\gamma_{2j-1}}\otimes g_{\gamma_{2j}} = F\otimes F-2^{-2\xi}\id\otimes \id$ if $g_{\gamma_{2j-1}}\otimes g_{\gamma_{2j}} = F\otimes F$ where $\gamma_j$ is the $j$-th component of $\Vec{\gamma}$. Using $O_{SR} = \bigotimes_i O_{S_iR_i}$ and Lemma~\ref{lemma:purity} below, it is straightforward to upper bound the coefficient $c_{\vec{\gamma}}$:
\begin{equation}\label{eq:OSRinequality}
c_{\vec{\gamma}} = \frac{1}{[(2^{2\xi})^2-1]^N} \tr(O_{SR}g_{\vec{\gamma}})\leq (2^{-2\xi})^{2N}\tr(O_{SR}\vec{\gamma}) = 2^{-2n}\tr(O_{SR}\vec{\gamma}).
\end{equation}

\begin{lemma}[Lemma 3.6 in Ref.~\cite{Dupuis2014Decoupling}]\label{lemma:purity}
Given a Hilbert space $\mathcal{H}_{AB}$, let $\widetilde{\rho}_{AB}$ be a non-negative operation on $\mathcal{H}_{AB}$. We have that
\begin{equation}
\frac{1}{\abs{A}}\leq \frac{\tr(\widetilde{\rho}^2_{AB})}{\tr(\widetilde{\rho}^2_B)}\leq \abs{A}.
\end{equation}
\end{lemma}

The proof of Eq.~\eqref{eq:OSRinequality} is mainly based on
\begin{equation}
\begin{split}
\tr(O_{S_{2i-1}R_{2i-1}}\otimes O_{S_{2i}R_{2i}} (\id\otimes\id-2^{-2\xi}F\otimes F)) &= \tr( \widetilde{\rho}^2_{R_{2i-1}}\widetilde{\rho}^2_{R_{2i}}) - 2^{-2\xi}\tr(\widetilde{\rho}^2_{S_{2i-1}R_{2i-1}}\widetilde{\rho}^2_{S_{2i}R_{2i}})\\
&\leq (1-(2^{-2\xi})^2)\tr( \widetilde{\rho}^2_{R_{2i-1}}\widetilde{\rho}^2_{R_{2i}});\\
\tr(O_{S_{2i-1}R_{2i-1}}\otimes O_{S_{2i}R_{2i}} (F\otimes F-2^{-2\xi}\id\otimes\id)) &= \tr(\widetilde{\rho}^2_{S_{2i-1}R_{2i-1}}\widetilde{\rho}^2_{S_{2i}R_{2i}})- 2^{-2\xi}\tr( \widetilde{\rho}^2_{R_{2i-1}}\widetilde{\rho}^2_{R_{2i}}) \\
&\leq (1-(2^{-2\xi})^2)\tr(\widetilde{\rho}^2_{S_{2i-1}R_{2i-1}}\widetilde{\rho}^2_{S_{2i}R_{2i}}).
\end{split}
\end{equation}
where the inequalities are based on Lemma~\ref{lemma:purity}.

The next step is to evaluate $M_2M_1(O_{SR}) = \sum_{\vec{\gamma}}c_{\vec{\gamma}}M_2(\vec{\gamma})$. Based on the tensor product structure of $U_2$, we can analyze the action of each block part of $M_2$ on $\vec{\gamma}$. Since $U_2 = \bigotimes_{i=1}^{N-1} U_C^{2i,2i+1}$, $M_2$ can be written as $\bigotimes_{i=1}^{N-1} M^{2i,2i+1}$ where $M^{2i,2i+1}$ is the second-order twirling operation on subsystems $S_{2i}S_{2i+1}$. Based on Eq.~\eqref{eq:secondtwirling}, the action of $M^{2i,2i+1}$ on $\{\id, F\}^{\otimes 2}$ is
\begin{align}
M^{2i,2i+1}(\id_{2i}\otimes \id_{2i+1}) &= \id_{2i}\otimes \id_{2i+1},\\
M^{2i,2i+1}(F_{2i}\otimes F_{2i+1}) &= F_{2i}\otimes F_{2i+1},\\
M^{2i,2i+1}(\id_{2i}\otimes F_{2i+1}) &= \eta(\id_{2i}\otimes \id_{2i+1}+F_{2i}\otimes F_{2i+1}),\\
M^{2i,2i+1}(F_{2i}\otimes \id_{2i+1}) &= \eta(\id_{2i}\otimes \id_{2i+1}+F_{2i}\otimes F_{2i+1}).
\end{align}
Here, $\eta = 2^{\xi}/(2^{2\xi}+1)$. Note that only if $\gamma_{2i}\neq \gamma_{2i+1}$, $M^{2i,2i+1}$ can alter $\vec{\gamma}$.

Note that $\vec{\gamma}\in \{\id\otimes \id, F\otimes F\}^{\otimes N}$. We always have $\gamma_{2i-1}=\gamma_{2i}$. Given a vector $\vec{\gamma}\in \{\id\otimes \id, F\otimes F\}^{\otimes N}$, we count all indexes $j$ such that $\gamma_{2j}\neq \gamma_{2j+1}$ and denote this set as $\mathcal{G} = \{j_1, j_2, \cdots, j_g\}$. Obviously, $0\leq g\leq N-1$, and there is $\binom{N-1}{g}$ choices of $\mathcal{G}$ after fixing $g$. We also denote the index set $\mathcal{A} = \{2j,2j+1|j\in \mathcal{G}\}$. We also define $S_{\mathcal{A}} = \bigcup_{i\in \mathcal{A}} S_i$ and $S_{\Bar{\mathcal{A}}} = \bigcup_{i\in \Bar{\mathcal{A}}} S_i$. Then, $M_2$ only alters $\vec{\gamma}$ within the support of $S_{\mathcal{A}}$ and leaves the remained part unchanged. We use $\gamma_{\Bar{\mathcal{A}}}$ to denote this unchanged part. Then,
\begin{equation}
M_2(\vec{\gamma}) = \eta^g \gamma_{\Bar{\mathcal{A}}}\otimes \bigotimes_{j\in \mathcal{G}}( \id_{2j}\otimes \id_{2j+1}+F_{2j}\otimes F_{2j+1} ).
\end{equation}
Thus, the final target can be evaluated by
\begin{equation}
\begin{split}
\tr(M_2M_1(O_{SR})O_{SE}) &= \sum_{\vec{\gamma\in \{\id\otimes \id, F\otimes F\}^{\otimes N}}} c_{\vec{\gamma}} \tr( M_2(\vec{\gamma}) O_{SE}) \\
&= \sum_{\gamma_1 = \id, F}\sum_{g=0}^{N-1} \sum_{\mathcal{G}}  c_{\vec{\gamma}} \tr( M_2(\vec{\gamma}) O_{SE})\\
&= \sum_{\gamma_1 = \id, F}\sum_{g=0}^{N-1} \eta^g \sum_{\mathcal{G}}  c_{\vec{\gamma}}  \tr( \gamma_{\Bar{\mathcal{A}}}\otimes \bigotimes_{j\in \mathcal{G}}( \id_{2j}\otimes \id_{2j+1}+F_{2j}\otimes F_{2j+1} ) O_{SE}).
\end{split}
\end{equation}
Note that given a vector $\vec{\gamma} = F_{A}\otimes \id_{\Bar{A}}$, we always have
\begin{equation}
\begin{split}
\tr(\vec{\gamma}O_{SE}) &= \tr( \widetilde{\mathcal{T}}^{\otimes 2}_{S\rightarrow E}(\vec{\gamma}) F_E )\\
&= 2^{2n}\tr(\widetilde{\tau}^{\otimes 2}_{SE} F_{A}\otimes \id_{\Bar{A}} \otimes F_E )\\
&= 2^{2n}\tr_{AE} (\tr_{\Bar{A}}\widetilde{\tau}_{SE})^2\\
&\geq 0.
\end{split}
\end{equation}
Thus, using Eq.~\eqref{eq:OSRinequality}, we have
\begin{equation}
\begin{split}
\tr(M_2M_1(O_{SR})O_{SE})
\leq& \sum_{\gamma_1 = \id, F}\sum_{g=0}^{N-1} \eta^g \sum_{\mathcal{G}} 2^{-2n} \tr(O_{SR}\vec{\gamma})  \tr( \gamma_{\Bar{\mathcal{A}}}\otimes \bigotimes_{j\in \mathcal{G}}( \id_{2j}\otimes \id_{2j+1}+F_{2j}\otimes F_{2j+1} ) O_{SE})\\
\leq& 2^{-2n}(\tr(O_{SR}\id^{\otimes 2N})\tr(O_{SE}\id^{\otimes 2N}) + \tr(O_{SR}F^{\otimes 2N})\tr(O_{SE}F^{\otimes 2N}))\\
&+\sum_{\gamma_1 = \id, F}\sum_{g=1}^{N-1} (2\eta)^g \sum_{\mathcal{G}} 2^{-2n} \max_{\gamma'_{\mathcal{A}}=\gamma_{\mathcal{A}}, \gamma'_{\Bar{\mathcal{A}}}\in \{\id\otimes \id, F\otimes F\}^{\otimes g}} \tr(O_{SR}\gamma_{\mathcal{A}}\otimes \gamma_{\Bar{\mathcal{A}}})  \tr(O_{SE}\gamma'_{\mathcal{A}}\otimes \gamma'_{\Bar{\mathcal{A}}})\\
\leq& \tr(\widetilde{\tau}^2_E)\tr(\widetilde{\rho}^2_R)+\tr(\widetilde{\tau}^2_{SE})\tr(\widetilde{\rho}^2_{SR})\\
&+ \sum_{\gamma_1 = \id, F}\sum_{g=1}^{N-1} (2\eta)^g \sum_{\mathcal{G}} 2^{-2n} \max_{\gamma'_{\mathcal{A}}=\gamma_{\mathcal{A}}, \gamma'_{\Bar{\mathcal{A}}}\in \{\id\otimes \id, F\otimes F\}^{\otimes g}} \tr(O_{SR}\vec{\gamma'})  \tr(O_{SE}\vec{\gamma'}) \frac{\tr(O_{S_{\mathcal{\Bar{A}}}R_{\mathcal{\Bar{A}}}}\gamma_{\mathcal{\Bar{A}}})}{\tr(O_{S_{\mathcal{\Bar{A}}}R_{\mathcal{\Bar{A}}}}\gamma'_{\mathcal{\Bar{A}}})}\\
\leq& \tr(\widetilde{\tau}^2_E)\tr(\widetilde{\rho}^2_R)+\tr(\widetilde{\tau}^2_{SE})\tr(\widetilde{\rho}^2_{SR})\\
&+ \sum_{\gamma_1 = \id, F}\sum_{g=1}^{N-1} (2\eta)^g \sum_{\mathcal{G}} 2^{-2n} \max_{\vec{\gamma}\in \{\id, F\}^{\otimes 2N}} (\tr(O_{SR}\vec{\gamma})  \tr(O_{SE}\vec{\gamma}))
(\max_i(\frac{\tr(O_{S_iR_i} \id_i)}{\tr(O_{S_iR_i} F_i)}, \frac{\tr(O_{S_iR_i} F_i)}{\tr(O_{S_iR_i} \id_i)}))^g\\
\leq& \tr(\widetilde{\tau}^2_E)\tr(\widetilde{\rho}^2_R)+\tr(\widetilde{\tau}^2_{SE})\tr(\widetilde{\rho}^2_{SR})\\
&+ 2\sum_{g=1}^{N-1} (2\eta\max_i(\frac{\tr(O_{S_iR_i} \id_i)}{\tr(O_{S_iR_i} F_i)}, \frac{\tr(O_{S_iR_i} F_i)}{\tr(O_{S_iR_i} \id_i)}))^g \binom{N-1}{g} \max_{\vec{\gamma}} (2^{-2n}\tr(O_{SR}\vec{\gamma})  \tr(O_{SE}\vec{\gamma}))\\
=& \tr(\widetilde{\tau}^2_E)\tr(\widetilde{\rho}^2_R)+\tr(\widetilde{\tau}^2_{SE})\tr(\widetilde{\rho}^2_{SR})\\
&+ 2\sum_{g=1}^{N-1} (2\eta\max_i(\frac{\tr(\widetilde{\rho}_{R_i}^2)}{\tr(\widetilde{\rho}_{S_iR_i}^2)}, \frac{\tr(\widetilde{\rho}_{S_iR_i}^2)}{\tr(\widetilde{\rho}_{R_i}^2)}))^g \binom{N-1}{g} \max_{A\subseteq [2N]} ( \tr_{AR}(\tr_{\Bar{A}} \widetilde{\rho}_{SR})^2 \tr_{AE}(\tr_{\Bar{A}} \widetilde{\tau}_{SE})^2 ) \\
=& \tr(\widetilde{\tau}^2_E)\tr(\widetilde{\rho}^2_R)+\tr(\widetilde{\tau}^2_{SE})\tr(\widetilde{\rho}^2_{SR}) + c\max_{A\subseteq [2N]} ( \tr_{AR}(\tr_{\Bar{A}} \widetilde{\rho}_{SR})^2 \tr_{AE}(\tr_{\Bar{A}} \widetilde{\tau}_{SE})^2 ),
\end{split}
\end{equation}
where
\begin{equation}
c = 2((1+2\eta\max_i(\max(\frac{\tr(\widetilde{\rho}_{R_i}^2)}{\tr(\widetilde{\rho}_{S_iR_i}^2)},\frac{\tr(\widetilde{\rho}_{S_iR_i}^2)}{\tr(\widetilde{\rho}_{R_i}^2)})))^{N-1}-1).
\end{equation}

Consequently,
\begin{equation}
\begin{split}
&\mathbb{E}_U\Vert \mathcal{T}_{S\rightarrow E}(U_S\rho_{SR}U_{S}^{\dagger}) - \tau_E\otimes \rho_R \Vert_1\\
\leq& \sqrt{\mathbb{E}_U\tr[ \left( \widetilde{\mathcal{T}}_{S\rightarrow E}(U_S\widetilde{\rho}_{SR}U_{S}^{\dagger}) - \widetilde{\tau}_E\otimes \widetilde{\rho}_R  \right)^2 ]}\\
\leq& \sqrt{-\tr(\widetilde{\tau}^2_E)\tr(\widetilde{\rho}^2_R)+\tr(\widetilde{\tau}^2_E)\tr(\widetilde{\rho}^2_R)+\tr(\widetilde{\tau}^2_{SE})\tr(\widetilde{\rho}^2_{SR}) + c \max_{A\subseteq [2N]} ( \tr_{AR}(\tr_{\Bar{A}} \widetilde{\rho}_{SR})^2 \tr_{AE}(\tr_{\Bar{A}} \widetilde{\tau}_{SE})^2 )}\\
\leq& \sqrt{\tr(\widetilde{\tau}^2_{SE})\tr(\widetilde{\rho}^2_{SR}) + c \max_{A\subseteq [2N]} ( \tr_{AR}(\tr_{\Bar{A}} \widetilde{\rho}_{SR})^2 \tr_{AE}(\tr_{\Bar{A}} \widetilde{\tau}_{SE})^2 )}.
\end{split}
\end{equation}

Note that $\widetilde{\rho}_{SR} = \bigotimes_{i=1}^{2N}\widetilde{\rho}_{S_iR_i}$ and $\widetilde{\rho}_{S_iR_i} = \zeta_{R_i}^{-1/4}\rho_{S_iR_i}\zeta_{R_i}^{-1/4}$. We are free to choose any state as $\zeta_{R_i}$. By the definition of conditional collision entropy in Definition~\ref{def:renyientropy}, we can choose $\zeta_{R_i} = \rho_{R_i}$ or $\zeta_{R} = \rho_{R}$ to let
\begin{equation}
\tr(\widetilde{\rho}^2_{R_i}) = 1, \tr(\widetilde{\rho}^2_{S_iR_i}) = 2^{-H_2(S_i|R_i)_{\rho_{S_iR_i}}}.
\end{equation}
Thus,
\begin{equation}
\tr_{AR}(\tr_{\Bar{A}} \widetilde{\rho}_{SR})^2 = \tr_{AR}( \widetilde{\rho}_{AR}^2 ) = 2^{-H_2(A|R)_{\rho_{AR}}}.
\end{equation}

Also, we can choose $\sigma_E = \tau_E$ such that
\begin{equation}
\tr(\widetilde{\tau}^2_{E}) = 1, \tr(\widetilde{\tau}^2_{SE}) = 2^{-H_2(S|E)_{\tau_{SE}}}.
\end{equation}
In this case,
\begin{equation}
\tr_{AE}(\tr_{\Bar{A}} \widetilde{\tau}_{SE})^2 = \tr_{AE}( \widetilde{\tau}_{AE}^2 ) = 2^{-H_2(A|E)_{\tau_{AE}}}.
\end{equation}
Finally, we get
\begin{equation}
\begin{split}
&\mathbb{E}_U\Vert \mathcal{T}_{S\rightarrow E}(U_S\rho_{SR}U_{S}^{\dagger}) - \tau_E\otimes \rho_R \Vert_1\\
\leq& \sqrt{2^{-H_2(S|E)_{\tau_{SE}}-H_2(S|R)_{\rho_{SR}}} + c \max_{A\subseteq [2N]} 2^{-H_2(A|R)_{\rho_{AR}}-H_2(A|E)_{\tau_{AE}}}},
\end{split}
\end{equation}
where
\begin{equation}
c = 2\left((1+2\eta\max_i(\max(2^{-H_2(S_i|R_i)_{\rho_{S_iR_i}}}, 2^{H_2(S_i|R_i)_{\rho_{S_iR_i}}})))^{N-1}-1\right).
\end{equation}

Furthermore, if $\mathcal{T}_{S\rightarrow E}$ has a tensor product structure, $\mathcal{T}_{S\rightarrow E} = \bigotimes_{i=1}^{2N}\mathcal{T}_{S_i\rightarrow E_i}$, then $\widetilde{\tau}_{SE} = \bigotimes_{i=1}^{2N}\widetilde{\tau}_{S_iE_i}$ with $\widetilde{\tau}_{S_iE_i} = \tau_{E_i}^{-1/4}\tau_{S_iE_i} \tau_{E_i}^{-1/4}$.
In this case,
\begin{align}
H_2(A|R)_{\rho_{AR}} &= \sum_{i\in A} H_2(S_i|R_i)_{\rho_{S_iR_i}};\\
H_2(A|E)_{\tau_{AE}} &= \sum_{i\in A} H_2(S_i|E_i)_{\tau_{S_iE_i}}.
\end{align}
Thus, we get the following decoupling result
\begin{equation}
\begin{split}
&\mathbb{E}_U\Vert \mathcal{T}_{S\rightarrow E}(U_S\rho_{SR}U_{S}^{\dagger}) - \tau_E\otimes \rho_R \Vert_1\\
\leq& \sqrt{2^{-\sum_{i=1}^{2N}H_2(S_i|E_i)_{\tau_{S_iE_i}}-\sum_{i=1}^{2N}H_2(S_i|R_i)_{\rho_{S_iR_i}}} + c\prod_{i=1}^{2N} \max(1, 2^{-H_2(S_i|E_i)_{\tau_{S_iE_i}}-H_2(S_i|R_i)_{\rho_{S_iR_i}}})}.
\end{split}
\end{equation}
Proof is done.
\end{proof}

\subsection{Proof of smooth decoupling theorem or Theorem~\ref{thm:smoothdecoupling}}\label{app:smoothing}
\begin{proof}
Theorem~\ref{thm:nonsmoothdecoupling} already tells us that
\begin{equation}
\begin{split}
&\mathbb{E}_U\Vert \mathcal{T}_{S\rightarrow E}(U_S\rho_{SR}U_{S}^{\dagger}) - \tau_E\otimes \rho_R \Vert_1\\
\leq& \sqrt{2^{-\sum_{i=1}^{2N}H_2(S_i|E_i)_{\tau_{S_iE_i}}-\sum_{i=1}^{2N}H_2(S_i|R_i)_{\rho_{S_iR_i}}}+ c\prod_{i=1}^{2N} \max(1, 2^{-H_2(S_i|E_i)_{\tau_{S_iE_i}}-H_2(S_i|R_i)_{\rho_{S_iR_i}}})},
\end{split}
\end{equation}
where
\begin{equation}
c = 2((1+2\eta\max_i(\max(2^{-H_2(S_i|R_i)_{\rho_{S_iR_i}}}, 2^{H_2(S_i|R_i)_{\rho_{S_iR_i}}})))^{N-1}-1).
\end{equation}
The remaining part considers the smoothing of the conditional collision entropy. Take $\hat{\rho}_{S_iR_i}\in \mathcal{B}^{\delta}(\rho_{S_iR_i})$ such that $H^{\delta}_{2}(S_i|R_i)_{\rho_{S_iR_i}} = H_{2}(S_i|R_i)_{\hat{\rho}_{S_iR_i}}$ and $\hat{\tau}_{S_iE_i}\in \mathcal{B}^{\delta}(\tau_{S_iE_i})$ such that $H^{\delta}_{2}(S_i|E_i)_{\tau_{S_iE_i}} = H_{2}(S_i|E_i)_{\hat{\tau}_{S_iE_i}}$. By Eq.~\eqref{eq:tracedistanceineq}, we have $\Vert \rho_{S_iR_i} - \hat{\rho}_{S_iR_i} \Vert_1\leq 2\delta$ and $\Vert \tau_{S_iE_i} - \hat{\tau}_{S_iE_i} \Vert_1\leq 2\delta$. Thus,
\begin{equation}
\begin{split}
&\sqrt{2^{-\sum_{i=1}^{2N}H^{\delta}_{2}(S_i|E_i)_{\tau_{S_iE_i}}-\sum_{i=1}^{2N}H^{\delta}_{2}(S_i|R_i)_{\rho_{S_iR_i}}} + c'\prod_{i=1}^{2N} \max(1, 2^{-H^{\delta}_{2}(S_i|E_i)_{\tau_{S_iE_i}}-H^{\delta}_{2}(S_i|R_i)_{\rho_{S_iR_i}}})}\\
=&\sqrt{2^{-\sum_{i=1}^{2N}H_{2}(S_i|E_i)_{\hat{\tau}_{S_iE_i}}-\sum_{i=1}^{2N}H_{2}(S_i|R_i)_{\hat{\rho}_{S_iR_i}}}+ c'\prod_{i=1}^{2N} \max(1, 2^{-H_{2}(S_i|E_i)_{\hat{\tau}_{S_iE_i}}-H_{2}(S_i|R_i)_{\hat{\rho}_{S_iR_i}}})}\\
\geq& \mathbb{E}_U\Vert \hat{\mathcal{T}}_{S\rightarrow E}(U_S\hat{\rho}_{SR}U_{S}^{\dagger}) - \hat{\tau}_E\otimes \hat{\rho}_R \Vert_1\\
\geq& \mathbb{E}_U\Vert \hat{\mathcal{T}}_{S\rightarrow E}(U_S\hat{\rho}_{SR}U_{S}^{\dagger}) - \tau_E\otimes \rho_R \Vert_1 - \Vert \tau_E\otimes \rho_R - \hat{\tau}_E\otimes \hat{\rho}_R \Vert_1\\
\geq& \mathbb{E}_U\Vert \hat{\mathcal{T}}_{S\rightarrow E}(U_S\hat{\rho}_{SR}U_{S}^{\dagger}) - \tau_E\otimes \rho_R \Vert_1 - \sum_{i=1}^{2N}\Vert \tau_{E_i}\otimes \rho_{R_i} - \hat{\tau}_{E_i}\otimes \hat{\rho}_{R_i} \Vert_1\\
\geq& \mathbb{E}_U\Vert \hat{\mathcal{T}}_{S\rightarrow E}(U_S\hat{\rho}_{SR}U_{S}^{\dagger}) - \tau_E\otimes \rho_R \Vert_1 - \sum_{i=1}^{2N}(\Vert \tau_{E_i} - \hat{\tau}_{E_i} \Vert_1 +\Vert \rho_{R_i} - \hat{\rho}_{R_i} \Vert_1)\\
\geq& \mathbb{E}_U\Vert \hat{\mathcal{T}}_{S\rightarrow E}(U_S\hat{\rho}_{SR}U_{S}^{\dagger}) - \tau_E\otimes \rho_R \Vert_1 - \sum_{i=1}^{2N}(\Vert \tau_{S_iE_i} - \hat{\tau}_{S_iE_i} \Vert_1 + \Vert \rho_{S_iR_i} - \hat{\rho}_{S_iR_i} \Vert_1)\\
\geq& \mathbb{E}_U\Vert \hat{\mathcal{T}}_{S\rightarrow E}(U_S\hat{\rho}_{SR}U_{S}^{\dagger}) - \tau_E\otimes \rho_R \Vert_1 - 8N\delta\\
\geq& \mathbb{E}_U\Vert \mathcal{T}_{S\rightarrow E}(U_S\rho_{SR}U_{S}^{\dagger}) - \tau_E\otimes \rho_R \Vert_1 - \mathbb{E}_U\Vert \hat{\mathcal{T}}_{S\rightarrow E}(U_S\hat{\rho}_{SR}U_{S}^{\dagger}) - \hat{\mathcal{T}}_{S\rightarrow E}(U_S\rho_{SR}U_{S}^{\dagger})\Vert_1\\
&- \mathbb{E}_U\Vert \hat{\mathcal{T}}_{S\rightarrow E}(U_S\rho_{SR}U_{S}^{\dagger}) - \mathcal{T}_{S\rightarrow E}(U_S\rho_{SR}U_{S}^{\dagger}) \Vert_1 - 8N\delta,\\
\end{split}
\end{equation}
where the coefficient $c'$ is
\begin{equation}
\begin{split}
c' &= 2((1+2\eta\max_i(\max(2^{-H_{2}(S_i|R_i)_{\hat{\rho}_{S_iR_i}}}, 2^{H_2(S_i|R_i)_{\hat{\rho}_{S_iR_i}}})))^{N-1}-1)\\
&= 2((1+2\eta\max_i(\max(2^{-H^{\delta}_{2}(S_i|R_i)_{\rho_{S_iR_i}}}, 2^{H^{\delta}_2(S_i|R_i)_{\rho_{S_iR_i}}})))^{N-1}-1).
\end{split}
\end{equation}
In the derivation, we repeatedly use the triangle inequality of the one norm. In the fifth line, we utilize the fact that for quantum states $\rho_i$ and $\sigma_i$
\begin{equation}
\begin{split}
\Vert \bigotimes_{i=1}^{r} \rho_i - \bigotimes_{i=1}^{r} \sigma_i \Vert_1
&= \Vert \sum_{i=1}^{r} \bigotimes_{j=1}^{i-1} \rho_j  \otimes (\rho_i - \sigma_i) \otimes  \bigotimes_{j=i+1}^{r} \sigma_j \Vert_1\\
&\leq \sum_{i=1}^{r} \Vert \bigotimes_{j=1}^{i-1} \rho_j  \otimes (\rho_i - \sigma_i) \otimes  \bigotimes_{j=i+1}^{r} \sigma_j \Vert_1\\
&\leq \sum_{i=1}^r \prod_{j=1}^{i-1} \Vert\rho_j\Vert_{\infty}  \Vert \rho_i - \sigma_i \Vert_1 \prod_{j=i+1}^{r} \Vert\sigma_j\Vert_{\infty}\\
&\leq \sum_{i=1}^r \Vert \rho_i - \sigma_i \Vert_1.
\end{split}
\end{equation}
Note that $\Vert AB\Vert_1\leq \Vert A\Vert_1\Vert B\Vert_{\infty}$ and $\Vert \rho\Vert_{\infty}\leq 1$ for quantum state $\rho$.

Meanwhile, let $\delta^+_{SR}-\delta^-_{SR} = \hat{\rho}_{SR}-\rho_{SR}$ where $\delta^+_{SR}$ and $\delta^-_{SR}$ have orthogonal support and they are both positive. Since the two have orthogonal support, we have
\begin{equation}
\tr(\delta^+_{SR}) = \tr(\delta^-_{SR})  = \frac{1}{2}\Vert\hat{\rho}_{SR}-\rho_{SR} \Vert_1 \leq \frac{1}{2} \sum_{i=1}^{2N} \Vert\hat{\rho}_{S_iR_i}-\rho_{S_iR_i} \Vert_1 \leq 2N\delta.\\
\end{equation}
% \begin{align}
% \tr(\delta^+_{SR}) &\leq \Vert\hat{\rho}_{SR}-\rho_{SR} \Vert_1 \leq \sum_{i=1}^{2N} \Vert\hat{\rho}_{S_iR_i}-\rho_{S_iR_i} \Vert_1 \leq 2N\delta,\\
% \tr(\delta^-_{SR}) &\leq \Vert\hat{\rho}_{SR}-\rho_{SR} \Vert_1 \leq \sum_{i=1}^{2N} \Vert\hat{\rho}_{S_iR_i}-\rho_{S_iR_i} \Vert_1 \leq 2N\delta.
% \end{align}
Thus,
\begin{equation}
\begin{split}
&\mathbb{E}_U\Vert \hat{\mathcal{T}}_{S\rightarrow E}(U_S\hat{\rho}_{SR}U_{S}^{\dagger}) - \hat{\mathcal{T}}_{S\rightarrow E}(U_S\rho_{SR}U_{S}^{\dagger})\Vert_1\\
=&\mathbb{E}_U\Vert \hat{\mathcal{T}}_{S\rightarrow E}(U_S\delta^+_{SR}U_{S}^{\dagger}) - \hat{\mathcal{T}}_{S\rightarrow E}(U_S\delta^-_{SR}U_{S}^{\dagger})\Vert_1\\
\leq& \mathbb{E}_U\Vert \hat{\mathcal{T}}_{S\rightarrow E}(U_S\delta^+_{SR}U_{S}^{\dagger})\Vert_1 + \mathbb{E}_U\Vert \hat{\mathcal{T}}_{S\rightarrow E}(U_S\delta^-_{SR}U_{S}^{\dagger})\Vert_1\\
=& \mathbb{E}_U\tr[\hat{\mathcal{T}}_{S\rightarrow E}(U_S\delta^+_{SR}U_{S}^{\dagger})] + \mathbb{E}_U\tr[\hat{\mathcal{T}}_{S\rightarrow E}(U_S\delta^-_{SR}U_{S}^{\dagger})]\\
=& \tr(\delta^+_{SR}+\delta^-_{SR})\\
\leq& 4N\delta.
\end{split}
\end{equation}
In the fifth line, we utilize that $\hat{\mathcal{T}}_{S\rightarrow E}$ is trace-preserving. Similarly, let $\Delta^+_{SE}-\Delta^-_{SE} = \hat{\tau}_{SE}-\tau_{SE}$ where $\Delta^+_{SE}$ and $\Delta^-_{SE}$ have orthogonal support and they are both positive. We also have
\begin{equation}
\tr(\Delta^+_{SE}) = \tr(\Delta^-_{SE}) = \frac{1}{2} \Vert\hat{\tau}_{SE}-\tau_{SE} \Vert_1 \leq \frac{1}{2}\sum_{i=1}^{2N} \Vert\hat{\tau}_{S_iE_i}-\tau_{S_iE_i} \Vert_1 \leq 2N\delta.
\end{equation}
% \begin{align}
% \tr(\Delta^+_{SE}) &\leq \Vert\hat{\tau}_{SE}-\tau_{SE} \Vert_1 \leq \sum_{i=1}^{2N} \Vert\hat{\tau}_{S_iE_i}-\tau_{S_iE_i} \Vert_1 \leq 2N\delta,\\
% \tr(\Delta^-_{SE}) &\leq \Vert\hat{\tau}_{SE}-\tau_{SE} \Vert_1 \leq \sum_{i=1}^{2N} \Vert\hat{\tau}_{S_iE_i}-\tau_{S_iE_i} \Vert_1 \leq 2N\delta.
% \end{align}
Meanwhile, we denote the Choi-Jamio{\l}kowski representations of $\Delta^+_{SE}$ and $\Delta^-_{SE}$ as $\mathcal{D}^+_{S\rightarrow E}$ and $\mathcal{D}^-_{S\rightarrow E}$, respectively. Thus,
\begin{equation}
\begin{split}
&\mathbb{E}_U\Vert \hat{\mathcal{T}}_{S\rightarrow E}(U_S\rho_{SR}U_{S}^{\dagger}) - \mathcal{T}_{S\rightarrow E}(U_S\rho_{SR}U_{S}^{\dagger}) \Vert_1\\
=&\mathbb{E}_U\Vert \mathcal{D}^+_{S\rightarrow E}(U_S\rho_{SR}U_{S}^{\dagger}) - \mathcal{D}^-_{S\rightarrow E}(U_S\rho_{SR}U_{S}^{\dagger})\Vert_1\\
\leq& \mathbb{E}_U\Vert \mathcal{D}^+_{S\rightarrow E}(U_S\rho_{SR}U_{S}^{\dagger})\Vert_1 + \mathbb{E}_U\Vert \mathcal{D}^-_{S\rightarrow E}(U_S\rho_{SR}U_{S}^{\dagger})\Vert_1\\
=& \mathbb{E}_U\tr[\mathcal{D}^+_{S\rightarrow E}(U_S\rho_{SR}U_{S}^{\dagger})] + \mathbb{E}_U\tr[\mathcal{D}^-_{S\rightarrow E}(U_S\rho_{SR}U_{S}^{\dagger})]\\
=& \tr[\mathcal{D}^+_{S\rightarrow E}(\frac{\id_S}{\abs{S}}\otimes \rho_R)] + \tr[\mathcal{D}^-_{S\rightarrow E}(\frac{\id_S}{\abs{S}}\otimes \rho_R)]\\
=&\tr(\Delta^+_{SE}\otimes \rho_R)+\tr(\Delta^-_{SE}\otimes \rho_R)\\
=& \tr(\Delta^+_{SE}+\Delta^-_{SE})\\
\leq& 4N\delta.
\end{split}
\end{equation}
In the fifth line, we utilize the 1-design property of the random unitary such that
\begin{equation}
\mathbb{E}_UU_S\rho_{SR}U_{S}^{\dagger}=\frac{\id_S}{\abs{S}}\otimes \rho_R.
\end{equation}
Finally, we get
\begin{equation}
\begin{split}
&\mathbb{E}_U\Vert \mathcal{T}_{S\rightarrow E}(U_S\rho_{SR}U_{S}^{\dagger}) - \tau_E\otimes \rho_R \Vert_1\leq 16N\delta+\\
&\sqrt{2^{-\sum_{i=1}^{2N}H^{\delta}_{2}(S_i|E_i)_{\tau_{S_iE_i}}-\sum_{i=1}^{2N}H^{\delta}_{2}(S_i|R_i)_{\rho_{S_iR_i}}} + c'\prod_{i=1}^{2N} \max(1, 2^{-H^{\delta}_{2}(S_i|E_i)_{\tau_{S_iE_i}}-H^{\delta}_{2}(S_i|R_i)_{\rho_{S_iR_i}}})},
\end{split}
\end{equation}
where
\begin{equation}
c' = 2((1+2\eta\max_i(\max(2^{-H^{\delta}_{2}(S_i|R_i)_{\rho_{S_iR_i}}}, 2^{H^{\delta}_2(S_i|R_i)_{\rho_{S_iR_i}}})))^{N-1}-1).
\end{equation}
\end{proof}

Note that we can only smooth part of conditional entropies in the smoothing process. For instance, we only smooth $H_2(S_i|E_i)$ while maintaining $H_2(S_i|R_i)$. Then, we get
\begin{equation}\label{eq:smoothSE}
\begin{split}
&\mathbb{E}_U\Vert \mathcal{T}_{S\rightarrow E}(U_S\rho_{SR}U_{S}^{\dagger}) - \tau_E\otimes \rho_R \Vert_1\leq 8N\delta+\\
&\sqrt{2^{-\sum_{i=1}^{2N}H^{\delta}_{2}(S_i|E_i)_{\tau_{S_iE_i}}-\sum_{i=1}^{2N}H_{2}(S_i|R_i)_{\rho_{S_iR_i}}} + c\prod_{i=1}^{2N} \max(1, 2^{-H^{\delta}_{2}(S_i|E_i)_{\tau_{S_iE_i}}-H_{2}(S_i|R_i)_{\rho_{S_iR_i}}})},
\end{split}
\end{equation}
where
\begin{equation}
c = 2((1+2\eta\max_i(\max(2^{-H_2(S_i|R_i)_{\rho_{S_iR_i}}}, 2^{H_2(S_i|R_i)_{\rho_{S_iR_i}}})))^{N-1}-1).
\end{equation}
Similarly, we only smooth $H_2(S_i|R_i)$ while maintaining $H_2(S_i|E_i)$ to get
\begin{equation}
\begin{split}
&\mathbb{E}_U\Vert \mathcal{T}_{S\rightarrow E}(U_S\rho_{SR}U_{S}^{\dagger}) - \tau_E\otimes \rho_R \Vert_1\leq 8N\delta+\\
&\sqrt{2^{-\sum_{i=1}^{2N}H_{2}(S_i|E_i)_{\tau_{S_iE_i}}-\sum_{i=1}^{2N}H^{\delta}_{2}(S_i|R_i)_{\rho_{S_iR_i}}} + c'\prod_{i=1}^{2N} \max(1, 2^{-H_{2}(S_i|E_i)_{\tau_{S_iE_i}}-H^{\delta}_{2}(S_i|R_i)_{\rho_{S_iR_i}}})},
\end{split}
\end{equation}
where
\begin{equation}
c' = 2((1+2\eta\max_i(\max(2^{-H^{\delta}_{2}(S_i|R_i)_{\rho_{S_iR_i}}}, 2^{H^{\delta}_2(S_i|R_i)_{\rho_{S_iR_i}}})))^{N-1}-1).
\end{equation}
The proof of the above formulas is the same as that of Theorem~\ref{thm:smoothdecoupling}.

\subsection{Proof of Corollary~\ref{coro:pauli_nonsmoothing} by applying non-smooth decoupling theorem to Pauli noise}\label{app:pauli_nonsmoothing}
\begin{proof}
We adopt the complementary channel to analyze the Choi error. Specifically,
\begin{equation}
\epsilon_{\mathrm{Choi}} = \min_{\zeta} P\left( (\widehat{\mathcal{N}\circ \mathcal{E}}_{L\rightarrow E}\otimes I_R)(\ketbra{\hat{\phi}}_{LR}), (\mathcal{T}^{\zeta}_{L\rightarrow E}\otimes I_R)(\ketbra{\hat{\phi}}_{LR}) \right).
\end{equation}
Since $\mathcal{T}^{\zeta}_{L\rightarrow E}$ maps all states to a fixed state $\zeta$, $(\mathcal{T}^{\zeta}_{L\rightarrow E}\otimes I_R)(\ketbra{\hat{\phi}}_{LR}) = \zeta_E\otimes \tr_L(\ketbra{\hat{\phi}}_{LR}) = \zeta_E\otimes \frac{\id_R}{2^k}$. Meanwhile, $\widehat{\mathcal{N}\circ \mathcal{E}}_{L\rightarrow E} = \hat{\mathcal{N}}_{S\rightarrow E}\circ \mathcal{E}_{L\rightarrow S}$ with $\mathcal{E}_{L\rightarrow S}$ the encoding map satisfying
\begin{equation}
\mathcal{E}_{L\rightarrow S}(\rho) = U_S (\rho_L\otimes \ketbra{0}^{n-k}) U_S^{\dagger},
\end{equation}
where $U_S$ is the encoding unitary, generated by 1D $\log(n/\varepsilon)$-depth Clifford circuit ensemble $\mathfrak{C}_n^{\varepsilon}$ in our case. The term $\ketbra{0}^{n-k}$ is the ancillary qubits used for quantum error correction. Hence,
\begin{equation}
\begin{split}
\epsilon_{\mathrm{Choi}} =& \min_{\zeta} P\left( \hat{\mathcal{N}}_{S\rightarrow E}(U\ketbra{\hat{\phi}}_{LR}\otimes \ketbra{0^{n-k}} U^{\dagger}), \zeta_E\otimes \frac{\id_R}{2^k} \right)\\
\leq& P\left( \hat{\mathcal{N}}_{S\rightarrow E}(U\ketbra{\hat{\phi}}_{LR}\otimes \ketbra{0^{n-k}} U^{\dagger}), \mathbb{E}_{U\sim \mathfrak{U}_n}\hat{\mathcal{N}}_{S\rightarrow E}(U\ketbra{\hat{\phi}}_{LR}\otimes \ketbra{0^{n-k}} U^{\dagger}) \right)\\
&+ \min_{\zeta} P\left(\mathbb{E}_{U\sim \mathfrak{U}_n} \hat{\mathcal{N}}_{S\rightarrow E}(U\ketbra{\hat{\phi}}_{LR}\otimes \ketbra{0^{n-k}} U^{\dagger}), \zeta_E\otimes \frac{\id_R}{2^k} \right)\\
=& P\left( \hat{\mathcal{N}}_{S\rightarrow E}(U\ketbra{\hat{\phi}}_{LR}\otimes \ketbra{0^{n-k}} U^{\dagger}), \hat{\mathcal{N}}_{S\rightarrow E}(\frac{\id_S}{2^n})\otimes \frac{\id_R}{2^k} \right) + \min_{\zeta} P\left(\hat{\mathcal{N}}_{S\rightarrow E}(\frac{\id_S}{2^n})\otimes \frac{\id_R}{2^k}, \zeta_E\otimes \frac{\id_R}{2^k} \right)\\
=& P\left( \hat{\mathcal{N}}_{S\rightarrow E}(U\ketbra{\hat{\phi}}_{LR}\otimes \ketbra{0^{n-k}} U^{\dagger}), \hat{\mathcal{N}}_{S\rightarrow E}(\frac{\id_S}{2^n})\otimes \frac{\id_R}{2^k} \right).
\end{split}
\end{equation}
Here, in the first inequality, we utilize the triangle inequality of the purified distance. In the second equality, we utilize the property of the twirling over Haar random unitary gates. In the third equality, we take the state $\zeta$ as $\hat{\mathcal{N}}_{S\rightarrow E}(\frac{\id_S}{2^n})$.

To further bound $\mathbb{E}_{U\sim \mathfrak{C}_n^{\varepsilon}}\epsilon_{\mathrm{Choi}}$, we utilize the relation between purified distance and trace distance in Eq.~\eqref{eq:tracedistanceineq} to get
\begin{equation}\label{eq:Paulichoi}
\begin{split}
\mathbb{E}_{U} \epsilon_{\mathrm{Choi}} \leq& \mathbb{E}_{U} \sqrt{\Vert \hat{\mathcal{N}}_{S\rightarrow E}(U\ketbra{\hat{\phi}}_{LR}\otimes \ketbra{0^{n-k}} U^{\dagger}) - \hat{\mathcal{N}}_{S\rightarrow E}(\frac{\id_S}{2^n})\otimes \frac{\id_R}{2^k} \Vert_1}\\
\leq& \sqrt{\mathbb{E}_{U} \Vert \hat{\mathcal{N}}_{S\rightarrow E}(U\ketbra{\hat{\phi}}_{LR}\otimes \ketbra{0^{n-k}} U^{\dagger}) - \hat{\mathcal{N}}_{S\rightarrow E}(\frac{\id_S}{2^n})\otimes \frac{\id_R}{2^k} \Vert_1}.
\end{split}
\end{equation}
For simplicity, we omit $\mathfrak{C}_n^{\varepsilon}$ in the subscript of expectation. For i.i.d.~Pauli noise, we have
\begin{equation}
\hat{\mathcal{N}}_{S\rightarrow E} = (\hat{\mathcal{N}}_p)^{\otimes n}
\end{equation}
with $\hat{\mathcal{N}}_p$ as the complementary channel of the single-qubit Pauli noise $\mathcal{N}_p(\rho) = p_I\rho+p_XX\rho X+p_YY\rho Y+p_ZZ\rho Z$.
% Note that the complementary channel~\cite{Beny2010AQEC} is
% \begin{equation}
% \hat{\mathcal{N}}_p = \sum_{ij}\tr(K_i\rho K_j^{\dagger})\ketbra{i}{j},
% \end{equation}
% with $K_i$ the Kraus operators for the noise channel. For the single-qubit Pauli noise, the Kraus operators are
% \begin{equation}
% K_0 = \sqrt{p_I}I, K_1 = \sqrt{p_X}X, K_2 = \sqrt{p_Y}Y, K_3 = \sqrt{p_Z}Z.
% \end{equation}

Now we apply the non-smooth decoupling theorem. Based on Theorem~\ref{thm:nonsmoothdecoupling}, the term in the square root in Eq.~\eqref{eq:Paulichoi} satisfies
\begin{equation}
\begin{split}
&\mathbb{E}_{U} \Vert \hat{\mathcal{N}}_{S\rightarrow E}(U\ketbra{\hat{\phi}}_{LR}\otimes \ketbra{0^{n-k}} U^{\dagger}) - \hat{\mathcal{N}}_{S\rightarrow E}(\frac{\id_S}{2^n})\otimes \frac{\id_R}{2^k} \Vert_1\\
\leq&\sqrt{2^{-\sum_{i=1}^{2N}H_{2}(S_i|E_i)_{\tau_{S_iE_i}}-\sum_{i=1}^{2N}H_{2}(S_i|R_i)_{\rho_{S_iR_i}}} + c\prod_{i=1}^{2N} \max(1, 2^{-H_{2}(S_i|E_i)_{\tau_{S_iE_i}}-H_{2}(S_i|R_i)_{\rho_{S_iR_i}}})},\\
\end{split}
\end{equation}
where $\tau_{S_iE_i} = (\hat{\mathcal{N}}_p(\ketbra{\hat{\phi}}))^{\otimes \xi}$, $\rho_{S_iR_i} = \ketbra{\hat{\phi}}^{\frac{k}{n}\xi}\otimes \ketbra{0^{\frac{n-k}{n}\xi}}$, and
\begin{equation}
c = 2((1+2\eta\max_i(\max(2^{-H_{2}(S_i|R_i)_{\rho_{S_iR_i}}}, 2^{H_2(S_i|R_i)_{\rho_{S_iR_i}}})))^{N-1}-1).
\end{equation}
Note that $H_{2}(S_i|R_i)_{\rho_{S_iR_i}}$ is given by
\begin{equation}\label{eq:EPRcollisionentropy}
\begin{split}
H_2(S_i|R_i)_{\rho_{S_iR_i}} &= H_2(L_i|R_i)_{\ketbra{\hat{\phi}}_{L_iR_i}}\\
&= \frac{k}{n}\xi H_2(L|R)_{\ketbra{\hat{\phi}}}\\
&= -\frac{k}{n}\xi.
\end{split}
\end{equation}

Also, given state $\ketbra{\psi}_{SS'}$ and the purification of $\hat{\mathcal{N}}_{S'\rightarrow E}(\ketbra{\psi}_{SS'})$ denoted as $\ket{\phi}_{SS'E}$, from Lemma~\ref{lemma:entropydual} we can get
\begin{equation}\label{eq:dualcollisionentropy}
\begin{split}
H_2(S|E)_{\hat{\mathcal{N}}_{S'\rightarrow E}(\ketbra{\psi}_{SS'})} &= H_2(S|E)_{\tr_{S'}(\ketbra{\phi})}\\
&= -\log \tr_{S'}( (\tr_S( \sqrt{\tr_E(\ketbra{\phi})} ))^2)\\
&=-\log \tr_{S'}( (\tr_S(\sqrt{\mathcal{N}_{S'}(\ketbra{\psi}_{SS'})}))^2).
\end{split}
\end{equation}
Thus,
\begin{equation}\label{eq:paulicollision}
\begin{split}
H_2(S_i|E_i)_{\tau_{S_iE_i}} &= \xi H_2(S|E)_{\hat{\mathcal{N}}_p(\ketbra{\hat{\phi}})}\\
&= -\xi \log \tr_{S'}( (\tr_S(\sqrt{\mathcal{N}_{p}(\ketbra{\hat{\phi}}_{SS'})}))^2)\\
&= -\xi \log \tr_{S'}( (\tr_S(\sqrt{p_I \ketbra{\hat{\phi}} + p_X X\ketbra{\hat{\phi}}X + p_Y Y\ketbra{\hat{\phi}}Y + p_Z Z\ketbra{\hat{\phi}}Z}))^2)\\
&= -\xi \log \tr_{S'}( (\tr_S(\sqrt{p_I} \ketbra{\hat{\phi}} + \sqrt{p_X} X\ketbra{\hat{\phi}}X + \sqrt{p_Y} Y\ketbra{\hat{\phi}}Y + \sqrt{p_Z} Z\ketbra{\hat{\phi}}Z))^2)\\
&= -\xi \log \tr_{S'}( ((\sqrt{p_I} + \sqrt{p_X} + \sqrt{p_Y} + \sqrt{p_Z})\frac{\id}{2})^2)\\
&= \xi(1 - 2\log (\sqrt{p_I} + \sqrt{p_X} + \sqrt{p_Y} + \sqrt{p_Z})).
\end{split}
\end{equation}
Denote $f(\Vec{p}) = 2\log (\sqrt{p_I} + \sqrt{p_X} + \sqrt{p_Y} + \sqrt{p_Z})$, we have $H_2(S_i|E_i)_{\tau_{S_iE_i}} = 1-f(\Vec{p})$. Below, we consider the region of $\Vec{p}$ such that $0 < f(\Vec{p}) < 1$. For sufficiently large $n$, utilizing $(1+\frac{x}{n})^n\leq 1+2x$ for small $x$, we have that
\begin{equation}\label{eq:boundofcoefc}
\begin{split}
c &\leq 8N\max(2^{-\xi(1+\frac{k}{n})}, 2^{-\xi(1-\frac{k}{n})})\\
&\leq \frac{4n2^{-\xi (1-\frac{k}{n})}}{\xi}\\
&= \frac{4\varepsilon^{1-\frac{k}{n}}n^{\frac{k}{n}}}{\log (n/\varepsilon)}.
\end{split}
\end{equation}
The condition that $x$ is small can be achieved by $(\frac{\varepsilon}{n})^{1-\frac{k}{n}} = o(1/n)$.

Thus, for sufficiently large $n$, we have that
\begin{equation}\label{eq:twolayerNonsmooth}
\begin{split}
&\mathbb{E}_{U} \Vert \hat{\mathcal{N}}_{S\rightarrow E}(U\ketbra{\hat{\phi}}_{LR}\otimes \ketbra{0^{n-k}} U^{\dagger}) - \hat{\mathcal{N}}_{S\rightarrow E}(\frac{\id_S}{2^n})\otimes \frac{\id_R}{2^k} \Vert_1\\
\leq&\sqrt{2^{-\sum_{i=1}^{2N}H_{2}(S_i|E_i)_{\tau_{S_iE_i}}-\sum_{i=1}^{2N}H_{2}(S_i|R_i)_{\rho_{S_iR_i}}} + c\prod_{i=1}^{2N} \max(1, 2^{-H_{2}(S_i|E_i)_{\tau_{S_iE_i}}-H_{2}(S_i|R_i)_{\rho_{S_iR_i}}})}\\
\leq&\sqrt{2^{-n(1-f(\Vec{p})-\frac{k}{n})} + \frac{4\varepsilon^{1-\frac{k}{n}}n^{\frac{k}{n}}}{\log (n/\varepsilon)}\prod_{i=1}^{2N} \max(1, 2^{-\xi(1-f(\Vec{p})-\frac{k}{n})})}.
\end{split}
\end{equation}
From the above equation, as long as $\varepsilon$ decays faster than $n^{-\frac{k}{n-k}}$, and encoding rate $k/n\leq 1-f(\Vec{p})$, the Choi error decays polynomially. Hence, we conclude that 1D $O(\log(n))$-depth circuit achieves encoding rate $1-f(\Vec{p})$. As a reminder, this rate is lower than the hashing bound $1-h(\Vec{p})$.
\end{proof}

\subsection{Proof of Corollary~\ref{coro:pauli} that
1D low-depth circuit achieves the hashing bound with smooth decoupling theorem}\label{app:hashingbound}
\begin{proof}

We follow the derivation in the last subsection but substitute the non-smooth decoupling theorem with the smooth decoupling theorem. Based on Eq.~\eqref{eq:smoothSE}, the term in the square root in Eq.~\eqref{eq:Paulichoi} satisfies
\begin{equation}
\begin{split}
&\mathbb{E}_{U} \Vert \hat{\mathcal{N}}_{S\rightarrow E}(U\ketbra{\hat{\phi}}_{LR}\otimes \ketbra{0^{n-k}} U^{\dagger}) - \hat{\mathcal{N}}_{S\rightarrow E}(\frac{\id_S}{2^n})\otimes \frac{\id_R}{2^k} \Vert_1\leq 8\delta+\\
&\sqrt{2^{-\sum_{i=1}^{2N}H^{\delta/N}_{2}(S_i|E_i)_{\tau_{S_iE_i}}-\sum_{i=1}^{2N}H_{2}(S_i|R_i)_{\rho_{S_iR_i}}} + c\prod_{i=1}^{2N} \max(1, 2^{-H^{\delta/N}_{2}(S_i|E_i)_{\tau_{S_iE_i}}-H_2(S_i|R_i)_{\rho_{S_iR_i}}})}.
\end{split}
\end{equation}
where $\tau_{S_iE_i} = \hat{\mathcal{N}}_{S'\rightarrow E}(\ketbra{\hat{\phi}}_{SS'}) = (\hat{\mathcal{N}}_p(\ketbra{\hat{\phi}}))^{\otimes \xi}$, $\rho_{S_iR_i} = \ketbra{\hat{\phi}}_{L_iR_i}\otimes \ketbra{0^{\frac{n-k}{n}\xi}} = \ketbra{\hat{\phi}}^{\frac{k}{n}\xi}\otimes \ketbra{0^{\frac{n-k}{n}\xi}}$, and
\begin{equation}
c = 2((1+2\eta\max_i(\max(2^{-H_{2}(S_i|R_i)_{\rho_{S_iR_i}}}, 2^{H_2(S_i|R_i)_{\rho_{S_iR_i}}})))^{N-1}-1).
\end{equation}

By Eq.~\eqref{eq:EPRcollisionentropy}, we obtain that $H_2(S_i|R_i)_{\rho_{S_iR_i}} = -\frac{k}{n}\xi$. Below, we analyze $H^{\delta/N}_{2}(S_i|E_i)_{\tau_{S_iE_i}}$. Note that for any $0 < \delta < 1$ as long as $\xi\geq -\frac{8}{5}\log(1-\sqrt{1-\delta^2/N^2})$, we have
\begin{equation}
H(S|E)_{\rho} - \frac{4\log \gamma \sqrt{\log \frac{2N^2}{\delta^2}}}{\sqrt{\xi}} \leq \frac{1}{\xi}H^{\delta/N}_{2}(S|E)_{\rho^{\otimes \xi}} \leq H(S|E)_{\rho} + \frac{4\log \gamma \sqrt{\log \frac{2N^2}{\delta^2}}}{\sqrt{\xi}}.
\end{equation}
with
\begin{equation}
\gamma \leq \sqrt{2^{-H^{\delta}_{\mathrm{min}}(S|E)_{\rho}}}+\sqrt{2^{H^{\delta}_{\mathrm{max}}(S|E)_{\rho}}}+1.
\end{equation}
The condition $\xi\geq -\frac{8}{5}\log(1-\sqrt{1-\delta^2/N^2})$ can be fulfilled as long as $n^{\frac{11}{8}}\epsilon^{\frac{5}{8}}$ decays with $n$. In this case, given any $0 < \delta < 1$, for sufficiently large $n$, we have
\begin{equation}
2N^2(\frac{\epsilon}{n})^{\frac{5}{8}}\leq \delta^2,
\end{equation}
which implies
\begin{equation}
(1-(\frac{\epsilon}{n})^{\frac{5}{8}})^2\geq 1-\frac{\delta^2}{N^2},
\end{equation}
and
\begin{equation}
\xi\geq -\frac{8}{5}\log(1-\sqrt{1-\delta^2/N^2}).
\end{equation}

Note that based on the definition of the complementary channel~\cite{Beny2010AQEC}, given state $\ketbra{\psi}_{SS'}$ and the purification of $\hat{\mathcal{N}}_{S'\rightarrow E}(\ketbra{\psi}_{SS'})$ denoted as $\ket{\phi}_{SS'E}$, we can get
\begin{equation}
\tr_E \ketbra{\phi}_{SS'E} = \mathcal{N}_{S'}(\ketbra{\psi}_{SS'}).
\end{equation}
Thus,
\begin{equation}\label{eq:dualvonneumann}
\begin{split}
H(S|E)_{\hat{\mathcal{N}}_{S'\rightarrow E}(\ketbra{\psi}_{SS'})} &= H(SE)-H(E)\\
&= H(S')-H(SS')\\
&=-H(S|S')_{\mathcal{N}_{S'}(\ketbra{\psi}_{SS'})}.
\end{split}
\end{equation}
With direct calculation, we have that
\begin{equation}
\mathcal{N}_p(\ketbra{\hat{\phi}}_{SS'}) = p_I \ketbra{\hat{\phi}} + p_X X\ketbra{\hat{\phi}}X + p_Y Y\ketbra{\hat{\phi}}Y + p_Z Z\ketbra{\hat{\phi}}Z,
\end{equation}
and
\begin{equation}
\tr_{S}\mathcal{N}_p(\ketbra{\hat{\phi}}_{SS'}) = \frac{\id}{2}.
\end{equation}Note that $\ketbra{\hat{\phi}}$, $X\ketbra{\hat{\phi}}X$, $Y\ketbra{\hat{\phi}}Y$, and $Z\ketbra{\hat{\phi}}Z$ are mutually orthogonal. Hence, $H(S|E)_{\hat{\mathcal{N}}_p(\ketbra{\hat{\phi}})} = H(\tr_{S}\mathcal{N}_p(\ketbra{\hat{\phi}}_{SS'})) - H(\mathcal{N}_p(\ketbra{\hat{\phi}}_{SS'})) = 1-h(\Vec{p})$. Meanwhile, for any $0<\delta<1$,
\begin{equation}
\lim_{n\rightarrow \infty}4\log\gamma\sqrt{\frac{\log \frac{2N^2}{\delta^2}}{\xi}} = \lim_{n\rightarrow \infty}4\log\gamma\sqrt{\frac{\log \frac{2N^2}{\delta^2}}{\log \frac{n}{\epsilon}}} = 4\log\gamma\sqrt{2\lim_{n\rightarrow \infty} \frac{\log n}{\log \frac{n}{\varepsilon}}}.
\end{equation}
As long as $\lim_{n\rightarrow \infty} \frac{\log n}{\log \frac{n}{\varepsilon}} = 0$, for sufficiently large $n$, we have that
\begin{equation}\label{eq:deltaineq}
4\log\gamma\sqrt{\frac{\log \frac{2N^2}{\delta^2}}{\xi}} \leq \delta.
\end{equation}
This can be achieved by choosing $\log n/\varepsilon$ as a function increasing slightly faster than $\log n$ like $(\log n)\log\log\log n$. This additional factor $\log\log\log n$ is normally negligible for a practical number $n$. Also, recall that by Eq.~\eqref{eq:boundofcoefc}, $c\leq \frac{4\varepsilon^{1-\frac{k}{n}}n^{\frac{k}{n}}}{\log (n/\varepsilon)}$.
Thus, for sufficiently large $n$ and any $0<\delta<1$, we have that
\begin{equation}
\begin{split}
&\mathbb{E}_{U} \Vert \hat{\mathcal{N}}_{S\rightarrow E}(U\ketbra{\hat{\phi}}_{LR}\otimes \ketbra{0^{n-k}} U^{\dagger}) - \hat{\mathcal{N}}_{S\rightarrow E}(\frac{\id_S}{2^n})\otimes \frac{\id_R}{2^k} \Vert_1\\
\leq&8\delta+\sqrt{2^{-\sum_{i=1}^{2N}\xi(1-h(\Vec{p})-\frac{k}{n}-\delta)} + \frac{4\varepsilon^{1-\frac{k}{n}}n^{\frac{k}{n}}}{\log (n/\varepsilon)}\prod_{i=1}^{2N} \max(1, 2^{-\xi(1-h(\Vec{p})-\frac{k}{n}-\delta)})}\\
=&8\delta+\sqrt{2^{-n(1-h(\Vec{p})-\frac{k}{n}-\delta)} + \frac{4\varepsilon^{1-\frac{k}{n}}n^{\frac{k}{n}}}{\log (n/\varepsilon)}\prod_{i=1}^{2N} \max(1, 2^{-\xi(1-h(\Vec{p})-\frac{k}{n}-\delta)})}.
\end{split}
\end{equation}
When $1-h(\Vec{p})-\frac{k}{n}-\delta > 0$, we have that
\begin{equation}\label{eq:twolayerhashing}
\mathbb{E}_{U} \Vert \hat{\mathcal{N}}_{S\rightarrow E}(U\ketbra{\hat{\phi}}_{LR}\otimes \ketbra{0^{n-k}} U^{\dagger}) - \hat{\mathcal{N}}_{S\rightarrow E}(\frac{\id_S}{2^n})\otimes \frac{\id_R}{2^k} \Vert_1\leq 8\delta+\sqrt{2^{-n(1-h(\Vec{p})-\frac{k}{n}-\delta)} + \frac{4\varepsilon^{1-\frac{k}{n}}n^{\frac{k}{n}}}{\log (n/\varepsilon)}}.
\end{equation}
Substituting the above equation into Eq.~\eqref{eq:Paulichoi} completes the proof. The term $2^{-n(1-h(\Vec{p})-\frac{k}{n}-\delta)}$ is the error associated with the random Clifford encoding, as demonstrated in Appendix~\ref{app:Clif}.

As a remark, the decoupling error within Eq.~\eqref{eq:twolayerhashing} is dominated by $8\delta$. Thus, the scaling of $\delta$ controls the decay rate of the decoupling error and the associated Choi error. From Eq.~\eqref{eq:deltaineq}, we obtain that $\delta$ can be chosen as the scaling of $O(\sqrt{\frac{\log n}{\log \frac{n}{\varepsilon}}})$. Thus, the scaling of the Choi error is $O((\frac{\log n}{\log \frac{n}{\varepsilon}})^{\frac{1}{4}})$.

\end{proof}

\subsection{Proof of Corollary~\ref{coro:iiderasure_nonsmoothing} by applying non-smooth decoupling theorem to i.i.d.~erasure error}\label{app:iiderasure_nondecouple}
\begin{proof}
Following the previous derivation, the expectation of Choi error over $\mathfrak{C}_n^{\varepsilon}$ can be bounded by:
\begin{equation}\label{eq:erasurechoi}
\mathbb{E}_{U} \epsilon_{\mathrm{Choi}} \leq \sqrt{\mathbb{E}_{U} \Vert \hat{\mathcal{N}}_{S\rightarrow E}(U\ketbra{\hat{\phi}}_{LR}\otimes \ketbra{0^{n-k}} U^{\dagger}) - \hat{\mathcal{N}}_{S\rightarrow E}(\frac{\id_S}{2^n})\otimes \frac{\id_R}{2^k} \Vert_1}.
\end{equation}
For i.i.d.~erasure error, we have
\begin{equation}
\hat{\mathcal{N}}_{S\rightarrow E} = (\hat{\mathcal{N}}_e)^{\otimes n}
\end{equation}
with $\hat{\mathcal{N}}_e$ as the complementary channel of the single-qubit erasure error $\mathcal{N}_e(\rho) = (1-p)\rho + p\ketbra{2}$.

Applying the non-smooth decoupling theorem or Theorem~\ref{thm:nonsmoothdecoupling}, the term in the square root in Eq.~\eqref{eq:Paulichoi} satisfies
\begin{equation}
\begin{split}
&\mathbb{E}_{U} \Vert \hat{\mathcal{N}}_{S\rightarrow E}(U\ketbra{\hat{\phi}}_{LR}\otimes \ketbra{0^{n-k}} U^{\dagger}) - \hat{\mathcal{N}}_{S\rightarrow E}(\frac{\id_S}{2^n})\otimes \frac{\id_R}{2^k} \Vert_1\\
\leq&\sqrt{2^{-\sum_{i=1}^{2N}H_{2}(S_i|E_i)_{\tau_{S_iE_i}}-\sum_{i=1}^{2N}H_{2}(S_i|R_i)_{\rho_{S_iR_i}}} + c\prod_{i=1}^{2N} \max(1, 2^{-H_{2}(S_i|E_i)_{\tau_{S_iE_i}}-H_{2}(S_i|R_i)_{\rho_{S_iR_i}}})},\\
\end{split}
\end{equation}
where $\tau_{S_iE_i} = (\hat{\mathcal{N}}_e(\ketbra{\hat{\phi}}))^{\otimes \xi}$, $\rho_{S_iR_i} = \ketbra{\hat{\phi}}^{\frac{k}{n}\xi}\otimes \ketbra{0^{\frac{n-k}{n}\xi}}$, and
\begin{equation}
c = 2((1+2\eta\max_i(\max(2^{-H_{2}(S_i|R_i)_{\rho_{S_iR_i}}}, 2^{H_2(S_i|R_i)_{\rho_{S_iR_i}}})))^{N-1}-1).
\end{equation}
Note that $H_{2}(S_i|R_i)_{\rho_{S_iR_i}} = -\frac{k}{n}\xi$ given by Eq.~\eqref{eq:EPRcollisionentropy}. Also, from Eq.~\eqref{eq:dualcollisionentropy} we can get
\begin{equation}
\begin{split}
H_2(S_i|E_i)_{\tau_{S_iE_i}} &= \xi H_2(S|E)_{\hat{\mathcal{N}}_e(\ketbra{\hat{\phi}})}\\
&= -\xi \log \tr_{S'}( (\tr_S(\sqrt{\mathcal{N}_{e}(\ketbra{\hat{\phi}}_{SS'})}))^2)\\
&= -\xi \log \tr_{S'}( (\tr_S(\sqrt{ (1-p)\ketbra{\hat{\phi}}_{SS'} + p \frac{\id_S}{2}\otimes \ketbra{2}_{S'} }))^2)\\
&= -\xi \log \tr_{S'}( (\tr_S( \sqrt{1-p}\ketbra{\hat{\phi}}_{SS'} + \sqrt{\frac{p}{2}}\id_S\otimes \ketbra{2}_{S'} ) )^2)\\
&= -\xi \log \tr_{S'}( ( \sqrt{1-p}\frac{\id_{S'}}{2} + \sqrt{2p} \ketbra{2}_{S'} )^2)\\
&= -\xi \log \tr_{S'}( \frac{1-p}{4}\id_{S'} + 2p \ketbra{2}_{S'} )\\
&= -\xi \log( \frac{1-p}{2} + 2p )\\
&= \xi (1-\log(1+3p)). \\
\end{split}
\end{equation}
By Eq.~\eqref{eq:boundofcoefc}, $c\leq \frac{4\varepsilon^{1-\frac{k}{n}}n^{\frac{k}{n}}}{\log (n/\varepsilon)}$.
Thus, for sufficiently large $n$, we have that
\begin{equation}
\begin{split}
&\mathbb{E}_{U} \Vert \hat{\mathcal{N}}_{S\rightarrow E}(U\ketbra{\hat{\phi}}_{LR}\otimes \ketbra{0^{n-k}} U^{\dagger}) - \hat{\mathcal{N}}_{S\rightarrow E}(\frac{\id_S}{2^n})\otimes \frac{\id_R}{2^k} \Vert_1\\
\leq&\sqrt{2^{-\sum_{i=1}^{2N}H_{2}(S_i|E_i)_{\tau_{S_iE_i}}-\sum_{i=1}^{2N}H_{2}(S_i|R_i)_{\rho_{S_iR_i}}} + c\prod_{i=1}^{2N} \max(1, 2^{-H_{2}(S_i|E_i)_{\tau_{S_iE_i}}-H_{2}(S_i|R_i)_{\rho_{S_iR_i}}})}\\
\leq&\sqrt{2^{-n(1-\log(1+3p)-\frac{k}{n})} + \frac{4\varepsilon^{1-\frac{k}{n}}n^{\frac{k}{n}}}{\log (n/\varepsilon)}\prod_{i=1}^{2N} \max(1, 2^{-\xi(1-\log(1+3p)-\frac{k}{n})})}.
\end{split}
\end{equation}
Substituting the above equation into Eq.~\eqref{eq:erasurechoi} completes the proof. From the above equation, we have that as long as $\varepsilon$ decays faster than $n^{-\frac{k}{n-k}}$, and encoding rate $k/n\leq 1-\log(1+3p)$, the Choi error decays polynomially. Hence, we conclude that 1D $O(\log(n))$-depth circuit achieves encoding rate $1-\log(1+3p)$ for i.i.d.~erasure error.
\end{proof}

\subsection{Proof of Corollary~\ref{coro:iiderasure} by applying smooth decoupling theorem to i.i.d.~erasure error}\label{app:iiderasure}
\begin{proof}
We follow the derivation in the last subsection but substitute the non-smooth decoupling theorem with the smooth decoupling theorem. Based on Eq.~\eqref{eq:smoothSE}, we have
\begin{equation}
\begin{split}
&\mathbb{E}_{U} \Vert \hat{\mathcal{N}}_{S\rightarrow E}(U\ketbra{\hat{\phi}}_{LR}\otimes \ketbra{0^{n-k}} U^{\dagger}) - \hat{\mathcal{N}}_{S\rightarrow E}(\frac{\id_S}{2^n})\otimes \frac{\id_R}{2^k} \Vert_1\leq 8\delta+\\
&\sqrt{2^{-\sum_{i=1}^{2N}H^{\delta/N}_{2}(S_i|E_i)_{\tau_{S_iE_i}}-\sum_{i=1}^{2N}H_{2}(S_i|R_i)_{\rho_{S_iR_i}}} + c\prod_{i=1}^{2N} \max(1, 2^{-H^{\delta/N}_{2}(S_i|E_i)_{\tau_{S_iE_i}}-H_2(S_i|R_i)_{\rho_{S_iR_i}}})}.
\end{split}
\end{equation}
where $\tau_{S_iE_i} = \hat{\mathcal{N}}_{S'\rightarrow E}(\ketbra{\hat{\phi}}_{SS'}) = (\hat{\mathcal{N}}_e(\ketbra{\hat{\phi}}))^{\otimes \xi}$, $\rho_{S_iR_i} = \ketbra{\hat{\phi}}_{L_iR_i}\otimes \ketbra{0^{\frac{n-k}{n}\xi}} = \ketbra{\hat{\phi}}^{\frac{k}{n}\xi}\otimes \ketbra{0^{\frac{n-k}{n}\xi}}$, and
\begin{equation}
c = 2((1+2\eta\max_i(\max(2^{-H_{2}(S_i|R_i)_{\rho_{S_iR_i}}}, 2^{H_2(S_i|R_i)_{\rho_{S_iR_i}}})))^{N-1}-1).
\end{equation}
Note that $H_{2}(S_i|R_i)_{\rho_{S_iR_i}} = -\frac{k}{n}\xi$ given by Eq.~\eqref{eq:EPRcollisionentropy}. Same with Appendix~\ref{app:hashingbound}, as long as $\lim_{n\rightarrow \infty} \frac{\log n}{\log \frac{n}{\varepsilon}} = 0$, for sufficiently large $n$, we have that
\begin{equation}
\abs{\frac{1}{\xi}H^{\delta/N}_{2}(S|E)_{\rho^{\otimes \xi}} - H(S|E)_{\rho}}\leq \delta.
\end{equation}
Meanwhile, from Eq.~\eqref{eq:dualvonneumann} we have
\begin{equation}
\begin{split}
H(S_i|E_i)_{\tau_{S_iE_i}} &= \xi H(S|E)_{\hat{\mathcal{N}}_e(\ketbra{\hat{\phi}})}\\
&= -\xi H(S|S')_{\mathcal{N}_e(\ketbra{\hat{\phi}})}\\
&= -\xi H(S|S')_{(1-p)\ketbra{\hat{\phi}}_{SS'} + p \frac{\id_S}{2}\otimes \ketbra{2}_{S'}}\\
&= \xi(H(S')_{\frac{1-p}{2}\id_{S'}+p\ketbra{2}_{S'}} - H(SS')_{(1-p)\ketbra{\hat{\phi}}_{SS'} + p \frac{\id_S}{2}\otimes \ketbra{2}_{S'}})\\
&= \xi(-p\log p- (1-p)\log\frac{1-p}{2} + (1-p)\log(1-p) + p\log \frac{p}{2})\\
&= \xi(1-2p).\\
\end{split}
\end{equation}
Also, by Eq.~\eqref{eq:boundofcoefc}, $c\leq \frac{4\varepsilon^{1-\frac{k}{n}}n^{\frac{k}{n}}}{\log (n/\varepsilon)}$. Thus, for sufficiently large $n$ and $0<\delta<1$, we have that
\begin{equation}
\begin{split}
&\mathbb{E}_{U} \Vert \hat{\mathcal{N}}_{S\rightarrow E}(U\ketbra{\hat{\phi}}_{LR}\otimes \ketbra{0^{n-k}} U^{\dagger}) - \hat{\mathcal{N}}_{S\rightarrow E}(\frac{\id_S}{2^n})\otimes \frac{\id_R}{2^k} \Vert_1\\
\leq&8\delta+\sqrt{2^{-\sum_{i=1}^{2N}\xi(1-2p-\frac{k}{n}-\delta)} + \frac{4\varepsilon^{1-\frac{k}{n}}n^{\frac{k}{n}}}{\log (n/\varepsilon)}\prod_{i=1}^{2N} \max(1, 2^{-\xi(1-2p-\frac{k}{n}-\delta)})}\\
=&8\delta+\sqrt{2^{-n(1-2p-\frac{k}{n}-\delta)} + \frac{4\varepsilon^{1-\frac{k}{n}}n^{\frac{k}{n}}}{\log (n/\varepsilon)}\prod_{i=1}^{2N} \max(1, 2^{-\xi(1-2p-\frac{k}{n}-\delta)})}.
\end{split}
\end{equation}
When $1-2p-\frac{k}{n}-\delta > 0$, we have that
\begin{equation}
\mathbb{E}_{U} \Vert \hat{\mathcal{N}}_{S\rightarrow E}(U\ketbra{\hat{\phi}}_{LR}\otimes \ketbra{0^{n-k}} U^{\dagger}) - \hat{\mathcal{N}}_{S\rightarrow E}(\frac{\id_S}{2^n})\otimes \frac{\id_R}{2^k} \Vert_1\leq 8\delta+\sqrt{2^{-n(1-2p-\frac{k}{n}-\delta)} + \frac{4\varepsilon^{1-\frac{k}{n}}n^{\frac{k}{n}}}{\log (n/\varepsilon)}}.
\end{equation}
Substituting the above equation into Eq.~\eqref{eq:erasurechoi} completes the proof. From the above equation, we have that the encoding rate $k/n$ can achieve $1-2p$ for i.i.d.~erasure error with 1D $\omega(\log(n))$-depth circuit. This encoding rate is also the quantum channel capacity of the erasure error.
\end{proof}

\subsection{Proof of Corollary~\ref{coro:amp_nonsmoothing} by applying non-smooth decoupling theorem to local amplitude damping noise}\label{app:amp_nondecouple}
\begin{proof}
Following the previous derivation for local Pauli noise and erasure errors, the key point is to evaluate $H_{2}(S_i|R_i)_{\rho_{S_iR_i}}$ and $H_2(S_i|E_i)_{\tau_{S_iE_i}}$. Similar to previous analysis, $H_{2}(S_i|R_i)_{\rho_{S_iR_i}} = -\frac{k}{n}\xi$.

For local amplitude damping noise, we have
\begin{equation}
\hat{\mathcal{N}}_{S\rightarrow E} = (\hat{\mathcal{N}}_a)^{\otimes n},
\end{equation}
with $\hat{\mathcal{N}}_a$ as the complementary channel of the amplitude damping noise. Thus,
\begin{equation}
\begin{split}
H_2(S_i|E_i)_{\tau_{S_iE_i}}
&= -\xi \log \tr_{S'}( (\tr_S(\sqrt{\mathcal{N}_{a}(\ketbra{\hat{\phi}}_{SS'})}))^2)\\
&= -\xi \log \tr_{S'}( (\tr_S(\sqrt{ \frac{1}{2}(\ket{00}+\sqrt{1-p}\ket{11})(\bra{00}+\sqrt{1-p}\bra{11})_{SS'}+\frac{p}{2}\ketbra{10}{10}_{SS'} }))^2)\\
&= -\xi \log \tr_{S'}( (\tr_S(\sqrt{ \frac{2-p}{2}(\sqrt{\frac{1}{2-p}}\ket{00}+\sqrt{\frac{1-p}{2-p}}\ket{11})(\sqrt{\frac{1}{2-p}}\bra{00}+\sqrt{\frac{1-p}{2-p}}\bra{11})_{SS'}+\frac{p}{2}\ketbra{10}{10}_{SS'} }))^2)\\
&= -\xi \log \tr_{S'}( (\tr_S(\sqrt{\frac{2-p}{2}}(\sqrt{\frac{1}{2-p}}\ket{00}+\sqrt{\frac{1-p}{2-p}}\ket{11})(\sqrt{\frac{1}{2-p}}\bra{00}+\sqrt{\frac{1-p}{2-p}}\bra{11})_{SS'}+\sqrt{\frac{p}{2}}\ketbra{10}{10}_{SS'} ))^2)\\
&= -\xi \log \tr_{S'}( ( (\frac{1}{\sqrt{4-2p}}+\sqrt{\frac{p}{2}})\ketbra{0}{0}_{S'}+\frac{1-p}{\sqrt{4-2p}}\ketbra{1}{1}_{S'} )^2)\\
&= -\xi \log \tr_{S'}(  (\frac{1}{\sqrt{4-2p}}+\sqrt{\frac{p}{2}})^2\ketbra{0}{0}_{S'}+\frac{(1-p)^2}{4-2p}\ketbra{1}{1}_{S'} )\\
&= -\xi \log( (\frac{1}{\sqrt{4-2p}}+\sqrt{\frac{p}{2}})^2 + \frac{(1-p)^2}{4-2p} )\\
&= -\xi \log( \frac{1}{2-p} + \sqrt{\frac{p}{2-p}} ). \\
\end{split}
\end{equation}
By Eq.~\eqref{eq:boundofcoefc}, $c\leq \frac{4\varepsilon^{1-\frac{k}{n}}n^{\frac{k}{n}}}{\log (n/\varepsilon)}$.
Thus, for sufficiently large $n$, we have that
\begin{equation}
\begin{split}
&\mathbb{E}_{U} \Vert \hat{\mathcal{N}}_{S\rightarrow E}(U\ketbra{\hat{\phi}}_{LR}\otimes \ketbra{0^{n-k}} U^{\dagger}) - \hat{\mathcal{N}}_{S\rightarrow E}(\frac{\id_S}{2^n})\otimes \frac{\id_R}{2^k} \Vert_1\\
\leq&\sqrt{2^{-\sum_{i=1}^{2N}H_{2}(S_i|E_i)_{\tau_{S_iE_i}}-\sum_{i=1}^{2N}H_{2}(S_i|R_i)_{\rho_{S_iR_i}}} + c\prod_{i=1}^{2N} \max(1, 2^{-H_{2}(S_i|E_i)_{\tau_{S_iE_i}}-H_{2}(S_i|R_i)_{\rho_{S_iR_i}}})}\\
\leq&\sqrt{2^{-n(-\log( \frac{1}{2-p} + \sqrt{\frac{p}{2-p}} )-\frac{k}{n})} + \frac{4\varepsilon^{1-\frac{k}{n}}n^{\frac{k}{n}}}{\log (n/\varepsilon)}\prod_{i=1}^{2N} \max(1, 2^{-\xi(-\log( \frac{1}{2-p} + \sqrt{\frac{p}{2-p}}-\frac{k}{n})})}.
\end{split}
\end{equation}
From the above equation, we have that as long as $\varepsilon$ decays faster than $n^{-\frac{k}{n-k}}$, and encoding rate $k/n\leq -\log( \frac{1}{2-p} + \sqrt{\frac{p}{2-p}} )$, the Choi error decays polynomially. Hence, we conclude that 1D $O(\log(n))$-depth circuit achieves encoding rate $-\log( \frac{1}{2-p} + \sqrt{\frac{p}{2-p}} )$ for local amplitude damping noise.
\end{proof}

\subsection{Proof of Corollary~\ref{coro:amp} by applying smooth decoupling theorem to local amplitude damping noise}\label{app:ampsmooth}
\begin{proof}
Following the previous derivation, the key point is to evaluate $H_{2}(S_i|R_i)_{\rho_{S_iR_i}}$ and $H(S_i|E_i)_{\tau_{S_iE_i}}$. Similar to previous analysis, $H_{2}(S_i|R_i)_{\rho_{S_iR_i}} = -\frac{k}{n}\xi$. Meanwhile, from Eq.~\eqref{eq:dualvonneumann} we have
\begin{equation}
\begin{split}
H(S_i|E_i)_{\tau_{S_iE_i}} &= -\xi H(S|S')_{\mathcal{N}_a(\ketbra{\hat{\phi}})}\\
&= -\xi H(S|S')_{\frac{1}{2}(\ket{00}+\sqrt{1-p}\ket{11})(\bra{00}+\sqrt{1-p}\bra{11})_{SS'}+\frac{p}{2}\ketbra{10}{10}_{SS'}}\\
&= \xi(H(S')_{\frac{1+p}{2}\ketbra{0}_{S'}+\frac{1-p}{2}\ketbra{1}_{S'}} - H(SS')_{\frac{1}{2}(\ket{00}+\sqrt{1-p}\ket{11})(\bra{00}+\sqrt{1-p}\bra{11})_{SS'}+\frac{p}{2}\ketbra{10}{10}_{SS'}})\\
&= \xi(h(\frac{1-p}{2})-h(\frac{p}{2})).\\
\end{split}
\end{equation}
Also, by Eq.~\eqref{eq:boundofcoefc}, $c\leq \frac{4\varepsilon^{1-\frac{k}{n}}n^{\frac{k}{n}}}{\log (n/\varepsilon)}$. Thus, for sufficiently large $n$ and $0<\delta<1$, we have that
\begin{equation}
\begin{split}
&\mathbb{E}_{U} \Vert \hat{\mathcal{N}}_{S\rightarrow E}(U\ketbra{\hat{\phi}}_{LR}\otimes \ketbra{0^{n-k}} U^{\dagger}) - \hat{\mathcal{N}}_{S\rightarrow E}(\frac{\id_S}{2^n})\otimes \frac{\id_R}{2^k} \Vert_1\\
\leq&8\delta+\sqrt{2^{-n(h(\frac{1-p}{2})-h(\frac{p}{2})-\frac{k}{n}-\delta)} + \frac{4\varepsilon^{1-\frac{k}{n}}n^{\frac{k}{n}}}{\log (n/\varepsilon)}\prod_{i=1}^{2N} \max(1, 2^{-\xi(h(\frac{1-p}{2})-h(\frac{p}{2})-\frac{k}{n}-\delta)})}.
\end{split}
\end{equation}
When $h(\frac{1-p}{2})-h(\frac{p}{2})-\frac{k}{n}-\delta > 0$, we have that
\begin{equation}
\mathbb{E}_{U} \Vert \hat{\mathcal{N}}_{S\rightarrow E}(U\ketbra{\hat{\phi}}_{LR}\otimes \ketbra{0^{n-k}} U^{\dagger}) - \hat{\mathcal{N}}_{S\rightarrow E}(\frac{\id_S}{2^n})\otimes \frac{\id_R}{2^k} \Vert_1\leq 8\delta+\sqrt{2^{-n(h(\frac{1-p}{2})-h(\frac{p}{2})-\frac{k}{n}-\delta)} + \frac{4\varepsilon^{1-\frac{k}{n}}n^{\frac{k}{n}}}{\log (n/\varepsilon)}}.
\end{equation}
From the above equation, we have that the encoding rate $k/n$ can achieve $h(\frac{1-p}{2})-h(\frac{p}{2})$ for local amplitude damping noise with 1D $\omega(\log(n))$-depth circuit, which is the same as the bound with random stabilizer codes.
\end{proof}

\subsection{Proof of Corollary~\ref{coro:zzcouple} by applying the non-smooth decoupling theorem to nearest-neighbour $ZZ$-coupling noise}\label{app:proof_corr}
\begin{proof}
Following the previous derivation, the expectation of Choi error over $\mathfrak{C}_n^{\varepsilon}$ can be bounded by:
\begin{equation}\label{eq:zzcouplechoi}
\mathbb{E}_{U} \epsilon_{\mathrm{Choi}} \leq \sqrt{\mathbb{E}_{U} \Vert \hat{\mathcal{N}}_{S\rightarrow E}(U\ketbra{\hat{\phi}}_{LR}\otimes \ketbra{0^{n-k}} U^{\dagger}) - \hat{\mathcal{N}}_{S\rightarrow E}(\frac{\id_S}{2^n})\otimes \frac{\id_R}{2^k} \Vert_1}.
\end{equation}
For 1D nearest neighbor $ZZ$-coupling noise, we have
\begin{equation}
\mathcal{N}_{S} = \mathcal{N}_{zz} = \circ_{i=1}^{n-1}((1-p)\mathcal{I}+p\mathcal{Z}_i\mathcal{Z}_{i+1}).
\end{equation}

Applying the non-smooth decoupling theorem or Theorem~\ref{thm:nonsmoothdecoupling}, the term in the square root in Eq.~\eqref{eq:Paulichoi} satisfies
\begin{equation}
\begin{split}
&\mathbb{E}_{U} \Vert \hat{\mathcal{N}}_{S\rightarrow E}(U\ketbra{\hat{\phi}}_{LR}\otimes \ketbra{0^{n-k}} U^{\dagger}) - \hat{\mathcal{N}}_{S\rightarrow E}(\frac{\id_S}{2^n})\otimes \frac{\id_R}{2^k} \Vert_1\\
\leq&\sqrt{2^{-H_2(S|E)_{\tau_{SE}}-H_2(S|R)_{\rho_{SR}}} + c \max_{A\subseteq [2N]} 2^{-H_2(A|R)_{\rho_{AR}}-H_2(A|E)_{\tau_{AE}}}},
\end{split}
\end{equation}
where $\tau_{SE} = \mathcal{N}_{zz}(\ketbra{\hat{\phi}}^{\otimes n})$, $\rho_{S_iR_i} = \ketbra{\hat{\phi}}^{\frac{k}{n}\xi}\otimes \ketbra{0^{\frac{n-k}{n}\xi}}$, $\tau_{AE} = \tr_{\Bar{A}}\tau_{SE}$, $\rho_{SR} = \bigotimes_{i=1}^{2N}\rho_{S_iR_i}$, and
\begin{equation}
c = 2((1+2\eta\max_i(\max(2^{-H_{2}(S_i|R_i)_{\rho_{S_iR_i}}}, 2^{H_2(S_i|R_i)_{\rho_{S_iR_i}}})))^{N-1}-1).
\end{equation}
Note that $H_{2}(S_i|R_i)_{\rho_{S_iR_i}} = -\frac{k}{n}\xi$ given by Eq.~\eqref{eq:EPRcollisionentropy}. Thus, $H_2(A|R)_{\rho_{AR}} = -\frac{k\abs{A}}{n}\xi$.

Also, from Eq.~\eqref{eq:dualcollisionentropy} we can get
\begin{equation}
\begin{split}
H_2(A|E)_{\tau_{AE}}
&= -\log \tr_{S'}[\tr_A\sqrt{ \tr_{\Bar{A}} \mathcal{N}_{zz}(\ketbra{\hat{\phi}}^{\otimes n})} ]\\
&= -\log \tr_{A'}[\tr_A\sqrt{\prod_{i: i\in S_A, i+1\in S_{\Bar{A}}}[(1-p)\mathcal{I}+p\mathcal{Z}_i] \prod_{\{i,i+1\}\subseteq S_A}[(1-p)\mathcal{I}+p\mathcal{Z}_i\mathcal{Z}_{i+1}] (\ketbra{\hat{\phi}}_{AA'})} ]\\
&= \abs{A}\xi-2\log ( \sum_{\vec{Z}\in \{\mathcal{I},\mathcal{Z}\}^{\otimes \abs{A}\xi}}\Pr(\vec{Z}) )\\
&\geq \abs{A}\xi - 2\log( (\sqrt{1-p}+\sqrt{p})^{\abs{A}(\xi+1)})\\
&= \xi\abs{A}[1 - 2(1+\xi^{-1})\log(\sqrt{1-p}+\sqrt{p})]\\
&\approx \xi \abs{A}[1 -2\log(\sqrt{1-p}+\sqrt{p})].
\end{split}
\end{equation}
Here, $\Pr(\vec{Z})$ is the probability of pattern $\vec{Z}$ generated by $\prod_{i: i\in S_A, i+1\in S_{\Bar{A}}}[(1-p)\mathcal{I}+p\mathcal{Z}_i] \prod_{\{i,i+1\}\subseteq S_A}[(1-p)\mathcal{I}+p\mathcal{Z}_i\mathcal{Z}_{i+1}]$. Note that this probability can be obtained by coarse-graining the binomial distribution $B(\abs{A}(\xi+1), p)$. This can be seen by counting the generators of this probability distribution:
\begin{equation}
\abs{A}(\xi+1)\geq \abs{\{ i | i\in S_A, i+1\in S_{\Bar{A}} \}} + \abs{ \{ i| \{i,i+1\}\subseteq S_A \} }.
\end{equation}
The coarse-graining operation can only decrease $\sum \Pr(\vec{Z})$, which gives us the inequality in the fourth line.

By Eq.~\eqref{eq:boundofcoefc}, $c\leq \frac{4\varepsilon^{1-\frac{k}{n}}n^{\frac{k}{n}}}{\log (n/\varepsilon)}$.
Thus, for sufficiently large $n$, we have $\xi^{-1} \approx 0$ and obtain that
\begin{equation}
\begin{split}
&\mathbb{E}_{U} \Vert \hat{\mathcal{N}}_{S\rightarrow E}(U\ketbra{\hat{\phi}}_{LR}\otimes \ketbra{0^{n-k}} U^{\dagger}) - \hat{\mathcal{N}}_{S\rightarrow E}(\frac{\id_S}{2^n})\otimes \frac{\id_R}{2^k} \Vert_1\\
\leq&\sqrt{2^{-H_2(S|E)_{\tau_{SE}}-H_2(S|R)_{\rho_{SR}}} + c \max_{A\subseteq [2N]} 2^{-H_2(A|R)_{\rho_{AR}}-H_2(A|E)_{\tau_{AE}}}}\\
\leq&\sqrt{2^{-n(1-2\log(\sqrt{1-p}+\sqrt{p})-\frac{k}{n})} + \frac{4\varepsilon^{1-\frac{k}{n}}n^{\frac{k}{n}}}{\log (n/\varepsilon)}\prod_{i=1}^{2N} \max(1, 2^{-\xi(1-2\log(\sqrt{1-p}+\sqrt{p})-\frac{k}{n})})}.
\end{split}
\end{equation}
Substituting the above equation into Eq.~\eqref{eq:zzcouplechoi} completes the proof. From the above equation, we have that as long as $\varepsilon$ decays faster than $n^{-\frac{k}{n-k}}$, and encoding rate $k/n < 1-2\log(\sqrt{1-p}+\sqrt{p})$, the Choi error decays polynomially. Hence, we conclude that a 1D $O(\log(n))$-depth circuit achieves an encoding rate $1-2\log(\sqrt{1-p}+\sqrt{p})$ for nearest neighbor $ZZ$-coupling noise, which is the same with that of random Clifford encoding scheme, manifested by Eq.~\eqref{eq:zzcouplingrandomClif}.
\end{proof}

\subsection{Approximate quantum error correction against the fixed-number random erasure errors for $\mathfrak{C}_n^{\varepsilon}$}\label{app:proof_erasure}
We first present the results of AQEC performance of $\mathfrak{C}_n^{\varepsilon}$ for random $t$-erasure error below. Then we give the full proof. The proof idea is similar to that of the non-decoupling theorem.
\begin{theorem}[AQEC performance of $\mathfrak{C}_n^{\varepsilon}$ for random $t$-erasure error]\label{thm:erasure}
In the large $n$ limit, suppose $k$ and $t$ satisfy $1 - \frac{k}{n} - \log(1+\frac{3t}{n}) \geq 0$, and $\varepsilon$ satisfies that $(\frac{\varepsilon}{n})^{1-\frac{k}{n}} = o(\frac{1}{n})$, then the expected Choi error of the random codes from $\mathfrak{C}_n^{\varepsilon}$ against the random $t$ erasure errors satisfies
\begin{equation}
\mathbb{E}\epsilon_{\mathrm{Choi}} \leq \left(2^{-(n-2t-k)}+\frac{4\varepsilon^{1-\frac{k}{n}}n^{\frac{k}{n}}}{\log (n/\varepsilon)}\right)^{\frac{1}{4}}.
\end{equation}
\end{theorem}

\begin{proof}
Here, we analyze the random erasure error with a fixed number of erasures. We utilize the complementary channel to analyze the Choi error. Specifically,
\begin{equation}
\epsilon_{\mathrm{Choi}} = \min_{\zeta} P\left( (\widehat{\mathcal{N}\circ \mathcal{E}}_{L\rightarrow E}\otimes I_R)(\ketbra{\hat{\phi}}_{LR}), (\mathcal{T}^{\zeta}_{L\rightarrow E}\otimes I_R)(\ketbra{\hat{\phi}}_{LR}) \right).
\end{equation}
Since $\mathcal{T}^{\zeta}_{L\rightarrow E}$ maps all states to a fixed state $\zeta$, $(\mathcal{T}^{\zeta}_{L\rightarrow E}\otimes I_R)(\ketbra{\hat{\phi}}_{LR}) = \zeta_E\otimes \tr_L(\ketbra{\hat{\phi}}_{LR}) = \zeta_E\otimes \frac{\id_R}{2^k}$. Meanwhile, $\widehat{\mathcal{N}\circ \mathcal{E}}_{L\rightarrow E} = \hat{\mathcal{N}}_{S\rightarrow E}\circ \mathcal{E}_{L\rightarrow S}$ with $\mathcal{E}_{L\rightarrow S}$ the encoding map satisfying
\begin{equation}
\mathcal{E}_{L\rightarrow S}(\rho) = U_S (\rho_L\otimes \ketbra{0}^{n-k}) U_S^{\dagger},
\end{equation}
where $U_S$ is the encoding unitary, generated by 1D $\log(n/\varepsilon)$-depth Clifford circuits in our case. The term $\ketbra{0}^{n-k}$ is the ancillary qubits used for quantum error correction. For a fixed $t$-erasure errors, $\hat{\mathcal{N}}_{S\rightarrow E}(\rho) = \tr_{n-t}(\rho)$. For a random $t$-erasure error, which is a mixing of $t$-erasure errors at different locations, $\hat{\mathcal{N}}_{S\rightarrow E}(\rho) = \frac{1}{\binom{n}{t}}\sum_{\abs{T}=t}\ketbra{T}\otimes \tr_{S\backslash T}(\rho)$~\cite{Faist2020AQEC} where $\ket{T}$ labels the location of the erasure errors.
Hence,
\begin{equation}
\begin{split}
\epsilon_{\mathrm{Choi}} =& \min_{\zeta} P\left( \hat{\mathcal{N}}_{S\rightarrow E}(U\ketbra{\hat{\phi}}_{LR}\otimes \ketbra{0^{n-k}} U^{\dagger}), \zeta_E\otimes \frac{\id_R}{2^k} \right)\\
\leq& P\left( \hat{\mathcal{N}}_{S\rightarrow E}(U\ketbra{\hat{\phi}}_{LR}\otimes \ketbra{0^{n-k}} U^{\dagger}), \mathbb{E}_{U\sim \mathfrak{U}_n}\hat{\mathcal{N}}_{S\rightarrow E}(U\ketbra{\hat{\phi}}_{LR}\otimes \ketbra{0^{n-k}} U^{\dagger}) \right)\\
&+ \min_{\zeta} P\left(\mathbb{E}_{U\sim \mathfrak{U}_n} \hat{\mathcal{N}}_{S\rightarrow E}(U\ketbra{\hat{\phi}}_{LR}\otimes \ketbra{0^{n-k}} U^{\dagger}), \zeta_E\otimes \frac{\id_R}{2^k} \right)\\
=& P\left( \hat{\mathcal{N}}_{S\rightarrow E}(U\ketbra{\hat{\phi}}_{LR}\otimes \ketbra{0^{n-k}} U^{\dagger}), \hat{\mathcal{N}}_{S\rightarrow E}(\frac{\id_S}{2^n})\otimes \frac{\id_R}{2^k} \right) + \min_{\zeta} P\left(\hat{\mathcal{N}}_{S\rightarrow E}(\frac{\id_S}{2^n})\otimes \frac{\id_R}{2^k}, \zeta_E\otimes \frac{\id_R}{2^k} \right)\\
=& P\left( \hat{\mathcal{N}}_{S\rightarrow E}(U\ketbra{\hat{\phi}}_{LR}\otimes \ketbra{0^{n-k}} U^{\dagger}), \hat{\mathcal{N}}_{S\rightarrow E}(\frac{\id_S}{2^n})\otimes \frac{\id_R}{2^k} \right).
\end{split}
\end{equation}
Here, in the first inequality, we utilize the triangle inequality of the purified distance. In the second equality, we utilize the property of the twirling over Haar random unitary gates. In the third equality, we take the state $\zeta$ as $\hat{\mathcal{N}}_{S\rightarrow E}(\frac{\id_S}{2^n})$.

In our case, the encoding unitary is generated by random 1D $\log(n/\varepsilon)$-depth Clifford circuits with the ensemble denoted as $\mathfrak{C}_n^{\varepsilon}$. Then, we aim to bound $\mathbb{E}_{U\sim \mathfrak{C}_n^{\varepsilon}}\epsilon_{\mathrm{Choi}}$. For simplicity, we omit $\mathfrak{C}_n^{\varepsilon}$ in the subscript. Utilizing Eq.~\eqref{eq:tracedistanceineq} we get
\begin{equation}
\begin{split}
\mathbb{E}_{U} \epsilon_{\mathrm{Choi}} \leq& \mathbb{E}_{U} \sqrt{\Vert \hat{\mathcal{N}}_{S\rightarrow E}(U\ketbra{\hat{\phi}}_{LR}\otimes \ketbra{0^{n-k}} U^{\dagger}) - \hat{\mathcal{N}}_{S\rightarrow E}(\frac{\id_S}{2^n})\otimes \frac{\id_R}{2^k} \Vert_1}\\
\leq& \sqrt{\mathbb{E}_{U} \Vert \hat{\mathcal{N}}_{S\rightarrow E}(U\ketbra{\hat{\phi}}_{LR}\otimes \ketbra{0^{n-k}} U^{\dagger}) - \hat{\mathcal{N}}_{S\rightarrow E}(\frac{\id_S}{2^n})\otimes \frac{\id_R}{2^k} \Vert_1}.
\end{split}
\end{equation}
When the complementary channel is $\mathbb{E}_i \ketbra{i}\otimes \hat{\mathcal{N}}^i_{S\rightarrow E}$ where $\hat{\mathcal{N}}^i_{S\rightarrow E}$ is the complementary channel of $\mathcal{N}^i_{S\rightarrow E}$, we have that
\begin{equation}
\begin{split}
&\Vert \hat{\mathcal{N}}_{S\rightarrow E}(U\ketbra{\hat{\phi}}_{LR}\otimes \ketbra{0^{n-k}} U^{\dagger}) - \hat{\mathcal{N}}_{S\rightarrow E}(\frac{\id_S}{2^n})\otimes \frac{\id_R}{2^k} \Vert_1\\
= &\mathbb{E}_i\Vert \hat{\mathcal{N}}^i_{S\rightarrow E}(U\ketbra{\hat{\phi}}_{LR}\otimes \ketbra{0^{n-k}} U^{\dagger}) - \hat{\mathcal{N}}^i_{S\rightarrow E}(\frac{\id_S}{2^n})\otimes \frac{\id_R}{2^k} \Vert_1,
\end{split}
\end{equation}
due to that $\Vert \mathbb{E}_i \ketbra{i} \otimes A_i \Vert_1 = \mathbb{E}_i\Vert A_i \Vert_1$. Below, we use $\mathbb{E}_{\mathcal{N}}$ to represent a case of mixing noise channels. For the mixing noise channel, we have
\begin{equation}
\mathbb{E}_{U} \epsilon_{\mathrm{Choi}} \leq \sqrt{\mathbb{E}_{U} \mathbb{E}_{\mathcal{N}}\Vert \hat{\mathcal{N}}_{S\rightarrow E}(U\ketbra{\hat{\phi}}_{LR}\otimes \ketbra{0^{n-k}} U^{\dagger}) - \hat{\mathcal{N}}_{S\rightarrow E}(\frac{\id_S}{2^n})\otimes \frac{\id_R}{2^k} \Vert_1}.
\end{equation}
For further evaluation, we use $\Phi_{Haar}$ to denote $\hat{\mathcal{N}}_{S\rightarrow E}(\frac{\id_S}{2^n})\otimes \frac{\id_R}{2^k}$. Using $\Vert A\Vert_1 \leq \sqrt{\abs{A}}\Vert A\Vert_2$, we further get
\begin{equation}\label{eq:choierrorerasurebound}
\begin{split}
\mathbb{E}_{U} \epsilon_{\mathrm{Choi}}
\leq& \sqrt{\mathbb{E}_{U} \mathbb{E}_{\mathcal{N}}\Vert \hat{\mathcal{N}}_{S\rightarrow E}(U\ketbra{\hat{\phi}}_{LR}\otimes \ketbra{0^{n-k}} U^{\dagger}) - \Phi_{Haar} \Vert_1}\\
\leq& \sqrt{\sqrt{\dim E \dim R} \mathbb{E}_{\mathcal{N}}\mathbb{E}_{U} \Vert \hat{\mathcal{N}}_{S\rightarrow E}(U\ketbra{\hat{\phi}}_{LR}\otimes \ketbra{0^{n-k}} U^{\dagger}) - \Phi_{Haar} \Vert_2}\\
\leq& (\dim E \dim R)^{\frac{1}{4}}\left(\mathbb{E}_{\mathcal{N}}\mathbb{E}_{U} \tr(\hat{\mathcal{N}}_{S\rightarrow E}(U\ketbra{\hat{\phi}}_{LR}\otimes \ketbra{0^{n-k}} U^{\dagger}) - \Phi_{Haar} )^2 \right)^{\frac{1}{4}}\\
=& (\dim E \dim R)^{\frac{1}{4}}( \mathbb{E}_{\mathcal{N}}[\tr \mathbb{E}_{U}[\hat{\mathcal{N}}_{S\rightarrow E}(U\ketbra{\hat{\phi}}_{LR}\otimes \ketbra{0^{n-k}} U^{\dagger})]^2 + \tr \Phi^2_{Haar}\\
&- 2\tr \hat{\mathcal{N}}_{S\rightarrow E}(\mathbb{E}_{U}[U\ketbra{\hat{\phi}}_{LR}\otimes \ketbra{0^{n-k}} U^{\dagger}])\Phi_{Haar}] )^{\frac{1}{4}}.\\
\end{split}
\end{equation}
In the last line, we expand the 2-norm. The error terms can be divided into two parts:
\begin{equation}\label{eq:seconderror}
\tr \mathbb{E}_{U}[\hat{\mathcal{N}}_{S\rightarrow E}(U\ketbra{\hat{\phi}}_{LR}\otimes \ketbra{0^{n-k}} U^{\dagger})]^2,
\end{equation}
and
\begin{equation}\label{eq:firsterror}
\tr \Phi^2_{Haar}- 2\tr \hat{\mathcal{N}}_{S\rightarrow E}(\mathbb{E}_{U}[U\ketbra{\hat{\phi}}_{LR}\otimes \ketbra{0^{n-k}} U^{\dagger}])\Phi_{Haar}.
\end{equation}
Eqs.~\eqref{eq:seconderror} and~\eqref{eq:firsterror} can be interpreted as the error brought by the second and the first moments of random unitary operations, respectively. Provided that $\mathfrak{C}_n^{\varepsilon}$ forms a unitary 1-design, we have
\begin{equation}
\mathbb{E}_{U}[U\ketbra{\hat{\phi}}_{LR}\otimes \ketbra{0^{n-k}} U^{\dagger}] = \frac{\id_S}{2^n}\otimes \frac{\id_R}{2^k}.
\end{equation}
Thus,
\begin{equation}\label{eq:erasure_linearterm}
\tr\Phi^2_{Haar} - 2\tr \hat{\mathcal{N}}_{S\rightarrow E}(\mathbb{E}_{U}[U\ketbra{\hat{\phi}}_{LR}\otimes \ketbra{0^{n-k}} U^{\dagger}])\Phi_{Haar} = -\tr\Phi^2_{Haar},
\end{equation}
where
\begin{equation}
\begin{split}
\tr\Phi^2_{Haar} &= \tr([\hat{\mathcal{N}}_{S\rightarrow E}(\frac{\id_S}{2^n})]^2\otimes \frac{\id_R}{2^{2k}})\\
&= \frac{1}{2^{k}} \tr([\hat{\mathcal{N}}_{S\rightarrow E}(\frac{\id_S}{2^n})]^2).
\end{split}
\end{equation}
For the erasure channel, $\hat{\mathcal{N}}_{S\rightarrow E} = \tr_{n-t}$, $\tr\Phi^2_{Haar} = \frac{1}{2^{k+t}} = \frac{1}{\dim E\dim R}$, and $\mathbb{E}_{\mathcal{N}}\tr\Phi^2_{Haar} = \frac{1}{2^{k+t}}$.

The remained part to bound $\mathbb{E}_U \epsilon_{\mathrm{Choi}}$ is $\mathbb{E}_{\mathcal{N}}\tr \mathbb{E}_{U}[\hat{\mathcal{N}}_{S\rightarrow E}(U\ketbra{\hat{\phi}}_{LR}\otimes \ketbra{0^{n-k}} U^{\dagger})]^2$. We first evaluate $\tr \mathbb{E}_{U}[\hat{\mathcal{N}}_{S\rightarrow E}(U\ketbra{\hat{\phi}}_{LR}\otimes \ketbra{0^{n-k}} U^{\dagger})]^2$ and then take the expectation $\mathbb{E}_{\mathcal{N}}$. We use $\tr A^2 = \tr A^{\otimes 2}F$ to evaluate it with $F$ being a SWAP operator. Hence,
\begin{equation}
\begin{split}
&\tr \mathbb{E}_{U}[\hat{\mathcal{N}}_{S\rightarrow E}(U\ketbra{\hat{\phi}}_{LR}\otimes \ketbra{0^{n-k}} U^{\dagger})]^2\\
=& \tr \mathbb{E}_{U}[(\hat{\mathcal{N}}_{S\rightarrow E}(U\ketbra{\hat{\phi}}_{LR}\otimes \ketbra{0^{n-k}} U^{\dagger}))^{\otimes 2}] F_E\otimes F_R\\
=& \tr \mathbb{E}_{U}[U^{\otimes 2}(\ketbra{\hat{\phi}}_{LR}\otimes \ketbra{0^{n-k}})^{\otimes 2} U^{\dagger \otimes 2}] (\hat{\mathcal{N}}_{S\rightarrow E}^{\dagger})^{\otimes 2}(F_E)\otimes F_R,
\end{split}
\end{equation}
where $F_E$ denotes the SWAP operator acting on two copies of system $E$, and $F_R$ is defined similarly. The term $\hat{\mathcal{N}}_{S\rightarrow E}^{\dagger}$ is the adjoint map of $\hat{\mathcal{N}}_{S\rightarrow E}$. To further evaluate the above quantity, we need to investigate the structure of $\mathfrak{C}_n^{\varepsilon}$. Note that the random encoding unitary $U$ can be generated by two layers, as shown in Fig.~\ref{fig:1Dlowdepthcircuit}. The qubits are divided into $2N$ parts, labelled from $1$ to $2N$, with each part having an equal size of $\xi = \frac{n}{2N}$ qubits. The first layer of the encoding unitary is a tensor product of random Clifford gates acting on the $2i-1$ and $2i$ parts, and the second layer is a tensor product of random Clifford gates acting on the $2i$ and $2i+1$ parts. That is,
\begin{equation}
U = U_2\cdot U_1,
\end{equation}
where
\begin{equation}\label{eq:1Dunitary}
U_1 = \bigotimes_{i=1}^N U_{C}^{2i-1, 2i}, U_2 = \bigotimes_{i=1}^{N-1} U_{C}^{2i, 2i+1}.
\end{equation}
The gate $U_C^{i,i+1}$ is uniformly and randomly sampled from the Clifford group in the regions $i$ and $i+1$. As a result,
\begin{equation}
\begin{split}
&\tr \mathbb{E}_{U}[\hat{\mathcal{N}}_{S\rightarrow E}(U\ketbra{\hat{\phi}}_{LR}\otimes \ketbra{0^{n-k}} U^{\dagger})]^2\\
=& \tr \mathbb{E}_{U_1}[U_1^{\otimes 2}(\ketbra{\hat{\phi}}_{LR}\otimes \ketbra{0^{n-k}})^{\otimes 2} U_1^{\dagger \otimes 2}] \mathbb{E}_{U_2}[U_2^{\dagger\otimes 2}(\hat{\mathcal{N}}_{S\rightarrow E}^{\dagger})^{\otimes 2}(F_E)U_2^{\otimes 2}]\otimes F_R\\
=&\tr( \tr_R(\mathbb{E}_{U_1}[U_1^{\otimes 2}(\ketbra{\hat{\phi}}_{LR}\otimes \ketbra{0^{n-k}})^{\otimes 2} U_1^{\dagger \otimes 2}]F_R) \mathbb{E}_{U_2}[U_2^{\dagger\otimes 2}(\hat{\mathcal{N}}_{S\rightarrow E}^{\dagger})^{\otimes 2}(F_E)U_2^{\otimes 2}]).
\end{split}
\end{equation}
We make the two layers of random unitary gates act on different terms for further evaluation.  We first focus on the first term,
\begin{equation}
\tr_R(\mathbb{E}_{U_1}[U_1^{\otimes 2}(\ketbra{\hat{\phi}}_{LR}\otimes \ketbra{0^{n-k}})^{\otimes 2} U_1^{\dagger \otimes 2}]F_R).
\end{equation}
For simplicity, we consider each part to have the same number of encoding qubits and ancillary qubits. That is, in part $i$, there will be $\frac{k}{n}\xi = \frac{k}{2N}$ encoding qubits and $\frac{n-k}{n}\xi = \frac{n-k}{2N}$ ancillary qubits. Thus, based on Eqs.~\eqref{eq:secondtwirling} and~\eqref{eq:1Dunitary}, we get
\begin{equation}
\begin{split}
&\mathbb{E}_{U_1}[U_1^{\otimes 2}(\ketbra{\hat{\phi}}_{LR}\otimes \ketbra{0^{n-k}})^{\otimes 2} U_1^{\dagger \otimes 2}]\\
=&\bigotimes_{i=1}^N \left(\frac{\id_{d^2_{S_{2i-1, 2i}}}}{{d^2_{S_{2i-1, 2i}}}}\otimes \frac{\id_{d^2_{R_{2i-1, 2i}}}}{{d^2_{R_{2i-1, 2i}}}} + \frac{F_{S_{2i-1,2i}} - \frac{\id_{d^2_{S_{2i-1, 2i}}}}{{d_{S_{2i-1, 2i}}}}}{d^2_{S_{2i-1, 2i}}-1}\otimes (F_{R_{2i-1,2i}}-\frac{1}{d_{S_{2i-1,2i}}})\frac{\id_{d^2_{R_{2i-1, 2i}}}}{{d^2_{R_{2i-1, 2i}}}} \right).
\end{split}
\end{equation}
Since $F_R = \bigotimes_{i=1}^N F_{R_{2i-1,2i}}$, we have
\begin{equation}
\begin{split}
&\tr_R(\mathbb{E}_{U_1}[U_1^{\otimes 2}(\ketbra{\hat{\phi}}_{LR}\otimes \ketbra{0^{n-k}})^{\otimes 2} U_1^{\dagger \otimes 2}]F_R)\\
=&\bigotimes_{i=1}^N \left(\frac{\id_{d^2_{S_{2i-1, 2i}}}}{{d^2_{S_{2i-1, 2i}}d_{R_{2i-1, 2i}}}} + \frac{F_{S_{2i-1,2i}} - \frac{\id_{d^2_{S_{2i-1, 2i}}}}{{d_{S_{2i-1, 2i}}}}}{d^2_{S_{2i-1, 2i}}-1}(1-\frac{1}{d_{S_{2i-1,2i}}d_{R_{2i-1,2i}}}) \right)\\
=&\bigotimes_{i=1}^N \left(\frac{\id_{d^2_{S_{2i-1, 2i}}}}{{d^4_{S_i}d^2_{R_i}}} + \frac{F_{S_{2i-1,2i}} - \frac{\id_{d^2_{S_{2i-1, 2i}}}}{{d^2_{S_i}}}}{d^4_{S_i}-1}(1-\frac{1}{d^2_{S_i}d^2_{R_i}}) \right)\\
=&\bigotimes_{i=1}^N (a\id_{d^2_{S_{2i-1}}}\otimes \id_{d^2_{S_{2i}}}+bF_{S_{2i-1}}\otimes F_{S_{2i}})\\
=&\bigotimes_{i=1}^N (a\id_{2i-1}\id_{2i}+bF_{2i-1}F_{2i}).
\end{split}
\end{equation}
Here, the last line is a simple notation of the third line. That is, we use $\id_i$ and $F_i$ to denote $\id_{d^2_{S_{i}}}$ and $F_{S_i}$, respectively, and we omit the tensor product; $d_{S_i} = 2^{\xi}$ represents the dimension of part $i$ and correspondingly, $d_{R_i} = 2^{\frac{k}{n}\xi}$. The coefficients
\begin{equation}
\begin{split}
a &= \frac{1}{{d^4_{S_i}d^2_{R_i}}} - \frac{1}{d^2_{S_i}(d^4_{S_i}-1)}(1-\frac{1}{d^2_{S_i}d^2_{R_i}})\\
&= \frac{1}{2^{4\xi+\frac{2k}{n}\xi}} - \frac{1}{2^{2\xi}(2^{4\xi}-1)}(1-\frac{1}{2^{2\xi+\frac{2k}{n}\xi}}),
\end{split}
\end{equation}
and
\begin{equation}
\begin{split}
b &= \frac{1}{d^4_{S_i}-1}(1-\frac{1}{d^2_{S_i}d^2_{R_i}})\\
&= \frac{1}{2^{4\xi}-1}(1-\frac{1}{2^{2\xi+\frac{2k}{n}\xi}}).
\end{split}
\end{equation}
It is easy to see that $a\leq 2^{-4\xi-\frac{2k}{n}\xi}$ and $b\leq 2^{-4\xi}$.

Next, we consider the second term $\mathbb{E}_{U_2}[U_2^{\dagger\otimes 2}(\hat{\mathcal{N}}_{S\rightarrow E}^{\dagger})^{\otimes 2}(F_E)U_2^{\otimes 2}]$ with $\mathcal{N}$ a fixed $t$-erasure channel. Note that $(\hat{\mathcal{N}}_{S\rightarrow E}^{\dagger})^{\otimes 2}(F_E) = F_E\otimes \id_{\Bar{E}}$, where $\id_{\Bar{E}}$ denotes the identity operator on two copies of subsystem $\Bar{E}$.
Combined with Eqs.~\eqref{eq:secondtwirling} and~\eqref{eq:1Dunitary}, we have
\begin{equation}
\mathbb{E}_{U_2}[U_2^{\dagger\otimes 2}(\hat{\mathcal{N}}_{S\rightarrow E}^{\dagger})^{\otimes 2}(F_E)U_2^{\otimes 2}] = (F_E\otimes \id_{\Bar{E}})_1 \bigotimes_{i=1}^{N-1} (a_{2i, 2i+1}\id_{2i}\id_{2i+1}+b_{2i,2i+1}F_{2i}F_{2i+1})(F_E\otimes \id_{\Bar{E}})_{2N},
\end{equation}
where
\begin{equation}
a_{2i, 2i+1} = \frac{2^{n^{2i}_I+n^{2i+1}_I}}{2^{2\xi}}\frac{2^{4\xi}-2^{2n^{2i}_F+2n^{2i+1}_F}}{2^{4\xi}-1}, b_{2i, 2i+1} = \frac{2^{n^{2i}_I+n^{2i+1}_I}(2^{2n^{2i}_F+2n^{2i+1}_F}-1)}{2^{4\xi}-1}.
\end{equation}
Here, $(F_E\otimes \id_{\Bar{E}})_i$ denotes the restriction of $F_E\otimes \id_{\Bar{E}}$ is part $i$. Note that $(F_E\otimes \id_{\Bar{E}})_i$ is a tensor product of identity operators and SWAP operators. The terms $n^i_I$ and $n^i_F$ denote the number of identity operators and SWAP operators in $(F_E\otimes \id_{\Bar{E}})_i$, respectively. Hence, $n^i_I+n^i_F = \xi$ and $\sum_i n^i_F = t$.

Combined the results above, we get
\begin{equation}
\begin{split}
&\tr \mathbb{E}_{U}[\hat{\mathcal{N}}_{S\rightarrow E}(U\ketbra{\hat{\phi}}_{LR}\otimes \ketbra{0^{n-k}} U^{\dagger})]^2\\
=&\tr( \tr_R(\mathbb{E}_{U_1}[U_1^{\otimes 2}(\ketbra{\hat{\phi}}_{LR}\otimes \ketbra{0^{n-k}})^{\otimes 2} U_1^{\dagger \otimes 2}]F_R) \mathbb{E}_{U_2}[U_2^{\dagger\otimes 2}(\hat{\mathcal{N}}_{S\rightarrow E}^{\dagger})^{\otimes 2}(F_E)U_2^{\otimes 2}])\\
=&\tr(\left(\bigotimes_{i=1}^N (a\id_{2i-1}\id_{2i}+bF_{2i-1}F_{2i})\right) \left( (F_E\otimes \id_{\Bar{E}})_1 \bigotimes_{i=1}^{N-1} (a_{2i, 2i+1}\id_{2i}\id_{2i+1}+b_{2i,2i+1}F_{2i}F_{2i+1})(F_E\otimes \id_{\Bar{E}})_{2N}\right))\\
=& \begin{pmatrix}
2^{2n^{2N}_I+n^{2N}_F} & 2^{n^{2N}_I+2n^{2N}_F}
\end{pmatrix}\left(\prod_{i=1}^{N-1}
\begin{pmatrix}
a & 0\\
0 & b
\end{pmatrix}
\begin{pmatrix}
2^{2\xi} & 2^{\xi}\\
2^{\xi} & 2^{2\xi}
\end{pmatrix}
\begin{pmatrix}
a_{2i, 2i+1} & 0\\
0 & b_{2i, 2i+1}
\end{pmatrix}
\begin{pmatrix}
2^{2\xi} & 2^{\xi}\\
2^{\xi} & 2^{2\xi}
\end{pmatrix}
\right)
\begin{pmatrix}
a2^{2n^{1}_I+n^{1}_F} \\ b2^{n^{1}_I+2n^{1}_F}
\end{pmatrix}\\
=& 2^n\begin{pmatrix}
2^{n^{2N}_I} & 2^{n^{2N}_F}
\end{pmatrix}\left(\prod_{i=1}^{N-1}
\begin{pmatrix}
a & 0\\
0 & b
\end{pmatrix}
\begin{pmatrix}
u_{2i}u_{2i+1} & \eta(u_{2i}u_{2i+1}+v_{2i}v_{2i+1})\\
\eta(u_{2i}u_{2i+1}+v_{2i}v_{2i+1}) & v_{2i}v_{2i+1}
\end{pmatrix}
\right)
\begin{pmatrix}
a & 0\\
0 & b
\end{pmatrix}
\begin{pmatrix}
2^{n^{1}_I} \\ 2^{n^{1}_F}
\end{pmatrix}
.
\end{split}
\end{equation}
Here, the third line is a direct substitution. The fourth line is done by direct calculation. Note that
\begin{equation}
\tr_1((a\id_1\id_2+bF_1F_2) (F_E\otimes \id_{\Bar{E}})_1) = (a2^{2n^1_I+n^1_F}\id_2 + b2^{n^1_I+2n^1_F}F_2),
\end{equation}
and
\begin{equation}\label{eq:partialtrace}
\tr_1((x\id_1+yF_1)(a\id_1\id_2+bF_1F_2)) = (a2^{2\xi}x+a2^{\xi}y)\id_2+(b2^{\xi}x+b2^{2\xi}y)F_2.
\end{equation}
By using the above two equalities, we derive the fourth line. For the fifth line, we use that $n^i_I+n^i_F = \xi$ and $2N\xi = n$. Also, we calculate
\begin{equation}
\begin{split}
\begin{pmatrix}
2^{\xi} & 1\\
1 & 2^{\xi}
\end{pmatrix}
\begin{pmatrix}
a_{2i, 2i+1} & 0\\
0 & b_{2i, 2i+1}
\end{pmatrix}
\begin{pmatrix}
2^{\xi} & 1\\
1 & 2^{\xi}
\end{pmatrix} &= \begin{pmatrix}
2^{2\xi}a_{2i,2i+1}+b_{2i,2i+1} & 2^{\xi}(a_{2i,2i+1}+b_{2i,2i+1})\\
2^{\xi}(a_{2i,2i+1}+b_{2i,2i+1}) & a_{2i,2i+1}+2^{2\xi}b_{2i,2i+1}
\end{pmatrix}\\
&= \begin{pmatrix}
u_{2i}u_{2i+1} & \eta(u_{2i}u_{2i+1}+v_{2i}v_{2i+1})\\
\eta(u_{2i}u_{2i+1}+v_{2i}v_{2i+1}) & v_{2i}v_{2i+1}
\end{pmatrix},
\end{split}
\end{equation}
where
\begin{equation}
u_i = 2^{n^i_I}, v_i = 2^{n^i_F}, \eta = \frac{2^{\xi}}{2^{2\xi}+1}\leq 2^{-\xi}.
\end{equation}
Denote $D_{2i,2i+1} = \begin{pmatrix}
u_{2i}u_{2i+1} & 0\\
0 & v_{2i}v_{2i+1}
\end{pmatrix}$, $E_{2i,2i+1} = \eta(u_{2i}u_{2i+1}+v_{2i}v_{2i+1})\begin{pmatrix}
0 & 1\\
1 & 0
\end{pmatrix}$, and $C = \begin{pmatrix}
a & 0\\ 0 & b
\end{pmatrix}$, we have
\begin{equation}
\tr \mathbb{E}_{U}[\hat{\mathcal{N}}_{S\rightarrow E}(U\ketbra{\hat{\phi}}_{LR}\otimes \ketbra{0^{n-k}} U^{\dagger})]^2
= 2^n\begin{pmatrix}
u_{2N} & v_{2N}
\end{pmatrix}\left(\prod_{i=1}^{N-1}
C(D_{2i, 2i+1}+E_{2i,2i+1})
\right)
C
\begin{pmatrix}
u_1 \\ v_1
\end{pmatrix}
.
\end{equation}
This equation is similar to $(D+E)^{N-1}$, which enables us to expand it according to the number of ``$E$." We first consider the expansion term with no $E$. That is,
\begin{equation}\label{eq:erasure_gzero}
\begin{split}
2^n\begin{pmatrix}
u_{2N} & v_{2N}
\end{pmatrix}\left(\prod_{i=1}^{N-1}
CD_{2i, 2i+1}
\right)
C
\begin{pmatrix}
u_1 \\ v_1
\end{pmatrix} &= 2^n(a^N\prod_{i=1}^{2N}u_i+b^N\prod_{i=1}^{2N}v_i)\\
&= 2^n(a^N2^{n-t}+b^N2^t)\\
&\leq 2^n(2^{-2n-k}2^{n-t}+2^{-2n}2^t)\\
&= 2^{-t-k}+2^{-n+t}\\
\end{split}
\end{equation}
The first term $2^{-t-k}$ in this equation equals $\tr\Phi^2_{Haar}$ and will be eliminated in the final result. The second term $2^{-n+t}$ exponentially decays with respect to $n$. Combining with the coefficient $\dim E\dim R = 2^{t+k}$, this is $2^{-(n-2t-k)}$ corresponding to a result of random Clifford encoding, as demonstrated in Appendix~\ref{app:Clif}. In fact, if the encoding unitary is drawn from the whole unitary group instead of $\mathfrak{C}_n^{\varepsilon}$, $\eta$ will quickly decay to 0, and it is unnecessary to consider $E$ in the expansion. The expansion terms with $E$ can be viewed as extra errors.

Below, we consider the case there are $g$ numbers of $E$ in the expansion. We set the location of $E$ as $j_1, j_2, \cdots, j_g$ with $1\leq j_1 < j_2 < \cdots < j_g \leq N-1$. That is,
\begin{equation}
\begin{split}
&2^n\begin{pmatrix}
u_{2N} & v_{2N}
\end{pmatrix}\left(\prod_{i=j_g+1}^{N-1}
CD_{2i, 2i+1}
\right)
CE_{2j_g,2j_g+1}
\cdots
CE_{2j_2,2j_2+1}
\left(\prod_{i=j_1+1}^{j_2-1}
CD_{2i, 2i+1}
\right)
CE_{2j_1,2j_1+1}
\left(\prod_{i=1}^{j_1-1}
CD_{2i, 2i+1}
\right)
C
\begin{pmatrix}
u_1 \\ v_1
\end{pmatrix}\\
=& 2^n\eta^g \prod_{r=1}^g(u_{2j_r}u_{2j_r+1}+v_{2j_r}v_{2j_r+1})\times \\
&
\begin{pmatrix}
b^{N-j_g}v_{2j_g+2\rightarrow 2N} & a^{N-j_g}u_{2j_g+2\rightarrow 2N}
\end{pmatrix}
\left(\prod_{s=1}^{g-1}
\begin{pmatrix}
0 & a^{j_{s+1}-j_s}u_{2j_s+2\rightarrow 2j_{s+1}-1}\\
b^{j_{s+1}-j_s}v_{2j_s+2\rightarrow 2j_{s+1}-1} & 0 \\
\end{pmatrix}
\right)
\begin{pmatrix}
a^{j_1}u_{1\rightarrow 2j_1-1} \\ b^{j_1}v_{1\rightarrow 2j_1-1}
\end{pmatrix}
\\
=& 2^n\eta^g \prod_{r=1}^g(u_{2j_r}u_{2j_r+1}+v_{2j_r}v_{2j_r+1})
(
a^{-j_1+j_2-j_3+\cdots}b^{j_1-j_2+j_3-\cdots}u_{2j_1+2\rightarrow 2j_2-1, 2j_3+2\rightarrow 2j_4-1, \cdots}v_{1\rightarrow 2j_1-1, 2j_2+2\rightarrow 2j_3-1, \cdots}
\\
&+
a^{j_1-j_2+j_3-\cdots}b^{-j_1+j_2-j_3+\cdots}u_{1\rightarrow 2j_1-1, 2j_2+2\rightarrow 2j_3-1, \cdots}v_{2j_1+2\rightarrow 2j_2-1, 2j_3+2\rightarrow 2j_4-1, \cdots}
)\\
=& 2^n\eta^g \prod_{r=1}^g(\frac{2^{2\xi}}{v_{2j_r}v_{2j_r+1}}+v_{2j_r}v_{2j_r+1})\times
\\
% &
% a^{N-j}b^{j}2^{(2N-2j-g)\xi}\frac{v_{1\rightarrow 2j_1-1, 2j_2+2\rightarrow 2j_3-1, \cdots, 2j_{g-1}+2\rightarrow 2j_g-1}}{v_{2j_1+2\rightarrow 2j_2-1\cdots 2j_{g-2}+2\rightarrow 2j_{g-1}-1, 2j_g+2\rightarrow 2N}}+
% a^{j}b^{N-j}2^{(2j-g)\xi}\frac{v_{2j_1+2\rightarrow 2j_2-1\cdots 2j_{g-2}+2\rightarrow 2j_{g-1}-1, 2j_g+2\rightarrow 2N}}{v_{1\rightarrow 2j_1-1, 2j_2+2\rightarrow 2j_3-1, \cdots, 2j_{g-1}+2\rightarrow 2j_g-1}}
% \\
% &
% a^{N-j}b^{j}2^{(2j-g)\xi}\frac{v_{1\rightarrow 2j_1-1, \cdots, 2j_{g-1}+2\rightarrow 2j_g-1, 2j_g+2\rightarrow 2N}}{v_{2j_1+2\rightarrow 2j_2-1, 2j_3+2\rightarrow 2j_4-1, \cdots 2j_{g-1}+2\rightarrow 2j_{g}-1}}+
% a^{j}b^{N-j}2^{(2N-2j-g)\xi}\frac{v_{2j_1+2\rightarrow 2j_2-1, 2j_3+2\rightarrow 2j_4-1, \cdots 2j_{g-1}+2\rightarrow 2j_{g}-1}}{v_{1\rightarrow 2j_1-1, \cdots, 2j_{g-1}+2\rightarrow 2j_g-1, 2j_g+2\rightarrow 2N}}
% \\
&
\begin{cases}
\begin{aligned}
&a^{N-j}b^{j}2^{(2N-2j-g)\xi}\frac{v_{1\rightarrow 2j_1-1, 2j_2+2\rightarrow 2j_3-1, \cdots, 2j_{g-1}+2\rightarrow 2j_g-1}}{v_{2j_1+2\rightarrow 2j_2-1\cdots 2j_{g-2}+2\rightarrow 2j_{g-1}-1, 2j_g+2\rightarrow 2N}}+\\
&a^{j}b^{N-j}2^{(2j-g)\xi}\frac{v_{2j_1+2\rightarrow 2j_2-1\cdots 2j_{g-2}+2\rightarrow 2j_{g-1}-1, 2j_g+2\rightarrow 2N}}{v_{1\rightarrow 2j_1-1, 2j_2+2\rightarrow 2j_3-1, \cdots, 2j_{g-1}+2\rightarrow 2j_g-1}}
\end{aligned}, & j = \sum_{s=1}^g(-1)^{s+1} j_s, \quad \text{if } \mod(g,2)=1, \\
\begin{aligned}
&a^{N-j}b^{j}2^{(2N-2j-g)\xi}\frac{v_{2j_1+2\rightarrow 2j_2-1, 2j_3+2\rightarrow 2j_4-1, \cdots 2j_{g-1}+2\rightarrow 2j_{g}-1}}{v_{1\rightarrow 2j_1-1, \cdots, 2j_{g-1}+2\rightarrow 2j_g-1, 2j_g+2\rightarrow 2N}}+\\
&a^{j}b^{N-j}2^{(2j-g)\xi}\frac{v_{1\rightarrow 2j_1-1, \cdots, 2j_{g-1}+2\rightarrow 2j_g-1, 2j_g+2\rightarrow 2N}}{v_{2j_1+2\rightarrow 2j_2-1, 2j_3+2\rightarrow 2j_4-1, \cdots 2j_{g-1}+2\rightarrow 2j_{g}-1}}
\end{aligned}, & j = \sum_{s=1}^g(-1)^{s} j_s, \quad \text{if } \mod(g,2)=0.
\end{cases}
\\
=& 2^{n-t}\eta^g \prod_{r=1}^g(2^{2\xi}+v^2_{2j_r}v^2_{2j_r+1})\times
\\
&\begin{cases}
\begin{aligned}
&a^{N-j}b^{j}2^{(2N-2j-g)\xi}v^2_{1\rightarrow 2j_1-1, 2j_2+2\rightarrow 2j_3-1, \cdots, 2j_{g-1}+2\rightarrow 2j_g-1}+\\
&a^{j}b^{N-j}2^{(2j-g)\xi}v^2_{2j_1+2\rightarrow 2j_2-1\cdots 2j_{g-2}+2\rightarrow 2j_{g-1}-1, 2j_g+2\rightarrow 2N},
\end{aligned} & j = \sum_{s=1}^g(-1)^{s+1} j_s, \quad \text{if } \mod(g,2)=1, \\
\begin{aligned}
&a^{N-j}b^{j}2^{(2N-2j-g)\xi}v^2_{2j_1+2\rightarrow 2j_2-1, 2j_3+2\rightarrow 2j_4-1, \cdots 2j_{g-1}+2\rightarrow 2j_{g}-1}+\\
&a^{j}b^{N-j}2^{(2j-g)\xi}v^2_{1\rightarrow 2j_1-1, \cdots, 2j_{g-1}+2\rightarrow 2j_g-1, 2j_g+2\rightarrow 2N},
\end{aligned} & j = \sum_{s=1}^g(-1)^{s} j_s, \quad \text{if } \mod(g,2)=0.
\end{cases}
\end{split}
\end{equation}
Here, $u_{i_1\rightarrow j_1, \cdots, i_s\rightarrow j_s}$ means $\prod_{r=1}^s \prod_{t=i_r}^{j_r}u_t$. The meaning is the same for $v_{i_1\rightarrow j_1, \cdots, i_s\rightarrow j_s}$. We also use $t = \sum_{i=1}^{2N}n^i_F$ or $2^t = \prod_{i=1}^{2N}v_i$ in the above derivation.

Now, we come back to considering the condition of random erasure noise channels and take the expectation $\mathbb{E}_{\mathcal{N}}$. Note that the erasure channel $\mathcal{N}$ only influences the value of $v_i$. For a quantity $v_{i_1\rightarrow j_1, \cdots, i_s\rightarrow j_s}$, we have
\begin{equation}
\begin{split}
\mathbb{E}_{\mathcal{N}} v^2_{i_1\rightarrow j_1, \cdots, i_s\rightarrow j_s} &= \mathbb{E}_{\mathcal{N}} 2^{2\sum_{r=1}^s\sum_{t=i_r}^{j_r}{n^t_F}}.
\end{split}
\end{equation}
Denote $X = \sum_{r=1}^s\sum_{t=i_r}^{j_r}{n^t_F}$. This is a random variable representing the number of erasures in the region $\{i_1,i_1+1\cdots,j_1,i_2,i_2+1\cdots,j_2,\cdots,i_s,i_s+1,\cdots,j_s\}$.  Since we consider the $t$ erasures to be equally distributed on all of the qubits, the random variable $X$ satisfies
\begin{equation}
\Pr_{\mathcal{N}}(X = x) = \frac{\binom{m}{x}\binom{n-m}{t-x}}{\binom{n}{t}} = \frac{\binom{t}{x}\binom{n-t}{m-x}}{\binom{n}{m}},
\end{equation}
where $m = \abs{\{i_1,i_1+1\cdots,j_1,i_2,i_2+1\cdots,j_2,\cdots,i_s,i_s+1,\cdots,j_s\}}\xi$. We prove that
\begin{equation}
\begin{split}
\mathbb{E}_{\mathcal{N}}2^{2X}&=\sum_{x=\max(0, m+t-n)}^{\max(m,t)}2^{2x}\frac{\binom{t}{x}\binom{n-t}{m-x}}{\binom{n}{m}}\\
&\leq (1+\frac{3t}{n})^m\\
&\leq 2^{m\log(1+\frac{3t}{n})}.
\end{split}
\end{equation}
This can be seen that $\sum_{x=\max(0, m+t-n)}^{\max(m,t)}2^{2x}\binom{t}{x}\binom{n-t}{m-x}$ is the coefficient of $z^m$ in the expansion of the function $(1+4z)^{t}(1+z)^{n-t}$. With Lemma~\ref{lemma:coef} proved later, the coefficients of $z^m$ of $(1+4z)^{t}(1+z)^{n-t}$ is not larger than that of $(1+\frac{4t+n-t}{n}z)^n$, of which the coefficient of $z^m$ is $(1+\frac{3t}{n})^m\binom{n}{m}$.

Thus, under the expectation of the location of the erasures, we get
\begin{equation}
\begin{split}
&\mathbb{E}_{\mathcal{N}}2^n\begin{pmatrix}
u_{2N} & v_{2N}
\end{pmatrix}\left(\prod_{i=j_g+1}^{N-1}
CD_{2i, 2i+1}
\right)
CE_{2j_g,2j_g+1}
\cdots
CE_{2j_2,2j_2+1}
\left(\prod_{i=j_1+1}^{j_2-1}
CD_{2i, 2i+1}
\right)
CE_{2j_1,2j_1+1}
\left(\prod_{i=1}^{j_1-1}
CD_{2i, 2i+1}
\right)
C
\begin{pmatrix}
u_1 \\ v_1
\end{pmatrix}\\
\leq& 2^{n-t}\eta^g \sum_{h=0}^g\binom{g}{h}(a^{N-j}b^j2^{(2N-2j-g+2h)\xi+(2j+g-2h)\xi\log(1+\frac{3t}{n})}+a^{j}b^{N-j}2^{(2j-g+2h)\xi+(2N-2j+g-2h)\xi\log(1+\frac{3t}{n})})\\
\leq& 2^{n-t}2^{-g\xi} \sum_{h=0}^g\binom{g}{h}(2^{-2n-\frac{N-j}{N}k}2^{(2N-2j-g+2h)\xi+(2j+g-2h)\xi\log(1+\frac{3t}{n})}+2^{-2n-\frac{j}{N}k}2^{(2j-g+2h)\xi+(2N-2j+g-2h)\xi\log(1+\frac{3t}{n})})\\
\leq& 2^{-t-k}2^{-\frac{g}{2N}n} \sum_{h=0}^g\binom{g}{h}(2^{\frac{2j\xi}{n}k}2^{-(2j+g-2h)\xi+(2j+g-2h)\xi\log(1+\frac{3t}{n})}+2^{(1-\frac{2j\xi}{n})k}2^{-(2N-2j+g-2h)\xi+(2N-2j+g-2h)\xi\log(1+\frac{3t}{n})})\\
=& 2^{-t-k}2^{-\frac{g}{2N}(n-k)} \sum_{h=0}^g\binom{g}{h}(2^{\frac{(2j-g)\xi}{n}k}2^{-(2j+g-2h)\xi+(2j+g-2h)\xi\log(1+\frac{3t}{n})}+2^{(1-\frac{(2j+g)\xi}{n})k}2^{-(2N-2j+g-2h)\xi+(2N-2j+g-2h)\xi\log(1+\frac{3t}{n})})\\
=& 2^{-t-k}2^{-\frac{g}{2N}(n-k)} \sum_{h=0}^g\binom{g}{h}(2^{\frac{(2h-2g)\xi}{n}k}2^{-(2j+g-2h)\xi(1-\frac{k}{n}-\log(1+\frac{3t}{n}))}+2^{\frac{(2h-2g)\xi}{n}k}2^{-(2N-2j+g-2h)\xi(1-\frac{k}{n}-\log(1+\frac{3t}{n}))}).
\end{split}
\end{equation}
Note that $j = \sum_{s=1}^g(-1)^{s+1} j_s$ or $j = \sum_{s=1}^g(-1)^{s} j_s$ imply $j\geq \frac{g}{2}$. As long as
\begin{equation}\label{eq:erasurebound}
1-\frac{k}{n}-\log(1+\frac{3t}{n}) \geq 0,
\end{equation}
we have
\begin{equation}\label{eq:erasure_gnonzero}
\begin{split}
&\mathbb{E}_{\mathcal{N}}2^n\begin{pmatrix}
u_{2N} & v_{2N}
\end{pmatrix}\left(\prod_{i=j_g+1}^{N-1}
CD_{2i, 2i+1}
\right)
CE_{2j_g,2j_g+1}
\cdots
CE_{2j_2,2j_2+1}
\left(\prod_{i=j_1+1}^{j_2-1}
CD_{2i, 2i+1}
\right)
CE_{2j_1,2j_1+1}
\left(\prod_{i=1}^{j_1-1}
CD_{2i, 2i+1}
\right)
C
\begin{pmatrix}
u_1 \\ v_1
\end{pmatrix}\\
&\leq 2^{1-t-k}2^{-\frac{g}{2N}(n-k)} \sum_{h=0}^g\binom{g}{h}2^{\frac{(2h-2g)\xi}{n}k}\\
&= 2^{1-t-k} (2^{-\xi(1+\frac{k}{n})} + 2^{-\xi(1-\frac{k}{n})})^g.
\end{split}
\end{equation}
Combining Eq.~\eqref{eq:erasure_gzero} and \eqref{eq:erasure_gnonzero}, we obtain
\begin{equation}\label{eq:erasure_nonlinearterm}
\begin{split}
&\mathbb{E}_{\mathcal{N}}\tr \mathbb{E}_{U}[\hat{\mathcal{N}}_{S\rightarrow E}(U\ketbra{\hat{\phi}}_{LR}\otimes \ketbra{0^{n-k}} U^{\dagger})]^2\\
&= 2^n\begin{pmatrix}
u_{2N} & v_{2N}
\end{pmatrix}\left(\prod_{i=1}^{N-1}
C(D_{2i, 2i+1}+E_{2i,2i+1})
\right)
C
\begin{pmatrix}
u_1 \\ v_1
\end{pmatrix}\\
\leq& 2^{-t-k}+2^{-n+t}+\sum_{g=1}^{N-1}\binom{N-1}{g}2^{1-t-k}(2^{-\xi(1+\frac{k}{n})} + 2^{-\xi(1-\frac{k}{n})})^g\\
=& 2^{-t-k}+2^{-n+t}+2^{1-t-k}((1+2^{-\xi(1+\frac{k}{n})} + 2^{-\xi(1-\frac{k}{n})})^{N-1}-1).
\end{split}
\end{equation}
Note that $\xi = \log(n/\varepsilon)$ and $N = n/2\xi$. Utilizing $(1+x/n)^n\leq e^x\leq 1+2x$ for sufficiently small $x$, we get
\begin{equation}
\begin{split}
(1+2^{-\xi(1+\frac{k}{n})} + 2^{-\xi(1-\frac{k}{n})})^{N-1}-1 &\leq 2N( 2^{-\xi(1+\frac{k}{n})} + 2^{-\xi(1-\frac{k}{n})} )\\
&\leq 4N2^{-\xi(1-\frac{k}{n})}\\
&= \frac{2n}{\log(n/\varepsilon)}2^{-\log(n/\varepsilon)(1-\frac{k}{n})}\\
&= \frac{2\varepsilon^{1-\frac{k}{n}}n^{\frac{k}{n}}}{\log(n/\varepsilon)}.
\end{split}
\end{equation}
The condition that $x$ is small can be achieved by $(\frac{\varepsilon}{n})^{1-\frac{k}{n}} = o(\frac{1}{n})$.

Integrating Eqs.~\eqref{eq:choierrorerasurebound},~\eqref{eq:erasure_linearterm}, and~\eqref{eq:erasure_nonlinearterm}, we obtain
\begin{equation}
\begin{split}
\mathbb{E}_U\epsilon_{\mathrm{Choi}}&\leq \left(2^{t+k}(-2^{-t-k}+2^{-t-k}+2^{-n+t}+2^{1-t-k}((1+2^{-\xi(1+\frac{k}{n})} + 2^{-\xi(1-\frac{k}{n})})^{N-1}-1))\right)^{\frac{1}{4}}\\
&= \left(2^{-(n-2t-k)}+2((1+2^{-\xi(1+\frac{k}{n})} + 2^{-\xi(1-\frac{k}{n})})^{N-1}-1)\right)^{\frac{1}{4}}\\
&\leq \left(2^{-(n-2t-k)}+\frac{4\varepsilon^{1-\frac{k}{n}}n^{\frac{k}{n}}}{\log(n/\varepsilon)}\right)^{\frac{1}{4}}.
\end{split}
\end{equation}
Proof is done.
\end{proof}

\begin{lemma}\label{lemma:coef}
Given $a, b > 0$, $m, n\in \mathbb{Z}_+$, the coefficient of $z^t$ in the expansion of $(1+az)^m(1+bz)^n$ is not larger than that in the expansion of $(1+(am+bn)z/(m+n))^{m+n}$ for any $0\leq t\leq m+n$.
\end{lemma}
\begin{proof}
We prove Lemma~\ref{lemma:coef} by induction. For any $a, b > 0$, $m\in \mathbb{Z}_+$, when $n=1$, the coefficient of $z^t$ is
\begin{equation}
\binom{m}{t}a^t + \binom{m}{t-1}a^{t-1}b,
\end{equation}
for $(1+az)^m(1+bz)^n$, and is
\begin{equation}
\binom{m+1}{t}(\frac{ma+b}{m+1})^t
\end{equation}
for $(1+(am+bn)z/(m+n))^{m+n}$; $t$ takes values from $0$ to $m+1$. Now we need to prove
\begin{equation}
\binom{m}{t}a^t + \binom{m}{t-1}a^{t-1}b\leq \binom{m+1}{t}(\frac{ma+b}{m+1})^t,
\end{equation}
which is equivalent to
\begin{equation}\label{eq:binominequ}
1+\frac{t}{m+1}(\frac{b}{a}-1)\leq (1+\frac{1}{m+1}(\frac{b}{a}-1))^t.
\end{equation}
Note that $\frac{1}{m+1}(\frac{b}{a}-1) > -1$ and when $x> -1$, $(1+x)^r\geq 1+rx$ for any $r\geq 1$ and $r=0$. This can be obtained by considering $g(r) = (1+x)^r - (1+rx)$. First, $g(0) = g(1) = 0$. Meanwhile, for $r\geq 1$,
\begin{equation}
\begin{split}
g'(r) &= (1+x)^r\ln (1+x) - x\\
&\geq (1+x)^r\frac{x}{1+x}-x\\
&= ((1+x)^{r-1}-1)x\\
&\geq 0.
\end{split}
\end{equation}
In the second line, we use the inequality $\ln(1+x)\geq x/(1+x)$. Thus, for any $r\geq 1$, $g(r)\geq 0$. As a consequence, Eq.~\eqref{eq:binominequ} is correct, and we prove the case for $n=1$.

Now assuming that for any $a, b > 0$ and $m\in \mathbb{Z}_+$, the claim is correct for $n\leq n_0$, we try to prove the claim for $n = n_0+1$. First, using the case $(m, n_0)$, we have
\begin{equation}
(1+az)^{m}(1+bz)^{n_0}\leq (1+\frac{am+bn_0}{m+n_0}z)^{m+n_0}.
\end{equation}

Here, $p_1(z) \leq p_2(z)$ means that each coefficient of $z^t$ in $p_1(z)$ is less or equal to that in $p_2(z)$. Next, using the case $(m+n_0, 1)$, we have
\begin{equation}
(1+\frac{am+bn_0}{m+n_0}z)^{m+n_0}(1+bz)\leq (1+\frac{a(m+1)+bn_0}{m+1+n_0}z)^{m+n_0+1}.
\end{equation}
Combining the equations above, we prove the case for $n = n_0+1$. The whole proof is done by induction.
\end{proof}

\subsection{Proof of Proposition \ref{prop:pauli_twirling}} \label{app:arbitrary_input}
In this part, we prove Proposition~\ref{prop:pauli_twirling}, which shows that applying random Pauli matrices before encoding and after decoding ensures that the recovery fidelity for any arbitrary input state is lower-bounded by the Choi fidelity.

\begin{proof}[Proof of Proposition~\ref{prop:pauli_twirling}]
Consider the Stinespring dilation of the channel $\cN \circ \cE$, which can be represented as an isometry $V_{\cN}$ from $L$ to $SE$. The pure state $V_{\cN} \ket{\psi}_{LR}$ is a purification of $\cN \circ \cE(\psi_{LR})$. A recovery channel $\cD$ can also be represented as a unitary $U_{\cD}$ acting on $S$ and an ancillary system $A$. By Uhlmann's Theorem \cite{Uhlmann1976transition, Beny2010AQEC}, the Choi fidelity can be expressed as:
\begin{equation}
\begin{split}
F_{\mathrm{Choi}} &= \max_{U'_{(S\backslash L), AE}} \abs{(\bra{\phi}_{LR} \otimes \bra{0}_{(S\backslash L), AE}) U'_{(S\backslash L), AE} U_{\cD}V_{\cN} (\ket{\phi}_{LR} \otimes \ket{0}_{A})} \\
&= \frac{1}{d_L}  \max_{U'_{(S\backslash L), AE}} \abs{\tr_L[(\bra{0}_{(S\backslash L), AE}) U'_{(S\backslash L), AE} U_{\cD}V_{\cN} (I_L\otimes \ket{0}_{A}) ]} .
\end{split}
\end{equation}
The recovery fidelity for input state $\psi$ using the recovery channel $\mathcal{\cD}$ can be written as
\begin{equation}
F_{\psi} = \max_{U'_{(S\backslash L), AE}} \abs{(\bra{\psi}_{LR} \otimes \bra{0}_{(S\backslash L), AE}) U'_{(S\backslash L), AE} U_{\cD}V_{\cN} (\ket{\psi}_{LR} \otimes \ket{0}_{A})}.
\end{equation}

Now, we evaluate the expected recovery fidelity when applying a random Pauli matrix. The expression for the expected recovery fidelity is given by:
\begin{equation}
\begin{split}
\bE_{P_L} F_{\psi} &= \bE_{P_L} \max_{U'_{(S\backslash L), AE}} \abs{(\bra{\psi}_{LR} \otimes \bra{0}_{(S\backslash L), AE}) U'_{(S\backslash L), AE} P^{\dagger}_L U_{\cD}V_{\cN} P_L (\ket{\psi}_{LR} \otimes \ket{0}_{A})} \\
&\ge \max_{U'_{(S\backslash L), AE}} \bE_{P_L}  \abs{(\bra{\psi}_{LR} \otimes \bra{0}_{(S\backslash L), AE}) U'_{(S\backslash L), AE}P^{\dagger}_L U_{\cD}V_{\cN} P_L (\ket{\psi}_{LR} \otimes \ket{0}_{A})} \\
&\ge  \max_{U'_{(S\backslash L), AE}}   \abs{ \bE_{P_L}(\bra{\psi}_{LR} \otimes \bra{0}_{(S\backslash L), AE}) U'_{(S\backslash L), AE} P^{\dagger}_L U_{\cD}V_{\cN} P_L (\ket{\psi}_{LR} \otimes \ket{0}_{A})} \\
&= \frac{1}{d_L} \max_{U'_{(S\backslash L), AE}} \abs{\tr_L[(\bra{0}_{(S\backslash L), AE}) U'_{(S\backslash L), AE} U_{\cD_0}V_{\cN} (I_L\otimes \ket{0}_{A}) ]} \\
&= F_{\mathrm{Choi}}.
\end{split}
\end{equation}
In the first line, we apply $P_L$ before encoding and set $P^{\dagger}_L U_{\cD}$ as the recovery channel. The second line swaps the maximum outside, which reduces the average fidelity. In the third line, we move the expectation inside the absolute value, further reducing the fidelity. Finally, the fourth line applies Pauli twirling. Specifically, for any state $\ket{\psi_{LR}}$, we have
\begin{equation}
\bE_{P_L} [P_L \psi_{LR} P^{\dagger}_L] = \tr_L(\psi_{LR}) \otimes \frac{I_L}{d_L}.
\end{equation}

\end{proof}

\subsection{The Choi error of random Clifford encoding}\label{app:Clif}
In this part, we introduce the bound of the Choi error of random Clifford encoding. The results can be viewed as a baseline to measure the performance of other encoding methods. Note that the random Clifford encoding is the random stabilizer code~\cite{gottesman1997stabilizer}, which is a good code with a constant encoding rate and linear code distance in general.

Recall that the Choi error is bounded by
\begin{equation}
\mathbb{E}_{C}\epsilon_{\mathrm{Choi}}\leq \sqrt{ \mathbb{E}_{C} \Vert \hat{\mathcal{N}}_{S\rightarrow E}(U\ketbra{\hat{\phi}}_{LR}\otimes \ketbra{0^{n-k}} U^{\dagger}) - \hat{\mathcal{N}}_{S\rightarrow E}(\frac{\id_S}{2^n})\otimes \frac{\id_R}{2^k} \Vert_1},
\end{equation}
where $\mathbb{E}_{C}$ is the expectation over the $n$-qubit Clifford group. We can apply the non-smooth decoupling theorem for the 2-design group, or Lemma~\ref{lemma:haardecoupling} for further evaluation:
\begin{equation}
\mathbb{E}_{C} \Vert \hat{\mathcal{N}}_{S\rightarrow E}(U\ketbra{\hat{\phi}}_{LR}\otimes \ketbra{0^{n-k}} U^{\dagger}) - \hat{\mathcal{N}}_{S\rightarrow E}(\frac{\id_S}{2^n})\otimes \frac{\id_R}{2^k} \Vert_1\leq 2^{-\frac{1}{2}H_2(S|R)_{\ketbra{\hat{\phi}}_{LR}\otimes \ketbra{0^{n-k}}}-\frac{1}{2}H_2(S|E)_{\hat{\mathcal{N}}_{S'\rightarrow E}(\ketbra{\hat{\phi}}_{SS'})}}.
\end{equation}
Another choice is the smooth decoupling theorem shown in Lemma~\ref{lemma:haardecoupling}:
\begin{equation}
\mathbb{E}_{C} \Vert \hat{\mathcal{N}}_{S\rightarrow E}(U\ketbra{\hat{\phi}}_{LR}\otimes \ketbra{0^{n-k}} U^{\dagger}) - \hat{\mathcal{N}}_{S\rightarrow E}(\frac{\id_S}{2^n})\otimes \frac{\id_R}{2^k} \Vert_1\leq 2^{-\frac{1}{2}H_2(S|R)_{\ketbra{\hat{\phi}}_{LR}\otimes \ketbra{0^{n-k}}}-\frac{1}{2}H^{\delta}_{\min}(S|E)_{\hat{\mathcal{N}}_{S'\rightarrow E}(\ketbra{\hat{\phi}}_{SS'})}}+8\delta.
\end{equation}
Note that we only smooth $H_2(S|E)$ with $H^{\delta}_{\min}(S|E)$ while maintaining $H_2(S|R)$ to evaluate a tighter bound like what did in the proof of Corollary~\ref{coro:pauli}. Particularly,
\begin{equation}
\begin{split}
H_2(S|R)_{\ketbra{\hat{\phi}}_{LR}\otimes \ketbra{0^{n-k}}} &= H_2(L|R)_{\ketbra{\hat{\phi}}_{LR}}\\
&= -k.
\end{split}
\end{equation}

Below, we continue the evaluation for specific noise models. For the Pauli noise,
\begin{equation}
\hat{\mathcal{N}}_{S\rightarrow E} = (\hat{\mathcal{N}}_p)^{\otimes n}
\end{equation}
with $\hat{\mathcal{N}}_p$ as the complementary channel of the single-qubit Pauli noise $\mathcal{N}_p(\rho) = p_I\rho+p_XX\rho X+p_YY\rho Y+p_ZZ\rho Z$.  Particularly,
\begin{equation}
\hat{\mathcal{N}}_{S'\rightarrow E}(\ketbra{\hat{\phi}}_{SS'}) = (\hat{\mathcal{N}}_p(\ketbra{\hat{\phi}}))^{\otimes n}.
\end{equation}
When applying the non-smooth decoupling theorem, we first calculate
\begin{equation}
H_2(S|E)_{(\hat{\mathcal{N}}_p(\ketbra{\hat{\phi}}))^{\otimes n}(\ketbra{\hat{\phi}}_{SS'})} = n(1-f(\Vec{p}))
\end{equation}
by Eq.~\eqref{eq:paulicollision}. Thus,
\begin{equation}
\mathbb{E}_{C}\epsilon_{\mathrm{Choi}}\leq 2^{-\frac{n}{4}(1-f(\Vec{p})-\frac{k}{n})}.
\end{equation}

Next, we apply the smooth decoupling theorem to Pauli noise. Note that
\begin{equation}
\frac{1}{n}H^{\delta}_{\mathrm{min}}(S|E)_{\rho^{\otimes n}} \geq H(S|E)_{\rho} - \frac{4\log \gamma \sqrt{\log \frac{2}{\delta^2}}}{\sqrt{n}},
\end{equation}
and
\begin{equation}
\lim_{n\rightarrow \infty} \frac{4\log \gamma \sqrt{\log \frac{2}{\delta^2}}}{\sqrt{n}} = 0.
\end{equation}
For sufficiently large $n$, we have $\frac{4\log \gamma \sqrt{\log \frac{2}{\delta^2}}}{\sqrt{n}} < \delta$. Meanwhile, $H(S|E)_{\hat{\mathcal{N}}_p(\ketbra{\hat{\phi}})} = 1-h(\Vec{p})$. Hence, for any $\delta > 0$,
\begin{equation}
\mathbb{E}_{C} \Vert \hat{\mathcal{N}}_{S\rightarrow E}(U\ketbra{\hat{\phi}}_{LR}\otimes \ketbra{0^{n-k}} U^{\dagger}) - \hat{\mathcal{N}}_{S\rightarrow E}(\frac{\id_S}{2^n})\otimes \frac{\id_R}{2^k} \Vert_1\leq 2^{-\frac{n}{2}(1-h(\Vec{p})-\frac{k}{n}-\delta)}+8\delta;
\end{equation}
\begin{equation}
\mathbb{E}_{C}\epsilon_{\mathrm{Choi}}\leq \sqrt{8\delta + 2^{-\frac{n}{2}(1-h(\Vec{p})-\frac{k}{n}-\delta)}}.
\end{equation}

The derivation for the i.i.d.~erasure error and local amplitude damping noise is similar. For strength-$p$ i.i.d.~erasure error, we can get
\begin{equation}
\mathbb{E}_{C}\epsilon_{\mathrm{Choi}}\leq 2^{-\frac{n}{4}(1-\log(1+3p)-\frac{k}{n})}
\end{equation}
with the non-smooth decoupling theorem, and
\begin{equation}
\mathbb{E}_{C}\epsilon_{\mathrm{Choi}}\leq \sqrt{8\delta + 2^{-\frac{n}{2}(1-2p-\frac{k}{n}-\delta)}}
\end{equation}
with the smooth decoupling theorem.

For strength-$p$ local amplitude damping noise, we can get
\begin{equation}
\mathbb{E}_{C}\epsilon_{\mathrm{Choi}}\leq 2^{-\frac{n}{4}(-\log( \frac{1}{2-p} + \sqrt{\frac{p}{2-p}} )-\frac{k}{n})}
\end{equation}
with the non-smooth decoupling theorem, and
\begin{equation}
\mathbb{E}_{C}\epsilon_{\mathrm{Choi}}\leq \sqrt{8\delta + 2^{-\frac{n}{2}(h(\frac{1-p}{2})-h(\frac{p}{2})-\frac{k}{n}-\delta)}}
\end{equation}
with the smooth decoupling theorem.

For strength-$p$ 1D nearest neighbor $ZZ$-coupling noise, the analysis is done by the non-smooth decoupling theorem. The $H_2(S|R)$ term is still equal to $-k$, and the $H_2(S|E)$ term is given by
\begin{equation}
\begin{split}
H_2(S|E)_{\tau_{SE}}
&= -\log \tr_{S'}[\tr_S\sqrt{ \mathcal{N}_{zz}(\ketbra{\hat{\phi}}^{\otimes n})} ]\\
&= -\log \tr_{S'}[\tr_S \prod_{\{i,i+1\}\subseteq [n]}[(1-p)\mathcal{I}+p\mathcal{Z}_i\mathcal{Z}_{i+1}] (\ketbra{\hat{\phi}}_{SS'}) ]\\
&= n - 2\log( (\sqrt{1-p}+\sqrt{p})^{n-1})\\
&\geq n [1 -2\log(\sqrt{1-p}+\sqrt{p})].
\end{split}
\end{equation}
Thus, for $ZZ$-coupling noise, we have
\begin{equation}\label{eq:zzcouplingrandomClif}
\mathbb{E}_{C}\epsilon_{\mathrm{Choi}}\leq 2^{-\frac{n}{4}( 1-2\log(\sqrt{1-p}+\sqrt{p})-\frac{k}{n})}.
\end{equation}

Finally, we consider the $t$-erasure errors. The random Clifford encoding is robust to a fixed $t$-erasure error rather than a random erasure error. The number of erasure errors should satisfy the Singleton bound $n-2t-k\geq 0$~\cite{Cerf1997singleton}. The corresponding complementary channel is $\hat{\mathcal{N}}_{S\rightarrow E} = \tr_{n-t}$. By applying the non-smooth decoupling theorem, we have that
\begin{equation}
\hat{\mathcal{N}}_{S'\rightarrow E}(\ketbra{\hat{\phi}}_{SS'}) = \ketbra{\hat{\phi}}^{\otimes t}_{TE} \otimes \frac{\id_{\Bar{T}}}{2^{n-t}};
\end{equation}
\begin{equation}
\begin{split}
H_2(S|E)_{\hat{\mathcal{N}}_{S'\rightarrow E}(\ketbra{\hat{\phi}}_{SS'})} &= -\log \tr( ((\frac{\id_{E}}{2^t})^{-\frac{1}{4}}
\ketbra{\hat{\phi}}^{\otimes t}_{TE} \otimes \frac{\id_{\Bar{T}}}{2^{n-t}} (\frac{\id_{E}}{2^t})^{-\frac{1}{4}})^2)\\
&= n - 2t.
\end{split}
\end{equation}
Here, $T$ represents the region of the erasure errors. Combining the results above, we get the upper bound of the Choi error against a fixed $t$-erasure error for the random Clifford encoding:
\begin{equation}
\mathbb{E}_{C}\epsilon_{\mathrm{Choi}}\leq 2^{-\frac{1}{4}(n-2t-k)}.
\end{equation}

\subsection{Comparison to the block-encoding scheme}\label{app:block}
In this part, we compare the double-layer blocked encoding scheme and the block-encoding scheme. We consider two types of noise models. For local Pauli noise, we provide evidence that the double-layer blocked encoding scheme might have an advantage. For local erasure error, we rigorously establish an advantage of the double-layer blocked encoding scheme over the block-encoding scheme in certain error regimes.

\textbf{Local Pauli noise.} We first consider the local Pauli noise. As mentioned in the main text, the encoding unitary operation of the block-encoding method is
\begin{equation}
U = \bigotimes_{i=1}^{\frac{N}{2}}U_C^{4i-3,4i-2,4i-1,4i},
\end{equation}
where $U_C^{4i-3,4i-2,4i-1,4i}$ is a random Clifford gate on regions $4i-3,4i-2,4i-1,4i$. The set of the random block-encoding unitary operations is denoted as $\mathfrak{S}_n^B$.

To compare the performance of the two encoding methods, we only need to evaluate
\begin{equation}
\mathbb{E}_{U} \Vert \hat{\mathcal{N}}_{S\rightarrow E}(U\ketbra{\hat{\phi}}_{LR}\otimes \ketbra{0^{n-k}} U^{\dagger}) - \hat{\mathcal{N}}_{S\rightarrow E}(\frac{\id_S}{2^n})\otimes \frac{\id_R}{2^k} \Vert_1
\end{equation}
under two encoding ensembles $\mathfrak{C}_n^{\varepsilon}$ and $\mathfrak{S}_n^B$. For simplicity, here we only analyze the results associated with the non-smooth decoupling theorem. The case with the smooth decoupling theorem is similar.

The result of $\mathfrak{C}_n^{\varepsilon}$ is given by Eq.~\eqref{eq:twolayerNonsmooth}. The result of $\mathfrak{S}_n^B$ can be established by considering the error of each block and the triangle inequality. Specifically,
\begin{equation}
\begin{split}
&\mathbb{E}_{U\sim \mathfrak{S}_n^B} \Vert \hat{\mathcal{N}}_{S\rightarrow E}(U\ketbra{\hat{\phi}}_{LR}\otimes \ketbra{0^{n-k}} U^{\dagger}) - \hat{\mathcal{N}}_{S\rightarrow E}(\frac{\id_S}{2^n})\otimes \frac{\id_R}{2^k} \Vert_1\\
\leq&\frac{n}{4\xi} \mathbb{E}_{U\sim \mathfrak{S}_{4\xi}^B} \Vert \hat{\mathcal{N}}_{S\rightarrow E}(U\ketbra{\hat{\phi}}_{LR}\otimes \ketbra{0^{\frac{4\xi(n-k)}{n}}} U^{\dagger}) - \hat{\mathcal{N}}_{S\rightarrow E}(\frac{\id_S}{2^{4\xi}})\otimes \frac{\id_R}{2^{\frac{4\xi k}{n}}} \Vert_1\\
\leq&\frac{n}{4\xi}2^{-\frac{1}{2}H_2(L|R)_{\ketbra{\hat{\phi}}^{\otimes 4\xi k/n}}-\frac{1}{2}H_2(S|E)_{(\hat{\mathcal{N}}_p(\ketbra{\hat{\phi}}))^{\otimes 4\xi}}}.
\end{split}
\end{equation}

Using $H_2(S|E)_{\hat{\mathcal{N}}_p(\ketbra{\hat{\phi}})} = 1-f(\Vec{p})$, we have that
\begin{equation}
\begin{split}
&\mathbb{E}_{U\sim \mathfrak{S}_n^B} \Vert \hat{\mathcal{N}}_{S\rightarrow E}(U\ketbra{\hat{\phi}}_{LR}\otimes \ketbra{0^{n-k}} U^{\dagger}) - \hat{\mathcal{N}}_{S\rightarrow E}(\frac{\id_S}{2^n})\otimes \frac{\id_R}{2^k} \Vert_1\\
\leq&\frac{n}{4\xi}(2^{-2\xi(1-f(\Vec{p})-\frac{k}{n})})\\
=&\frac{n}{4\log\frac{n}{\varepsilon}}(\frac{\varepsilon}{n})^{2(1-f(\Vec{p})-\frac{k}{n})}.
\end{split}
\end{equation}
Recall that the result of the double-layer blocked encoding is given by
\begin{equation}
\begin{split}
&\mathbb{E}_{U\sim \mathfrak{C}_n^{\varepsilon}} \Vert \hat{\mathcal{N}}_{S\rightarrow E}(U\ketbra{\hat{\phi}}_{LR}\otimes \ketbra{0^{n-k}} U^{\dagger}) - \hat{\mathcal{N}}_{S\rightarrow E}(\frac{\id_S}{2^n})\otimes \frac{\id_R}{2^k} \Vert_1\\
\leq& \sqrt{2^{-n(1-f(\Vec{p})-\frac{k}{n})} + \frac{4\varepsilon^{1-\frac{k}{n}}n^{\frac{k}{n}}}{\log (n/\varepsilon)}}.
\end{split}
\end{equation}
For block encoding, the dependence of error on $n$ is given by $\frac{n}{4\log\frac{n}{\varepsilon}}(\frac{\varepsilon}{n})^{2(1-f(\Vec{p})-\frac{k}{n})}$. For double-layer blocked encoding, it is approximately given by two terms $\sqrt{2^{-n(1-f(\Vec{p})-\frac{k}{n})}}$ and $\sqrt{\frac{4\varepsilon^{1-\frac{k}{n}}n^{\frac{k}{n}}}{\log (n/\varepsilon)}}$. Note that $\varepsilon$ is normally chosen to make $\frac{\varepsilon}{n} = 1/\poly(n)$. Meanwhile, $1-f(\Vec{p})-\frac{k}{n}$ is often a constant. Thus, an exponentially decay term must vanish faster:
\begin{equation}
\sqrt{2^{-n(1-f(\Vec{p})-\frac{k}{n})}} \ll \frac{n}{4\log\frac{n}{\varepsilon}}(\frac{\varepsilon}{n})^{2(1-f(\Vec{p})-\frac{k}{n})}.
\end{equation}
Also, when $1-f(\Vec{p})-\frac{k}{n} < \frac{1}{4}(1-\frac{k}{n})$, we have
\begin{equation}
\lim_{n\rightarrow \infty} \frac{\sqrt{\frac{4\varepsilon^{1-\frac{k}{n}}n^{\frac{k}{n}}}{\log (n/\varepsilon)}}}{\frac{n}{4\log\frac{n}{\varepsilon}}(\frac{\varepsilon}{n})^{2(1-f(\Vec{p})-\frac{k}{n})}} = 8\lim_{n\rightarrow \infty} \sqrt{\frac{\log\frac{n}{\varepsilon}}{n}} (\frac{\varepsilon}{n})^{2f(\vec{p})-\frac{3}{2}(1-\frac{k}{n})} = 0.
\end{equation}
Thus,
\begin{equation}
\sqrt{\frac{4\varepsilon^{1-\frac{k}{n}}n^{\frac{k}{n}}}{\log (n/\varepsilon)}} \ll \frac{n}{4\log\frac{n}{\varepsilon}}(\frac{\varepsilon}{n})^{2(1-f(\Vec{p})-\frac{k}{n})}.
\end{equation}
That is, given the noise strength $\vec{p}$, when the encoding rate satisfies $\frac{k}{n}>1-\frac{4}{3}f(\vec{p})$, the error given by the double-layer blocked encoding decays much faster. In other words, if we fix a required Choi error decay scaling, the encoding rate offered by the double-layer blocked scheme can be higher. This provides evidence that the double-layer blocked encoding scheme can have an advantage in some error regimes. Nonetheless, it is a comparison of the upper bound of the error. Below, we consider local erasure errors and evaluate a lower bound of the Choi error given by the block-encoding method. This establishes a rigorous advantage of the double-layer blocked encoding scheme.

\textbf{Local Erasure error.} To establish the lower bound of the Choi error given by the block-encoding method, we only need to evaluate the lower bound of the Choi error of each block, which corresponds to a random stabilizer encoding.

For a random stabilizer code $[[n, k]]$, correctability under erasure error requires the number of erasure errors $t$ to satisfy the singleton bound:
\begin{equation}
2t + k \le n \quad \Rightarrow \quad t \le \left\lfloor \frac{n - k}{2} \right\rfloor.
\end{equation}
Define $\tau = \frac{1}{2}(1-\frac{k}{n})$. Then any erasure pattern erasing more than $\tau n$ qubits is uncorrectable. Moreover, as long as the number of erasure errors is $t > \left\lfloor \frac{n - k}{2} \right\rfloor$, we can prove that there exists a lower bound of the Choi error for random stabilizer codes. The result is summarized as follows.

\begin{theorem}\label{thm:choi_random_erasure}
Let $t\le n/2$. Suppose $Q$ is a subset of $n$ qubits with $\abs{Q}=t$. Then the expectation of Choi error against an erasure error on $Q$ over $[[n, k]]$ random stabilizer codes has a lower bound:
\begin{equation}
\mathbb{E}_{C}\epsilon_Q
\ge\frac{2t-(n-k)-e^{-1}}{8k},
\end{equation}
where $\epsilon_Q$ is the Choi error $\epsilon_{\mathrm{Choi}}$ for erasure on $Q$.
\end{theorem}

We postpone the proof of Theorem~\ref{thm:choi_random_erasure} to the end of this subsection. With the number of erasure errors increasing, the Choi error also monotonically increases. Thus, we have the following corollary.
\begin{corollary}
For $t>n/2$,
\begin{equation}
\mathbb{E}_{C}\epsilon_Q
\ge
\mathbb{E}_{C}\epsilon_{Q'}
=\frac18-\frac{1}{ek},
\end{equation}
where $\abs{Q'} = n/2$.
\end{corollary}

For strength-$p$ local erasure errors, each qubit is erased independently with probability $p$. The number of erased qubits satisfies a binomial distribution $B(n, p)$, where we denote this number or random variable as $E_n$. Based on Theorem~\ref{thm:choi_random_erasure}, the Choi error is always lower bounded by $\frac{1-e^{-1}}{8k}\geq \frac{1}{16n}$ when $E_n > \tau n$. Then the Choi error of a random stabilizer code is lower bounded by $\frac{1}{16n}\Pr\left[ E_n > \tau n \right]$. Using Chernoff bounds, for $\tau > p$ or $1-\frac{k}{n}-2p>0$, we have
\begin{equation}
\Pr\left[ E_n > \tau n \right] = \exp\left( -n \cdot D\left( \tau \,\|\, p \right) + o(n) \right),
\end{equation}
where $D(a \| b)$ is the binary relative entropy:
\begin{equation}
D(a \| b) = a \log_2\left( \frac{a}{b} \right) + (1 - a) \log_2\left( \frac{1 - a}{1 - b} \right).
\end{equation}

Thus, the Choi error satisfies
\begin{equation}
\mathbb{E}_{C}\epsilon_{\mathrm{Choi}} \ge \frac{1}{16n}\exp\left( -n \cdot D\left( \tau \Big\| p \right) + o(n) \right) = \exp\left( -n D\left( \tau \Big\| p \right) + o(n) \right).
\end{equation}

To further give a lower bound of the Choi error for the block-encoding scheme, we introduce the near-optimal fidelity~\cite{Guo2024NearOptimal}, denoted as $\widetilde{F}_{\mathrm{Choi}}$, with the corresponding near-optimal Choi error denoted as $\widetilde{\epsilon}_{\mathrm{Choi}}$. Compared to the definition of original Choi fidelity,
\begin{equation}
F_{\mathrm{Choi}} = \max_{\mathcal{D}} F\left(\hat{\phi}_{LR}, [(\mathcal{D}\circ\mathcal{N}\circ\mathcal{E})_L\otimes I_R](\hat{\phi}_{LR})\right),
\end{equation}
the near-optimal fidelity uses the transpose channel recovery for decoding,
\begin{equation}
\widetilde{F}_{\mathrm{Choi}} = F\left(\hat{\phi}_{LR}, [(\mathcal{D}^{TC}\circ\mathcal{N}\circ\mathcal{E})_L\otimes I_R](\hat{\phi}_{LR})\right).
\end{equation}
The transpose channel recovery $\mathcal{D}^{TC}$ is defined as~\cite{Guo2024NearOptimal}
\begin{equation}
\mathcal{D}^{TC} = P_L\circ \mathcal{N}^{\dagger} \circ \mathcal{N}^{-\frac{1}{2}}(P_L).
\end{equation}
Here, $\mathcal{N}$ is the noise channel. The map $P_L$ satisfies $P_L(\rho) = \Pi_L \rho \Pi_L$, where $\Pi_L$ is the projector to the logical space. The map $\mathcal{N}^{-\frac{1}{2}}(P_L)$ satisfies $[\mathcal{N}^{-\frac{1}{2}}(P_L)](\rho) = \mathcal{N}^{-\frac{1}{2}}(\Pi_L)\rho \mathcal{N}^{-\frac{1}{2}}(\Pi_L)$.

It has been proved that there exists two-sided bounds between the optimal fidelity and the near-optimal fidelity~\cite{Barnum2002Reversing,Guo2024NearOptimal}:
\begin{equation}
\frac{1}{2}(1-\widetilde{F}_{\mathrm{Choi}})\leq 1-F_{\mathrm{Choi}}\leq 1-\widetilde{F}_{\mathrm{Choi}}.
\end{equation}
Using the relation $F = \sqrt{1-\epsilon^2}$, we have that for sufficiently small $\epsilon_{\mathrm{Choi}}$,
\begin{equation}
\epsilon_{\mathrm{Choi}}\leq \widetilde{\epsilon}_{\mathrm{Choi}}\leq \sqrt{2}\epsilon_{\mathrm{Choi}}.
\end{equation}
The left inequality follows from the near-optimality. The right inequality is derived by
\begin{equation}
\begin{split}
&\frac{1}{2}(1-\sqrt{1-\widetilde{\epsilon}^2_{\mathrm{Choi}}})\leq 1-\sqrt{1-\epsilon^2_{\mathrm{Choi}}}\\
\Leftrightarrow &  4(1-\epsilon^2_{\mathrm{Choi}})\leq 2-\widetilde{\epsilon}^2_{\mathrm{Choi}} + 2\sqrt{1-\widetilde{\epsilon}^2_{\mathrm{Choi}}}\leq 2-\widetilde{\epsilon}^2_{\mathrm{Choi}} + 2(1-\frac{1}{2}\widetilde{\epsilon}^2_{\mathrm{Choi}})\\
\Leftrightarrow & \widetilde{\epsilon}_{\mathrm{Choi}}\leq \sqrt{2}\epsilon_{\mathrm{Choi}}.
\end{split}
\end{equation}
Note that in the second line we use $\sqrt{1-x}\leq 1-\frac{1}{2}x$ for sufficiently small $x$.

From the definition of the transpose recovery map, it can be seen that for a block-encoding scheme $\mathcal{E} = \bigotimes_i \mathcal{E}_i$ and block noise model $\mathcal{N} = \bigotimes_i \mathcal{N}_i$, the transpose recovery map is also a tensor product, $\mathcal{D}^{TC} = \bigotimes_i \mathcal{D}^{TC}_i$. Thus, given a block-encoding scheme, since each block is independent, and the noise also acts independently on each block, the near-optimal fidelity for the whole block-encoding scheme is the product of the near-optimal fidelity for each block.
\begin{equation}
\widetilde{F}_{\mathrm{Choi}} = \prod_{i=1}^{\frac{n}{4\xi}}\widetilde{F}_{\mathrm{block}_i},
\end{equation}
or equivalently,
\begin{equation}
\widetilde{F}^2_{\mathrm{Choi}} = \prod_{i=1}^{\frac{n}{4\xi}}\widetilde{F}^2_{\mathrm{block}_i}.
\end{equation}
Then,
\begin{equation}
1-\widetilde{\epsilon}_{\mathrm{Choi}}^2 = \prod_{i=1}^{\frac{n}{4\xi}}(1-\widetilde{\epsilon}_{\mathrm{block}_i}^2),
\end{equation}
where
\begin{equation}
\mathbb{E}_{U\sim \mathfrak{S}_n^B}\widetilde{\epsilon}_{\mathrm{block}_i}\ge \mathbb{E}_{U\sim \mathfrak{S}_n^B}\epsilon_{\mathrm{block}_i} \ge \exp\left( -4\xi \cdot D\left( \tau \Big\| p \right) + o(\xi) \right).
\end{equation}
Thus, we obtain the lower bound of the Choi error for the block-encoding scheme:
\begin{equation}
\begin{split}
\mathbb{E}_{U\sim \mathfrak{S}_n^B}\epsilon_{\mathrm{Choi}}
&\geq \frac{1}{\sqrt{2}} \mathbb{E}_{U\sim \mathfrak{S}_n^B}\widetilde{\epsilon}_{\mathrm{Choi}}\\
&\geq \frac{1}{\sqrt{2}} \mathbb{E}_{U\sim \mathfrak{S}_n^B}\widetilde{\epsilon}^2_{\mathrm{Choi}}\\
&= \frac{1}{\sqrt{2}}\mathbb{E}_{U\sim \mathfrak{S}_n^B}(1-\prod_i(1-\widetilde{\epsilon}_{\mathrm{block}_i}^2))\\
&= \frac{1}{\sqrt{2}}(1-\prod_i(1-\mathbb{E}_{U\sim \mathfrak{S}_n^B}\widetilde{\epsilon}_{\mathrm{block}_i}^2))\\
&\geq \frac{1}{\sqrt{2}}(1-\prod_i(1-[\mathbb{E}_{U\sim \mathfrak{S}_n^B}\widetilde{\epsilon}_{\mathrm{block}_i}]^2))\\
&\geq \frac{1}{\sqrt{2}}(1-(1-\exp\left( -8\xi \cdot D\left( \tau \Big\| p \right) + o(\xi) \right))^{\frac{n}{4\xi}})\\
&\geq\min \left\{ \frac{n}{8\sqrt{2}\xi}\exp\left( -8\xi \cdot D\left( \tau \Big\| p \right) + o(\xi) \right), \Omega(1) \right\}\\
&=\min\left\{ n\exp\left( -8\xi D\left( \tau \Big\| p \right) + o(\xi) \right), \Omega(1) \right\}.
\end{split}
\end{equation}
Here, the fourth and fifth lines use the independence of each block and $\mathbb{E}x^2\geq (\mathbb{E}x)^2$, respectively. The last inequality comes as follows.  Let $m\coloneqq \frac{n}{4\xi}$ and $x \coloneqq m \exp(-8\xi \cdot D\left( \tau \Big\| p \right) + o(\xi))$. Then
\begin{equation}
    (1-(1-\exp\left( -8\xi \cdot D\left( \tau \Big\| p \right) + o(\xi) \right))^{\frac{n}{4\xi}}) = \left(1 - (1-\frac{x}{m})^m\right),
\end{equation}
Using $(1-\frac{x}{m})^{m}\le 1- \frac{x}{2}$ for $0\le x\le 1$, we obtain the first term; while for $x>1$, $(1-(1-\frac{x}{m})^m) \ge (1-(1-\frac{1}{m})^m) = \Omega(1)$, we obtain the second term.

Substituting $\xi$ with $\log(n/\varepsilon)$, we get that
\begin{equation}
\begin{split}
\mathbb{E}_{U\sim \mathfrak{S}_n^B}\epsilon_{\mathrm{Choi}}
\geq \min \left\{ n \left(\frac{\varepsilon}{n}\right)^{8D( \frac{1-\frac{k}{n}}{2} \| p )+o(1)}, \Omega(1) \right\}.
\end{split}
\end{equation}
This equation shows that, for the block-encoding scheme, given the noise strength $p$, a trade-off relation between the encoding rate $\frac{k}{n}$ and the Choi error exists. If one fixes an encoding rate, the Choi error will have a lower bound, and vice versa. Below, we prove that there exists a regime of the strength parameter $p$ where the double-layer blocked encoding scheme achieves both a higher encoding rate and a lower Choi error. Moreover, there exist regimes of noise strength and encoding rate such that the double-layer blocked encoding scheme requires a lower depth to make the Choi error decay.

Recall Corollary~\ref{coro:iiderasure_nonsmoothing}, the Choi error for the double-layer blocked encoding scheme is:
\begin{equation}
\mathbb{E}_{U\sim \mathfrak{C}_n^{\varepsilon}}\epsilon_{\mathrm{Choi}} \leq (2^{-n(1-\log(1+3p)-\frac{k}{n})} + \frac{4\varepsilon^{1-\frac{k}{n}}n^{\frac{k}{n}}}{\log (n/\varepsilon)})^{\frac{1}{4}}.
\end{equation}
The main contributions are $2^{-\frac{n}{4}(1-\log(1+3p)-\frac{k}{n})}$ and $(\frac{4\varepsilon^{1-\frac{k}{n}}n^{\frac{k}{n}}}{\log (n/\varepsilon)})^{\frac{1}{4}}$. In the following, we discuss the case of $1-\log(1+3p)-\frac{k}{n} > 0$. For simplicity, we set $r = \frac{k}{n}$ and $\varepsilon = n^{-\frac{\alpha+r}{1-r}}$ where $r$ and $\alpha$ and both positive constants; $p$ is also a constant. Then, $\frac{n}{\varepsilon} = n^{\frac{1+\alpha}{1-r}}$. The block size parameter $\xi = \log\frac{n}{\varepsilon} = \frac{1+\alpha}{1-r}\log n$. The upper bound of the Choi error for the double-layer blocked encoding scheme is simplified to:
\begin{equation}
\mathbb{E}_{U\sim \mathfrak{C}_n^{\varepsilon}}\epsilon_{\mathrm{Choi}} \leq (2^{-n(1-\log(1+3p)-r)} + \frac{4(1-r)}{(1+\alpha)n^{\alpha}\log n})^{\frac{1}{4}}=O(n^{-\frac{\alpha}{4}}).
\end{equation}
The lower bound of the Choi error for the block-encoding scheme is simplified to:
\begin{equation}\label{eq:blocklowerboundsimple}
\mathbb{E}_{U\sim \mathfrak{S}_n^B}\epsilon_{\mathrm{Choi}}
\geq \min\left\{n^{1 - \frac{8(1+\alpha)}{1-r}D(\frac{1-r}{2}\|p)+o(1)} = n^{c_{\mathrm{block}}},\Omega(1)\right\}.
\end{equation}

In the case that $0<p<1/3,0<r<1,\alpha>0$ and $1-\log(1+3p)-r > 0$, the Choi error of the double-layer blocked encoding scheme always decays. The decay scale is approximately $O(n^{-\frac{\alpha}{4}})$ where we ignore $\log n = n^{o(1)}$. On the other hand, the exponential coefficient within Eq.~\eqref{eq:blocklowerboundsimple} is
\begin{equation}
c_{\mathrm{block}} = 1 - \frac{8(1+\alpha)}{1-r}D(\frac{1-r}{2}\|p)+o(1).
\end{equation}
Given noise strength $p$, under the constraint of $r < 1-\log(1+3p)$ and $\alpha > 0$, $c_{\mathrm{block}}$ can be greater than 0. For instance,
\begin{equation}
p = 0.25, r = 0.2, \alpha = 0.1\Rightarrow c_{\mathrm{block}} \approx 0.14 \Rightarrow \mathbb{E}_{U\sim \mathfrak{S}_n^B}\epsilon_{\mathrm{Choi}}
\geq \min\{n^{0.14}, \Omega(1) \} = \Omega(1).
\end{equation}
That is, given the noise strength $p = 0.25$ and encoding rate $r = 0.1$, we obtain that the double-layer blocked encoding scheme can have an error decay $n^{-0.025}$, and the block-encoding scheme cannot make the Choi error decay. In this case, $\xi = 1.375\log n$. As mentioned in the main text, the depths for the two encoding schemes are the same after fixing $\xi$. This manifests that for $p = 0.25$ and $r = 0.1$, the double-layer blocked encoding scheme requires a lower depth to make the Choi error vanish. On the other hand, this implies that at the same circuit depth, the double-layer blocked encoding scheme can achieve a higher encoding rate.

More concretely, one can find parameter regimes such that $c_{\mathrm{block}} > 0$. In this case, the double-layer blocked encoding scheme can correct errors, but the block-encoding scheme cannot. One can also find parameter regimes such that $c_{\mathrm{block}} > -\frac{\alpha}{4}$. In this case, the double-layer blocked encoding scheme always outperforms the block-encoding scheme by manifesting a faster error decay. We plot the regimes that $c_{\mathrm{block}} > 0$ and $c_{\mathrm{block}} > -\frac{\alpha}{4}$ in Figs.~\ref{fig:blockfail} and~\ref{fig:blockcompare}, respectively.

\begin{figure}[htbp!]
\raggedright
% \raggedleft
\begin{minipage}[b]{0.98\linewidth}
\subfigure[]{\label{fig:blockfail}
\includegraphics[width=8.5cm]{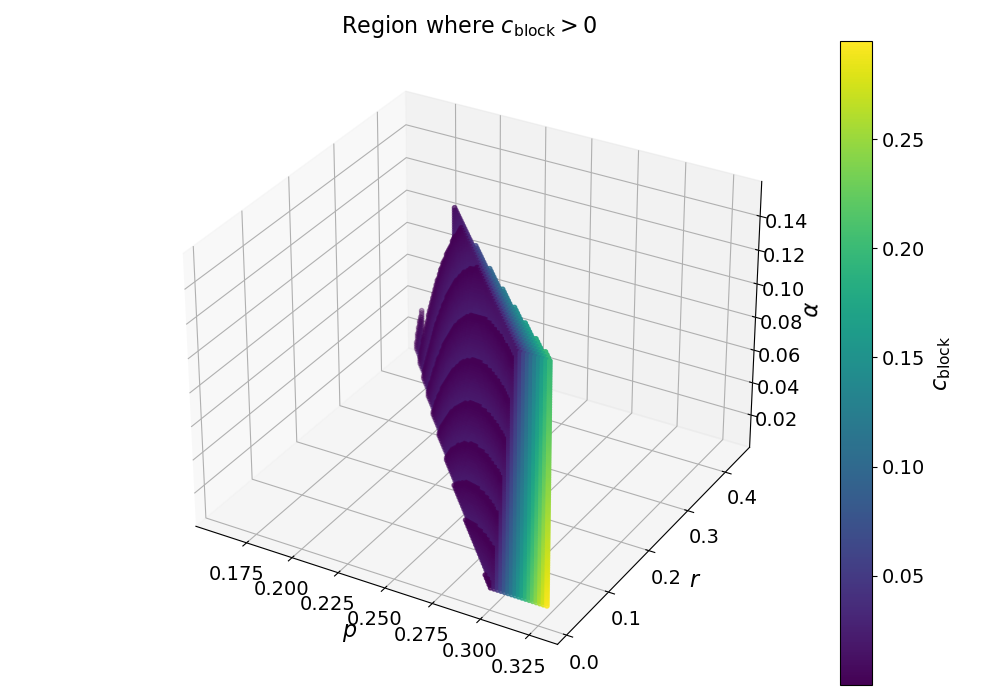}
}
\subfigure[]{
\label{fig:blockcompare}
\includegraphics[width=8.5cm]{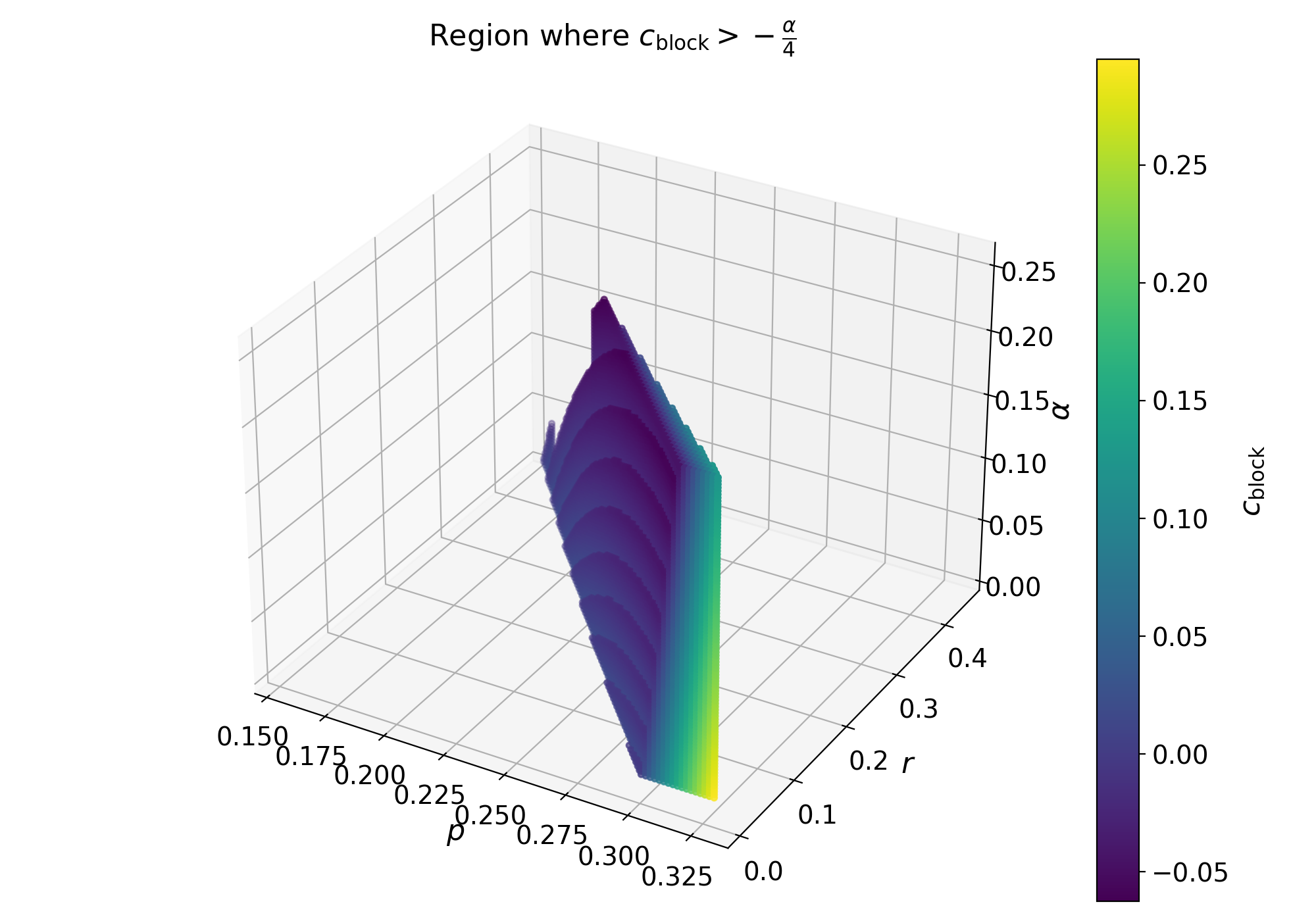}}
\end{minipage}
\caption{(a) Parameter regime where $c_{\mathrm{block}} > 0$. In this case, the double-layer blocked encoding scheme can make the Choi error decay, but the block-encoding scheme cannot. (b) Parameter regime where $c_{\mathrm{block}} > -\frac{\alpha}{4}$. In this case, the double-layer blocked encoding scheme outperforms the block-encoding scheme.}
\end{figure}

Below, we show the proof of Theorem~\ref{thm:choi_random_erasure}.
\begin{proof}
We first observe that the choice of subset $Q$ is irrelevant when considering the expectation over the random stabilizer encoding scheme.

\begin{lemma}\label{lem:fix_subsystem}
For any fixed subsets $Q_1, Q_2\subseteq [n]$ with $|Q_1|=|Q_2|=t$,
\begin{equation}
\mathbb{E}_{C}\epsilon_{C,Q_1} =\mathbb{E}_{C}\epsilon_{C,Q_2},
\end{equation}
where $\epsilon_{C,Q_i}$ is the Choi error for encoding unitary $C$ and erasure errors on $Q_i$.
\end{lemma}

\begin{proof}
Because the SWAP operation is Clifford, there exists a Clifford $S_{12}$ that swaps $Q_1$ to $Q_2$. Hence
\begin{equation}
\begin{split}
\mathbb{E}_{C}\epsilon_{C,Q_2}
&=\mathbb{E}_{C}\epsilon_{S_{12}C,Q_1}\\
&=\mathbb{E}_{C}\epsilon_{C,Q_1}.
\end{split}
\end{equation}
\end{proof}

Fix the subsets $Q_1=[t]$, $Q_2=\{t+1,\dots,2t\}$, and $Q_3=[n]\backslash(Q_1\cup Q_2)$ so that $Q_3$ contains the remaining $n-2t$ physical qubits. We now prove the following.
\begin{lemma}\label{lem:choi_two_t_subsystem}
For any unitary encoding $C$,
\begin{equation}
\frac12\bigl(\epsilon_{Q_1}+\epsilon_{Q_2}\bigr)
\ge
\frac{2t-(n-k)-e^{-1}}{8k}.
\end{equation}
\end{lemma}

\begin{proof}
Let $\epsilon_1\coloneqq\epsilon_{Q_1}$ and $\epsilon_2\coloneqq\epsilon_{Q_2}$. Also, recall that we denote $L$, $R$, and $S$ as the logical, reference, and physical systems, respectively. The initial input state is the $\hat{\phi}_{LR}$, and after encoding and noise, the state is denoted as $\hat{\phi}_{LR}$.

By definition of $\epsilon_Q$, there exists a recovery map $\mathcal{D}_{Q_2Q_3\!\rightarrow L}$ such that
\begin{equation}
\sigma=\mathcal{D}\!\bigl(\hat{\phi}_{LR}\bigr),
\qquad
P\bigl(\sigma_{LR},\hat{\phi}_{LR}\bigr)=\epsilon_1.
\end{equation}
Because the purified distance bounds the trace distance,
\begin{equation}
d_{\tr}\bigl(\sigma_{LR},\hat{\phi}_{LR}\bigr)\le\epsilon_1 .
\end{equation}
Applying the Fannes--Audenaert inequality~\cite{nielsen2010quantum} with subsystem size $|LR|=2k$ gives
\begin{equation}\label{eq:delta_1}
S(\sigma_{LR})
\le
S(\hat{\phi}_{LR})+2\epsilon_1 (2k)+\tfrac1e
=
4k\epsilon_1+\tfrac1e
\coloneqq \delta_1 .
\end{equation}
Therefore,
\begin{equation}
I_\sigma(Q_1:R) \le I_{\sigma}(Q_1:LR) \le 2S(\sigma_{LR}) \le 2 \delta_1.
\end{equation}
Here, $I_{\sigma}(A:B)$ is the mutual information between subsystems $A$ and $B$ for state $\sigma$. The first inequality comes from the data processing inequality for mutual information, and the second uses $I_\sigma(A:B)\le2S(\sigma_B)$ for arbitrary $\sigma$.
Let $\rho\coloneqq\hat{\phi}_{LR}$.  Because the recovery map $\mathcal{D}$ acts only on $Q_2Q_3$, the registers $Q_1$ and $R$ are untouched, so
\begin{equation}
I_{\rho}(Q_1:R) = I_{\sigma}(Q_1:R) \le 2\delta_1
\end{equation}
Expanding mutual information and using that $\rho$ is pure, we obtain
\begin{equation}
S(\rho_{Q_1})+S(\rho_R) \le S(\rho_{Q_1R})+2\delta_1 = S(\rho_{Q_2Q_3}) + 2\delta_1 \le S(\rho_{Q_2}) + S(\rho_{Q_3}) + 2\delta_1,
\end{equation}
where the last step uses subadditivity. Repeating the same argument for $Q_2$ gives
\begin{equation}
S(\rho_{Q_2})+S(\rho_R) \le S(\rho_{Q_2R})+2\delta_2 = S(\rho_{Q_1Q_3}) + 2\delta_2 \le S(\rho_{Q_1}) + S(\rho_{Q_3}) + 2\delta_2.
\end{equation}
Adding the two inequalities yields
\begin{equation}
2 S(\rho_R) \le 2S(\rho_{Q_3}) +2( \delta_1 + \delta_2).
\end{equation}
With $S(\rho_R)=k$ and $S(\rho_{Q_3})\le|Q_3|=n-2t$, we obtain
\begin{equation}
\delta_1 + \delta_2 \ge 2t - (n-k).
\end{equation}
Substituting \eqref{eq:delta_1}, we have
\begin{equation}
\frac{1}{2}(\epsilon_1 + \epsilon_2) \ge \frac{2t-(n-k) - e^{-1}}{8k}.
\end{equation}
\end{proof}

Based on Lemma~\ref{lem:fix_subsystem} and Lemma~\ref{lem:choi_two_t_subsystem}, we prove Theorem~\ref{thm:choi_random_erasure} as follows:
\begin{equation}
\mathbb{E}_{C}\epsilon_{Q_1} = \frac{1}{2}(\mathbb{E}_{C}\epsilon_{Q_1}+\mathbb{E}_{C}\epsilon_{Q_2}) = \mathbb{E}_{C}\frac{\epsilon_{Q_1}+\epsilon_{Q_2}}{2}\ge
\frac{2t-(n-k)-e^{-1}}{8k}.
\end{equation}
\end{proof}

\section{Analysis for 1D brickwork circuits}\label{appendssc:1DLRC}
In this part, we prove that the decoupling can be achieved by 1D log-depth brickwork circuits. We start with introducing preliminaries, which discuss how to evaluate $\tr O_2 \Phi^2_{\mathcal{E}} (O_1)$ by transforming it into a stochastic process. Then, we present the main results of the decoupling theorem for 1D brickwork circuits and prove Theorem~\ref{thm:1Dlocalrandomcircuit}. Finally, we apply this decoupling theorem to quantum error correction and show that there exists a constant threshold for encoding schemes with such circuits against local noise. Note that here we prove the result for a generic local dimension $q$.

\subsection{Random quantum circuits to stochastic processes}\label{sc:stochastic}
Here, we show how to evaluate $\tr O_2 \Phi^2_{\mathcal{E}} (O_1)$ by transforming it into a stochastic process for certain observables $O_1$ and $O_2$, which lays a foundation for further analyzing the decoupling error. The evaluation is equivalent to evaluating a statistical-mechanical partition function. Most of the analysis is similar to that in Ref.~\cite{Dalzell2022Anticoncentrate}. Nonetheless, in Ref.~\cite{Dalzell2022Anticoncentrate}, $O_1$ and $O_2$ take specific forms of $\ketbra{1^n}^{\otimes 2}$, which makes the boundary condition in the statistical-mechanical mapping easier. For general operators $O_1$ and $O_2$, the boundary condition is more complicated, and a more refined analysis is necessary.

We first evaluate $\Phi^2_{\mathcal{E}} (O_1)$. By definition,
\begin{equation}
\Phi^2_{\mathcal{E}} (O_1) = \mathbb{E}_{U\sim \mathcal{E}} U^{\otimes 2} O_1 U^{\dagger\otimes 2},
\end{equation}
where
\begin{equation}
U = U^{(s)}_{A^{(s)}}U^{(s-1)}_{A^{(s-1)}}\cdots U^{(1)}_{A^{(1)}}.
\end{equation}
We would like to evaluate terms like $\mathbb{E}_{U} U^{\otimes k} O U^{\dagger\otimes k}$. Also, $U^{(i)}$ is a Haar-random gate of dimension $q^2$. Inserting random single-qudit operations before all two-qudit gates does not change the expectation. Thus, we further consider the random unitary operation of the following form.
\begin{equation}\label{eq:eq:randomtwounitaryplusone}
U = U^{(s)}_{A^{(s)}}U^{(s-1)}_{A^{(s-1)}}\cdots U^{(1)}_{A^{(1)}}\bigotimes_{i=1}^n U^{(-i)}_{\{i\}}.
\end{equation}
Here, $U^{(-i)}_{\{i\}}$ is a random single-qudit operation with dimension $q\times q$ acting on qudit $i$. It is worth noting that when analyzing 1D brickwork circuits, we also insert a random single-qudit operation $\bigotimes_{i=1}^n U^{(-i)}_{\{i\}}$ before the random two-qudit unitary operations within Eq.~\eqref{eq:1DLQC}.

We continue the analysis by evaluating the actions of each part within $U$. The action of $\bigotimes_{i=1}^n U^{(-i)}_{\{i\}}$ on $O_1$ can be evaluated by Eq.~\eqref{eq:secondtwirling}. We obtain that
\begin{equation}
\mathbb{E}_{ U^{(-i)}}(\bigotimes_{i=1}^n U^{(-i)\otimes 2}_{\{i\}}) O_1 (\bigotimes_{i=1}^n U^{(-i)\dagger\otimes 2}_{\{i\}}) = \frac{1}{(q^2-1)^n}\sum_{\Vec{\gamma}\in \{I,S\}^n} \tr(O_1 \bigotimes_{j=1}^n g_{\gamma_j}) \bigotimes_{j=1}^n \gamma_j.
\end{equation}
Here, $\gamma_j$ is the $j$-th component of $\Vec{\gamma}$ and equals identity operation $\id$ or SWAP operation $F$. The term $g_{\gamma_j} = \id-q^{-1}F$ if $\gamma_j = \id$, and $g_{\gamma_j} = F-q^{-1}\id$ if $\gamma_j = F$. The vector $\gamma$ is called a configuration, corresponding to the spin configuration is the statistical-mechanical mapping. For simplicity, we denote $c_{\Vec{\gamma}}(O_1) = \frac{1}{(q^2-1)^n}\tr(O_1 \bigotimes_{j=1}^n g_{\gamma_j})$ and use $\Vec{\gamma}$ to represent $\bigotimes_{j=1}^n \gamma_j$ in the above equation. That is,
\begin{equation}
\mathbb{E}_{ U^{(-i)}}(\bigotimes_{i=1}^n U^{(-i)\otimes 2}_{\{i\}}) O_1 (\bigotimes_{i=1}^n U^{(-i)\dagger\otimes 2}_{\{i\}}) = \sum_{\Vec{\gamma}\in \{\id,F\}^n} c_{\Vec{\gamma}}(O_1)\Vec{\gamma}.
\end{equation}

The next step is evaluating the action of a random two-qudit gate $U^{(i)}_{A^{(i)}}$ on $\Vec{\gamma}$, which was previously studied in Ref.~\cite{Dalzell2022Anticoncentrate}. Since $A^{(i)}$ only contains two qudits, it suffices to analyze how $U^{(i)}_{A^{(i)}}$ acts on $\gamma$ within support $A^{(i)}$. Denote $M^{(i)}(O) = \mathbb{E}_{U^{(i)}}U^{(i)\otimes 2}_{A^{(i)}}O U^{(i)\dagger\otimes 2}_{A^{(i)}}$, we have that
\begin{equation}
M^{(i)}(O) = \tr_{A^{(i)}}(O\frac{\id_{A^(i)}-q^{-2}F_{A^(i)}}{q^4-1})\id_{A^{(i)}} + \tr_{A^{(i)}}(O\frac{F_{A^(i)}-q^{-2}\id_{A^(i)}}{q^4-1})F_{A^{(i)}}.
\end{equation}
Particularly, we consider that $O$ equals $\id\otimes \id$, $\id\otimes F$, $F\otimes \id$, and $F\otimes F$ within support $A^{(i)}$. We have that
\begin{align}
M^{(i)}(O_{\overline{A^{(i)}}}\otimes (\id\otimes \id)_{A^{(i)}}) &= O_{\overline{A^{(i)}}}\otimes (\id\otimes \id)_{A^{(i)}},\\
M^{(i)}(O_{\overline{A^{(i)}}}\otimes (F\otimes F)_{A^{(i)}}) &= O_{\overline{A^{(i)}}}\otimes (F\otimes F)_{A^{(i)}},\\
M^{(i)}(O_{\overline{A^{(i)}}}\otimes (\id\otimes F)_{A^{(i)}}) &= O_{\overline{A^{(i)}}}\otimes\left( \frac{q}{q^2+1}(\id\otimes \id)_{A^{(i)}} + \frac{q}{q^2+1}(F\otimes F)_{A^{(i)}} \right),\\
M^{(i)}(O_{\overline{A^{(i)}}}\otimes (F\otimes \id)_{A^{(i)}}) &= O_{\overline{A^{(i)}}}\otimes\left( \frac{q}{q^2+1}(\id\otimes \id)_{A^{(i)}} + \frac{q}{q^2+1}(F\otimes F)_{A^{(i)}} \right).
\end{align}
Here, $O_{\overline{A^{(i)}}}$ only has support outside $A^{(i)}$. From the above equation, we obtain that if $O\in \{\id,F\}^n$, then $M^{(i)}(O)\in \{\id,F\}^n$. Particularly, $M^{(i)}$ is a linear transformation over $\{\id,F\}^n$, and we use $M^{(i)}_{\Vec{\nu}\Vec{\gamma}}$ to denote the matrix element such that
\begin{equation}
M^{(i)}(\Vec{\gamma}) = \sum_{\Vec{\nu}\in \{\id,F\}^n}   M^{(i)}_{\Vec{\nu}\Vec{\gamma}}\Vec{\nu}.
\end{equation}
Suppose $A^{(i)}=\{a_i,b_i\}$, the matrix element $M^{(i)}_{\Vec{\nu}\Vec{\gamma}}$ satisfies
\begin{equation}\label{eq:transitionmatrix}
M^{(i)}_{\Vec{\nu}\Vec{\gamma}} =
\begin{cases}
\begin{aligned}
1,
\end{aligned} \quad &\text{if } \gamma_{a_i}=\gamma_{b_i}\text{ and } \Vec{\gamma}=\Vec{\nu},\\
\begin{aligned}
\frac{q}{q^2+1},
\end{aligned} \quad &\text{if } \gamma_{a_i}\neq\gamma_{b_i}, \nu_{a_i}=\nu_{b_i}, \text{and } \gamma_c=\nu_c, \forall c\in [n]\backslash\{a_i, b_i\}, \\
\begin{aligned}
0,
\end{aligned} \quad &\text{otherwise.}
\end{cases}
\end{equation}
It can be seen that each action of $M^{(i)}$ maintains the configuration $\Vec{\gamma}$ with weight $1$, or flips one term in $\Vec{\gamma}$ with weight $\frac{q}{q^2+1}$. Thus, $\Phi^2_{\mathcal{E}} (O_1)$ can be expressed as follows.
\begin{equation}
\begin{split}
\Phi^2_{\mathcal{E}} (O_1) &= M^{(s)}\cdots M^{(2)}M^{(1)}( \sum_{\Vec{\gamma}\in \{\id,F\}^n} c_{\Vec{\gamma}}(O_1)\Vec{\gamma} )\\
&= \sum_{\Vec{\gamma}\in \{\id,F\}^n} c_{\Vec{\gamma}}(O_1) M^{(s)}\cdots M^{(2)}M^{(1)}(\Vec{\gamma})\\
&= \sum_{\Gamma=(\Vec{\gamma}^{(0)},\cdots,\Vec{\gamma}^{(s)})\in \{\id,F\}^{n\times (s+1)}} c_{\Vec{\gamma}^{(0)}}(O_1)\Vec{\gamma}^{(s)} \prod_{i=1}^{s} M^{(i)}_{\Vec{\gamma}^{(i)}\Vec{\gamma}^{(i-1)}}.
\end{split}
\end{equation}
Here, $\Gamma=(\Vec{\gamma}^{(0)},\cdots,\Vec{\gamma}^{(s)})$ is called a trajectory, which is a sequence of $s+1$ configurations. We also use $\Gamma_i$ to denote its $i$-th component. Here, $\Gamma_i = \Vec{\gamma}^{(i-1)}$. Furthermore,
\begin{equation}
\begin{split}
\tr (O_2\Phi^2_{\mathcal{E}} (O_1)) &= \sum_{\Gamma}\weight(\Gamma) c_{\Vec{\gamma}^{(0)}}(O_1) \tr ( O_2\Vec{\gamma}^{(s)})\\
&= \sum_{\Gamma}\weight(\Gamma) c_{\Vec{\gamma}^{(0)}\Vec{\gamma}^{(s)}}(O_1,O_2),
\end{split}
\end{equation}
where
\begin{align}
\weight(\Gamma) &= \prod_{i=1}^{s} M^{(i)}_{\Vec{\gamma}^{(i)}\Vec{\gamma}^{(i-1)}},\\
\label{eq:cgamma0gammas}c_{\Vec{\gamma}^{(0)}\Vec{\gamma}^{(s)}}(O_1,O_2) &= \frac{1}{(q^2-1)^n}  \tr(O_1 \bigotimes_{j=1}^n g_{\gamma^{(0)}_j}) \tr(O_2 \bigotimes_{j=1}^n \gamma^{(s)}_j).
\end{align}
When $O_1=O_2=\ketbra{1^n}^{\otimes 2}$, $c_{\Vec{\gamma}^{(0)}\Vec{\gamma}^{(s)}}(O_1,O_2) = \frac{1}{(q+1)^n}$, which is the case studied in Ref.~\cite{Dalzell2022Anticoncentrate}. By mapping $\{\id,F\}$ to $\{+1, -1\}$, a configuration $\Vec{\gamma}$ is a spin variable, and the trajectory $\Gamma$ is a two-dimensional Ising spin configuration. Correspondingly, $\tr (O_2\Phi^2_{\mathcal{E}} (O_1))$ is mapped into a partition function of a classical Ising spin model, which is a weighted summation of different trajectories. For each trajectory $\Gamma$, the total weight depends on two terms. The first one is $\weight(\Gamma)$, where the value depends on the whole trajectory. The second term is $c_{\Vec{\gamma}^{(0)}\Vec{\gamma}^{(s)}}(O_1,O_2)$, and the value only depends on the initial and ending configurations of the trajectory. Note that $\weight(\Gamma)$ is always non-negative since the matrix element $M^{(i)}_{\Vec{\gamma}^{(i)}\Vec{\gamma}^{(i-1)}}$ is always non-negative.

For further analysis, we write
\begin{equation}\label{eq:randomcircuittrace}
\tr (O_2\Phi^2_{\mathcal{E}} (O_1)) = \sum_{\Vec{\gamma}^{(0)}\Vec{\gamma}^{(s)}} c_{\Vec{\gamma}^{(0)}\Vec{\gamma}^{(s)}}(O_1,O_2) \sum_{\Gamma\backslash\Vec{\gamma}^{(0)}\Vec{\gamma}^{(s)}}\weight(\Gamma).
\end{equation}
Here $\sum_{\Gamma\backslash\Vec{\gamma}^{(0)}\Vec{\gamma}^{(s)}}$ takes the summation of trajectories that start with $\Vec{\gamma}^{(0)}$ and end with $\Vec{\gamma}^{(s)}$.
In the following, we denote $Z_{\Vec{\gamma}^{(0)}\Vec{\gamma}^{(s)}} = \sum_{\Gamma\backslash\Vec{\gamma}^{(0)}\Vec{\gamma}^{(s)}}\weight(\Gamma)$ and evaluate this term to analyze $\tr (O_2\Phi^2_{\mathcal{E}} (O_1))$. We start with analyzing $\weight(\Gamma)$ with the method of random walk.

Given a fixed set of supports $\{A^{(1)}, A^{(2)}, \cdots, A^{(s)}\}$, the action of $\prod_{i=1}^{s}M^{(i)}$ on a configuration $\Vec{\gamma}^{(0)}$ results in a length-$(s+1)$ random walk $\Gamma$ over $\{\id, F\}^n$. The initial location of the random walk is a fixed configuration, $\Vec{\gamma}^{(0)}$, and at time step $t$, the location transitions from $\Vec{\gamma}^{(t-1)}$ to $\Vec{\gamma}^{(t)}$. The configuration after step $1$ may not be fixed. Based on Eq.~\eqref{eq:transitionmatrix}, the transition must follow the following rule: the configuration is unchanged from $\Vec{\gamma}^{(t-1)}$ to $\Vec{\gamma}^{(t)}$ if the operators of $\Vec{\gamma}^{(t-1)}$ within support $A^{(t)}$ are $\id\,\id$ and $FF$; there is an identity operation flipped into a SWAP operation or a SWAP operation flipped into an identity operation if the operators of $\Vec{\gamma}^{(t-1)}$ within support $A^{(t)}$ are $\id F$ and $F \id$, and the probability of each case of flip is $\frac{1}{2}$. If the configuration is unchanged, the contributed weight of this transition to $\weight(\Gamma)$ is $1$, and if there is a flip, the contributed weight is $\frac{q}{q^2+1}$. As a result, $\weight(\Gamma)$ depends on how many flips within trajectory $\Gamma$. That is,
\begin{equation}
\weight(\Gamma) = (\frac{q}{q^2+1})^{\#F_{\Gamma}},
\end{equation}
where $\#F_{\Gamma}$ is the number of flips within $\Gamma$. It can also be written as
\begin{equation}\label{eq:probgamma}
\begin{split}
\weight(\Gamma) &= (\frac{1}{2})^{\#F_{\Gamma}}(\frac{2q}{q^2+1})^{\#F_{\Gamma}}\\
&= \Pr(\Gamma)(\frac{2q}{q^2+1})^{\#F_{\Gamma}},
\end{split}
\end{equation}
where $\Pr(\Gamma) = (\frac{1}{2})^{\#F_{\Gamma}}$ is the probability of trajectory $\Gamma$ in this random walk.

To calculate $Z_{\Vec{\gamma}^{(0)}\Vec{\gamma}^{(s)}}$, we need to count all trajectories satisfying the transition rule and require the start and ending configurations to be $\Vec{\gamma}^{(0)}$ and $\Vec{\gamma}^{(s)}$, respectively. Then, we obtain that
\begin{equation}
Z_{\Vec{\gamma}^{(0)}\Vec{\gamma}^{(s)}} = \sum_{\Gamma}\weight(\Gamma)\mathbf{1}_{\Gamma_{1}=\Vec{\gamma}^{(0)}}\mathbf{1}_{\Gamma_{s+1}=\Vec{\gamma}^{(s)}}.
\end{equation}
Here, $\mathbf{1}_A$ is the indicator function, which equals $1$ if $A$ is satisfied and 0 otherwise. Using Eq.~\eqref{eq:probgamma}, we have that
\begin{equation}
Z_{\Vec{\gamma}^{(0)}\Vec{\gamma}^{(s)}} = \mathbb{E}_{\Gamma}(\frac{2q}{q^2+1})^{\#F_{\Gamma}}\mathbf{1}_{\Gamma_{1}=\Vec{\gamma}^{(0)}}\mathbf{1}_{\Gamma_{s+1}=\Vec{\gamma}^{(s)}}.
\end{equation}

Note that the weight of a particular walk or trajectory in the above random walk relates to the number of flips during the walk, which brings challenges in evaluating its value. To tackle this issue, we turn the above unbiased random walk into a biased random walk. In the original unbiased random walk, the flip has an equal probability from $\id$ to $F$ and from $F$ to $\id$. In the new biased random walk, we set the probability from $\id$ to $F$ as $\frac{1}{q^2+1}$, and the probability from $F$ to $\id$ as $\frac{q^2}{q^2+1}$. Accordingly, we can modify the form of the weight of a trajectory into the following.
\begin{equation}\label{eq:probgammabias}
\begin{split}
\weight(\Gamma) = &(\frac{q}{q^2+1})^{\#F_{\Gamma}}\\
= &(\frac{q}{q^2+1})^{\#F^{\id\rightarrow F}_{\Gamma}}(\frac{q}{q^2+1})^{\#F^{F\rightarrow \id}_{\Gamma}}\\
= &(\frac{1}{q^2+1})^{\#F^{\id\rightarrow F}_{\Gamma}}(\frac{q^2}{q^2+1})^{\#F^{F\rightarrow \id}_{\Gamma}}q^{\#F^{\id\rightarrow F}_{\Gamma}-\#F^{F\rightarrow \id}_{\Gamma}}\\
= &\widetilde{\Pr}(\Gamma)q^{\abs{\Vec{\gamma}^{(s)}}-\abs{\Vec{\gamma}^{(0)}}},
\end{split}
\end{equation}
where $\#F^{\id\rightarrow F}_{\Gamma}$ is the number of flips from $\id$ to $F$ within $\Gamma$, $\#F^{F\rightarrow \id}_{\Gamma}$ is the number of flips from $F$ to $\id$ within $\Gamma$, and $\widetilde{\Pr}(\Gamma) = (\frac{1}{q^2+1})^{\#F^{\id\rightarrow F}_{\Gamma}}(\frac{q^2}{q^2+1})^{\#F^{F\rightarrow \id}_{\Gamma}}$ is the probability of trajectory $\Gamma$ in the biased random walk, and $\abs{\Vec{\gamma}^{(i)}}$ is the number of SWAP operators in $\Vec{\gamma}^{(i)}$. Then,
\begin{equation}\label{eq:Zgamma}
\begin{split}
Z_{\Vec{\gamma}^{(0)}\Vec{\gamma}^{(s)}} &= \sum_{\Gamma}\weight(\Gamma)\mathbf{1}_{\Gamma_{1}=\Vec{\gamma}^{(0)}}\mathbf{1}_{\Gamma_{s+1}=\Vec{\gamma}^{(s)}}\\
&= \sum_{\Gamma}\widetilde{\Pr}(\Gamma)q^{\abs{\Vec{\gamma}^{(s)}}-\abs{\Vec{\gamma}^{(0)}}} \mathbf{1}_{\Gamma_{1}=\Vec{\gamma}^{(0)}}\mathbf{1}_{\Gamma_{s+1}=\Vec{\gamma}^{(s)}}\\
&= q^{\abs{\Vec{\gamma}^{(s)}}-\abs{\Vec{\gamma}^{(0)}}}\widetilde{\mathbb{E}}_{\Gamma}\mathbf{1}_{\Gamma_{1}=\Vec{\gamma}^{(0)}}\mathbf{1}_{\Gamma_{s+1}=\Vec{\gamma}^{(s)}}\\
&= q^{\abs{\Vec{\gamma}^{(s)}}-\abs{\Vec{\gamma}^{(0)}}}\widetilde{\Pr}(\Gamma: \Vec{\gamma}^{(0)}\rightarrow \Vec{\gamma}^{(s)}).
\end{split}
\end{equation}
Here, $\widetilde{\mathbb{E}}_{\Gamma}$ represents the expectation over biased random walk, and $\widetilde{\Pr}(\Gamma: \Vec{\gamma}^{(0)}\rightarrow \Vec{\gamma}^{(s)})$ represents the probability of this biased random walk transitioning from $\Vec{\gamma}^{(0)}$ to $\Vec{\gamma}^{(s)}$.

For an infinite-size circuit or $s\rightarrow\infty$, the final configuration of the random walk can only take $\id^n$ and $F^n$ since they are the fixed points of the biased random walk. It was found that in this case the probability of $\Vec{\gamma}^{(0)}$ transitioning into $\id^n$ or $F^n$ only depends on the weight of $\Vec{\gamma}^{(0)}$ or $\abs{\Vec{\gamma}^{(0)}}$~\cite{Dalzell2022Anticoncentrate}. The concrete values are given by
\begin{align}
\label{eq:probgammatoI}\widetilde{\Pr}(\Gamma: \Vec{\gamma}^{(0)}\rightarrow \id^n) &= \frac{1-q^{-2n+2\abs{\Vec{\gamma}^{(0)}}}}{1-q^{-2n}}\\
\label{eq:probgammatoS}\widetilde{\Pr}(\Gamma: \Vec{\gamma}^{(0)}\rightarrow F^n) &= \frac{q^{-2n+2\abs{\Vec{\gamma}^{(0)}}}-q^{-2n}}{1-q^{-2n}}.
\end{align}

We can use the above results of the infinite-size circuit to derive the $\tr O_2\Phi^2_{Haar}(O_1)$ for the Haar-random ensemble. This can be used to check the soundness of the derivation.

From Eq.~\eqref{eq:secondtwirling}, we have that
\begin{equation}\label{eq:Haar}
\begin{split}
\tr O_2\Phi^2_{Haar}(O_1) &= \tr O_2 \left( \tr(O_1 \frac{\id^n-q^{-n}F^n}{q^{2n}-1}) \id^n + \tr(O_1 \frac{F^n-q^{-n}\id^n}{q^{2n}-1}) F^n\right)\\
&= \frac{1}{q^{2n}-1} (\tr O_1\id^n\tr O_2\id^n + \tr O_1F^n\tr O_2F^n - q^{-n}\tr O_1\id^n\tr O_2F^n - q^{-n}\tr O_1F^n\tr O_2\id^n),
\end{split}
\end{equation}
where $\id$ and $F$ are the identity and SWAP operators between two qudits, respectively.

Based on Eqs.~\eqref{eq:cgamma0gammas},~\eqref{eq:randomcircuittrace},~\eqref{eq:Zgamma},~\eqref{eq:probgammatoI}, and~\eqref{eq:probgammatoS}, we have that
\begin{equation}
\begin{split}
\tr O_2\Phi^2_{Haar} (O_1) =& \sum_{\Vec{\gamma}^{(0)}\Vec{\gamma}^{(s)}} c_{\Vec{\gamma}^{(0)}\Vec{\gamma}^{(s)}}(O_1,O_2) Z_{\Vec{\gamma}^{(0)}\Vec{\gamma}^{(s)}}\\
=& \sum_{\Vec{\gamma}^{(0)}\Vec{\gamma}^{(s)}} c_{\Vec{\gamma}^{(0)}\Vec{\gamma}^{(s)}}(O_1,O_2) q^{\abs{\Vec{\gamma}^{(s)}}-\abs{\Vec{\gamma}^{(0)}}}\widetilde{\Pr}(\Gamma: \Vec{\gamma}^{(0)}\rightarrow \Vec{\gamma}^{(s)})\\
=& \sum_{\Vec{\gamma}^{(0)}} c_{\Vec{\gamma}^{(0)}\id^n}(O_1,O_2) q^{-\abs{\Vec{\gamma}^{(0)}}}\widetilde{\Pr}(\Gamma: \Vec{\gamma}^{(0)}\rightarrow \id^n) + \sum_{\Vec{\gamma}^{(0)}} c_{\Vec{\gamma}^{(0)}F^n}(O_1,O_2) q^{n-\abs{\Vec{\gamma}^{(0)}}}\widetilde{\Pr}(\Gamma: \Vec{\gamma}^{(0)}\rightarrow F^n)\\
=& \sum_{\Vec{\gamma}^{(0)}} c_{\Vec{\gamma}^{(0)}\id^n}(O_1,O_2) q^{-\abs{\Vec{\gamma}^{(0)}}}\frac{1-q^{-2n+2\abs{\Vec{\gamma}^{(0)}}}}{1-q^{-2n}} + \sum_{\Vec{\gamma}^{(0)}} c_{\Vec{\gamma}^{(0)}F^n}(O_1,O_2) q^{n-\abs{\Vec{\gamma}^{(0)}}}\frac{q^{-2n+2\abs{\Vec{\gamma}^{(0)}}}-q^{-2n}}{1-q^{-2n}}\\
=& \sum_{\Vec{\gamma}^{(0)}} \frac{1}{(q^2-1)^n}  ( \tr(O_1 \bigotimes_{j=1}^n g_{\gamma^{(0)}_j}) \tr(O_2 \id^n)) q^{-\abs{\Vec{\gamma}^{(0)}}}\frac{1-q^{-2n+2\abs{\Vec{\gamma}^{(0)}}}}{1-q^{-2n}}\\
&+ \sum_{\Vec{\gamma}^{(0)}} \frac{1}{(q^2-1)^n}  ( \tr(O_1 \bigotimes_{j=1}^n g_{\gamma^{(0)}_j}) \tr(O_2 F^n)) q^{n-\abs{\Vec{\gamma}^{(0)}}}\frac{q^{-2n+2\abs{\Vec{\gamma}^{(0)}}}-q^{-2n}}{1-q^{-2n}}\\
=& \frac{1}{q^{2n}-1} \tr(\tr O_1\id^n\tr O_2\id^n + \tr O_1F^n\tr O_2F^n - q^{-n}\tr O_1\id^n\tr O_2F^n - q^{-n}\tr O_1F^n\tr O_2\id^n).
\end{split}
\end{equation}
In the final equation, we use the fact that
\begin{equation}\label{eq:summation_gamma_0}
\begin{split}
\sum_{\Vec{\gamma}^{(0)}} q^{-\abs{\Vec{\gamma}^{(0)}}} \bigotimes_{j=1}^n g_{\gamma^{(0)}_j} &= (1-q^{-2})^n \id^n,\\
\sum_{\Vec{\gamma}^{(0)}} q^{\abs{\Vec{\gamma}^{(0)}}} \bigotimes_{j=1}^n g_{\gamma^{(0)}_j} &= (q-q^{-1})^n F^n.
\end{split}
\end{equation}
Thus, our derivation from the infinite-size circuit is in accordance with that from the Haar-random ensemble, verifying the soundness of the derivation. In the next section, we will focus on the 1D log-depth brickwork circuits.

\subsection{Decoupling with 1D brickwork circuits and proof of Theorem~\ref{thm:1Dlocalrandomcircuit}}\label{appendssc:1D}
Below, we give the formal proof of Theorem~\ref{thm:1Dlocalrandomcircuit}.

As shown in Appendix~\ref{app:proof_nonsmooth}, since $\mathfrak{B}_n^{\varepsilon}$ is a unitary-$1$ design, we have that
\begin{equation}
\mathbb{E}_U\Vert \mathcal{T}_{S\rightarrow E}(U_S\rho_{SR}U_{S}^{\dagger}) - \tau_E\otimes \rho_R \Vert_1 \leq \sqrt{\mathbb{E}_U\tr[ \left( \widetilde{\mathcal{T}}_{S\rightarrow E}(U_S\widetilde{\rho}_{SR}U_{S}^{\dagger}) \right)^2] -\tr(\widetilde{\tau}^2_E)\tr(\widetilde{\rho}^2_R)},
\end{equation}
and
\begin{equation}
\begin{split}
\mathbb{E}_U\tr[ \left( \widetilde{\mathcal{T}}_{S\rightarrow E}(U_S\widetilde{\rho}_{SR}U_{S}^{\dagger}) \right)^2] = \tr_S( \Phi^2_{\mathfrak{B}_n^{\varepsilon}}( \tr_R(\widetilde{\rho}_{SR}^{\otimes 2} F_R) )  (\widetilde{\mathcal{T}}_{S\rightarrow E}^{\dagger})^{\otimes 2}(F_E)),
\end{split}
\end{equation}
where $\widetilde{\mathcal{T}}_{S\rightarrow E} = \sigma_E^{-1/4}\mathcal{T}_{S\rightarrow E} \sigma_E^{-1/4}$ with the corresponding Choi-Jomio{\l}kowski representation $\widetilde{\tau}_{SE} = \sigma_E^{-1/4}\tau_{SE} \sigma_E^{-1/4}$ and $\widetilde{\rho}_{SR} = \zeta_R^{-1/4}\rho_{SR}\zeta_R^{-1/4}$. In the following, we denote $O_{SR} = \tr_R(\widetilde{\rho}_{SR}^{\otimes 2} F_R)$ and $O_{SE} = (\widetilde{\mathcal{T}}_{S\rightarrow E}^{\dagger})^{\otimes 2}(F_E)$. The following analysis is to give the upper bound of $\tr O_{SE}\Phi^2_{\mathcal{L}_n^d} (O_{SR})$.

Note that for 1D brickwork circuits, the random unitary is denoted as
\begin{equation}
U = \prod_{l=1}^{d} U^{[l]}.
\end{equation}
Thus, instead of defining $\Vec{\gamma}^{(i)}$ by transforming $\Vec{\gamma}^{(i-1)}$ with a two-qudit gate $U^{(i)}$, we define $\Vec{\gamma}^{i}$ by transforming $\Vec{\gamma}^{i-1}$ with a layer of two-qudit gates $U^{i}$. Particularly, we set $\Vec{\gamma}^{0}$ equaling $\Vec{\gamma}^{(0)}$ and define $M^i$ such that
\begin{equation}
M^i(O) = \mathbb{E}_{U^i} U^{i\otimes 2} O U^{i\dagger\otimes 2}.
\end{equation}
Then,
\begin{equation}
\begin{split}
\Phi^2_{\mathcal{L}_n^d} (O_{SR}) &= M^{d}\cdots M^{2}M^{1}( \sum_{\Vec{\gamma}\in \{\id,F\}^n} c_{\Vec{\gamma}}(O_{SR})\Vec{\gamma} )\\
&= \sum_{\Vec{\gamma}\in \{\id,F\}^n} c_{\Vec{\gamma}}(O_{SR}) M^{d}\cdots M^{2}M^{1}(\Vec{\gamma})\\
&= \sum_{\Gamma=(\Vec{\gamma}^{0},\cdots,\Vec{\gamma}^{d})\in \{\id,F\}^{n\times (d+1)}} c_{\Vec{\gamma}^{0}}(O_{SR})\Vec{\gamma}^{d} \prod_{i=1}^{d} M^{i}_{\Vec{\gamma}^{i}\Vec{\gamma}^{i-1}}.
\end{split}
\end{equation}
Here, $\Gamma=(\Vec{\gamma}^{0},\cdots,\Vec{\gamma}^{d})$ is called a trajectory, which is a sequence of $d+1$ configurations. Note that we view $M^{i}$ as a transfer matrix as follows.
\begin{equation}
M^{i}(\Vec{\gamma}) = \sum_{\Vec{\nu}\in \{\id,F\}^n}   M^{i}_{\Vec{\nu}\Vec{\gamma}}\Vec{\nu}.
\end{equation}
Correspondingly,
\begin{equation}
\tr O_{SE}\Phi^2_{\mathcal{L}_n^d} (O_{SR}) = \sum_{\Gamma}\weight(\Gamma) c_{\Vec{\gamma}^{(0)}\Vec{\gamma}^{(s)}}(O_{SR},O_{SE}),
\end{equation}
where
\begin{equation}
\weight(\Gamma) = \prod_{i=1}^{d} M^{i}_{\Vec{\gamma}^{i}\Vec{\gamma}^{i-1}},
\end{equation}
and $c_{\Vec{\gamma}^{(0)}\Vec{\gamma}^{(s)}}(O_{SR},O_{SE})$ is given by Eq.~\eqref{eq:cgamma0gammas}.

Following the analysis in Appendix~\ref{sc:stochastic}, we have that
\begin{equation}
\tr O_{SE}\Phi^2_{\mathcal{L}_n^d} (O_{SR}) = \sum_{\Vec{\gamma}^{0}\Vec{\gamma}^{d}} c_{\Vec{\gamma}^{0}\Vec{\gamma}^{d}}(O_{SR},O_{SE}) Z_{\Vec{\gamma}^{0}\Vec{\gamma}^{d}},
\end{equation}
where
\begin{equation}
Z_{\Vec{\gamma}^{0}\Vec{\gamma}^{d}} = \sum_{\Gamma}(\frac{q}{q^2+1})^{\#F_{\Gamma}}\mathbf{1}_{\Gamma_{1}=\Vec{\gamma}^{(0)}}\mathbf{1}_{\Gamma_{d+1}=\Vec{\gamma}^{d}}.
\end{equation}

In the following, we will use the domain wall method to analyze the above quantity.

\subsubsection{Definition of domain wall}
Given a configuration $\vec{\gamma}\in \{\id, F\}^n$, we define the associated set of domain-wall locations as
\begin{equation}
DW(\vec{\gamma}):= \{e\in \{1, 2, \cdots, n\}: \gamma_e\neq \gamma_{e+1}\}.
\end{equation}
Define $g^i = DW(\vec{\gamma}^i)$. The number of domain walls within $g^i$ is denoted as $\abs{g^i}$. For a trajectory $\Gamma=(\Vec{\gamma}^{0},\cdots,\Vec{\gamma}^{d})$, we can define its domain-wall trajectory $G(\Gamma)=(g^{0},\cdots,g^{d})$. If $\weight(\Gamma)\neq 0$, then its associated domain-wall trajectory must satisfy the following rule:
If there is a domain wall at location $e$, and a gate acts on qudits $\{e, e+1\}$, then the domain wall must move from $e$ to $e-1$ or $e+1$. If there already exists a domain wall at these locations, the two domain walls will annihilate with each other. On the other hand, a pair of domain walls cannot be created from a vacuum. Meanwhile, due to the boundary condition, the number of domain walls must be even.

Based on Eq.~\eqref{eq:1DLQC}, the $l$-th layer quantum gate $U^{[l]}$ only contains two-qudit gates acting on $\{2i-\mod(l, 2), 2i+1-\mod(l, 2)\}$. Thus, the domain walls within $g^l$ must have opposite parity with respect to $l$. For a specific domain wall within $g^0$, it may move or stay still at the first step. Starting with the second step, the domain wall will keep moving and survive until the final step or annihilate with another domain wall during the movement. The number of domain walls cannot increase within a trajectory.

For each trajectory $\Gamma$, its weight, $\weight(\Gamma)=(\frac{q}{q^2+1})^{\#F_{\Gamma}}$, is uniquely determined by its associated domain-wall trajectory $G(\Gamma)$. Each movement of a domain wall within $G(\Gamma)$ corresponds to a flip. Hence, we denote $\weight(\Gamma) = \weight(G(\Gamma))$.

A domain-wall trajectory $G$ can always be decomposed into two disjoint parts. That is, we define $G_1 = (g^{0}_1,\cdots,g^{d}_1)$ and $G_2 = (g^{0}_2,\cdots,g^{d}_2)$ such that
\begin{equation}
G = G_1\sqcup G_2 \equiv (g^{0}_1\sqcup g^{0}_2,\cdots,g^{d}_1\sqcup g^{d}_2),
\end{equation}
where $\sqcup$ is the disjoint union. Thus, $g^{l}_1\cap g^{l}_2 = \emptyset$ for all $l$. This decomposition satisfies that $\weight(G) = \weight(G_1)\weight(G_2)$~\cite{Dalzell2022Anticoncentrate}. This relation is essential in evaluating the bound of $\weight(G)$.

For each trajectory $\Gamma$, we decompose its domain-wall trajectory $G(\Gamma)$ into two disjoint domain-wall trajectories, $G(\Gamma) = G_U \sqcup G_0$, as shown in Fig.~\ref{fig:domainwall}. Here, $G_U$ has a conserved number of domain walls, and the domain walls of $G_0$ will annihilate before the final step. That is, $g_U^d = g^d$, and $g_0^{d} = \emptyset$. This decomposition is unique~\cite{Dalzell2022Anticoncentrate}. In the following, we separately give the analysis of the upper bound and lower bound for $\weight(\Gamma)$ and $Z_{\Vec{\gamma}^{0}\Vec{\gamma}^{d}}$. For further elaboration, we denote $\mathcal{G}_U$ and $\mathcal{G}_0$ as the subsets of domain-wall trajectories without annihilation and without a surviving domain wall at the end, respectively. Also, we denote $\mathcal{G}_{U,k}\subset \mathcal{G}_U$ as the set of domain-wall trajectories always with $k$ domain walls. For a 1D periodic circuit structure, $k$ must be even to make $\mathcal{G}_{U,k}\neq \emptyset$.

\input{figure/DomainWall}

\subsubsection{Upper bound of $Z_{\Vec{\gamma}^{0}\Vec{\gamma}^{d}}$}
In this part, we derive the upper bound of $Z_{\Vec{\gamma}^{0}\Vec{\gamma}^{d}}$, presented in the following lemma.

\begin{lemma}\label{lemma:Zgammabound}
Given the set of trajectories $\Gamma$ with initial and final configurations as $\vec{\gamma}^0$ and $\vec{\gamma}^d$, we divide these trajectories into different sets with fixed $G_U$,
\begin{equation}
\mathcal{S}_U^G = \{\Gamma | G(\Gamma) = G_U\sqcup G_0\}.
\end{equation}
Given the domain-wall trajectory $G_U$ with initial configuration $g_U^0=\{e^0_1, e^0_2, \cdots, e^0_{2k_0-1}, e^0_{2k_0}\}$. Here, $2k_0 = \abs{DW(\vec{\gamma}^d)}$ is even. Denote $E^0_{odd} = \{1, \cdots, e^0_1, e^0_2+1, \cdots, e^0_3, \cdots, e^0_{2k_0}+1, \cdots, n\}$ and $E^0_{even} = [n]\backslash E^0_{odd} = \overline{E^0_{odd}}$. We also denote $\vec{\gamma}_{A}$ as the component of $\vec{\gamma}$ within support $A$, and denote $\abs{\vec{\gamma}_A = F}$ and $\abs{\vec{\gamma}_A = \id}$ as the number of SWAP and identity operators within $\vec{\gamma}_{A}$, respectively. Then,
\begin{equation}
Z_{\Vec{\gamma}^{0}\Vec{\gamma}^{d}} \leq \sum_{\mathcal{S}_U^G} \eta^{(d-1)\abs{DW(\vec{\gamma}^d)}} q^{-\abs{\vec{\gamma}^0_{E^0_{odd}}=\overline{\vec{\gamma}^d_{\{1\}}}}-\abs{\vec{\gamma}^0_{E^0_{even}}=\vec{\gamma}^d_{\{1\}}}}.
\end{equation}
Here, we denote $\overline{F} = \id$ and $\overline{\id} = F$. And $\eta = q/(q^2+1)$.
\end{lemma}

\begin{proof}
Recall that
\begin{equation}
Z_{\Vec{\gamma}^{0}\Vec{\gamma}^{d}} = \sum_{\Gamma}\weight(\Gamma)\mathbf{1}_{\Gamma_{1}=\Vec{\gamma}^{0}}\mathbf{1}_{\Gamma_{d+1}=\Vec{\gamma}^{d}}.
\end{equation}
The summation is over all trajectories from $\vec{\gamma}^0$ to $\vec{\gamma}^d$ with non-zero weight. Note that $\weight(\Gamma) = \weight(G(\Gamma))$. We can analyze the domain-wall trajectories of all $\Gamma$ for evaluating the weight. Note that $\Gamma$ has a fixed ending configuration $\vec{\gamma}^d$, implying that $G(\Gamma)$ also has a fixed ending point $DW(\vec{\gamma}^d)$. Following the analysis of the previous subsection, the domain-wall trajectory $G(\Gamma)$ can be uniquely decomposed into the disjoint union of two domain-wall trajectories $G_U$ and $G_0$ such that $g_U^d = DW(\vec{\gamma}^d)$ and $g_0^d = \emptyset$. After the decomposition, we determine and denote $g_U^0$ as the initial point of the domain-wall trajectory $G_U$. Note that $G_U\in \mathcal{G}_{U, \abs{DW(\vec{\gamma}^d)}}$ and $\abs{g_U^0}=\abs{g_U^d}=\abs{DW(\vec{\gamma}^d)}$.

We then divide all trajectories $\Gamma$, equivalently, all $G(\Gamma)$, according to $G_U$. Particularly, we define the set of trajectories with fixed $G_U$:
\begin{equation}
\mathcal{S}_U^G = \{\Gamma | G(\Gamma) = G_U\sqcup G_0\}.
\end{equation}
Given a fixed $\vec{\gamma}^d$, the number of distinct $G_U$ or different sets $\mathcal{S}_U^G$ is upper bounded by $2^{d \abs{DW(\vec{\gamma}^d)}}$. For each domain wall $e$ inside $DW(\vec{\gamma}^d)$, moving it forward only has two directions. Moving it back to the first layer only has $2^d$ choices. Since there is only $\abs{DW(\vec{\gamma}^d)}$ domain walls, moving all of them back to the first layer only has $2^{d \abs{DW(\vec{\gamma}^d)}}$ options, which gives the upper bound of the number of distinct $G_U$ and $\mathcal{S}_U^G$. Then, we have
\begin{equation}\label{eq:Zgammadecompose}
Z_{\Vec{\gamma}^{0}\Vec{\gamma}^{d}} = \sum_{\mathcal{S}_U^G}\sum_{\Gamma\in \mathcal{S}_U^G}\weight(\Gamma)\mathbf{1}_{\Gamma_{1}=\Vec{\gamma}^{(0)}}\mathbf{1}_{\Gamma_{d+1}=\Vec{\gamma}^{d}}.
\end{equation}
Below, we analyze the bound of $\sum_{\Gamma\in \mathcal{S}_U^G}\weight(\Gamma)$. For each trajectory $\Gamma$, we have $G(\Gamma) = G_U\sqcup G_0$. Since there is no annihilation of domain walls within $G_U$, the number of domain-wall movements is at least $(d-1)\abs{DW(\vec{\gamma}^d)}$, implying that
$\weight(G_U)\leq \eta^{(d-1)\abs{DW(\vec{\gamma}^d)}}$ with $\eta = q/(q^2+1)$. Thus,
\begin{equation}
\weight(\Gamma) = \weight(G_U)\weight(G_0)\leq \eta^{(d-1)\abs{DW(\vec{\gamma}^d)}}\weight(G_0),
\end{equation}
and
\begin{equation}\label{eq:weightgamma}
\sum_{\Gamma\in \mathcal{S}_U^G}\weight(\Gamma)\leq \eta^{(d-1)\abs{DW(\vec{\gamma}^d)}} \sum_{\Gamma\in \mathcal{S}_U^G} \weight(G_0).
\end{equation}

The following derivation is to further evaluate the upper bound of $\sum_{\Gamma\in \mathcal{S}_U^G} \weight(G_0)$. Given the domain-wall trajectory $G_U$ with initial configuration $g_U^0$ and ending configuration $g_U^d = DW(\vec{\gamma}^d)$, we write $g_U^i$ as $\{e^i_1, e^i_2, \cdots, e^i_{2k_0-1}, e^i_{2k_0}\}$. Here, $2k_0 = \abs{DW(\vec{\gamma}^d)}$ is even. We denote the support $E^i_{odd} = \{1, \cdots, e^i_1, e^i_2+1, \cdots, e^i_3, \cdots, e^i_{2k_0}+1, \cdots, n\}$ and $E^i_{even} = [n]\backslash E^i_{odd} = \overline{E^i_{odd}}$ as demonstrated in Fig.~\ref{fig:domainwallupper}. We also denote $\vec{\gamma}_{A}$ as the component of $\vec{\gamma}$ within support $A$, and denote $\abs{\vec{\gamma}_A = F}$ and $\abs{\vec{\gamma}_A = \id}$ as the number of SWAP and identity operators within $\vec{\gamma}_{A}$, respectively. We prove the following lemma.

\input{figure/DomainWallUpperBounds}

\begin{lemma}\label{lemma:1Dweightup}
If $\vec{\gamma}^d_{\{1\}} = \id$,
\begin{equation}
\sum_{\Gamma\in \mathcal{S}_U^G} \weight(G_0) \leq q^{-\abs{\vec{\gamma}^0_{E^0_{odd}}=F}-\abs{\vec{\gamma}^0_{E^0_{even}}=\id}};
\end{equation}
if $\vec{\gamma}^d_{\{1\}} = F$,
\begin{equation}
\sum_{\Gamma\in \mathcal{S}_U^G} \weight(G_0) \leq q^{-\abs{\vec{\gamma}^0_{E^0_{odd}}=\id}-\abs{\vec{\gamma}^0_{E^0_{even}}=F}}.
\end{equation}
\end{lemma}
\begin{proof}
For $\Gamma\in \mathcal{S}^G_U$, $G_U$ within the decomposition of $\Gamma$ is fixed. After fixing $G_U$, we can separate $G_0$ into $\abs{DW(\vec{\gamma}^d)}$ regions.

In step $i$, the $j$-th region is defined as $R_j^i = \{e^i_{j-1}+1,e^i_{j-1}+2\cdots, e^i_j\}$. And we divide $g^i_0\in G_0$ into the union of its parts in different regions, $g^i_0 = g^i_0|_{R_1^i} + g^i_0|_{R_2^i} + \cdots + g^i_0|_{R_{\abs{DW(\vec{\gamma}^d)}}^i}$. In the overall picture, the region $j$ is the union of all $R_j^i$ across different steps. Correspondingly, we can decompose $G_0$ into the disjoint union of trajectories in different regions: $G_0 = \bigsqcup_{j=1}^{\abs{DW(\vec{\gamma}^d)}} G_0^j$ where $G_0^j = (g^0_0|_{R_j^0}, g^1_0|_{R_j^1}, \cdots, g^d_0|_{R_j^d})$. An illustration of regions is demonstrated in Fig.~\ref{fig:domainwallregion}.

\input{figure/DomainWallRegion}

Since $G_0 = \bigsqcup_{j=1}^{\abs{DW(\vec{\gamma}^d)}} G_0^j$, we have that
\begin{equation}
\weight(G_0) = \prod_{j=1}^{\abs{DW(\vec{\gamma}^d)}} \weight(G_0^j),
\end{equation}
and
\begin{equation}\label{eq:weightG0product}
\sum_{\Gamma\in \mathcal{S}^G_U}\weight(G_0) = \prod_{j=1}^{\abs{DW(\vec{\gamma}^d)}} \sum_{\Gamma|_{\text{Region } j}}\weight(G_0^j).
\end{equation}
In the following, we will derive the upper bound of $\sum_{\Gamma|_{\text{Region } j}}\weight(G_0^j)$. Since the domain walls within $G_0$ annihilate in the end, $\vec{\gamma}^d$ must be the tensor product of all identity operations or SWAP operations within a region. Depending on which case holds, we call this region a type-$\id$ region or a type-$F$ region. If Region $j$ is a type-$\id$ region, we will show that
\begin{equation}\label{eq:typeI}
\sum_{\Gamma|_{\text{Region } j}}\weight(G_0^j)\leq q^{-\abs{\vec{\gamma}^0_{R_j^0}=F}}.
\end{equation}
If Region $j$ is a type-$F$ region, we will show that
\begin{equation}\label{eq:typeF}
\sum_{\Gamma|_{\text{Region } j}}\weight(G_0^j)\leq q^{-\abs{\vec{\gamma}^0_{R_j^0}=\id}}.
\end{equation}

We first focus on the type-$\id$ region and prove Eq.~\eqref{eq:typeI}. An example of a type-$\id$ region is shown in Fig.~\ref{fig:typeIregion}. It is Region 1 in Fig.~\ref{fig:domainwallregion}. Given this region, we first extend it to a wider region. Particularly, we define $e_{j-1}^{\min} = \min_{i\in \{0, 1, \cdots, d\}} e_j^i$ , $e_{j}^{\max} = \max_{i\in \{0, 1, \cdots, d\}} e_j^i$, and $R_j = \{e_{j-1}^{\min}+1, e_{j-1}^{\min}+2, \cdots, e_{j}^{\max}\}$. Then for each configuration $\vec{\gamma}^i|_{R_j^i}$, we extend it to $\vec{\gamma'}^i|_{R_j}$ by filling the extra components with identity operations. Region $j$ is correspondingly extended to Region $j'$ as shown in Fig.~\ref{fig:typeIregion}. By extending the region, we make sure that the configurations of different steps have equal length. Also, since we extend the region, the number of valid domain-wall configurations within the new region can only increase. Thus, we have

\input{figure/TypeIRegion}

\begin{equation}
\sum_{\Gamma|_{\text{Region } j}}\weight(G_0^j)\leq \sum_{\Gamma'|_{\text{Region } j}}\weight(G_0^{\prime j}).
\end{equation}
Also, the above trajectory is induced by 1D depth-$d$ brickwork circuits. If we consider all trajectories induced by infinite-size circuits, the number of valid domain-wall configurations also increases. Thus,
\begin{equation}
\sum_{\Gamma'_{1D}, \Gamma'|_{\text{Region } j}}\weight(G_0^{\prime j})\leq \sum_{\Gamma'_{inf}, \Gamma'|_{\text{Region } j}}\weight(G_0^{\prime j}).
\end{equation}

For the extended configurations, using Eqs.~\eqref{eq:Zgamma} and~\eqref{eq:probgammatoI}, we have that
\begin{equation}
\sum_{\Gamma'_{inf}|_{\text{Region } j}}\weight(G_0^{\prime j}) = q^{- \abs{\vec{\gamma'}^0|_{\text{Region } j} = F}}\frac{1-q^{-2n+2\abs{\vec{\gamma}^0}}}{1-q^{-2n}}\leq q^{- \abs{\vec{\gamma'}^0|_{\text{Region } j} = F}} = q^{- \abs{\vec{\gamma}^0|_{R_j^0} = F}}.
\end{equation}
Hence,
\begin{equation}
\sum_{\Gamma|_{\text{Region } j}}\weight(G_0^j)\leq q^{- \abs{\vec{\gamma}^0|_{R_j^0} = F}}.
\end{equation}

Similarly, an example of a type-$F$ region is shown in Fig.~\ref{fig:typeFregion}, which is Region 2 in Fig.~\ref{fig:domainwallregion}. Similar to previous arguments, for each configuration $\vec{\gamma}^i|_{R_j^i}$, we extend it to $\vec{\gamma'}^i|_{R_j}$ by filling the extra components with SWAP operations. Region $j$ is correspondingly extended to Region $j'$ as shown in Fig.~\ref{fig:typeIregion}. Then, we have

\input{figure/TypeFRegion}

\begin{equation}
\begin{split}
\sum_{\Gamma|_{\text{Region } j}}\weight(G_0^j)&\leq \sum_{\Gamma'|_{\text{Region } j}}\weight(G_0^{\prime j})\\
&\leq \sum_{\Gamma'_{inf}, \Gamma'|_{\text{Region } j}}\weight(G_0^{\prime j})\\
&= q^{- \abs{\vec{\gamma'}^0|_{\text{Region } j} = \id}} \frac{q^{ 2\abs{\vec{\gamma'}^0|_{\text{Region } j} = \id}}-q^{-2n}}{1-q^{-2n}}\\
&\leq q^{- \abs{\vec{\gamma'}^0|_{R_j^0} = \id}}\\
&= q^{- \abs{\vec{\gamma}^0|_{R_j^0} = \id}}.
\end{split}
\end{equation}

Hence, Eqs.~\eqref{eq:typeI} and~\eqref{eq:typeF} are proved. Combining these two equations with Eq.~\eqref{eq:weightG0product}, we obtain that
\begin{equation}\label{eq:weightG0upperbound}
\sum_{\Gamma\in \mathcal{S}^G_U}\weight(G_0) \leq \prod_{j: \text{ type-}\id}q^{- \abs{\vec{\gamma}^0|_{R_j^0} = F}} \prod_{j: \text{ type-}F}q^{- \abs{\vec{\gamma}^0|_{R_j^0} = \id}}.
\end{equation}

When $\vec{\gamma}^d_{\{1\}} = \id$, Eq.~\eqref{eq:weightG0upperbound} reduces to
\begin{equation}
\sum_{\Gamma\in \mathcal{S}_U^G} \weight(G_0) \leq q^{-\abs{\vec{\gamma}^0_{E^0_{odd}}=F}-\abs{\vec{\gamma}^0_{E^0_{even}}=\id}}.
\end{equation}
When $\vec{\gamma}^d_{\{1\}} = F$, Eq.~\eqref{eq:weightG0upperbound} reduces to
\begin{equation}
\sum_{\Gamma\in \mathcal{S}_U^G} \weight(G_0) \leq q^{-\abs{\vec{\gamma}^0_{E^0_{odd}}=\id}-\abs{\vec{\gamma}^0_{E^0_{even}}=F}}.
\end{equation}
Proof is done.
\end{proof}

Applying Lemma~\ref{lemma:1Dweightup} and Eq.~\eqref{eq:weightgamma} to Eq.~\eqref{eq:Zgammadecompose}, we get
\begin{equation}
Z_{\Vec{\gamma}^{0}\Vec{\gamma}^{d}} \leq \sum_{\mathcal{S}_U^G} \eta^{(d-1)\abs{DW(\vec{\gamma}^d)}} q^{-\abs{\vec{\gamma}^0_{E^0_{odd}}=\overline{\vec{\gamma}^d_{\{1\}}}}-\abs{\vec{\gamma}^0_{E^0_{even}}=\vec{\gamma}^d_{\{1\}}}}.
\end{equation}
Note that we denote $\overline{F} = \id$ and $\overline{\id} = F$.
\end{proof}

\subsubsection{Proof of Theorem~\ref{thm:1Dlocalrandomcircuit}}
Applying Lemma~\ref{lemma:Zgammabound}, we get the upper bound of $\tr O_{SE}\Phi^2_{\mathcal{L}_n^d} (O_{SR})$. Recall that
\begin{equation}
\tr O_{SE}\Phi^2_{\mathcal{L}_n^d} (O_{SR}) = \sum_{\Vec{\gamma}^{0}\Vec{\gamma}^{d}} c_{\Vec{\gamma}^{0}\Vec{\gamma}^{d}}(O_{SR},O_{SE}) Z_{\Vec{\gamma}^{0}\Vec{\gamma}^{d}},
\end{equation}
where
\begin{equation}
c_{\Vec{\gamma}^{0}\Vec{\gamma}^{d}}(O_{SR},O_{SE}) = \frac{1}{(q^2-1)^n}  \tr(O_{SR} \bigotimes_{j=1}^n g_{\gamma^{(0)}_j}) \tr(O_{SE} \bigotimes_{j=1}^n \gamma^{(s)}_j).
\end{equation}
Note that $O_{SR} = \tr_R(\widetilde{\rho}_{SR}^{\otimes 2} F_R)$ and $O_{SE} = (\widetilde{\mathcal{T}}_{S\rightarrow E}^{\dagger})^{\otimes 2}(F_E)$. Below, we utilize the tensor product structure of $\rho_{SR}$, $\rho_{SR} = \bigotimes_{i=1}^n \rho_{S_iR_i}$. Then, $\widetilde{\rho}_{SR} = \bigotimes_{i=1}^n \widetilde{\rho}_{S_iR_i}$ and $O_{SR} = \bigotimes_{i=1}^n O_{S_iR_i}$. Since $\tr(O_{S_iR_i} g_{\id}) = \tr(\widetilde{\rho}_{R_i}^2)-q^{-1}\tr(\widetilde{\rho}_{S_iR_i}^2)\geq 0$ and $\tr(O_{S_iR_i} g_F) = \tr(\widetilde{\rho}_{S_iR_i}^2)-q^{-1}\tr(\widetilde{\rho}_{R_i}^2)\geq 0$, we have $\tr(O_{SR} \bigotimes_{j=1}^n g_{\gamma^{(0)}_j})\geq 0$. Meanwhile, $\tr(O_{SE}F_A\otimes \id_{\Bar{A}}) = q^{2n}\tr_{AE} (\tr_{\Bar{A}}\widetilde{\tau}_{SE})^2\geq 0$. Thus, we will always have $c_{\Vec{\gamma}^0\Vec{\gamma}^d}(O_{SR},O_{SE})\geq 0$. Then,
\begin{equation}
\begin{split}
\tr O_{SE}\Phi^2_{\mathcal{L}_n^d} (O_{SR}) =& \sum_{\Vec{\gamma}^{0}\Vec{\gamma}^{d}} c_{\Vec{\gamma}^{0}\Vec{\gamma}^{d}}(O_{SR},O_{SE}) Z_{\Vec{\gamma}^{0}\Vec{\gamma}^{d}}\\
\leq& \sum_{\Vec{\gamma}^{0}\Vec{\gamma}^{d}} c_{\Vec{\gamma}^{0}\Vec{\gamma}^{d}}(O_{SR},O_{SE})\sum_{\mathcal{S}_U^G} \eta^{(d-1)\abs{DW(\vec{\gamma}^d)}} q^{-\abs{\vec{\gamma}^0_{E^0_{odd}}=\overline{\vec{\gamma}^d_{\{1\}}}}-\abs{\vec{\gamma}^0_{E^0_{even}}=\vec{\gamma}^d_{\{1\}}}}\\
=&
\sum_{\Vec{\gamma}^{0}\Vec{\gamma}^{d}} \frac{1}{(q^2-1)^n} \tr(O_{SR} \bigotimes_{j=1}^n g_{\gamma^0_j}) \tr(O_{SE} \bigotimes_{j=1}^n \gamma^d_j)\sum_{\mathcal{S}_U^G} \eta^{(d-1)\abs{DW(\vec{\gamma}^d)}} q^{-\abs{\vec{\gamma}^0_{E^0_{odd}}=\overline{\vec{\gamma}^d_{\{1\}}}}-\abs{\vec{\gamma}^0_{E^0_{even}}=\vec{\gamma}^d_{\{1\}}}}\\
=& \sum_{\Vec{\gamma}^{d}, \mathcal{S}_U^G} \frac{\eta^{(d-1)\abs{DW(\vec{\gamma}^d)}}}{(q^2-1)^n}\sum_{\vec{\gamma}^0} \tr(O_{SR} \bigotimes_{j=1}^n g_{\gamma^0_j}) \tr(O_{SE} \bigotimes_{j=1}^n \gamma^d_j)q^{-\abs{\vec{\gamma}^0_{E^0_{odd}}=\overline{\vec{\gamma}^d_{\{1\}}}}-\abs{\vec{\gamma}^0_{E^0_{even}}=\vec{\gamma}^d_{\{1\}}}}\\
=& \sum_{DW(\Vec{\gamma}^{d}), \mathcal{S}_U^G} \frac{\eta^{(d-1)\abs{DW(\vec{\gamma}^d)}}}{(q^2-1)^n}\sum_{\vec{\gamma}^0} \tr(O_{SR} \bigotimes_{j=1}^n g_{\gamma^0_j}) \tr(O_{SE} \bigotimes_{j=1}^n \gamma^d_j)(q^{-\abs{\vec{\gamma}^0_{E^0_{odd}}=\id}-\abs{\vec{\gamma}^0_{E^0_{even}}=F}}+q^{-\abs{\vec{\gamma}^0_{E^0_{odd}}=F}-\abs{\vec{\gamma}^0_{E^0_{even}}=\id}})\\
=& \sum_{DW(\Vec{\gamma}^{d}), \mathcal{S}_U^G}\eta^{(d-1)\abs{DW(\vec{\gamma}^d)}} q^{-2n}(\tr O_{SR} \id_{E^0_{odd}}F_{E^0_{even}} \tr O_{SE} \id_{E^d_{odd}}F_{E^d_{even}} + \tr O_{SR} \id_{E^0_{even}}F_{E^0_{odd}} \tr O_{SE} \id_{E^d_{even}}F_{E^d_{odd}}) \\
=& q^{-2n}(\tr O_{SR}\id^n \tr O_{SE} \id^n + \tr O_{SR}F^n \tr O_{SE} F^n)+q^{-2n} \sum_{\abs{DW(\Vec{\gamma}^{d})}\in \{2, 4, \cdots, 2\lfloor n/4 \rfloor \}}\eta^{(d-1)\abs{DW(\vec{\gamma}^d)}}\times\\
&\sum_{\mathcal{S}_U^G} (\tr O_{SR} \id_{E^0_{odd}}F_{E^0_{even}} \tr O_{SE} \id_{E^d_{odd}}F_{E^d_{even}} + \tr O_{SR} \id_{E^0_{even}}F_{E^0_{odd}} \tr O_{SE} \id_{E^d_{even}}F_{E^d_{odd}}).
\end{split}
\end{equation}
Here, $\id_A$ and $F_A$ means identity and SWAP operations on the support $A$ The last inequality comes from Eq.~\eqref{eq:summation_gamma_0}. Note that
\begin{align}
q^{-2n}\tr O_{SR}\id^n \tr O_{SE} \id^n &= \tr(\widetilde{\tau}^2_{E})\tr(\widetilde{\rho}^2_{R}),\\
q^{-2n}\tr O_{SR}F^n \tr O_{SE} F^n &= \tr(\widetilde{\tau}^2_{SE})\tr(\widetilde{\rho}^2_{SR}).
\end{align}
Also,
\begin{equation}
\begin{split}
q^{-2n}\tr O_{SR} \id_{E^0_{odd}}F_{E^0_{even}} \tr O_{SE} \id_{E^d_{odd}}F_{E^d_{even}}
&\leq \max_{E^0}(\frac{\tr O_{SR} \id_{E^0_{odd}}F_{E^0_{even}}}{\tr O_{SR} \id_{E^d_{odd}}F_{E^d_{even}}})\max_{E^d}(q^{-2n}\tr O_{SR} \id_{E^d_{odd}}F_{E^d_{even}} \tr O_{SE} \id_{E^d_{odd}}F_{E^d_{even}}) \\
&\leq \max_{i\in [n]}(\frac{\tr\widetilde{\rho}_{S_iR_i}^2}{\tr\widetilde{\rho}_{R_i}^2}, \frac{\tr\widetilde{\rho}_{R_i}^2}{\tr\widetilde{\rho}_{S_iR_i}^2})^{\abs{E^d_{even}-E^0_{even}}}\max_{A\subseteq [n]} ( \tr_{AR}(\tr_{\Bar{A}} \widetilde{\rho}_{SR})^2 \tr_{AE}(\tr_{\Bar{A}} \widetilde{\tau}_{SE})^2)\\
&\leq \max_{i\in [n]}(\frac{\tr\widetilde{\rho}_{S_iR_i}^2}{\tr\widetilde{\rho}_{R_i}^2}, \frac{\tr\widetilde{\rho}_{R_i}^2}{\tr\widetilde{\rho}_{S_iR_i}^2})^{d\abs{DW(\vec{\gamma}^d)}}\max_{A\subseteq [n]} ( \tr_{AR}(\tr_{\Bar{A}} \widetilde{\rho}_{SR})^2 \tr_{AE}(\tr_{\Bar{A}} \widetilde{\tau}_{SE})^2).\\
\end{split}
\end{equation}
Similarly,
\begin{equation}
\begin{split}
q^{-2n}\tr O_{SR} \id_{E^0_{even}}F_{E^0_{odd}} \tr O_{SE} \id_{E^d_{even}}F_{E^d_{odd}} \leq \max_{i\in [n]}(\frac{\tr\widetilde{\rho}_{S_iR_i}^2}{\tr\widetilde{\rho}_{R_i}^2}, \frac{\tr\widetilde{\rho}_{R_i}^2}{\tr\widetilde{\rho}_{S_iR_i}^2})^{d\abs{DW(\vec{\gamma}^d)}}\max_{A\subseteq [n]} ( \tr_{AR}(\tr_{\Bar{A}} \widetilde{\rho}_{SR})^2 \tr_{AE}(\tr_{\Bar{A}} \widetilde{\tau}_{SE})^2).
\end{split}
\end{equation}

Note that $\abs{\mathcal{S}_U^G}\leq 2^{d\abs{DW(\vec{\gamma}^d)}}$. Then,
\begin{equation}
\begin{split}
\tr O_{SE}\Phi^2_{\mathcal{L}_n^d} (O_{SR}) \leq& \tr(\widetilde{\tau}^2_E)\tr(\widetilde{\rho}^2_R)+\tr(\widetilde{\tau}^2_{SE})\tr(\widetilde{\rho}^2_{SR})+\sum_{\abs{DW(\Vec{\gamma}^{d})}\in \{2, 4, \cdots, 2\lfloor n/4 \rfloor \}}\eta^{(d-1)\abs{DW(\vec{\gamma}^d)}}\times\\
&(2\rho_m)^{d\abs{DW(\vec{\gamma}^d)}}\max_{A\subseteq [n]} ( \tr_{AR}(\tr_{\Bar{A}} \widetilde{\rho}_{SR})^2 \tr_{AE}(\tr_{\Bar{A}} \widetilde{\tau}_{SE})^2)\\
\leq& \tr(\widetilde{\tau}^2_E)\tr(\widetilde{\rho}^2_R)+\tr(\widetilde{\tau}^2_{SE})\tr(\widetilde{\rho}^2_{SR})+\\
&\max_{A\subseteq [n]} ( \tr_{AR}(\tr_{\Bar{A}} \widetilde{\rho}_{SR})^2 \tr_{AE}(\tr_{\Bar{A}} \widetilde{\tau}_{SE})^2\sum_{k_0=1}^{\lfloor n/4 \rfloor}\binom{n}{2k_0}((2\rho_m)^d\eta^{(d-1)})^{2k_0}),
\end{split}
\end{equation}
where we denote $\rho_m = \max_{i\in [n]}(\frac{\tr\widetilde{\rho}_{S_iR_i}^2}{\tr\widetilde{\rho}_{R_i}^2}, \frac{\tr\widetilde{\rho}_{R_i}^2}{\tr\widetilde{\rho}_{S_iR_i}^2})$ for simplicity. In the following, we analyze the coefficient
\begin{equation}
c = \sum_{k_0=1}^{\lfloor n/4 \rfloor}\binom{n}{2k_0}((2\rho_m)^d\eta^{(d-1)})^{2k_0}.
\end{equation}
From Ref.~\cite{Dalzell2022Anticoncentrate}, coefficient $c$ has an upper bound:
\begin{equation}
c = \sum_{k_0=1}^{\lfloor n/4 \rfloor}\binom{n}{2k_0}((2\rho_m)^d\eta^{(d-1)})^{2k_0}\leq (\eta (2\rho_m)^{\frac{d}{d-1}})^{d-d^*},
\end{equation}
where
\begin{equation}
d^* = \frac{\log n}{\log \frac{1}{2\eta}} + \frac{\log(e-1)}{\log \frac{1}{2\eta}}+1.
\end{equation}
Recall that we choose depth $d$ such that $d-d^* = \log \frac{n}{\varepsilon} = O(\log n)$. In the large $n$ limit, $\eta (2\rho_m)^{\frac{d}{d-1}} \leq 2\eta\rho_m+\delta$ with $\delta$ an arbitrarily small number. Then,
\begin{equation}
c\leq (2\eta\rho_m+\delta)^{\log \frac{n}{\varepsilon}} = (\frac{\varepsilon}{n})^{\log\frac{1}{2\eta\rho_m+\delta}}.
\end{equation}
As a result,
\begin{equation}
\tr O_{SE}\Phi^2_{\mathcal{L}_n^d} (O_{SR}) \leq \tr(\widetilde{\tau}^2_E)\tr(\widetilde{\rho}^2_R)+\tr(\widetilde{\tau}^2_{SE})\tr(\widetilde{\rho}^2_{SR})+(\frac{\varepsilon}{n})^{\log\frac{1}{2\eta\rho_m+\delta}}\max_{A\subseteq [n]} ( \tr_{AR}(\tr_{\Bar{A}} \widetilde{\rho}_{SR})^2 \tr_{AE}(\tr_{\Bar{A}} \widetilde{\tau}_{SE})^2).
\end{equation}
\begin{equation}
\mathbb{E}_U\Vert \mathcal{T}_{S\rightarrow E}(U_S\rho_{SR}U_{S}^{\dagger}) - \tau_E\otimes \rho_R \Vert_1 \leq \sqrt{\tr(\widetilde{\tau}^2_{SE})\tr(\widetilde{\rho}^2_{SR})+(\frac{\varepsilon}{n})^{\log\frac{1}{2\eta\rho_m+\delta}}\max_{A\subseteq [n]} ( \tr_{AR}(\tr_{\Bar{A}} \widetilde{\rho}_{SR})^2 \tr_{AE}(\tr_{\Bar{A}} \widetilde{\tau}_{SE})^2}).
\end{equation}
Note that $\widetilde{\rho}_{SR} = \bigotimes_{i=1}^{2N}\widetilde{\rho}_{S_iR_i}$ and $\widetilde{\rho}_{S_iR_i} = \zeta_{R_i}^{-1/4}\rho_{S_iR_i}\zeta_{R_i}^{-1/4}$. We can choose $\zeta_{R_i} = \rho_{R_i}$ or $\zeta_{R} = \rho_{R}$ to let
\begin{equation}
\tr(\widetilde{\rho}^2_{R_i}) = 1, \tr(\widetilde{\rho}^2_{S_iR_i}) = 2^{-H_2(S_i|R_i)_{\rho_{S_iR_i}}}.
\end{equation}
Thus,
\begin{equation}
\tr_{AR}(\tr_{\Bar{A}} \widetilde{\rho}_{SR})^2 = \tr_{AR}( \widetilde{\rho}_{AR}^2 ) = 2^{-H_2(A|R)_{\rho_{AR}}},
\end{equation}
and
\begin{equation}
\rho_m = \max_{i\in [n]}(2^{H_2(S_i|R_i)_{\rho_{S_iR_i}}}, 2^{-H_2(S_i|R_i)_{\rho_{S_iR_i}}}).
\end{equation}

Also, we can choose $\sigma_E = \tau_E$ such that
\begin{equation}
\tr(\widetilde{\tau}^2_{E}) = 1, \tr(\widetilde{\tau}^2_{SE}) = 2^{-H_2(S|E)_{\tau_{SE}}}.
\end{equation}
In this case,
\begin{equation}
\tr_{AE}(\tr_{\Bar{A}} \widetilde{\tau}_{SE})^2 = \tr_{AE}( \widetilde{\tau}_{AE}^2 ) = 2^{-H_2(A|E)_{\tau_{AE}}}.
\end{equation}

Then, for any $\delta > 0$, we have that
\begin{equation}
\mathbb{E}_U\Vert \mathcal{T}_{S\rightarrow E}(U_S\rho_{SR}U_{S}^{\dagger}) - \tau_E\otimes \rho_R \Vert_1 \leq \sqrt{2^{-H_2(S|E)_{\tau_{SE}}-H_2(S|R)_{\rho_{SR}}} + (\frac{\varepsilon}{n})^{\log\frac{1}{2\eta\rho_m+\delta}} \max_{A\subseteq [n]} 2^{-H_2(A|R)_{\rho_{AR}}-H_2(A|E)_{\tau_{AE}}}}.
\end{equation}
In the description of Theorem~\ref{thm:1Dlocalrandomcircuit}, we omit $\delta$ for simplicity. Proof is done.

\subsection{Application of Theorem~\ref{thm:1Dlocalrandomcircuit} in quantum error correction}\label{appendssc:1DLRCQEC}
Here, we present the results of applying Theorem~\ref{thm:1Dlocalrandomcircuit} in quantum error correction and show that the encoding scheme induced by $\mathfrak{B}$ can have a constant threshold against local noise.

For the quantum error correction task, there are $nR$ physical qudits in total. We first fix the encoding rate $R = k/n$, where $R$ is a constant. Then, we set the local dimension $q$ as $2^{R^{-1}}$ composed of $R^{-1}$ qubits. The total qudit number is correspondingly changed from $n$ to $nR$. Since $R$ is a constant, the local dimension is still a constant. Then, each $S_i$ is a qudit with local dimension $2^{R^{-1}}$, and $R_i$ is a qubit. Without loss of generality, the state $\rho_{S_iR_i}$ is set as $\ketbra{\hat{\phi}}\otimes \ketbra{0}^{R^{-1}-1}$ where $\ketbra{\hat{\phi}}$ is an EPR pair. Then, $H_2(S_i|R_i) = -1$, $\eta = \frac{2^{R^{-1}}}{2^{2R^{-1}}+1}$, $\rho_m = 2$, and
\begin{equation}
\frac{1}{2\eta \rho_m} = \frac{2^{2R^{-1}}+1}{2^{R^{-1}+2}} = \frac{2^{2n/k}+1}{2^{n/k+2}} > 2^{n/k-2}.
\end{equation}

Meanwhile, we consider local noise. The channel $\mathcal{T}_{S\rightarrow E}$ will be the complementary channel of local noise, $\mathcal{T}_{S\rightarrow E} = \hat{\mathcal{N}}^{\otimes n}$. In this case, the decoupling theorem can be simplified to
\begin{equation}
\begin{split}
&\mathbb{E}_{U_S\sim \mathfrak{B}_n^{\varepsilon}}\Vert \mathcal{T}_{S\rightarrow E}(U_S\rho_{SR}U_{S}^{\dagger}) - \tau_E\otimes \rho_R \Vert_1\\
\leq& \sqrt{2^{-\sum_{i=1}^{nR}H_2(S_i|E_i)_{\tau_{S_iE_i}}-\sum_{i=1}^{nR}H_2(S_i|R_i)_{\rho_{S_iR_i}}} + (\frac{\varepsilon}{nR})^{\log\frac{1}{2\eta\rho_m}} \prod_{i=1}^{nR} \max(1, 2^{-H_2(S_i|E_i)_{\tau_{S_iE_i}}-H_2(S_i|R_i)_{\rho_{S_iR_i}}})},
\end{split}
\end{equation}
where $H_2(S_i|E_i)_{\tau_{S_iE_i}} = R^{-1}(1-f(\vec{p}))$ with $f(\vec{p})=2\log (\sqrt{p_I} + \sqrt{p_X} + \sqrt{p_Y} + \sqrt{p_Z})$ for strength-$\vec{p}$ local Pauli noise, $H_2(S_i|E_i)_{\tau_{S_iE_i}} = R^{-1}(1-\log(1+3p))$ for strength-$p$ local erasure error, and $H_2(S_i|E_i)_{\tau_{S_iE_i}} = -R^{-1}\log( \frac{1}{2-p} + \sqrt{\frac{p}{2-p}} )$ for strength-$p$ local amplitude damping noise.

For local Pauli noise,
\begin{equation}
H_2(S_i|E_i)_{\tau_{S_iE_i}}+H_2(S_i|R_i)_{\rho_{S_iR_i}} = R^{-1}(1-f(\vec{p}))-1 = R^{-1}(1-f(\vec{p})-\frac{k}{n}).
\end{equation}
When $1-f(\vec{p})-\frac{k}{n} > 0$, the term $\max(1, 2^{-H_2(S_i|E_i)_{\tau_{S_iE_i}}-H_2(S_i|R_i)_{\rho_{S_iR_i}}})\leq 1$. It is also true for local erasure error when $1-\log(1+3p)-\frac{k}{n} > 0$ and for local amplitude damping noise when $-\log( \frac{1}{2-p} + \sqrt{\frac{p}{2-p}} )-\frac{k}{n} > 0$.

Note that the expected Choi error satisfies
\begin{equation}
\mathbb{E}_{U_S\sim \mathfrak{B}_n^{\varepsilon}} \epsilon_{\mathrm{Choi}} \leq \sqrt{\mathbb{E}_{U_S\sim \mathfrak{B}_n^{\varepsilon}}\Vert \mathcal{T}_{S\rightarrow E}(U_S\rho_{SR}U_{S}^{\dagger}) - \tau_E\otimes \rho_R \Vert_1}.
\end{equation}
We obtain that the Choi error is upper-bounded by
\begin{equation}
\mathbb{E}_{U_S\sim \mathfrak{B}_n^{\varepsilon}} \epsilon_{\mathrm{Choi}}\leq \left(2^{-n(1-f(\Vec{p})-\frac{k}{n})} + (\frac{\varepsilon}{k})^{\frac{n}{k}-2}\right)^{\frac{1}{4}}
\end{equation}
for strength-$\vec{p}$ local Pauli noise, upper-bounded by
\begin{equation}
\mathbb{E}_{U_S\sim \mathfrak{B}_n^{\varepsilon}} \epsilon_{\mathrm{Choi}}\leq \left(2^{-n(1-\log(1+3p)-\frac{k}{n})} + (\frac{\varepsilon}{k})^{\frac{n}{k}-2}\right)^{\frac{1}{4}}
\end{equation}
for strength-$p$ local erasure error, and upper-bounded by
\begin{equation}
\mathbb{E}_{U_S\sim \mathfrak{B}_n^{\varepsilon}} \epsilon_{\mathrm{Choi}}\leq \left(2^{-n(-\log( \frac{1}{2-p} + \sqrt{\frac{p}{2-p}} )-\frac{k}{n})} + (\frac{\varepsilon}{k})^{\frac{n}{k}-2}\right)^{\frac{1}{4}}
\end{equation}
for strength-$p$ local amplitude damping noise.

To let the Choi error vanish, we need to ensure $\frac{n}{k}-2>0$. Thus, we show that there exists a threshold $\min(\frac{1}{2}, 1-f(\vec{p}))$ for local Pauli noise, a threshold $\min(\frac{1}{2}, 1-\log(1+3p)))$ for local erasure error, and a threshold $\min(\frac{1}{2}, -\log( \frac{1}{2-p} + \sqrt{\frac{p}{2-p}} ))$ for local amplitude damping noise. Moreover, by choosing appropriate parameters, $(\frac{\varepsilon}{k})^{\frac{n}{k}-2}$ will polynomially decay with $n$.

For strength-$p$ nearest neighbor $ZZ$-coupling noise, we apply the original Theorem~\ref{thm:1Dlocalrandomcircuit}. Similar to the analysis in Appendix~\ref{app:proof_corr}, we obtain that
\begin{align}
H_2(A|R)_{\rho_{AR}} &= -\abs{A}\\
H_2(A|E)_{\tau_{AE}} &\geq R^{-1}\abs{A}[1-2(1+R)\log(\sqrt{1-p}+\sqrt{p})].
\end{align}
Thus, when $H_2(A|R)+H_2(A|E)\geq 0$, or equivalently, $1-R-2(1+R)\log(\sqrt{1-p}+\sqrt{p})\geq 0$, we have that
\begin{equation}
\mathbb{E}_{U_S\sim \mathfrak{B}_n^{\varepsilon}} \epsilon_{\mathrm{Choi}}\leq \left(2^{-n(1-\frac{k}{n}-2(1+\frac{k}{n})\log(\sqrt{1-p}+\sqrt{p}))} + (\frac{\varepsilon}{k})^{\frac{n}{k}-2}\right)^{\frac{1}{4}}.
\end{equation}
Hence, the achievable encoding rate for $\mathfrak{B}_n^{\varepsilon}$ is $1-2(1+\frac{k}{n})\log(\sqrt{1-p}+\sqrt{p})$.

Note that the above analysis also applies to the case of local dimension $2$. In this case, the total number of qubits is $n$, and each state $\rho_{S_iR_i}$ is defined as $\frac{k}{n}\ketbra{\hat{\phi}} + (1-\frac{k}{n})\ketbra{0}$. We then have $H_2(S_i|R_i) = -\frac{k}{n}$, while $H_2(S_i|E_i)$ depends on the noise parameter, typically taking the form $1 - f(\vec{p})$ for local Pauli noise. With $\eta = \frac{2}{5}$ and $\rho_m = 2^{\frac{k}{n}}$, substituting these values into the decoupling theorem yields the same bound, except that the factor $(\frac{\varepsilon}{k})^{\frac{n}{k}-2}$ is replaced by $(\frac{\varepsilon}{n})^{\log\frac{5}{4}-\frac{k}{n}}$.

\section{Lower bound for circuit depth in AQEC}\label{app:lower_bound}
Here, we analyze the lower bound on the circuit depth required for a unitary encoding of
$k$ logical qubits that achieves a Choi error $\epsilon_{\mathrm{Choi}}$. For the $2k$-qubit maximally entangled state $\hat{\phi}_{LR}$ shared between the reference system $R$ and the logical system $L$, where both $L$ and $R$ consist of $k$ qubits each, we denote the $i$-th qubit in $L$ by $L_{i,0}$ and the qubit in $R$ that is maximally entangled with $L_{i,0}$ by $R_{i,0}$. We define the physical system $S$ as the union of $L$ and the ancillary system initialized in the state $\ket{0^{n-k}}$.

Consider a unitary circuit $U$ of depth $d$ acting on $S$, decomposed as
\begin{equation}
U = \prod_{i=1}^d \left(\bigotimes_j V_{i,j}\right),
\end{equation}
where each $\{V_{i,j}\}_j$ represents the set of two-qubit gates in the $i$-th layer of $U$ with disjoint supports. For each layer $j$, where $1 \le j \le d$, let $L_{i,j}$ denote the union of the supports of all two-qubit gates that intersect with $L_{i,j-1}$:
\begin{equation}
L_{i,j} = \bigcup \left\{\mathrm{supp}(V_{i,j}) \;\big|\; \mathrm{supp}(V_{i,j}) \cap L_{i,j-1} \neq \emptyset\right\}.
\end{equation}
In other words, $L_{i,j}$ represents the light cone of the $i$-th qubit in the first $j$ layers of $U$.

Similarly, we define the backward light cone of the $i$-th physical qubit, $i \in S$, recursively as
\begin{equation}
\begin{split}
B_{i,d} &= \{i\}, \\
B_{i,j} &= \bigcup \left\{\mathrm{supp}(V_{i,j+1}) \;\big|\; \mathrm{supp}(V_{i,j+1}) \cap B_{i,j+1} \neq \emptyset\right\}, \quad 0 \le j < d.
\end{split}
\end{equation}
By the definitions of $L_{i,j}$ and $B_{i,j}$, for any logical qubit $i \in L$ and physical qubit $k \in S$, the qubit $k$ is included in $L_{i,d}$ if and only if the logical qubit $i$ is included in $B_{k,0}$.

Finally, we define $M$ as an upper bound on the size of all light cones:
\begin{equation}
M = \max\left(\max_{i \in L} |L_{i,d}|, \max_{i \in S} |B_{i,0}| \right).
\end{equation}
With this definition, we can establish a lower bound on the number of physical qubits in $L$ that have disjoint light cones, as stated in the following lemma.

\begin{lemma}[Disjoint light cones]\label{lem:disjoint}
For any unitary circuit \( U = \prod_{i=1}^d \left(\bigotimes_j V_{i,j}\right) \), there exists a subset \( J \subseteq L \) such that:
\begin{enumerate}
\item For any distinct qubits \( i, j \in J\), their light cones are disjoint: \( L_{i,d} \cap L_{j,d} = \emptyset \).
\item The size of \( J \) satisfies \( |J| \ge \left\lceil \frac{k}{M^2} \right\rceil \).
\end{enumerate}
\end{lemma}

\begin{proof}
We construct the set \( J \) iteratively. Start with an empty set \( J_0 = \emptyset \). At each step \( a \), assume we have a set \( J_a \) with disjoint light cones. Define \( L_{J_a,d} = \bigcup_{i \in J_a} L_{i,d} \), the union of the light cones of qubits in \( J_a \).

For a qubit \( j \in L \), the condition that \( L_{j,d} \) does not intersect with any \( L_{i,d} \) for \( i \in J_a \),
\begin{equation}
L_{j,d} \cap L_{J_a,d} = \emptyset,
\end{equation}
is equivalent to
\begin{equation}
j \notin \bigcup_{m \in L_{J_a,d}} B_{m,0}.
\end{equation}
Since each \( L_{i,d} \) has at most \( M \) qubits, the size of \( L_{J_a,d} \) is at most \( |J_a| M \). Consequently, the set \( \bigcup_{m \in L_{J_a,d}} B_{m,0} \) has size at most \( |J_a| M^2 \).

If \( k > |J_a| M^2 \), there exists at least one logical qubit \( i \in L - J_a \) such that \( i \notin \bigcup_{m \in L_{J_a,d}} B_{m,0} \). We can then add this qubit to \( J_a \) to form \( J_{a+1} = J_a \cup \{i\} \), maintaining the property that the light cones of qubits in \( J_{a+1} \) remain disjoint.

This process continues until \( k \le |J_a|M^2 = aM^2 \). Therefore, the final set \( J_a \) satisfies that the light cones of any two distinct qubits in \( J \) are disjoint, and $|J_a| = a \ge \frac{k}{M^2}$.
\end{proof}

\begin{proof}[Proof of Theorem \ref{thm:depolarizing_lower_bound}]
We now combine Lemma \ref{lem:disjoint} with the properties of depolarizing noise to analyze circuit lower bound. Recall that the Choi fidelity is defined as
\begin{equation}
F_{\mathrm{Choi}} = \max_{\mathcal{D}} F(\hat{\phi}_{LR}, (\cD \circ \cN \circ \cU)_L (\hat{\phi}_{LR})).
\end{equation}
Here, and in the following, we omit the identity channel \( I_R \) on the reference system. The encoding channel \( \mathcal{U} \) maps a state \( \rho \) on system \( L \) to another state \( \mathcal{U}(\rho) = U (\rho \otimes \ketbra{0^{n-k}}) U^{\dagger} \) on system \( S \). The noise channel \( \mathcal{N} \) represents i.i.d.~local depolarizing noise on system \( S \), specifically \( \mathcal{N} = \mathcal{N}_d^{S} \), where \( \mathcal{N}_d^{S} = \bigotimes_{i \in S} \mathcal{N}_d^{(i)} \) and \( \mathcal{N}_d^{(i)} \) is the strength-$p$ depolarizing channel acting on the \( i \)-th qubit. The recovery channel \( \mathcal{D} \) maps a state on system \( S \) back to a state on system \( L \).

Consider the subset \( J \) obtained in Lemma \ref{lem:disjoint} under the strength-$p$ local depolarizing noise defined in Eq.~\eqref{eq:depolarizing_channel}. For a single qubit \( i \in J \), the probability that all physical qubits in \( L_{i,d} \) are traced out and replaced by the maximally mixed state is \( p_L = p^{|L_{i,d}|} \ge p^{M} \). Since the light cones of logical qubits in \( J \) are disjoint, the events that these light cones are traced out are independent. Therefore, the probability that none of these light cones are fully traced out is
\begin{equation}\label{eq:pNone}
p_{\mathrm{None}} = \prod_{i \in J} \left(1 - p^{|L_{i,d}|}\right) \le \left(1 - p^{M}\right)^{|J|} \le \exp\left(-|J| \cdot p^{M}\right).
\end{equation}
The second inequality utilizes the bound \( (1 - x) \le e^{-x} \) for \( 0 \le x \le 1 \). This probability provides a lower bound for the Choi error, as later obtained in Eq.~\eqref{eq:ChoiErrorLowerboundJ}. To see this, consider the Choi fidelity under the local depolarizing channel:

\begin{equation}
F_{\mathrm{Choi}} = \max_{\mathcal{D}} F\left(\hat{\phi}_{LR}, \mathcal{D}_{S \rightarrow L} \circ \mathcal{N}_d^{S} \left[ U_d \left(\hat{\phi}_{LR} \otimes \ketbra{0^{n-k}}\right) U_d^{\dagger} \right]\right).
\end{equation}
Expanding the state after passing through the noise channel \( \mathcal{N}_d^{S} \), with probability \( 1 - p_{\mathrm{None}} \), the noise erases all information in at least one of the \( L_{i,d} \) systems for \( i \in J \), replacing it with a maximally mixed state. Consequently, the state after applying the local depolarizing noise is
\begin{equation}
\mathcal{N}_d^{S} \left[ U_d \left(\hat{\phi}_{LR} \otimes \ketbra{0^{n-k}}\right) U_d^{\dagger} \right] = \sum_{l} p_l \psi_l + p_{\mathrm{None}} \psi',
\end{equation}
where \( \sum_l p_l = 1 - p_{\mathrm{None}} \) represents the probability that at least one logical qubit's light cone is traced out. Here, the indices $l$ are used to distinguish states subject to errors on different subsets of physical qubits.
For each $\psi_l$, at least one subsystem \( L_{i,d} \) with $i \in J$ is replaced by the maximally mixed state.

The square of the Choi fidelity can then be expressed as
\begin{equation}\label{eq:square_Choi_fidelity}
\begin{split}
F^2_{\mathrm{Choi}} &= \max_{\mathcal{D}} F^2\left(\hat{\phi}_{LR}, \mathcal{D}_{S \rightarrow L} \circ \mathcal{N}_d^{S} \left[ U_d \left(\hat{\phi}_{LR} \otimes \ketbra{0^{n-k}}\right) U_d^{\dagger} \right]\right) \\
&= \max_{\mathcal{D}} F^2\left(\hat{\phi}_{LR}, \mathcal{D}_{S \rightarrow L} \left[\sum_l p_l \psi_l + p_{\mathrm{None}} \psi'\right]\right) \\
&= \max_{\mathcal{D}} \left\langle \hat{\phi}_{LR} \right| \mathcal{D}_{S \rightarrow L} \left[\sum_l p_l \psi_l + p_{\mathrm{None}} \psi'\right] \left| \hat{\phi}_{LR} \right\rangle \\
&= \max_{\mathcal{D}} \left( \sum_l p_l \left\langle \hat{\phi}_{LR} \right| \mathcal{D}_{S \rightarrow L}(\psi_l) \left| \hat{\phi}_{LR} \right\rangle + p_{\mathrm{None}} \left\langle \hat{\phi}_{LR} \right| \mathcal{D}_{S \rightarrow L}(\psi') \left| \hat{\phi}_{LR} \right\rangle \right) \\
&\le \max_{\mathcal{D}} \left( \sum_l p_l \left\langle \hat{\phi}_{LR} \right| \mathcal{D}_{S \rightarrow L}(\psi_l) \left| \hat{\phi}_{LR} \right\rangle + p_{\mathrm{None}} \right).
\end{split}
\end{equation}
In the third equality, we utilize Eq.~\eqref{eq:fidelity_pure_mixed}. The fourth equality follows from the linearity of quantum channels.

Now, consider the quantity \( \langle \hat{\phi}_{LR} | \mathcal{D}_{S \rightarrow L}(\psi_l) | \hat{\phi}_{LR} \rangle \). Since at least the light cone of one logical qubit in \( L \) is fully traced out, we have

\begin{equation}
\psi_l = \mathrm{Tr}_{T} \left[ U \left(\hat{\phi}_{LR} \otimes \ketbra{0^{n-k}}\right) U^{\dagger} \right] \otimes \frac{\mathbb{I}_T}{2^{|T|}},
\end{equation}
for some subset \( T \subseteq S \) such that \( L_{i,d} \subseteq T \) for some \( i \in L \). Crucially, when the light cone \( L_{i,d} \) is traced out, the unitaries applied within \( L_{i,d} \) are canceled. Therefore, we can replace the initial state \( \hat{\phi}_{LR} \) with \( \hat{\phi}_{L - L_{i,0}, R - R_{i,0}} \otimes \frac{\mathbb{I}_{R_{i,0}}}{2} \otimes \frac{\mathbb{I}_{L_{i,0}}}{2} \). Specifically,

\begin{equation}
\tr_{L_{i,d}} \left[ U \left(\hat{\phi}_{LR} \otimes \ketbra{0^{n-k}}\right) U^{\dagger} \right] = \tr_{L_{i,d}} \left[ U \left(\hat{\phi}_{L - L_{i,0}, R - R_{i,0}} \otimes \frac{\mathbb{I}_{L_{i,0}}}{2} \otimes \ketbra{0^{n-k}} \right) U^{\dagger} \right] \otimes \frac{\mathbb{I}_{R_{i,0}}}{2}.
\end{equation}
Thus, the state $\psi_l$ in Eq.~\eqref{eq:square_Choi_fidelity} can be written as
\begin{equation}
\begin{split}
\psi_l = \hat{\psi}_{l} \otimes \frac{\bI_{T}}{2^{|T|}} \otimes \frac{\bI_{R_0}}{2}.
\end{split}
\end{equation}
where $\hat{\psi}_{l} = \tr_{T \cup R_{i,0}}(\psi_l)$. We then have
\begin{equation}
\bra{\hat{\phi}_{LR}} \cD_{S \rightarrow L}  \left(\hat{\psi}_l \otimes \frac{\bI_{T}}{2^{|T|}} \right)  \otimes \frac{\bI_{R_0}}{2} \ket{\hat{\phi}_{LR}}
= \max_{\cD} \frac{1}{4} \bra{\hat{\phi}_{L-L_{i,0},R-R_{i,0}}} \cD_{S \rightarrow L}  \left(\hat{\psi} \otimes \frac{\bI_{T}}{2^{|T|}} \right) \ket{\hat{\phi}_{L-L_{i,0},R-R_{i,0}}} \le \frac{1}{4}.
\end{equation}
Therefore, the Choi fidelity in Eq.~\eqref{eq:square_Choi_fidelity} satisfies
\begin{equation}
\begin{split}
F^2_{\mathrm{Choi}} &\le \max_{\mathcal{D}} \left( \sum_l p_l \langle \hat{\phi}_{LR} | \mathcal{D}_{S \rightarrow L}(\psi_l) | \hat{\phi}_{LR} \rangle + p_{\mathrm{None}} \right) \\
&\le \sum_l \frac{p_l}{4} + p_{\mathrm{None}} = \frac{1}{4} + \frac{3}{4} p_{\mathrm{None}} \\
&\le \frac{1}{4} + \frac{3}{4} \exp\left(-|J| \cdot p^{M}\right),
\end{split}
\end{equation}
where the last inequality follows from Eq.~\eqref{eq:pNone}. Recall that \( 1 - \epsilon_{\mathrm{Choi}}^2 = F^2_{\mathrm{Choi}} \). For \( \epsilon_{\mathrm{Choi}} < 0.1 \), we have \( \exp\left(-|J| \cdot p^{M}\right) > 0.98 \). In this regime, using the inequality \( e^{-x} \le 1 - \frac{x}{2} \), we obtain

\begin{equation}\label{eq:ChoiErrorLowerboundJ}
\begin{split}
F^2_{\mathrm{Choi}} &\le \frac{1}{4} + \frac{3}{4} \left(1 - \frac{|J|}{2} p^{M}\right) = 1 - \frac{3 |J|}{8} p^{M}, \\
\epsilon^2_{\mathrm{Choi}} &\ge \frac{3 |J|}{8} p^{M}.
\end{split}
\end{equation}
Taking the logarithm of both sides, we obtain
\begin{equation}
M \log\left(\frac{1}{p}\right) - \log(|J|) \ge \log\left(\frac{3}{8 \epsilon^2_{\mathrm{Choi}}}\right).
\end{equation}
According to Lemma \ref{lem:disjoint}, \( |J| \ge \frac{k}{M^2} \). Therefore, we have
\begin{equation}
M \log\left(\frac{1}{p}\right) + 2 \log M \ge \log\left(\frac{3}{8 \epsilon^2_{\mathrm{Choi}}}\right) + \log k.
\end{equation}

For a \( D \)-dimensional quantum circuit, the size of the light cones can be upper bounded by \( M \leq (2d)^D \), since the length of each edge of the light-cone hypercube is upper-bounded by \( 2d \). Substituting this bound into the previous inequality, we obtain

\begin{equation}
(2d)^D \log\left(\frac{1}{p}\right) + 2D \log(2d) \geq \log\left(\frac{3}{8 \epsilon_{\mathrm{Choi}}^2}\right) + \log k.
\end{equation}

In the case of an all-to-all circuit, the size of the light cones is upper bounded by \( M \leq 2^d \). Therefore, the inequality becomes

\begin{equation}
2^d \log\left(\frac{1}{p}\right) + 2d \geq \log\left(\frac{3}{8 \epsilon_{\mathrm{Choi}}^2}\right) + \log k.
\end{equation}

\end{proof}

%%%%%%%%%%%%%%%%%%%%%%%%%%%%%%%%%%%%%%%%
% choose a style
%\bibliographystyle{ieeetr}
%\bibliographystyle{unsrt}
% \bibliographystyle{apsrev4-1}
%%%%%%%%%%%%%%%%%%%%%%%%%%%%%%%%%%%%%%%%

%%%%%%%%%%%%%%%%%%%%%%%%%%%%%%%%%%%%%%%%
% choose a .bib file
\bibliography{bibShortQEC}
%%%%%%%%%%%%%%%%%%%%%%%%%%%%%%%%%%%%%%%%

%\nocite{*}
\end{document}

%% file: figure/DomainWall.tex
\begin{figure}[hbt!]
\centering

\begin{tikzpicture}[scale=0.5]

% ================================
% Figure 1
% ================================
\begin{scope}[shift={(0,0)}]
\foreach \t/\row in {
0/{0,0,0,0,0,0,0,0},
1/{1,0,1,0,0,0,0,0},
2/{1,1,1,1,0,0,0,0},
3/{1,1,1,1,0,0,0,0},
4/{1,1,1,0,0,0,0,0},
5/{1,1,0,0,0,0,0,0},
6/{1,0,0,1,0,0,0,0},
7/{0,0,1,1,1,0,1,0},
8/{0,1,1,1,1,1,1,1},
9/{1,1,1,1,0,1,0,1},
10/{1,0,1,0,0,0,0,0},
11/{0,0,0,0,1,0,1,0},
12/{0,0,0,1,1,1,1,0},
13/{1,0,1,1,1,1,0,0},
14/{1,1,1,0,1,0,0,0},
15/{0,1,0,0,0,0,0,0},
16/{0,0,0,0,0,0,0,0},
} {
    \foreach \x [count=\i] in \row {
        \pgfmathtruncatemacro{\val}{\x}
        \ifnum\val=0
            \definecolor{cellcolor}{rgb}{0.8,0.9,1} % blue for I
            \def\labeltext{$I$}
        \else
            \definecolor{cellcolor}{rgb}{0.7,1.0,0.7} % green for S
            \def\labeltext{$F$}
        \fi
        \fill[cellcolor] (\i-1,-\t) rectangle ++(1,-1);
        \draw[gray] (\i-1,-\t) rectangle ++(1,-1);
        \node at (\i-0.5,-\t-0.5) {\footnotesize \labeltext};
    }
}

% Draw domain wall line (manually chosen)
\draw[red, thick]
(4.5,-3) -- (3.5,-2) -- (2.5,-1) -- (1.5,-2) -- (0.5,-1);

\draw[red, thick]
(4.5,-3) -- (3.5,-4) -- (2.5,-5) -- (1.5,-6) -- (0.5,-7);

\draw[red, thick]
(7.5,-8) -- (6.5,-7) -- (5.5,-8) -- (4.5,-7) -- (3.5,-6) -- (2.5,-7) -- (1.5,-8) -- (0.5,-9);

\draw[red, thick]
(7.5,-10) -- (6.5,-9) -- (5.5,-10) -- (4.5,-9) -- (3.5,-10) -- (2.5,-11) -- (1.5,-10) -- (0.5,-11);

\draw[red, thick]
(7.5,-12) -- (6.5,-11) -- (5.5,-12) -- (4.5,-11) -- (3.5,-12) -- (2.5,-13) -- (1.5,-14) -- (0.5,-13);

\draw[red, thick]
(7.5,-12) -- (6.5,-13) -- (5.5,-14) -- (4.5,-15) -- (3.5,-14) -- (2.5,-15) -- (1.5,-16) -- (0.5,-15);

% Add red dots at domain wall corners
\foreach \x/\y in {
    3.5/2,
    2.5/1,
    1.5/2,
    0.5/1
} {
    \filldraw[red] (\x,-\y) circle (4pt);
}

\foreach \x/\y in {
    3.5/4,
    2.5/5,
    1.5/6,
    0.5/7
} {
    \filldraw[red] (\x,-\y) circle (4pt);
}

\foreach \x/\y in {
    7.5/8,
    6.5/7,
    5.5/8,
    4.5/7,
    3.5/6,
    2.5/7,
    1.5/8,
    0.5/9
} {
    \filldraw[red] (\x,-\y) circle (4pt);
}

\foreach \x/\y in {
    7.5/10,
    6.5/9,
    5.5/10,
    4.5/9,
    3.5/10,
    2.5/11,
    1.5/10,
    0.5/11
} {
    \filldraw[red] (\x,-\y) circle (4pt);
}

\foreach \x/\y in {
    6.5/11,
    5.5/12,
    4.5/11,
    3.5/12,
    2.5/13,
    1.5/14,
    0.5/13
} {
    \filldraw[red] (\x,-\y) circle (4pt);
}

\foreach \x/\y in {
    6.5/13,
    5.5/14,
    4.5/15,
    3.5/14,
    2.5/15,
    1.5/16,
    0.5/15
} {
    \filldraw[red] (\x,-\y) circle (4pt);
}

% Qubit labels
\foreach \i in {0,...,7} {
    \node at (\i+0.6,0.5) {\scriptsize $\vec{\gamma}^{\i}$};
}

% Add time step labels
\foreach \t in {1,...,17} {
    \node[left] at (-0.1,-\t+0.48) {\footnotesize $q_{\t}$};
}

\node at (4,-17.5) {\large $G$};
\node at (9.8,-17.5) {\large $=$};

\node[right] at (9.1,-8) {\huge $=$};

\end{scope}

% ================================
% Figure 2
% ================================
\begin{scope}[shift={(12,0)}]
\foreach \t/\row in {
0/{0,0,0,0,0,0,0,0},
1/{0,0,0,0,0,0,0,0},
2/{0,0,0,0,0,0,0,0},
3/{0,0,0,0,0,0,0,0},
4/{0,0,0,0,0,0,0,0},
5/{0,0,0,0,0,0,0,0},
6/{0,0,0,1,0,0,0,0},
7/{0,0,1,1,1,0,1,0},
8/{0,1,1,1,1,1,1,1},
9/{1,1,1,1,0,1,0,1},
10/{1,0,1,0,0,0,0,0},
11/{0,0,0,0,0,0,0,0},
12/{0,0,0,0,0,0,0,0},
13/{0,0,0,0,0,0,0,0},
14/{0,0,0,0,0,0,0,0},
15/{0,0,0,0,0,0,0,0},
16/{0,0,0,0,0,0,0,0},
} {
    \foreach \x [count=\i] in \row {
        \pgfmathtruncatemacro{\val}{\x}
        \ifnum\val=0
            \definecolor{cellcolor}{rgb}{0.8,0.9,1} % blue for I
            \def\labeltext{$I$}
        \else
            \definecolor{cellcolor}{rgb}{0.7,1.0,0.7} % green for S
            \def\labeltext{$F$}
        \fi
        \fill[cellcolor] (\i-1,-\t) rectangle ++(1,-1);
        \draw[gray] (\i-1,-\t) rectangle ++(1,-1);
        \node at (\i-0.5,-\t-0.5) {\footnotesize \labeltext};
    }
}

% Draw domain wall line (manually chosen)

\draw[red, thick]
(7.5,-8) -- (6.5,-7) -- (5.5,-8) -- (4.5,-7) -- (3.5,-6) -- (2.5,-7) -- (1.5,-8) -- (0.5,-9);

\draw[red, thick]
(7.5,-10) -- (6.5,-9) -- (5.5,-10) -- (4.5,-9) -- (3.5,-10) -- (2.5,-11) -- (1.5,-10) -- (0.5,-11);

% Add red dots at domain wall corners

\foreach \x/\y in {
    7.5/8,
    6.5/7,
    5.5/8,
    4.5/7,
    3.5/6,
    2.5/7,
    1.5/8,
    0.5/9
} {
    \filldraw[red] (\x,-\y) circle (4pt);
}

\foreach \x/\y in {
    7.5/10,
    6.5/9,
    5.5/10,
    4.5/9,
    3.5/10,
    2.5/11,
    1.5/10,
    0.5/11
} {
    \filldraw[red] (\x,-\y) circle (4pt);
}

% Qubit labels
\foreach \i in {0,...,7} {
    \node at (\i+0.6,0.5) {\scriptsize $\vec{\gamma}^{\i}$};
}

% Add time step labels
\foreach \t in {1,...,17} {
    \node[left] at (-0.1,-\t+0.48) {\footnotesize $q_{\t}$};
}

\node at (4,-17.5) {\large $G_U$};
\node at (9.3,-17.5) {\large $\sqcup$};

\node at (9.2,-8) {\huge $\sqcup$};
\end{scope}

% ================================
% Figure 3
% ================================
\begin{scope}[shift={(23,0)}]
\foreach \t/\row in {
0/{0,0,0,0,0,0,0,0},
1/{1,0,1,0,0,0,0,0},
2/{1,1,1,1,0,0,0,0},
3/{1,1,1,1,0,0,0,0},
4/{1,1,1,0,0,0,0,0},
5/{1,1,0,0,0,0,0,0},
6/{1,0,0,0,0,0,0,0},
7/{0,0,0,0,0,0,0,0},
8/{0,0,0,0,0,0,0,0},
9/{0,0,0,0,0,0,0,0},
10/{0,0,0,0,0,0,0,0},
11/{0,0,0,0,1,0,1,0},
12/{0,0,0,1,1,1,1,0},
13/{1,0,1,1,1,1,0,0},
14/{1,1,1,0,1,0,0,0},
15/{0,1,0,0,0,0,0,0},
16/{0,0,0,0,0,0,0,0},
} {
    \foreach \x [count=\i] in \row {
        \pgfmathtruncatemacro{\val}{\x}
        \ifnum\val=0
            \definecolor{cellcolor}{rgb}{0.8,0.9,1} % blue for I
            \def\labeltext{$I$}
        \else
            \definecolor{cellcolor}{rgb}{0.7,1.0,0.7} % green for S
            \def\labeltext{$F$}
        \fi
        \fill[cellcolor] (\i-1,-\t) rectangle ++(1,-1);
        \draw[gray] (\i-1,-\t) rectangle ++(1,-1);
        \node at (\i-0.5,-\t-0.5) {\footnotesize \labeltext};
    }
}

% Draw domain wall line (manually chosen)
\draw[red, thick]
(4.5,-3) -- (3.5,-2) -- (2.5,-1) -- (1.5,-2) -- (0.5,-1);

\draw[red, thick]
(4.5,-3) -- (3.5,-4) -- (2.5,-5) -- (1.5,-6) -- (0.5,-7);

\draw[red, thick]
(7.5,-12) -- (6.5,-11) -- (5.5,-12) -- (4.5,-11) -- (3.5,-12) -- (2.5,-13) -- (1.5,-14) -- (0.5,-13);

\draw[red, thick]
(7.5,-12) -- (6.5,-13) -- (5.5,-14) -- (4.5,-15) -- (3.5,-14) -- (2.5,-15) -- (1.5,-16) -- (0.5,-15);

% Add red dots at domain wall corners
\foreach \x/\y in {
    3.5/2,
    2.5/1,
    1.5/2,
    0.5/1
} {
    \filldraw[red] (\x,-\y) circle (4pt);
}

\foreach \x/\y in {
    3.5/4,
    2.5/5,
    1.5/6,
    0.5/7
} {
    \filldraw[red] (\x,-\y) circle (4pt);
}

\foreach \x/\y in {
    6.5/11,
    5.5/12,
    4.5/11,
    3.5/12,
    2.5/13,
    1.5/14,
    0.5/13
} {
    \filldraw[red] (\x,-\y) circle (4pt);
}

\foreach \x/\y in {
    6.5/13,
    5.5/14,
    4.5/15,
    3.5/14,
    2.5/15,
    1.5/16,
    0.5/15
} {
    \filldraw[red] (\x,-\y) circle (4pt);
}

% Qubit labels
\foreach \i in {0,...,7} {
    \node at (\i+0.6,0.5) {\scriptsize $\vec{\gamma}^{\i}$};
}

% Add time step labels
\foreach \t in {1,...,17} {
    \node[left] at (-0.1,-\t+0.48) {\footnotesize $q_{\t}$};
}

\node at (4,-17.5) {\large $G_0$};
\end{scope}

\end{tikzpicture}
\caption{Diagram for the domain-wall trajectory decomposition. The total qubit number is $17$ and the circuit depth is $7$. The first figure is a trajectory $\Gamma$ with corresponding domain-wall trajectory $G(\Gamma)$. The number of domain walls within $G(\Gamma)$ does not conserve. $G(\Gamma)$ can be decomposed into two disjoint domain-wall trajectories, $G(\Gamma) = G_U \sqcup G_0$. $G_U$ always contains $2$ domain walls, and the domain walls of $G_0$ annihilate at the final depth.}
\label{fig:domainwall}
\end{figure}

%% file: figure/DomainWallUpperBounds.tex
\begin{figure}[hbt!]
\centering

\begin{tikzpicture}[scale=0.5]

% ================================
% Figure 1
% ================================
\begin{scope}[shift={(0,0)}]
\foreach \t/\row in {
0/{0,0,0,0,0,0,0,0},
1/{1,0,1,0,0,0,0,0},
2/{1,1,1,1,0,0,0,0},
3/{1,1,1,1,0,0,0,0},
4/{1,1,1,0,0,0,0,0},
5/{1,1,0,0,0,0,0,0},
6/{1,0,0,1,0,0,0,0},
7/{0,0,1,1,1,0,1,0},
8/{0,1,1,1,1,1,1,1},
9/{1,1,1,1,0,1,0,1},
10/{1,0,1,0,0,0,0,0},
11/{0,0,0,0,1,0,1,0},
12/{0,0,0,1,1,1,1,0},
13/{1,0,1,1,1,1,0,0},
14/{1,1,1,0,1,0,0,0},
15/{0,1,0,0,0,0,0,0},
16/{0,0,0,0,0,0,0,0},
} {
    \foreach \x [count=\i] in \row {
        \pgfmathtruncatemacro{\val}{\x}
        \ifnum\val=0
            \definecolor{cellcolor}{rgb}{0.8,0.9,1} % blue for I
            \def\labeltext{$I$}
        \else
            \definecolor{cellcolor}{rgb}{0.7,1.0,0.7} % green for S
            \def\labeltext{$F$}
        \fi
        \fill[cellcolor] (\i-1,-\t) rectangle ++(1,-1);
        \draw[gray] (\i-1,-\t) rectangle ++(1,-1);
        \node at (\i-0.5,-\t-0.5) {\footnotesize \labeltext};
    }
}

% Draw domain wall line (manually chosen)
\draw[red, thick]
(4.5,-3) -- (3.5,-2) -- (2.5,-1) -- (1.5,-2) -- (0.5,-1);

\draw[red, thick]
(4.5,-3) -- (3.5,-4) -- (2.5,-5) -- (1.5,-6) -- (0.5,-7);

\draw[red, thick]
(7.5,-8) -- (6.5,-7) -- (5.5,-8) -- (4.5,-7) -- (3.5,-6) -- (2.5,-7) -- (1.5,-8) -- (0.5,-9);

\draw[red, thick]
(7.5,-10) -- (6.5,-9) -- (5.5,-10) -- (4.5,-9) -- (3.5,-10) -- (2.5,-11) -- (1.5,-10) -- (0.5,-11);

\draw[red, thick]
(7.5,-12) -- (6.5,-11) -- (5.5,-12) -- (4.5,-11) -- (3.5,-12) -- (2.5,-13) -- (1.5,-14) -- (0.5,-13);

\draw[red, thick]
(7.5,-12) -- (6.5,-13) -- (5.5,-14) -- (4.5,-15) -- (3.5,-14) -- (2.5,-15) -- (1.5,-16) -- (0.5,-15);

% Add red dots at domain wall corners
\foreach \x/\y in {
    3.5/2,
    2.5/1,
    1.5/2,
    0.5/1
} {
    \filldraw[red] (\x,-\y) circle (4pt);
}

\foreach \x/\y in {
    3.5/4,
    2.5/5,
    1.5/6,
    0.5/7
} {
    \filldraw[red] (\x,-\y) circle (4pt);
}

\foreach \x/\y in {
    7.5/8,
    6.5/7,
    5.5/8,
    4.5/7,
    3.5/6,
    2.5/7,
    1.5/8,
    0.5/9
} {
    \filldraw[red] (\x,-\y) circle (4pt);
}

\foreach \x/\y in {
    7.5/10,
    6.5/9,
    5.5/10,
    4.5/9,
    3.5/10,
    2.5/11,
    1.5/10,
    0.5/11
} {
    \filldraw[red] (\x,-\y) circle (4pt);
}

\foreach \x/\y in {
    6.5/11,
    5.5/12,
    4.5/11,
    3.5/12,
    2.5/13,
    1.5/14,
    0.5/13
} {
    \filldraw[red] (\x,-\y) circle (4pt);
}

\foreach \x/\y in {
    6.5/13,
    5.5/14,
    4.5/15,
    3.5/14,
    2.5/15,
    1.5/16,
    0.5/15
} {
    \filldraw[red] (\x,-\y) circle (4pt);
}

% Qubit labels
\foreach \i in {0,...,7} {
    \node at (\i+0.6,0.5) {\scriptsize $\vec{\gamma}^{\i}$};
}

% Add time step labels
\foreach \t in {1,...,9} {
    \node[left] at (-0.1,-\t+0.48) {\footnotesize \color{blue} $q_{\t}$};
}

\foreach \t in {10,11} {
    \node[left] at (-0.1,-\t+0.48) {\footnotesize \color{ao} $q_{\t}$};
}

\foreach \t in {12,...,17} {
    \node[left] at (-0.1,-\t+0.48) {\footnotesize \color{blue} $q_{\t}$};
}

% E^0_odd
\draw [decorate,decoration={brace,amplitude=5pt,mirror},thick,blue]
(-1.2,-0.2) -- (-1.2,-8.8);
\draw [decorate,decoration={brace,amplitude=5pt,mirror},thick,blue]
(-1.2,-11.2) -- (-1.2,-16.8);
\node[blue,left] at (-1.5,-4.5) {$E^0_{odd}\supseteq$};
\node[blue,left] at (-1.5,-14) {$E^0_{odd}\supseteq$};

% E^0_even
\draw [decorate,decoration={brace,amplitude=5pt,mirror},thick,ao]
(-1.2,-9.2) -- (-1.2,-10.8);
\node[ao,left] at (-1.5,-10) {$E^0_{even}\supseteq$};

% E^d_odd
\draw [decorate,decoration={brace,amplitude=5pt},thick,blue]
(8.2,-0.2) -- (8.2,-7.8);
\draw [decorate,decoration={brace,amplitude=5pt},thick,blue]
(8.2,-10.2) -- (8.2,-16.8);
\node[blue,left] at (11,-4) {$\subseteq E^d_{odd}$};
\node[blue,left] at (11,-13.5) {$\subseteq E^d_{odd}$};

% E^d_even
\draw [decorate,decoration={brace,amplitude=5pt},thick,ao]
(8.2,-8.2) -- (8.2,-9.8);
\node[ao,left] at (11,-9) {$\subseteq E^d_{\text{even}}$};

\node at (4,-17.5) {\large $G$};

\end{scope}

% ================================
% Figure 2
% ================================
\begin{scope}[shift={(15,0)}]
\foreach \t/\row in {
0/{0,0,0,0,0,0,0,0},
1/{0,0,0,0,0,0,0,0},
2/{0,0,0,0,0,0,0,0},
3/{0,0,0,0,0,0,0,0},
4/{0,0,0,0,0,0,0,0},
5/{0,0,0,0,0,0,0,0},
6/{0,0,0,1,0,0,0,0},
7/{0,0,1,1,1,0,1,0},
8/{0,1,1,1,1,1,1,1},
9/{1,1,1,1,0,1,0,1},
10/{1,0,1,0,0,0,0,0},
11/{0,0,0,0,0,0,0,0},
12/{0,0,0,0,0,0,0,0},
13/{0,0,0,0,0,0,0,0},
14/{0,0,0,0,0,0,0,0},
15/{0,0,0,0,0,0,0,0},
16/{0,0,0,0,0,0,0,0},
} {
    \foreach \x [count=\i] in \row {
        \pgfmathtruncatemacro{\val}{\x}
        \ifnum\val=0
            \definecolor{cellcolor}{rgb}{0.8,0.9,1} % blue for I
            \def\labeltext{$I$}
        \else
            \definecolor{cellcolor}{rgb}{0.7,1.0,0.7} % green for S
            \def\labeltext{$F$}
        \fi
        \fill[cellcolor] (\i-1,-\t) rectangle ++(1,-1);
        \draw[gray] (\i-1,-\t) rectangle ++(1,-1);
        \node at (\i-0.5,-\t-0.5) {\footnotesize \labeltext};
    }
}

% Draw domain wall line (manually chosen)

\draw[red, thick]
(7.5,-8) -- (6.5,-7) -- (5.5,-8) -- (4.5,-7) -- (3.5,-6) -- (2.5,-7) -- (1.5,-8) -- (0.5,-9);

\draw[red, thick]
(7.5,-10) -- (6.5,-9) -- (5.5,-10) -- (4.5,-9) -- (3.5,-10) -- (2.5,-11) -- (1.5,-10) -- (0.5,-11);

% Add red dots at domain wall corners

\foreach \x/\y in {
    7.5/8,
    6.5/7,
    5.5/8,
    4.5/7,
    3.5/6,
    2.5/7,
    1.5/8,
    0.5/9
} {
    \filldraw[red] (\x,-\y) circle (4pt);
}

\foreach \x/\y in {
    7.5/10,
    6.5/9,
    5.5/10,
    4.5/9,
    3.5/10,
    2.5/11,
    1.5/10,
    0.5/11
} {
    \filldraw[red] (\x,-\y) circle (4pt);
}

% Qubit labels
\foreach \i in {0,...,7} {
    \node at (\i+0.6,0.5) {\scriptsize $\vec{\gamma}^{\i}$};
}

% Add time step labels
\foreach \t in {1,...,17} {
    \node[left] at (-0.1,-\t+0.48) {\footnotesize $q_{\t}$};
}

\node at (4,-17.5) {\large $G_U$};

\end{scope}
\end{tikzpicture}
\caption{Diagram to show the supports $E^i_{odd}$ and $E^i_{even}$ within $G$. We also draw $G_U$ on the right as a reference.}
\label{fig:domainwallupper}
\end{figure}

%% file: figure/DomainWallRegion.tex
\begin{figure}[hbt!]
\centering

\begin{tikzpicture}[scale=0.5]

% ================================
% Figure 1
% ================================
\begin{scope}[shift={(0,0)}]
\foreach \t/\row in {
0/{0,0,0,0,0,0,0,0},
1/{1,0,1,0,0,0,0,0},
2/{1,1,1,1,0,0,0,0},
3/{1,1,1,1,0,0,0,0},
4/{1,1,1,0,0,0,0,0},
5/{1,1,0,0,0,0,0,0},
6/{1,0,0,1,0,0,0,0},
7/{0,0,1,1,1,0,1,0},
8/{0,1,1,1,1,1,1,1},
9/{1,1,1,1,0,1,0,1},
10/{1,0,1,0,0,0,0,0},
11/{0,0,0,0,1,0,1,0},
12/{0,0,0,1,1,1,1,0},
13/{1,0,1,1,1,1,0,0},
14/{1,1,1,0,1,0,0,0},
15/{0,1,0,0,0,0,0,0},
16/{0,0,0,0,0,0,0,0},
} {
    \foreach \x [count=\i] in \row {
        \pgfmathtruncatemacro{\val}{\x}
        \ifnum\val=0
            \definecolor{cellcolor}{rgb}{0.8,0.9,1} % blue for I
            \def\labeltext{$I$}
        \else
            \definecolor{cellcolor}{rgb}{0.7,1.0,0.7} % green for S
            \def\labeltext{$F$}
        \fi
        \fill[cellcolor] (\i-1,-\t) rectangle ++(1,-1);
        \draw[gray] (\i-1,-\t) rectangle ++(1,-1);
        \node at (\i-0.5,-\t-0.5) {\footnotesize \labeltext};
    }
}

% Draw domain wall line (manually chosen)
\draw[red, thick]
(4.5,-3) -- (3.5,-2) -- (2.5,-1) -- (1.5,-2) -- (0.5,-1);

\draw[red, thick]
(4.5,-3) -- (3.5,-4) -- (2.5,-5) -- (1.5,-6) -- (0.5,-7);

\draw[red, thick]
(7.5,-8) -- (6.5,-7) -- (5.5,-8) -- (4.5,-7) -- (3.5,-6) -- (2.5,-7) -- (1.5,-8) -- (0.5,-9);

\draw[red, thick]
(7.5,-10) -- (6.5,-9) -- (5.5,-10) -- (4.5,-9) -- (3.5,-10) -- (2.5,-11) -- (1.5,-10) -- (0.5,-11);

\draw[red, thick]
(7.5,-12) -- (6.5,-11) -- (5.5,-12) -- (4.5,-11) -- (3.5,-12) -- (2.5,-13) -- (1.5,-14) -- (0.5,-13);

\draw[red, thick]
(7.5,-12) -- (6.5,-13) -- (5.5,-14) -- (4.5,-15) -- (3.5,-14) -- (2.5,-15) -- (1.5,-16) -- (0.5,-15);

% Add red dots at domain wall corners
\foreach \x/\y in {
    3.5/2,
    2.5/1,
    1.5/2,
    0.5/1
} {
    \filldraw[red] (\x,-\y) circle (4pt);
}

\foreach \x/\y in {
    3.5/4,
    2.5/5,
    1.5/6,
    0.5/7
} {
    \filldraw[red] (\x,-\y) circle (4pt);
}

\foreach \x/\y in {
    7.5/8,
    6.5/7,
    5.5/8,
    4.5/7,
    3.5/6,
    2.5/7,
    1.5/8,
    0.5/9
} {
    \filldraw[red] (\x,-\y) circle (4pt);
}

\foreach \x/\y in {
    7.5/10,
    6.5/9,
    5.5/10,
    4.5/9,
    3.5/10,
    2.5/11,
    1.5/10,
    0.5/11
} {
    \filldraw[red] (\x,-\y) circle (4pt);
}

\foreach \x/\y in {
    6.5/11,
    5.5/12,
    4.5/11,
    3.5/12,
    2.5/13,
    1.5/14,
    0.5/13
} {
    \filldraw[red] (\x,-\y) circle (4pt);
}

\foreach \x/\y in {
    6.5/13,
    5.5/14,
    4.5/15,
    3.5/14,
    2.5/15,
    1.5/16,
    0.5/15
} {
    \filldraw[red] (\x,-\y) circle (4pt);
}

% Qubit labels
\foreach \i in {0,...,7} {
    \node at (\i+0.6,0.5) {\scriptsize $\vec{\gamma}^{\i}$};
}

% Add time step labels
\foreach \t in {1,...,17} {
    \node[left] at (-0.1,-\t+0.48) {\footnotesize $q_{\t}$};
}

\node[right] at (8.8,-8) {\huge $\Rightarrow$};

\end{scope}

% ================================
% Figure 2
% ================================
\begin{scope}[shift={(12,0)}]
\foreach \t/\row in {
0/{0,0,0,0,0,0,0,0},
1/{1,0,1,0,0,0,0,0},
2/{1,1,1,1,0,0,0,0},
3/{1,1,1,1,0,0,0,0},
4/{1,1,1,0,0,0,0,0},
5/{1,1,0,0,0,0,0,0},
6/{1,0,0,2,0,0,0,0},
7/{0,0,2,2,2,0,2,0},
8/{0,2,2,2,2,2,2,2},
9/{2,2,2,2,0,2,0,2},
10/{2,0,2,0,0,0,0,0},
11/{0,0,0,0,1,0,1,0},
12/{0,0,0,1,1,1,1,0},
13/{1,0,1,1,1,1,0,0},
14/{1,1,1,0,1,0,0,0},
15/{0,1,0,0,0,0,0,0},
16/{0,0,0,0,0,0,0,0},
} {
    \foreach \x [count=\i] in \row {
        \pgfmathtruncatemacro{\val}{\x}
        \ifnum\val=0
            \definecolor{cellcolor}{rgb}{0.8,0.9,1} % blue for I
            \def\labeltext{$I$}
        \else
            \ifnum\val=1
                \definecolor{cellcolor}{rgb}{0.7,1.0,0.7} % green for S
                \def\labeltext{$F$}
            \else
                \definecolor{cellcolor}{rgb}{1,1,1}
                \def\labeltext{}
            \fi
        \fi
        \fill[cellcolor] (\i-1,-\t) rectangle ++(1,-1);
        \draw[gray] (\i-1,-\t) rectangle ++(1,-1);
        \node at (\i-0.5,-\t-0.5) {\footnotesize \labeltext};
    }
}

% Draw domain wall line (manually chosen)
\draw[red, thick]
(4.5,-3) -- (3.5,-2) -- (2.5,-1) -- (1.5,-2) -- (0.5,-1);

\draw[red, thick]
(4.5,-3) -- (3.5,-4) -- (2.5,-5) -- (1.5,-6) -- (0.5,-7);

\draw[red, thick]
(7.5,-8) -- (6.5,-7) -- (5.5,-8) -- (4.5,-7) -- (3.5,-6) -- (2.5,-7) -- (1.5,-8) -- (0.5,-9);

\draw[red, thick]
(7.5,-10) -- (6.5,-9) -- (5.5,-10) -- (4.5,-9) -- (3.5,-10) -- (2.5,-11) -- (1.5,-10) -- (0.5,-11);

\draw[red, thick]
(7.5,-12) -- (6.5,-11) -- (5.5,-12) -- (4.5,-11) -- (3.5,-12) -- (2.5,-13) -- (1.5,-14) -- (0.5,-13);

\draw[red, thick]
(7.5,-12) -- (6.5,-13) -- (5.5,-14) -- (4.5,-15) -- (3.5,-14) -- (2.5,-15) -- (1.5,-16) -- (0.5,-15);

% Add red dots at domain wall corners
\foreach \x/\y in {
    3.5/2,
    2.5/1,
    1.5/2,
    0.5/1
} {
    \filldraw[red] (\x,-\y) circle (4pt);
}

\foreach \x/\y in {
    3.5/4,
    2.5/5,
    1.5/6,
    0.5/7
} {
    \filldraw[red] (\x,-\y) circle (4pt);
}

% \foreach \x/\y in {
%     7.5/8,
%     6.5/7,
%     5.5/8,
%     4.5/7,
%     3.5/6,
%     2.5/7,
%     1.5/8,
%     0.5/9
% } {
%     \filldraw[red] (\x,-\y) circle (4pt);
% }

% \foreach \x/\y in {
%     7.5/10,
%     6.5/9,
%     5.5/10,
%     4.5/9,
%     3.5/10,
%     2.5/11,
%     1.5/10,
%     0.5/11
% } {
%     \filldraw[red] (\x,-\y) circle (4pt);
% }

\foreach \x/\y in {
    6.5/11,
    5.5/12,
    4.5/11,
    3.5/12,
    2.5/13,
    1.5/14,
    0.5/13
} {
    \filldraw[red] (\x,-\y) circle (4pt);
}

\foreach \x/\y in {
    6.5/13,
    5.5/14,
    4.5/15,
    3.5/14,
    2.5/15,
    1.5/16,
    0.5/15
} {
    \filldraw[red] (\x,-\y) circle (4pt);
}

% Qubit labels
\foreach \i in {0,...,7} {
    \node at (\i+0.6,0.5) {\scriptsize $\vec{\gamma}^{\i}$};
}

% Add time step labels
\foreach \t in {1,...,17} {
    \node[left] at (-0.1,-\t+0.48) {\footnotesize $q_{\t}$};
}

\node at (4,-17.5) {\large Region $1$};

\node at (9.1,-8) {\huge $+$};

\end{scope}

% ================================
% Figure 3
% ================================
\begin{scope}[shift={(23,0)}]
\foreach \t/\row in {
0/{2,2,2,2,2,2,2,2},
1/{2,2,2,2,2,2,2,2},
2/{2,2,2,2,2,2,2,2},
3/{2,2,2,2,2,2,2,2},
4/{2,2,2,2,2,2,2,2},
5/{2,2,2,2,2,2,2,2},
6/{2,2,2,1,2,2,2,2},
7/{2,2,1,1,1,2,1,2},
8/{2,1,1,1,1,1,1,1},
9/{1,1,1,1,2,1,2,1},
10/{1,2,1,2,2,2,2,2},
11/{2,2,2,2,2,2,2,2},
12/{2,2,2,2,2,2,2,2},
13/{2,2,2,2,2,2,2,2},
14/{2,2,2,2,2,2,2,2},
15/{2,2,2,2,2,2,2,2},
16/{2,2,2,2,2,2,2,2},
} {
    \foreach \x [count=\i] in \row {
        \pgfmathtruncatemacro{\val}{\x}
        \ifnum\val=0
            \definecolor{cellcolor}{rgb}{0.8,0.9,1} % blue for I
            \def\labeltext{$I$}
        \else
            \ifnum\val=1
                \definecolor{cellcolor}{rgb}{0.7,1.0,0.7} % green for S
                \def\labeltext{$F$}
            \else
                \definecolor{cellcolor}{rgb}{1,1,1}
                \def\labeltext{}
            \fi
        \fi
        \fill[cellcolor] (\i-1,-\t) rectangle ++(1,-1);
        \draw[gray] (\i-1,-\t) rectangle ++(1,-1);
        \node at (\i-0.5,-\t-0.5) {\footnotesize \labeltext};
    }
}

\draw[red, thick]
(7.5,-8) -- (6.5,-7) -- (5.5,-8) -- (4.5,-7) -- (3.5,-6) -- (2.5,-7) -- (1.5,-8) -- (0.5,-9);

\draw[red, thick]
(7.5,-10) -- (6.5,-9) -- (5.5,-10) -- (4.5,-9) -- (3.5,-10) -- (2.5,-11) -- (1.5,-10) -- (0.5,-11);

% Qubit labels
\foreach \i in {0,...,7} {
    \node at (\i+0.6,0.5) {\scriptsize $\vec{\gamma}^{\i}$};
}

% Add time step labels
\foreach \t in {1,...,17} {
    \node[left] at (-0.1,-\t+0.48) {\footnotesize $q_{\t}$};
}

\node at (4,-17.5) {\large Region $2$};
\end{scope}

\end{tikzpicture}
\caption{Diagram of separating a trajectory into 2 regions.}
\label{fig:domainwallregion}
\end{figure}

%% file: figure/TypeIRegion.tex
\begin{figure}[hbt!]
\centering

\begin{tikzpicture}[scale=0.5]

\begin{scope}[shift={(0,0)}]
\foreach \t/\row in {
0/{0,0,0,0,0,0,0,0},
1/{1,0,1,0,0,0,0,0},
2/{1,1,1,1,0,0,0,0},
3/{1,1,1,1,0,0,0,0},
4/{1,1,1,0,0,0,0,0},
5/{1,1,0,0,0,0,0,0},
6/{1,0,0,2,0,0,0,0},
7/{0,0,2,2,2,0,2,0},
8/{0,2,2,2,2,2,2,2},
9/{2,2,2,2,0,2,0,2},
10/{2,0,2,0,0,0,0,0},
11/{0,0,0,0,1,0,1,0},
12/{0,0,0,1,1,1,1,0},
13/{1,0,1,1,1,1,0,0},
14/{1,1,1,0,1,0,0,0},
15/{0,1,0,0,0,0,0,0},
16/{0,0,0,0,0,0,0,0},
} {
    \foreach \x [count=\i] in \row {
        \pgfmathtruncatemacro{\val}{\x}
        \ifnum\val=0
            \definecolor{cellcolor}{rgb}{0.8,0.9,1} % blue for I
            \def\labeltext{$I$}
        \else
            \ifnum\val=1
                \definecolor{cellcolor}{rgb}{0.7,1.0,0.7} % green for S
                \def\labeltext{$F$}
            \else
                \definecolor{cellcolor}{rgb}{1,1,1}
                \def\labeltext{}
            \fi
        \fi
        \fill[cellcolor] (\i-1,-\t) rectangle ++(1,-1);
        \draw[gray] (\i-1,-\t) rectangle ++(1,-1);
        \node at (\i-0.5,-\t-0.5) {\footnotesize \labeltext};
    }
}

% Draw domain wall line (manually chosen)
\draw[red, thick]
(4.5,-3) -- (3.5,-2) -- (2.5,-1) -- (1.5,-2) -- (0.5,-1);

\draw[red, thick]
(4.5,-3) -- (3.5,-4) -- (2.5,-5) -- (1.5,-6) -- (0.5,-7);

\draw[red, thick]
(7.5,-8) -- (6.5,-7) -- (5.5,-8) -- (4.5,-7) -- (3.5,-6) -- (2.5,-7) -- (1.5,-8) -- (0.5,-9);

\draw[red, thick]
(7.5,-10) -- (6.5,-9) -- (5.5,-10) -- (4.5,-9) -- (3.5,-10) -- (2.5,-11) -- (1.5,-10) -- (0.5,-11);

\draw[red, thick]
(7.5,-12) -- (6.5,-11) -- (5.5,-12) -- (4.5,-11) -- (3.5,-12) -- (2.5,-13) -- (1.5,-14) -- (0.5,-13);

\draw[red, thick]
(7.5,-12) -- (6.5,-13) -- (5.5,-14) -- (4.5,-15) -- (3.5,-14) -- (2.5,-15) -- (1.5,-16) -- (0.5,-15);

\draw[red, dashed]
(8.5,-6) -- (-0.5,-6);

\draw[red, dashed]
(8.5,-8.9) -- (-0.5,-8.9);

\draw[red, dashed]
(8.5,-9.1) -- (-0.5,-9.1);

\draw[red, dashed]
(8.5,-11) -- (-0.5,-11);

% Add red dots at domain wall corners
\foreach \x/\y in {
    3.5/2,
    2.5/1,
    1.5/2,
    0.5/1
} {
    \filldraw[red] (\x,-\y) circle (4pt);
}

\foreach \x/\y in {
    3.5/4,
    2.5/5,
    1.5/6,
    0.5/7
} {
    \filldraw[red] (\x,-\y) circle (4pt);
}

% \foreach \x/\y in {
%     7.5/8,
%     6.5/7,
%     5.5/8,
%     4.5/7,
%     3.5/6,
%     2.5/7,
%     1.5/8,
%     0.5/9
% } {
%     \filldraw[red] (\x,-\y) circle (4pt);
% }

% \foreach \x/\y in {
%     7.5/10,
%     6.5/9,
%     5.5/10,
%     4.5/9,
%     3.5/10,
%     2.5/11,
%     1.5/10,
%     0.5/11
% } {
%     \filldraw[red] (\x,-\y) circle (4pt);
% }

\foreach \x/\y in {
    6.5/11,
    5.5/12,
    4.5/11,
    3.5/12,
    2.5/13,
    1.5/14,
    0.5/13
} {
    \filldraw[red] (\x,-\y) circle (4pt);
}

\foreach \x/\y in {
    6.5/13,
    5.5/14,
    4.5/15,
    3.5/14,
    2.5/15,
    1.5/16,
    0.5/15
} {
    \filldraw[red] (\x,-\y) circle (4pt);
}

% Qubit labels
\foreach \i in {0,...,7} {
    \node at (\i+0.6,0.5) {\scriptsize $\vec{\gamma}^{\i}$};
}

% Add time step labels
\foreach \t in {1,...,17} {
    \node[left] at (-0.1,-\t+0.48) {\footnotesize $q_{\t}$};
}

\node at (4,-17.5) {\large Region $1$};

\node at (9.6,-8) {\huge $\Rightarrow$};

\end{scope}

\begin{scope}[shift={(12,0)}]
\foreach \t/\row in {
0/{0,0,0,0,0,0,0,0},
1/{1,0,1,0,0,0,0,0},
2/{1,1,1,1,0,0,0,0},
3/{1,1,1,1,0,0,0,0},
4/{1,1,1,0,0,0,0,0},
5/{1,1,0,0,0,0,0,0},
6/{1,0,0,0,0,0,0,0},
7/{0,0,0,0,0,0,0,0},
8/{0,0,0,0,0,0,0,0},
9/{0,0,0,0,0,0,0,0},
10/{0,0,0,0,0,0,0,0},
11/{0,0,0,0,1,0,1,0},
12/{0,0,0,1,1,1,1,0},
13/{1,0,1,1,1,1,0,0},
14/{1,1,1,0,1,0,0,0},
15/{0,1,0,0,0,0,0,0},
16/{0,0,0,0,0,0,0,0},
} {
    \foreach \x [count=\i] in \row {
        \pgfmathtruncatemacro{\val}{\x}
        \ifnum\val=0
            \definecolor{cellcolor}{rgb}{0.8,0.9,1} % blue for I
            \def\labeltext{$I$}
        \else
            \ifnum\val=1
                \definecolor{cellcolor}{rgb}{0.7,1.0,0.7} % green for S
                \def\labeltext{$F$}
            \else
                \definecolor{cellcolor}{rgb}{1,1,1}
                \def\labeltext{}
            \fi
        \fi
        \fill[cellcolor] (\i-1,-\t) rectangle ++(1,-1);
        \draw[gray] (\i-1,-\t) rectangle ++(1,-1);
        \node at (\i-0.5,-\t-0.5) {\footnotesize \labeltext};
    }
}

% Draw domain wall line (manually chosen)
\draw[red, thick]
(4.5,-3) -- (3.5,-2) -- (2.5,-1) -- (1.5,-2) -- (0.5,-1);

\draw[red, thick]
(4.5,-3) -- (3.5,-4) -- (2.5,-5) -- (1.5,-6) -- (0.5,-7);

\draw[red, thick]
(7.5,-8) -- (6.5,-7) -- (5.5,-8) -- (4.5,-7) -- (3.5,-6) -- (2.5,-7) -- (1.5,-8) -- (0.5,-9);

\draw[red, thick]
(7.5,-10) -- (6.5,-9) -- (5.5,-10) -- (4.5,-9) -- (3.5,-10) -- (2.5,-11) -- (1.5,-10) -- (0.5,-11);

\draw[red, thick]
(7.5,-12) -- (6.5,-11) -- (5.5,-12) -- (4.5,-11) -- (3.5,-12) -- (2.5,-13) -- (1.5,-14) -- (0.5,-13);

\draw[red, thick]
(7.5,-12) -- (6.5,-13) -- (5.5,-14) -- (4.5,-15) -- (3.5,-14) -- (2.5,-15) -- (1.5,-16) -- (0.5,-15);

\draw[red, dashed]
(8.5,-6) -- (-0.5,-6);

\draw[red, dashed]
(8.5,-8.9) -- (-0.5,-8.9);

\draw[red, dashed]
(8.5,-9.1) -- (-0.5,-9.1);

\draw[red, dashed]
(8.5,-11) -- (-0.5,-11);

% Add red dots at domain wall corners
\foreach \x/\y in {
    3.5/2,
    2.5/1,
    1.5/2,
    0.5/1
} {
    \filldraw[red] (\x,-\y) circle (4pt);
}

\foreach \x/\y in {
    3.5/4,
    2.5/5,
    1.5/6,
    0.5/7
} {
    \filldraw[red] (\x,-\y) circle (4pt);
}

% \foreach \x/\y in {
%     7.5/8,
%     6.5/7,
%     5.5/8,
%     4.5/7,
%     3.5/6,
%     2.5/7,
%     1.5/8,
%     0.5/9
% } {
%     \filldraw[red] (\x,-\y) circle (4pt);
% }

% \foreach \x/\y in {
%     7.5/10,
%     6.5/9,
%     5.5/10,
%     4.5/9,
%     3.5/10,
%     2.5/11,
%     1.5/10,
%     0.5/11
% } {
%     \filldraw[red] (\x,-\y) circle (4pt);
% }

\foreach \x/\y in {
    6.5/11,
    5.5/12,
    4.5/11,
    3.5/12,
    2.5/13,
    1.5/14,
    0.5/13
} {
    \filldraw[red] (\x,-\y) circle (4pt);
}

\foreach \x/\y in {
    6.5/13,
    5.5/14,
    4.5/15,
    3.5/14,
    2.5/15,
    1.5/16,
    0.5/15
} {
    \filldraw[red] (\x,-\y) circle (4pt);
}

% Qubit labels
\foreach \i in {0,...,7} {
    \node at (\i+0.6,0.5) {\scriptsize $\vec{\gamma}^{\i}$};
}

% Add time step labels
\foreach \t in {1,...,17} {
    \node[left] at (-0.1,-\t+0.48) {\footnotesize $q_{\t}$};
}

\node at (4,-17.5) {\large Region $1'$};

\end{scope}

\end{tikzpicture}
\caption{Extending Region $1$ to Region $1'$ by filling the extra components with identity operations. The configurations in the new region all have equal lengths.}
\label{fig:typeIregion}
\end{figure}

%% file: figure/TypeFRegion.tex
\begin{figure}[hbt!]
\centering

\begin{tikzpicture}[scale=0.5]

\begin{scope}[shift={(0,0)}]
\foreach \t/\row in {
0/{2,2,2,2,2,2,2,2},
1/{2,2,2,2,2,2,2,2},
2/{2,2,2,2,2,2,2,2},
3/{2,2,2,2,2,2,2,2},
4/{2,2,2,2,2,2,2,2},
5/{2,2,2,2,2,2,2,2},
6/{2,2,2,1,2,2,2,2},
7/{2,2,1,1,1,2,1,2},
8/{2,1,1,1,1,1,1,1},
9/{1,1,1,1,2,1,2,1},
10/{1,2,1,2,2,2,2,2},
11/{2,2,2,2,2,2,2,2},
12/{2,2,2,2,2,2,2,2},
13/{2,2,2,2,2,2,2,2},
14/{2,2,2,2,2,2,2,2},
15/{2,2,2,2,2,2,2,2},
16/{2,2,2,2,2,2,2,2},
} {
    \foreach \x [count=\i] in \row {
        \pgfmathtruncatemacro{\val}{\x}
        \ifnum\val=0
            \definecolor{cellcolor}{rgb}{0.8,0.9,1} % blue for I
            \def\labeltext{$I$}
        \else
            \ifnum\val=1
                \definecolor{cellcolor}{rgb}{0.7,1.0,0.7} % green for S
                \def\labeltext{$F$}
            \else
                \definecolor{cellcolor}{rgb}{1,1,1}
                \def\labeltext{}
            \fi
        \fi
        \fill[cellcolor] (\i-1,-\t) rectangle ++(1,-1);
        \draw[gray] (\i-1,-\t) rectangle ++(1,-1);
        \node at (\i-0.5,-\t-0.5) {\footnotesize \labeltext};
    }
}

\draw[red, thick]
(7.5,-8) -- (6.5,-7) -- (5.5,-8) -- (4.5,-7) -- (3.5,-6) -- (2.5,-7) -- (1.5,-8) -- (0.5,-9);

\draw[red, thick]
(7.5,-10) -- (6.5,-9) -- (5.5,-10) -- (4.5,-9) -- (3.5,-10) -- (2.5,-11) -- (1.5,-10) -- (0.5,-11);

\draw[red, dashed]
(8.5,-6) -- (-0.5,-6);

\draw[red, dashed]
(8.5,-11) -- (-0.5,-11);

% Qubit labels
\foreach \i in {0,...,7} {
    \node at (\i+0.6,0.5) {\scriptsize $\vec{\gamma}^{\i}$};
}

% Add time step labels
\foreach \t in {1,...,17} {
    \node[left] at (-0.1,-\t+0.48) {\footnotesize $q_{\t}$};
}

\node at (4,-17.5) {\large Region $2$};

\node at (9.6,-8) {\huge $\Rightarrow$};

\end{scope}

\begin{scope}[shift={(12,0)}]
\foreach \t/\row in {
0/{2,2,2,2,2,2,2,2},
1/{2,2,2,2,2,2,2,2},
2/{2,2,2,2,2,2,2,2},
3/{2,2,2,2,2,2,2,2},
4/{2,2,2,2,2,2,2,2},
5/{2,2,2,2,2,2,2,2},
6/{1,1,1,1,1,1,1,1},
7/{1,1,1,1,1,1,1,1},
8/{1,1,1,1,1,1,1,1},
9/{1,1,1,1,1,1,1,1},
10/{1,1,1,1,1,1,1,1},
11/{2,2,2,2,2,2,2,2},
12/{2,2,2,2,2,2,2,2},
13/{2,2,2,2,2,2,2,2},
14/{2,2,2,2,2,2,2,2},
15/{2,2,2,2,2,2,2,2},
16/{2,2,2,2,2,2,2,2},
} {
    \foreach \x [count=\i] in \row {
        \pgfmathtruncatemacro{\val}{\x}
        \ifnum\val=0
            \definecolor{cellcolor}{rgb}{0.8,0.9,1} % blue for I
            \def\labeltext{$I$}
        \else
            \ifnum\val=1
                \definecolor{cellcolor}{rgb}{0.7,1.0,0.7} % green for S
                \def\labeltext{$F$}
            \else
                \definecolor{cellcolor}{rgb}{1,1,1}
                \def\labeltext{}
            \fi
        \fi
        \fill[cellcolor] (\i-1,-\t) rectangle ++(1,-1);
        \draw[gray] (\i-1,-\t) rectangle ++(1,-1);
        \node at (\i-0.5,-\t-0.5) {\footnotesize \labeltext};
    }
}

\draw[red, thick]
(7.5,-8) -- (6.5,-7) -- (5.5,-8) -- (4.5,-7) -- (3.5,-6) -- (2.5,-7) -- (1.5,-8) -- (0.5,-9);

\draw[red, thick]
(7.5,-10) -- (6.5,-9) -- (5.5,-10) -- (4.5,-9) -- (3.5,-10) -- (2.5,-11) -- (1.5,-10) -- (0.5,-11);

\draw[red, dashed]
(8.5,-6) -- (-0.5,-6);

\draw[red, dashed]
(8.5,-11) -- (-0.5,-11);

% Qubit labels
\foreach \i in {0,...,7} {
    \node at (\i+0.6,0.5) {\scriptsize $\vec{\gamma}^{\i}$};
}

% Add time step labels
\foreach \t in {1,...,17} {
    \node[left] at (-0.1,-\t+0.48) {\footnotesize $q_{\t}$};
}

\node at (4,-17.5) {\large Region $2'$};
\end{scope}

\end{tikzpicture}
\caption{Extending Region $2$ to Region $2'$ by filling the extra components with SWAP operations. The configurations in the new region all have equal lengths.}
\label{fig:typeFregion}
\end{figure}